\documentclass[11pt]{article}
\pdfoutput=1

\usepackage[letterpaper, portrait, margin=1in]{geometry}
\usepackage[colorlinks=true,linkcolor=blue,citecolor=ForestGreen]{hyperref}
\usepackage{algorithm}
\usepackage[noend]{algpseudocode}
\usepackage{url}
\usepackage{amsmath,amssymb,amsthm}
\usepackage{thmtools,thm-restate}
\usepackage[noabbrev,capitalise,nameinlink]{cleveref}
\usepackage{mathtools}
\usepackage{xspace}
\usepackage{verbatim}
\usepackage{mathrsfs}
\usepackage[usenames,dvipsnames,svgnames,table]{xcolor}
\usepackage{pgf}
\usepackage[dvipsnames]{xcolor}
\usepackage{todo}
\usepackage{tabularx}
\usepackage{enumitem}
\usepackage{derivative}
\usepackage{bm}
\usepackage{multirow}
\usepackage{diagbox}
\usepackage{nicematrix}

\usepackage{dsfont}

\newcommand*\ie{i.\kern.1em e.\ }
\newcommand*\eg{e.\kern.1em g.\ }
\newcommand*\cf{c.\kern.1em f.\ }
\newcommand*\almev{a.\kern.1em e.\ }

\theoremstyle{plain}
\newtheorem{theorem}{Theorem}[section]
\newtheorem{lemma}[theorem]{Lemma}
\newtheorem{fact}[theorem]{Fact}
\newtheorem{proposition}[theorem]{Proposition}
\newtheorem{claim}[theorem]{Claim}
\newtheorem{corollary}[theorem]{Corollary}

\newtheorem{observation}[theorem]{Observation}

\crefname{claim}{Claim}{Claims}
\crefname{fact}{Fact}{Facts}

\theoremstyle{definition}
\newtheorem{condition}[theorem]{Condition}
\newtheorem{definition}[theorem]{Definition}
\newtheorem{remark}[theorem]{Remark}

\theoremstyle{plain}

\newcommand{\ignore}[1]{}

\DeclareMathOperator{\supp}{supp}   

\newcommand{\dist}{\mathsf{dist}}

\newcommand{\Var}[1]{\mathrm{Var} \left[ #1 \right]}

\newcommand{\Inf}[1]{\mathrm{Inf} \left[ #1 \right]}

\newcommand{\Ex}[1]{\bE \left[ #1 \right]}
\newcommand{\Exu}[2]{\underset{#1} \bE \left[ #2 \right] }

\renewcommand{\Pr}[1]{\bP \left[ #1 \right]} 
\newcommand{\Pru}[2]{\underset{ #1 }\bP \left[ #2 \right]}
\newcommand{\Pruc}[3]{\underset{ #1 }\bP \left[ #2 \;\; \mid \;\; #3 \right]}

\newcommand{\define}{\vcentcolon=}

\newcommand{\floor}[1]{\ensuremath{\lfloor #1 \rfloor}}
\newcommand{\ceil}[1]{\ensuremath{\lceil #1 \rceil}}

\DeclarePairedDelimiter{\abs}{\lvert}{\rvert}

\newcommand{\ind}[1]{\mathds{1} \left[ #1 \right] }

\newcommand\ltquestion{\stackrel{\mathclap{\normalfont\mbox{\tiny ?}}}{<}}
\newcommand\lequestion{\stackrel{\mathclap{\normalfont\mbox{\tiny ?}}}{\le}}

\newcommand\gequestion{\stackrel{\mathclap{\normalfont\mbox{\tiny ?}}}{\ge}}
\newcommand\eqquestion{\stackrel{\mathclap{\normalfont\mbox{\tiny ?}}}{=}}

\newcommand{\zo}{\{0,1\}}

\newcommand{\cD}{\ensuremath{\mathcal{D}}}
\newcommand{\cE}{\ensuremath{\mathcal{E}}}

\newcommand{\cL}{\ensuremath{\mathcal{L}}}
\newcommand{\cM}{\ensuremath{\mathcal{M}}}

\newcommand{\cR}{\ensuremath{\mathcal{R}}}

\newcommand{\cU}{\ensuremath{\mathcal{U}}}

\newcommand{\cX}{\ensuremath{\mathcal{X}}}

\newcommand{\bE}{\ensuremath{\mathbb{E}}}

\newcommand{\bN}{\ensuremath{\mathbb{N}}}
\newcommand{\bP}{\ensuremath{\mathbb{P}}}
\newcommand{\bR}{\ensuremath{\mathbb{R}}}

\newcommand{\bZ}{\ensuremath{\mathbb{Z}}}

\usepackage{etoolbox}
\usepackage{enumitem}

\usepackage[backend=bibtex, style=alphabetic, maxnames=99, maxalphanames=99]{biblatex}
\setlist[description]{leftmargin=\parindent,labelindent=\parindent}

\renewcommand{\epsilon}{\varepsilon}

\newcommand{\grad}{\nabla}
\newcommand{\laplacian}{\Delta}

\newcommand{\weakto}{\rightharpoonup}

\newcommand{\const}{\mathsf{const}}
\newcommand{\mono}{\mathsf{mono}}

\DeclareMathOperator{\Lip}{\mathsf{Lip}}
\DeclareMathOperator{\Mon}{\mathsf{Mon}}

\DeclareMathOperator{\slope}{\mathsf{slope}}

\DeclareMathOperator*{\esssup}{ess\,sup}

\newcommand{\loc}{\mathrm{loc}}

\newcommand{\inp}[2]{\left\langle #1, #2 \right\rangle}
\newcommand{\inpspace}[3]{\left\langle #1, #2 \right\rangle_{#3}}

\newcommand{\Ber}{\mathsf{Ber}}

\newcommand{\upf}[1]{#1 {\uparrow}}
\newcommand{\downf}[1]{#1 {\downarrow}}
\newcommand{\onef}[1]{#1^{(1)}}
\newcommand{\twof}[1]{#1^{(2)}}

\newcommand{\transmap}[2]{{#1}_{\#}{#2}}

\newcommand{\cond}[2]{{#1}_{|#2}}
\newcommand{\marginal}[2]{{#1}_{#2}}

\newcommand{\mon}{\mathrm{mon}}
\newcommand{\idmap}{\mathsf{id}}

\newcommand{\stcolon}{\;:\;}

\newcommand{\closedInt}{[0,1]}

\newcommand{\dirH}{\vec H}

\algblockx{Repeat}{EndRepeat}[1]{\textbf{repeat} #1 \textbf{times}}[0]{\textbf{end repeat}}

\newtoggle{anonymous}
\togglefalse{anonymous}

\begin{document}

\title{Directed Isoperimetry and Monotonicity Testing: A Dynamical Approach}

\iftoggle{anonymous}{%
    \author{Anonymous Author}
    }{%
\author{%
  Renato Ferreira Pinto Jr.\thanks{Partly funded by an NSERC Canada Graduate Scholarship Doctoral
  Award.}\\
  University of Waterloo\\
  \texttt{renato.ferreira@uwaterloo.ca}}}

\date{}

\maketitle

\begin{abstract}
    This paper explores the connection between classical isoperimetric inequalities, their directed
    analogues, and monotonicity testing. We study the setting of real-valued functions $f :
    \closedInt^d \to \bR$ on the solid unit cube, where the goal is to test with respect to the
    $L^p$ distance. Our goals are twofold: to further understand the relationship between classical
    and directed isoperimetry, and to give a monotonicity tester with sublinear query complexity in
    this setting.

    Our main results are 1)~an $L^2$ monotonicity tester for $M$-Lipschitz functions with query
    complexity $\widetilde O(\sqrt{d} M^2 / \epsilon^2)$ and, behind this result, 2)~the directed
    Poincaré inequality $\dist^\mono_2(f)^2 \le C\, \Ex{|\grad^- f|^2}$, where the ``directed
    gradient'' operator $\grad^-$ measures the local violations of monotonicity of $f$.

    To prove the second result, we introduce a partial differential equation (PDE), the
    \emph{directed heat equation}, which takes a one-dimensional function $f$ into a monotone
    function $f^*$ over time and enjoys many desirable analytic properties. We obtain the directed
    Poincaré inequality by combining convergence aspects of this PDE with the theory of optimal
    transport. Crucially for our conceptual motivation, this proof is in complete analogy with the
    mathematical physics perspective on the classical Poincaré inequality, namely as characterizing
    the convergence of the standard heat equation toward equilibrium.
\end{abstract}

\thispagestyle{empty}
\setcounter{page}{0}
\newpage
{
\setcounter{tocdepth}{2} 
\tableofcontents
}
\thispagestyle{empty}
\setcounter{page}{0}
\newpage
\setcounter{page}{1}



\section{Introduction}

One of the central problems in the field of property testing is \emph{monotonicity testing}: given a
function $f$ defined over some partially ordered domain, decide whether $f$ is \emph{monotone}, \ie
$f(x) \le f(y)$ whenever $x \preceq y$, or \emph{$\epsilon$-far} from any monotone function under a
given distance metric. Since the introduction of this problem \cite{GGLRS00} and especially over the
last decade, a series of works has revealed striking connections between monotonicity testing and
directed analogues of well-studied and ubiquitous isoperimetric inequalities such as Poincaré,
Margulis, and Talagrand inequalities.

Let $p, q \ge 1$. Following the notation of \cite{Fer23}, we say a (classical) \emph{$(L^p,
\ell^q)$-Poincaré inequality} is an inequality of the form
\[
    \dist^\const_p(f)^p \le C\, \Ex{\|\grad f\|_q^p}
\]
for, say, all functions $f : \zo^d \to \bR$ on the Boolean cube, or all functions $f : \closedInt^d
\to \bR$ on the unit cube, or perhaps only the Boolean-valued functions on these domains. Here
$\dist^\const_p(f)$ denotes the $L^p$ distance of $f$ to the closest constant function when the
domain is given uniform probability measure. For example,
\begin{enumerate}
    \item for functions $f : \zo^d \to \bR$, the classical Poincaré inequality $\Var{f} \le C\,
        \Inf{f}$, where $\Inf{f}$ denotes the total influence of $f$, is an $(L^2, \ell^2)$-Poincaré
        inequality (see \eg \cite{OD14});
    \item in the same setting, the Talagrand inequality \cite{Tal93} is an $(L^1, \ell^2)$ inequality;
    \item for smooth functions $f : \closedInt^d \to \bR$, the $(L^2, \ell^2)$ inequality is often
        called \emph{the} Poincaré inequality, especially in mathematical analysis (see \eg
        \cite{BGL14}); and
    \item in the same setting, the $(L^1, \ell^2)$ inequality was proved by Bobkov \& Houdré
        \cite{BH97}.
\end{enumerate}

A series of works on monotonicity testing has shown that many of these inequalities enjoy natural
``directed analogues''\!\!, as identified by \cite{CS16}. Let $\dist^\mono_p(f)$ denote the $L^p$
distance of $f$ to the closest monotone function. Then a \emph{directed $(L^p, \ell^q)$-Poincaré
inequality} is an inequality of the form
\[
    \dist^\mono_p(f)^p \le C\, \Ex{\|\grad^- f\|_q^p} \,,
\]
where $\grad^- f \define \min\{0, \grad f\}$, the \emph{directed gradient} of $f$, captures the
local violations of monotonicity\footnote{Note that both the classical and directed inequalities
relate a local property (violations of the ``constant'' or ``monotone'' property, captured by the
gradient) to a global one (distance to a constant or monotone function).}\!\!. For example
(listing by type of inequality rather than in historical order),
\begin{enumerate}
    \item an $(L^1, \ell^1)$ inequality was proved by \cite{GGLRS00} for functions\footnote{For the
        purposes of exposition, we limit the present discussion to the Boolean cube and unit cube
        domains; in particular, we do not extend the notation and discussion to hypergrid domains.}
        $f : \zo^d \to \zo$;
    \item \cite{Fer23} gave the same inequality for Lipschitz functions $f : \closedInt^d \to \bR$;
    \item \cite{KMS18} proved the stronger $(L^1, \ell^2)$ inequality for functions $f : \zo^d \to
        \zo$, in analogy with Talagrand's result in the classical setting; and
    \item related directed isoperimetric statements, not in Poincaré form, include a directed
        analogue of the Margulis inequality by \cite{CS16}, and isoperimetric inequalities for
        real-valued functions on the Boolean cube \cite{BKR24} and for Boolean functions on the
        hypergrid \cite{BCS23,BKKM23}.
\end{enumerate}

What makes the story above especially compelling is that each new directed isoperimetric inequality
also enabled an \emph{algorithmic result}, namely a monotonicity tester with improved query
complexity in the same setting. While we refrain from a full review and refer to \cite{Bla23}
instead, one central example is that, for functions $f : \zo^d \to \zo$, the celebrated work of
\cite{KMS18} essentially resolved the question for nonadaptive testers by giving a $\widetilde
O(\sqrt{d}/\epsilon^2)$ query tester.

In this paper, we seek to further understand the connection between classical and directed
isoperimetry, and the role of the latter in monotonicity testing. Building upon \cite{Fer23}, we
pursue this goal by studying the fully continuous setting, namely functions $f : \closedInt^d \to
\bR$. Let us offer two reasons for this choice, which we view as connections along two conceptual
``axes'':
\begin{description}
    \item[Classical versus directed.]
        The continuous setting is central to the study of isoperimetric phenomena. Ever since the
        original ``isoperimetric problem'' about shapes in Euclidean space, isoperimetric
        inequalities have enjoyed a rich and fruitful history with connections to mathematical
        physics, geometry, probability theory, diffusion processes, optimal transport, and so on;
        and closest to our subject, the Poincaré inequality itself first appeared in the study of
        partial differential equations arising in mathematical physics \cite{Poi90}. Thus, if there
        are unifying principles underlying both classical and directed isoperimetric phenomena, it
        seems reasonable to expect such a principle to manifest itself in the continuous setting.
    \item[Discrete versus continuous.]
        Phenomena involving functions on Boolean and continuous domains are often intimately
        related\footnote{See \eg \cite[Chapter~11]{OD14} and recent works such as
        \cite{KOW16,CHHL19,DNS21,AHLVXY23,EMR23}.}\!\!, which we may interpret as a form of
        ``robustness'' of the phenomena. Thus, it is natural to ask about the full scope of this
        connection in the case of directed isoperimetry and monotonicity testing. While \cite{Fer23}
        started to answer this question by giving an $L^1$ monotonicity tester with $O(d)$ query
        complexity via a directed $(L^1, \ell^1)$-Poincaré inequality, many questions remain. For
        example, they left open the possibility of a tester with $O(\sqrt{d})$ query complexity,
        which would bring the continuous landscape closer to the discrete one.
\end{description}

Our main result in this paper is a directed $(L^2, \ell^2)$-Poincaré inequality for sufficiently
regular functions $f : \closedInt^d \to \bR$ (in particular, Lipschitz continuity suffices; see
\cref{section:pde-preliminaries}):

\begin{restatable}[Directed Poincaré inequality]{theorem}{thmdirectedpoincareinequality}
    \label{thm:directed-poincare-inequality}
    There exists a universal constant $C > 0$ such that, for all $f \in H^1((0,1)^d)$,
    \begin{equation}
        \label{eq:main}
        \dist^\mono_2(f)^2 \le C\, \Ex{\|\grad^- f\|_2^2} \,.
    \end{equation}
\end{restatable}

We highlight two related aspects of this result. First, it takes the same form as the most classical
form of the Poincaré inequality, namely the $(L^2, \ell^2)$ form. In contrast, \cite{Fer23} gave an
$(L^1, \ell^1)$ inequality, which does not have a natural classical counterpart\footnote{In the
classical settings, as noted above, we have the stronger $(L^1, \ell^2)$ inequality instead. Note
that the $\ell^2$ (Euclidean) norm enjoys many nice properties that the $\ell^1$ norm does not (\eg
rotation invariance, being self-dual).}\!\!.

Second, the main theme of our proof of this result is the study of the \emph{convergence of a
partial differential equation} (PDE), which is the original motivating problem for the Poincaré
inequality. The idea is to take one of the central properties of the Poincaré inequality---that it
characterizes the convergence of the \emph{heat equation} to equilibrium---and modify that PDE in a
natural way so that it converges to a monotone (rather than constant) equilibrium, then derive from
its (exponential) convergence the directed inequality. In fact, plugging in the unmodified heat
equation into our proof recovers a long but unsurprising proof of the classical Poincaré inequality.
We contend that, in the continuous setting, exponential convergence of a PDE to its
constant/monotone equilibrium is the unifying principle behind the classical and directed
inequalities.

As an application, we obtain a monotonicity tester for Lipschitz functions $f : \closedInt^d \to
\bR$, with respect to the $L^2$ distance, using roughly $\sqrt{d}$ queries:

\begin{restatable}{theorem}{thmltwomonotonicitytester}
    \label{thm:l2-monotonicity-tester}
    There exists a nonadaptive, directional derivative $L^2$ monotonicity tester for $M$-Lipschitz
    functions $f : \closedInt^d \to \bR$ with query complexity $\widetilde O( \sqrt{d} M^2 /
    \epsilon^2 )$ and one-sided error.
\end{restatable}

This result answers affirmatively (up to a logarithmic factor) the question asked by \cite{Fer23}.
We remark that the algorithm above is also an $L^p$ tester for any $p \in [1,2]$, that the
directional derivative queries may be replaced by value queries as long as $f$ is sufficiently
smooth, and that the dependence on $d$ is optimal among a natural generalization of pair testers to
the continuous setting; see the next section for definitions and details.

The rest of the introduction proceeds as follows. In \cref{section:tester-details}, we define our
testing model and describe our monotonicity tester; in \cref{section:proof-overview} we give an
overview of our proof of the directed Poincaré inequality; and in \cref{section:discussion}, we
discuss our results and some open questions.

\subsection{Monotonicity testing of Lipschitz functions}
\label{section:tester-details}

In this section, let $\Omega \define \closedInt^d$ for simplicity of notation. For a function $f \in
L^p(\Omega)$ and $p \ge 1$, we define the $L^p$ distance to monotonicity of $f$ as
\[
    \dist^\mono_p(f)
    \define \inf \left\{ \|f-g\|_{L^p(\Omega)} : g \in L^p(\Omega) \text{ monotone} \right\} \,,
\]
where $\|f-g\|_{L^p(\Omega)} = \left( \int_{\Omega} (f-g)^p \odif x \right)^{1/p}$ by definition.

As announced, our main algorithmic result is a monotonicity tester with respect to the $L^2$
distance, so we adopt the $L^p$ testing model of \cite{BRY14}. The specific problem we
consider---$L^p$ testing monotonicity of Lipschitz functions---is the same as \cite{Fer23}. Let us
formally define the model. We say that function $f$ is $M$-Lipschitz if $\abs*{f(x) - f(y)} \le M
|x-y|$ for all pairs of points $x, y$. Then, the testing model is as follows:

\begin{definition}[Monotonicity tester]
    Let $p \ge 1$. Given parameters $d \in \bN$ and $M, \epsilon > 0$, we say a randomized algorithm
    $A$ is an \emph{$L^p$ monotonicity tester} for $M$-Lipschitz functions with query complexity
    $m(d, M, \epsilon)$ if, for every $M$-Lipschitz input function $f : \closedInt^d \to \bR$,
    algorithm $A$ makes $m(d, M, \epsilon)$ oracle queries to $f$ and 1)~accepts with probability at
    least $2/3$ if $f$ is monotone; 2)~rejects with probability at least $2/3$ if $\dist^\mono_p(f)
    \ge \epsilon$.
\end{definition}

As usual, an algorithm is \emph{nonadaptive} if it decides all its queries in advance before seeing
any output from the oracle, and it has \emph{one-sided error} if it accepts monotone functions with
probability $1$. As in \cite{Fer23}, we allow two types of oracles queries:

\begin{description}
    \item[Value query:] Given point $x \in \Omega$, the oracle outputs the value $f(x)$.
    \item[Directional derivative query:] Given point $x \in \Omega$ and direction $v \in \bR^d$, the
        oracle outputs the directional derivative $\grad_v f(x) = v \cdot \grad f(x)$, or the symbol
        $\bot$ if $f$ is not differentiable at $x$.
\end{description}

We remark that Lipschitz functions on $\Omega$ are differentiable \emph{almost everywhere} (\ie
outside a set of measure zero) by Rademacher's theorem, so non-differentiability is not a concern,
and they are bounded and hence in $L^p(\Omega)$ for any $p \ge 1$, so $\dist^\mono_p(f)$ is always
well-defined.

\paragraph*{The algorithm.}
Our tester may be seen as the natural continuous analogue of the path tester of \cite{KMS18} for the
Boolean cube, which samples points $x \preceq y \in \zo^d$ connected by a path of length $2^k$, for
$k$ sampled uniformly from $[\log d]$, and rejects if $f(x) > f(y)$. In our case, a path is replaced
by a \emph{direction} $v \in \zo^d$, and the condition $f(x) > f(y)$ is replaced by the condition $v
\cdot \grad f(x) < 0$, for a uniformly random point $x \in \closedInt^d$, via a directional
derivative query. Such a tester accepts any monotone function, so the challenge is to ensure that it
detects the case when $\dist^\mono_2(f) \ge \epsilon$.

In this case, hiding constant factors, \cref{thm:directed-poincare-inequality} gives that
$\Ex{\|\grad^- f\|_2^2} \gtrsim \epsilon^2$. We would like the distribution of $v$ to be such that,
for any $x$, $\Pru{v}{v \cdot \grad f(x) < 0} \gtrsim \frac{\|\grad^- f(x)\|_2^2}{\sqrt{d} \|\grad
f(x)\|_2^2}$. Note that $\|\grad f(x)\|_2^2 \le M^2$ because $f$ is $M$-Lipschitz, so if we had such
a distribution for $v$, the tester would reject with probability at least $\Exu{x}{\frac{\|\grad^-
f(x)\|_2^2}{\sqrt{d} M^2}} \gtrsim \frac{\epsilon^2}{\sqrt{d} M^2}$, implying that $O(\sqrt{d} M^2 /
\epsilon^2)$ queries suffice.

A natural choice would be to sample $v \in \zo^d$ so that $v_i = 1$ corresponds to taking edge $i$
in the path tester. We make this concrete as follows: sample $p$ uniformly from $\left\{ 1,
\frac{1}{2}, \dotsc, \frac{1}{2^{\ceil{\log_2(4d)}}} \right\}$, and sample $v \in \zo^d$ by letting
each $v_i$ be $1$ independently with probability $p$. At a high level,
\cref{lemma:good-vector-distribution} shows that this distribution satisfies the condition above (up
to a logarithmic factor) by considering, for each \emph{threshold} $\tau$, the smallest value of $p$
such that, informally, after issuing all its queries the algorithm should expect to sample at least
one vector $v$ whose support intersects with at least one entry $i$ from $\grad f(x)$ satisfying
$(\grad f(x))_i < -\tau$ (this idea allows us to ``forget'' about all the other negative entries of
$\grad f(x)$, and focus on a simpler ``good event''). By ranging over all possible $\tau$, we show
that there exists $\tau$ and corresponding $p$ such that, with good probability, the contribution of
the positive entries of $\grad f(x)$ to $v \cdot \grad f(x)$ is smaller than $\tau$, so that $v
\cdot \grad f(x) < 0$.

\paragraph*{Remarks on the tester.} As observed in \cite{BRY14}, for $1 \le p \le q$ and any (say)
Lipschitz $f$ and $g$, Jensen's inequality gives $\|f-g\|_{L^p(\Omega)} \le \|f-g\|_{L^q(\Omega)}$.
Thus any $L^q$ monotonicity tester is also an $L^p$ tester, so that the $\widetilde O(\sqrt{d} M^2 /
\epsilon^2)$ upper bound also holds for $L^p$ testing with $p \in [1,2]$.

A second remark is that the directional derivative queries, while convenient for the analysis and in
our opinion conceptually clean, can be replaced with value queries under reasonable assumptions.
First, in \cref{remark:robust-tester} we note that the rejection condition $v \cdot \grad f(x) < 0$
may in fact be replaced with the more robust condition $v \cdot \grad f(x) < -\Theta(\epsilon/d)$.
Now, suppose the input function $f$ is promised to be twice-differentiable with second derivatives
bounded by any constant $\beta$. We may then replace any directional derivative query $v \cdot \grad
f(x)$ with the value queries $f(x)$ and $f(y)$, where $y = x + \alpha v$ for sufficiently small
$\alpha > 0$ as a function of $\beta$, $\epsilon$ and $d$, and reject if $f(y) < f(x)$. This is
possible because, if $v \cdot \grad f(x) < -\Theta(\epsilon/d)$ and for such small $\alpha$, the
directional derivative $v \cdot \grad f(z)$ remains negative on all points $z$ in the line segment
between $x$ and $y$, and hence $f(y) < f(x)$. Conceptually, we view this as confirmation that
directional derivative queries are not unreasonably powerful.

Finally, the $\sqrt{d}$ dependence in the query complexity is optimal for \emph{derivative-pair
testers}, which are randomized nonadaptive testers that, at each step, either sample points $x
\preceq y$ and use value queries to reject if $f(x) > f(y)$, or sample point $x$ and direction $v
\succeq 0$ and use a directional derivative query to reject if $v \cdot \grad f(x) < 0$. The proof
is not difficult, and is given in \cref{section:lower-bound}.

\subsection{Proof overview}
\label{section:proof-overview}

Our guiding ideal is to prove \cref{thm:directed-poincare-inequality} by identifying a robust
approach to the classical Poincaré inequality such that, by ``toggling'' a single aspect of the
approach, we can transform a proof of the classical statement into a proof of its directed
counterpart.

\subsubsection{Starting point: the heat equation}

The first step is to identify the right classical starting point. For example, the Poincaré
inequality can be proved using Fourier analysis, but since the directed problem is highly nonlinear
(due to the $\grad^-$ operator), this approach does not seem suitable. Instead, we take a
physically-motivated approach. Let $u = u(t,x)$ where we think of the first variable as time, with
each $f = u(t) = u(t, \cdot)$ a function over space. A remarkable property (see \eg
\cite[Chapter~4]{BGL14}) of the Poincaré inequality
\[
    \Var{f} \le C \, \Ex{\|\grad f\|_2^2}
\]
(where we have written $\Var{f}$ in the place of $\dist^\const_2(f)^2$, which is the same thing) is
that it is equivalent to exponential decay of variance in the heat equation
\begin{equation}
    \label{eq:overview-heat}
    \partial_t u = \laplacian u \,,
\end{equation}
where $\laplacian$ denotes the Laplacian operator, given appropriate ``no-flux'' boundary conditions
which we do not expand on at the moment. To see this equivalence, assume for simplicity the mean
zero condition $\int u(0) \odif x = \int u(0,x) \odif x = 0$ (which is then preserved over time),
compute $\partial_t \Var{u(t)}$, differentiate under the integral, apply \eqref{eq:overview-heat}
and integrate by parts to obtain
\[
    \partial_t \Var{u(t)}
    = \partial_t \int u(t)^2 \odif x
    = 2 \int u(t) \partial_t u(t) \odif x
    = 2 \int u(t) \laplacian u(t) \odif x
    = -2 \int \grad u(t) \cdot \grad u(t) \odif x \,.
\]
This means that $\partial_t \Var{u(t)} \le -C\, \Var{u(t)}$ if and only if $\Var{u(t)} \le
\tfrac{2}{C} \int \|\grad u(t)\|_2^2 \odif x$, \ie exponential decay of the variance is equivalent
to the Poincaré inequality.

We observe that solutions to the heat equation converge to a constant equilibrium (taking the
average value of $u$), and the associated Poincaré inequality is a bound on the distance to a
constant function. Accordingly, it seems intuitive that if we replace \eqref{eq:overview-heat} with
a PDE that converges to a \emph{monotone} function, we might learn something about the distance to a
monotone function instead.

\subsubsection{The directed heat equation}
\label{section:overview-directed-heat-equation}

One challenge we face is that analyzing directed analogues of \eqref{eq:overview-heat} in the
multidimensional setting proves challenging due to the interaction between nonlinearities in
multiple dimensions. Hence, let us focus on the one-dimensional case for now. In one dimension, the
heat equation is
\[
    \partial_t u = \partial_x \partial_x u \,.
\]
The surface reading of this PDE is that, if we focus on the value of $u$ at a single point $x$, the
PDE tells us how this value changes over time. Another useful perspective is to think of $u(t,x)$ as
the density of ``particles'' at each point, and ask about how these particles are ``moving'' over
time; this idea can be made formal via the so-called \emph{continuity equation}. Under this view, it
turns out that the inner expression $\partial_x u$ represents (up to a sign change) the \emph{flux}
(or momentum) field, \ie the total rate of particle movement at each point, and by taking the
derivative of this field, \ie $\partial_x \partial_x u$, we determine whether, on the balance, there
are more ``incoming'' or ``outgoing'' particles---quantitatively, the rate of change $\partial_t u$.

With this perspective, a natural candidate PDE---which we call the \emph{directed heat
equation}---is
\begin{equation}
    \label{eq:intro-directed-pde}
    \partial_t u = \partial_x \partial_x^- u \,,
\end{equation}
where $\partial_x^- u \define \min\{0, \partial_x u\}$. The idea is that the new flux $\partial_x^-
u$ is always nonpositive, which, up to a required sign change, means that particles are only allowed
to ``move to the right''. Observe also that the RHS of \eqref{eq:intro-directed-pde} is zero for any
monotone function, so monotone functions are stationary solutions, while for decreasing functions,
this PDE behaves exactly as the heat equation. Thus, intuitively, this PDE seems to take a function
$u(0)$ and move it toward a monotone limit over time.

We explained above that the Poincaré inequality is intimately connected to convergence of the heat
equation, and we would like to show a similar property in the directed case. One option would be to
study the rate of decay of the distance to monotonicity, $\dist^\mono_2$. Although this is possible,
it does not seem to lead to a proof strategy; also note that, while in the case of the variance we
also have tools such as Fourier analysis to reason about the rate of decay directly, here we only
have the PDE to work with. Fortunately, there is another relevant quantity which decays
exponentially under the heat equation, namely the \emph{Dirichlet energy}
\[
    \cE(f) = \frac{1}{2} \int (\partial_x f)^2 \odif x \,.
\]
Indeed, another formal computation shows that by differentiating under the integral, switching the
partial derivatives, applying \eqref{eq:overview-heat} and integrating by parts, we get
\begin{align*}
    \partial_t \cE(u(t))
    &= \frac{1}{2} \partial_t \int (\partial_x u(t))^2 \odif x
    = \int (\partial_x u(t)) (\partial_t \partial_x u(t)) \odif x
    = \int (\partial_x u(t)) (\partial_x \partial_t u(t)) \odif x \\
    &= \int (\partial_x u(t)) (\partial_x \partial_x \partial_x u(t)) \odif x
    = -\int (\partial_x \partial_x u(t))^2 \odif x
    \le -\frac{1}{C} \int (\partial_x u(t))^2 \odif x
    = -\frac{2}{C} \cE(u(t)) \,,
\end{align*}
where the inequality is an application of a version of the Poincaré inequality for functions that
are zero on the boundary, which will be the case for $\partial_x u(t)$ by the no-flux boundary
conditions.

Now, it seems reasonable to propose the following directed analogue of the Dirichlet energy:
\[
    \cE^-(f) = \frac{1}{2} \int (\partial_x^- f)^2 \odif x \,.
\]
If $\cE$ measures the ``locally non-constant'' activity of a function, $\cE^-$ measures the
``locally non-monotone'' activity. Using $\cE^-$ to recast the directed heat equation in the
language of gradient flows and maximal monotone operators \cite{Bre73}, we can show that
\begin{enumerate}
    \item this PDE has a solution (\cref{res:solution,cor:evolution-solution-gradient-to-neumann});
    \item the directed Dirichlet energy decays exponentially in time
        (\cref{prop:exponential-decay});
    \item the solution converges to a monotone equilibrium as $t \to \infty$
        (\cref{prop:convergence-to-equilibrium}); and
    \item the solution, including up to the limit above, satisfies several other essential analytic
        properties, such as nonexpansiveness (\cref{prop:p-infty-nonexpansive}) and order
        preservation (\cref{cor:p-infty-order-preserving}).
\end{enumerate}

\subsubsection{Transport-energy inequality in one dimension}
\label{section:overview-transport-energy-one-dimension}

It is tempting, now, to try and use these results to conclude the one-dimensional directed Poincaré
inequality. The problem is that this seems to lead to a dead end, because we do not know how to
\emph{tensorize} the one-dimensional inequality into a multidimensional one\footnote{In the
classical case, the \emph{law of total variance} does allow such tensorization
\cite[Chapter~4.3]{BGL14}, but this aspect of the problem does not seem to be robust to passing to
the directed setting. Also note that, if we only desired a one-dimensional directed Poincaré
inequality, then \cite{Fer23} offers a much shorter proof.}\!\!.

Instead, we keep leaning on what the PDEs naturally tell us. It turns out that there is a different
notion of distance, the \emph{Wasserstein} distance, which is much more closely connected to the
theory of evolution equations and better suited to the dynamical approach we are undertaking.
Informally, the squared Wasserstein distance $W_2^2(\varrho_0, \varrho_1)$ between probability
measures $\varrho_0$ and $\varrho_1$ is the minimum total cost of a ``transport plan'' (coupling)
moving particles from $\varrho_0$ to $\varrho_1$, where the cost of moving a particle from point $x$
to point $y$ is $|x-y|^2$. The connection to PDEs is via the \emph{Benamou-Brenier formula}, which
informally says that
\[
    W_2^2(\varrho_0, \varrho_1)
    = \min\left\{ \int_0^1 \|v_t\|_{L^2(\varrho_t)}^2 \odif t :
            \text{$v_t$ is a velocity field taking $\varrho_0$ to $\varrho_1$ from time $0$ to $1$}
            \right\} \,.
\]
Without going into details, it follows that if we can upper bound the directed Dirichlet energy
$\cE^-(u(t))$, then we can upper bound the total (weighted) magnitude of the velocity field, \ie
$\|v_t\|_{L^2(\varrho_t)}^2$ (recall that the momentum of particles in the directed heat equation is
essentially $\partial_x^- u$, so this connection is not arbitrary), and therefore upper bound the
Wasserstein distance between the initial state $\varrho_0$---which informally corresponds to the
function $f = u(0)$---and the final state $\varrho_1$---which informally corresponds to the monotone
equilibrium $f^* = \lim_{t \to \infty} u(t)$.

This strategy yields, at least for bounded functions, the following result. Let $\cU$ be a class of
``reasonable'' initial states which we do not define for now, and let the operator $P_\infty$ map
each initial state to its monotone equilibrium according to the directed heat equation. Then
\begin{restatable*}[Transport-energy inequality in one
    dimension]{theorem}{thmundirectedtransportenergy}
    \label{thm:undirected-transport-energy}
    There exists a constant $C > 0$ such that the following holds. Let $u \in \cU$ be positive,
    bounded away from zero, and satisfy $\int_{(0,1)} u \odif x = 1$. Define the measures $\odif \mu
    \define u \odif x$ and $\odif \mu_\infty \define (P_\infty u) \odif x$. Then
    \[
        W_2^2(\mu, \mu_\infty) \le \frac{C}{\inf u} \cE^-(u) \,.
    \]
\end{restatable*}

\subsubsection{Tensorizing the transport-energy inequality}
\label{section:overview-tensorizing}

The advantage of making this detour through the Wasserstein distance is that, modulo the required
technical work, tensorization becomes relatively straightforward. Indeed, it is well-known in the
theory of optimal transport (\cf \cite[Remark~6.6]{Vil09}) that the $W_2$ distance scales well with
the dimension and reflects useful geometric content. Here the main idea is that, since the cost
function for the $W_2$ distance is the squared Euclidean distance $|x-y|^2$ between points $x, y \in
\closedInt^d$, one way to extend the one-dimensional result above to, say, the unit square $[0,1]^2$
is to
\begin{enumerate}
    \item apply the one-dimensional result to each row, making the row restrictions monotone while
        paying cost $W_2^2(\varrho_0, \varrho_1) = a^2$; and then
    \item apply the one-dimensional result to each column, making the column restrictions monotone
        while paying cost $W_2^2(\varrho_1, \varrho_2) = b^2$.
\end{enumerate}
Then, by combining the transport of particles along rows and columns into a single transport plan,
we obtain via the Pythagorean theorem a plan with cost $a^2 + b^2$, so this quantity upper bounds
the overall squared distance $W_2^2(\varrho_0, \varrho_2)$. The same idea extends to higher
dimensions, and we prove:

\begin{restatable*}[Pythagorean composition of transport plans]{lemma}{lemmapythagoreancomposition}
    \label{lemma:pythagorean-composition}
    Let $I, J \subseteq [d]$ be nonempty disjoint sets. Let $\mu, \varrho, \nu \in P(\bR^d)$ be
    supported in bounded sets, and let $\gamma^+ \in \Pi_I(\mu \to \varrho)$ and $\gamma^- \in
    \Pi_J(\varrho \to \nu)$. Then there exists $\gamma \in \Pi_{I \cup J}(\mu \to \nu)$ satisfying
    $C_2(\gamma)^2 = C_2(\gamma^+)^2 + C_2(\gamma^-)^2$.
\end{restatable*}

In this statement, $C_2(\gamma)^2$ is the cost of transport plan $\gamma$, and $\Pi_I(\mu \to \nu)$
denotes the set of transport plans between probability measures $\mu$ to $\nu$, under two additional
restrictions: 1)~particles can only move along ``$I$-aligned'' lines, \eg along rows in the example;
and 2)~particles can only move ``up'' in the partial order on $\bR^d$, \ie a particle can move from
$x$ to $y$ only if $x \preceq y$---we call such a transport plan \emph{directed}. The second
restriction is important because, when we finally recover a Poincaré inequality from a
transport-energy inequality, the information that particles only moved up in the partial order will
be reflected in the appearance of the directed gradient $\grad^-$.

Now, informally, \cref{thm:undirected-transport-energy} tells us that each step of the above
strategy (\eg moving particles along rows) incurs cost bounded by the directed Dirichlet energy, \ie
the integral $\int (\partial_i^- f)^2 \odif x$ of the partial derivatives along the (say) rows.
Repeating for every $i \in [d]$ via \cref{lemma:pythagorean-composition}, we obtain the directed,
multidimensional transport-energy inequality:

\begin{restatable*}[Transport-energy inequality]{theorem}{thmtransportenergy}
    \label{thm:transport-energy}
    There exists a universal constant $C > 0$ such that the following holds. Let $a \in (0,1)$, and
    let $f \in \Lip$ satisfy $1-a \le f \le 1+a$ and $\int_{\closedInt^d} f \odif x = 1$. Define the
    probability measures $\odif \mu \define f \odif x$ and $\odif \mu^* \define f^* \odif x$ on
    $\closedInt^d$. Then
    \[
        W_2^2(\mu \to \mu^*)
        \le \frac{C(1+a)^2}{(1-a)^3} \int_{\closedInt^d} \left| \grad^- f \right|^2 \odif x \,.
    \]
\end{restatable*}

Moreover, by the order preservation property of our solution to our PDE, each operation along
direction $i$ preserves the monotonicity along the directions $j < i$, so the final function $f^*$
in the statement above, which is the result after all $d$ steps, is indeed monotone. The notation
$W(\mu \to \mu^*)$ indicates that this bound holds even for directed transport plans, as explained
above.

\subsubsection{Recovering a Poincaré inequality via optimal transport duality}
\label{section:overview-duality}

The final step is to recover, from the multidimensional transport-energy inequality, our desired
directed Poincaré inequality. Fortunately, there is an intimate connection between transport and
Poincaré inequalities (as well as Sobolev, logarithmic Sobolev, and other isoperimetric and
concentration inequalities), and this theory is well established in the classical case (\cf
\cite[Chapters~21~and~22]{Vil09}). Closest to the present approach, \cite{Liu20} proved, in the
classical setting, the \emph{equivalence} between transport-energy, Poincaré, and other related
inequalities. Our task, therefore, is to obtain at least an implication in the directed setting.

The main ingredient toward this goal is the notion of \emph{Kantorovich duality} for the Wasserstein
distance. In the classical setting, the weak Kantorovich duality says that, given two probability
measures $\mu$ and $\nu$, if we can find ``certificate functions'' $\phi$ and $\psi$ satisfying
\begin{equation}
    \label{eq:overview-duality-cond}
    \phi(y) - \psi(x) \le |x-y|^2 \qquad \text{for all points $x$ and $y$,}
\end{equation}
then
\begin{equation}
    \label{eq:overview-duality}
    W_2^2(\mu, \nu) \ge \int \phi(y) \odif \nu(y) - \int \psi(x) \odif \mu(x) \,.
\end{equation}
Villani recounts the following analogy by Caffarelli: suppose you operate a consortium of bakeries
and cafes in a city, and that $\mu(\odif x)$ is the supply of bread by a bakery at point $x$ while
$\nu(\odif y)$ is the demand by a cafe at point $y$. Then if transporting a unit of bread from $x$
to $y$ costs $|x-y|^2$, the minimum cost for transporting all the bread is by definition $W_2^2(\mu,
\nu)$. Suppose a transportation company offers to take over the transportation job by buying bread
from the bakeries at price $\psi(x)$ and selling it to cafes at price $\phi(y)$. If
\eqref{eq:overview-duality-cond} holds then, for each $x$ and $y$, hiring the transportation company
is no more expensive than handling the transportation yourself. Therefore the cost $W_2^2(\mu, \nu)$
is at least as large as how much the bakery-cafe consortium pays the company to fulfill the supply
and demand, which is the RHS of \eqref{eq:overview-duality}.

In the directed setting, where we can only transport mass from $x$ to $y$ if $x \preceq y$, we
expect the transport to be in general more expensive, so ``more'' certificates should be valid.
Indeed, it is not difficult to obtain the following natural directed version of weak duality
(\cref{lemma:weak-duality}): if
\[
    \phi(y) - \psi(x) \le |x-y|^2 \qquad \text{for all points $x \preceq y$,}
\]
then
\[
    W_2^2(\mu \to \nu) \ge \int \phi(y) \odif \nu(y) - \int \psi(x) \odif \mu(x) \,.
\]

The final main ingredient is an operator which, given a candidate function $\phi$, produces the best
possible $\psi$ for duality. This is accomplished by a directed analogue of the so-called
\emph{Hamilton-Jacobi operator}, defined as follows: for each function $h$ on the unit cube,
\[
    (\dirH_t h)(x) \define \begin{cases}
        h(x)                                                            & \text{if $t = 0$} \\
        \sup_{y \succeq x} \left\{ h(y) - \frac{1}{2t} |x-y|^2 \right\} & \text{otherwise.}
    \end{cases}
\]
Then, for each $h$ and setting $t=1$, this operator yields the following (\cref{prop:h-duality}):
\begin{equation}
    \label{eq:overview-h-duality}
    \tfrac{1}{2} W_2^2(\mu \to \nu)
    \ge \int h \odif \nu - \int (\dirH_1 h) \odif \mu \,.
\end{equation}
Note that, as $t \to 0^+$, $(\dirH_t h)(x)$ intuitively seeks the direction $y-x \succeq 0$ of
steepest ascent of $h$. This suggests a connection to the directed gradient, and indeed in
\cref{prop:properties-h} we show that
\begin{equation}
    \label{eq:overview-limsup}
    \limsup_{t \to 0^+} \frac{\dirH_t h(x) - h(x)}{t}
    \le \frac{|\grad^+ h(x)|^2}{2} \,,
\end{equation}
where $\grad^+ h \define \max\{0, \grad h\}$.

At this point, instead of trying to reproduce the calculations from \cref{section:perturbation}, let
us give some intuition for how the pieces above recover a Poincaré inequality. By homogeneity, it
suffices to consider mean-zero functions $h$. The key idea is to fix a small $t > 0$ and let $h$
play two roles at the same time: 1)~as the building block for measures $\odif \mu(x) = (1 + th)
\odif x$ and $\odif \mu^*(x) = (1 + th^*) \odif x$, where $h^*$ is the monotone function obtained
via the coordinate-wise application of the directed heat equation as in the previous section; and
2)~in the test function $-th$ for duality via \eqref{eq:overview-h-duality}.

From the fact that $h$ plays both of these roles\footnote{We may think of test function $-th$ as the
transportation company trying to profit from non-monotonicity of $h$.}\!\!,
\eqref{eq:overview-h-duality} ends up producing the ``interaction term'' $\int (h^2 - h h^*) \odif
x$, which can be appropriately bounded by $\int (h - h^*)^2 \odif x$; and since $h^*$ is monotone,
the inequality $\int (h-h^*) \odif x \ge \dist^\mono_2(h)^2$ explains why we should expect the
distance to monotonicity to appear. Moreover, after the appropriate calculations, the term involving
$\dirH$ in \eqref{eq:overview-h-duality} gives rise to the expression in \eqref{eq:overview-limsup},
so letting $t \to 0^+$ makes the directed gradient of $h$ appear. Then $W_2^2(\mu, \mu^*)$ can be
lower bounded by an expression involving $\dist^\mono_2(h)^2$ and $\int |\grad^- h|^2 \odif x$. On
the other hand, \cref{thm:transport-energy} gives that $W_2^2(\mu, \mu^*)$ is upper bounded by an
expression involving $\int |\grad^- h|^2 \odif x$. Therefore, chaining the inequalities, we have
precisely the terms required to put \eqref{eq:main} together.

\subsection{Discussion and open questions}
\label{section:discussion}

\newcommand{\parspace}{\vspace{-.7em}}

\paragraph*{Conceptual and technical aspects.}
We consider the dynamical approach to \cref{thm:directed-poincare-inequality}---which establishes
that the convergence properties of a PDE underlies both directed and classical isoperimetric
statements---to be the main conceptual contribution of this work. The role played by optimal
transport speaks to the intimate relation between optimal transport and such dynamical processes,
and the fact that some of the optimal transport theory seems to find natural directed counterparts
is also intriguing and could be of independent interest.

Much of our technical effort is in dealing with the nonlinear nature of the $\grad^-$ operator. For
example, in principle this rules out well-known Fourier analytic arguments. In fact, even many of
the tools from nonlinear PDEs fail to apply at first, because they require some sort of coercivity
that is not satisfied in our setting---at a very informal level, because $\partial^-_x u$ may remain
at zero even as $\partial_x u$ grows arbitrarily large, which ``opens the door'' for pathological
objects to obstruct the theory. We deal with this difficulty via a canonical decomposition $u =
\upf{u} + \downf{u}$ of $u$ into a nondecreasing $\upf{u}$ and a nonincreasing $\downf{u}$, so that
we can isolate the phenomena we can control in $\downf{u}$ (and in fact recover a bit of linearity),
and deal with the less well-behaved $\upf{u}$ only when necessary.

\parspace
\paragraph*{Comparison with \cite{Fer23}.}
In prior work, \cite{Fer23} gave an $L^1$ tester for functions $f$ satisfying $\Lip_1(f) \le L$ with
query complexity $O(dL/\epsilon)$, where $\Lip_1(f)$ is the Lipschitz constant of $f$ with respect
to the $\ell^1$ metric. In contrast, in this paper we parameterize the problem by the more natural
$\ell^2$ (Euclidean) metric\footnote{In fact, that a $\sqrt{d}$ tester would be parameterized in the
$\ell^2$ metric if it exists was already suggested in \cite{Fer23}.}\!\!. Our tester takes functions
satisfying $\Lip_2(f) \le M$ and has query complexity $\widetilde O(\sqrt{d} M^2/\epsilon^2)$, and
it is also a fortiori an $L^1$ tester as already remarked.

By monotonicity of $\ell^q$ norms and the Cauchy-Schwarz inequality, in general we have $\Lip_1(f)
\le \Lip_2(f) \le \sqrt{d} \Lip_1(f)$, which allows rough comparisons between the two results. Since
the Euclidean geometry is more natural than the $\ell^1$ geometry, we find the first inequality more
informative: it implies that the tester of \cite{Fer23} is also an $L^1$ tester for functions
satisfying $\Lip_2(f) \le M$ with query complexity $O(dM/\epsilon)$. Thus, for $L^1$ testing
$M$-Lipschitz functions and ignoring logarithmic factors, \cref{thm:l2-monotonicity-tester} is
better than the tester of \cite{Fer23} when $\frac{\epsilon}{M} > \frac{1}{\sqrt{d}}$ and
vice-versa. Curiously, something analogous is true of the testers of \cite{KMS18} and \cite{GGLRS00}
for the Boolean setting.

\cite{Fer23} asked about lower bounds (for general testers) in the present setting, and this
question remains open. We also do not resolve Conjecture~1.8 of \cite{Fer23}, which asks for a
directed $(L^1, \ell^2)$-Poincaré inequality for Lipschitz $f : \closedInt^d \to \bR$, \ie an
analogue to the inequalities of Talagrand and Bobkov \& Houdré, but rather obtain our $\widetilde
O(\sqrt{d} M^2/\epsilon^2)$ tester via the $(L^2, \ell^2)$ inequality.

\parspace
\paragraph*{Comparison with the path tester.}
As explained in \cref{section:tester-details}, our directional derivative tester is essentially the
natural continuous analogue of the path tester of \cite{KMS18}, where taking edge $i$ in the path
tester roughly corresponds to letting $v_i = 1$ in the directional derivative query $v \cdot \grad
f(x)$. This conceptual connection is intriguing, and we do not know whether there is a formal
connection between the continuous and discrete problems that could explain it.

\parspace
\paragraph*{Nonlinear discrete Laplacian.}
In the study of spectral theory for directed graphs, \cite{Yos16} introduced a notion of nonlinear
Laplacian operator, with an associated heat equation, and applied these ideas to network analysis
and spectral clustering of directed graphs. Our directed heat equation may be seen as a closely
related continuous counterpart to the dynamical process studied in \cite{Yos16}, and further
exploring this connection is an exciting direction for future work.

\parspace
\paragraph*{Further applications?}
Finally, we ask whether other problems in property testing---particularly in the case of continuous
domain---can benefit from the techniques and ideas in the proof of
\cref{thm:directed-poincare-inequality}. In property testing, the idea of comparing the input object
$f$ to some ``ideal'' object $f^*$ that satisfies the property is very natural, and in this paper we
offer techniques from partial differential equations and optimal transport as useful tools for
reasoning about this comparison smoothly in time when the problem is continuous in nature.
Therefore, it is plausible that these tools may have something to say about other property testing
problems of continuous nature.

\paragraph*{Organization.} The rest of the paper is organized as follows. In
\cref{section:preliminaries}, we introduce definitions and conventions used throughout the paper. In
\cref{section:upper-bound}, we give our monotonicity tester and prove
\cref{thm:l2-monotonicity-tester}. In \cref{section:lower-bound}, we prove the query complexity
lower bound for derivative-pair testers. Finally,
\cref{section:directed-heat-semigroup,section:transport-energy-inequality,section:poincare-from-transport-energy}
establish \cref{thm:directed-poincare-inequality}, following the outline given in the proof overview.

\section{Preliminaries}
\label{section:preliminaries}

In this paper, $\bN$ denotes the set of strictly positive integers $\{1, 2, \dotsc\}$. Throughout
the paper, $d \in \bN$ is an arbitrary natural number indicating the dimension of the ambient space
$\bR^d$ unless otherwise specified. For $m \in \bN$, we write $[m]$ to denote the set $\{i \in \bN :
i \le m\}$. For any $x \in \bR$, we write $x^+$ for $\max\{0, x\}$ and $x^-$ for $\max\{0, -x\}$,
and we extend this notation to vectors $u \in \bR^d$ in the natural way: $u^+, u^- \in \bR^d$ are
given by $u^+_i \define (u_i)^+$ and $u^-_i \define (u_i)^-$. For real numbers $a$ and $b$, we use
the notation $a \land b \define \min(a, b)$ and $a \lor b \define \max(a, b)$.

For a vector $u \in \bR^d$, we let $\supp(u) \subseteq [d]$ denote the set of indices where $u$ is
nonzero, and write $\|u\|_0 \define \abs*{\supp(u)}$. For two points $x, y \in \bR^d$, we write $x
\preceq y$ if $x_i \le y_i$ for every $i \in [d]$, and $y \succeq x$ if $x \preceq y$.

We denote the closure of a set $D \subset \bR^d$ by $\overline D$. For a measure space $(\Omega,
\Sigma, \mu)$ and measurable function $f : \Omega \to \bR$, we write $\int_\Omega f \odif \mu$ for
the Lebesgue integral of $f$ over this space when it exists. Then for $1 \le p < +\infty$, the space
$L^p(\Omega)$ is the set of measurable functions $f$ such that $\abs{f}^p$ is Lebesgue integrable,
\ie $\int_\Omega \abs{f}^p \odif \mu < +\infty$, and we write the $L^p$ norm of such functions as
$\|f\|_{L^p(\Omega)} = \left(\int_\Omega \abs{f}^p \odif \mu\right)^{1/p}$. We usually write $\odif
x$ for the Lebesgue measure on $\bR^d$ in the context of integration, \ie for Lebesgue measurable
set $\Omega \subset \bR^d$, we write $\int_\Omega f \odif x$ for the Lebesgue integral of integrable
$f$. When we need to refer to the Lebesgue measure of a set $\Omega$ explicitly, we write
$\cL(\Omega)$.

Throughout the paper, all measures on $\bR^d$ are Borel measures. We say that measure $\mu$ over
$\Omega \subset \bR^d$ is \emph{absolutely continuous} if it is absolutely continuous with respect
to the Lebesgue measure, \ie if there exists a measurable function $f : \Omega \to \bR_{\ge 0}$
satisfying $\odif \mu = f \odif x$. In this case, $f$ is called the Radon-Nikodym derivative, or
\emph{density} of $\mu$.

Given a set $\Omega \subset \bR^d$ and $M > 0$, we say $f : \Omega \to \bR$ is \emph{$M$-Lipschitz}
if $\abs*{f(x) - f(y)} \le M |x-y|$ for all $x, y \in \Omega$. The \emph{Lipschitz constant} of $f$
is the smallest $M$ for which $f$ is $M$-Lipschitz. We say that $f$ is Lipschitz if it is
$M$-Lipschitz for any $M > 0$.

We use the notation $a \lequestion b$, $a \eqquestion b$, etc.\ within a proof to denote
(in)equalities that have not yet been established.

\paragraph*{Notation for directed partial derivatives and gradients.} Let $\Omega \subset \bR^d$ be
an open set, and let $f : \Omega \to \bR$ be Lipschitz. Then by Rademacher's theorem $f$ is
differentiable almost everywhere in $\Omega$. For each $x \in \Omega$ where $f$ is differentiable,
let $\grad f(x) = (\partial_1 f(x), \dotsc, \partial_d f(x))$ denote its gradient, where $\partial_i
f(x)$ is the partial derivative of $f$ along the $i$-th coordinate at $x$. Then, let $\partial^-_i
\define 0 \land \partial_i$, \ie for every $x$ where $f$ is differentiable we have $\partial^-_i
f(x) = -\left( \partial_i f(x) \right)^-$. We call $\partial^-_i$ the \emph{directed partial
derivative} operator in direction $i$. Then we define the \emph{directed gradient} operator by
$\grad^- \define (\partial^-_1, \dotsc, \partial^-_d)$, again defined on every $x$ where $f$ is
differentiable. Note that $\grad^- f \preceq 0$. We also similarly define $\partial^+_i \define 0
\lor \partial_i$ and $\grad^+_i \define (\partial^+_1, \dotsc, \partial^+_d)$.

\section{Algorithm and upper bound}
\label{section:upper-bound}

\begin{algorithm}[t]
    \caption{$L^2$ monotonicity tester for Lipschitz functions using directional derivative queries}
    \hspace*{\algorithmicindent}
        \textbf{Input:} Directional derivative oracle access to $M$-Lipschitz function $f :
        \closedInt^d \to \bR$. \\
    \hspace*{\algorithmicindent}
        \textbf{Output:} Accept if $f$ is monotone, reject if $\dist^\mono_2(f) \ge \epsilon$.
    \begin{algorithmic}
        \Procedure{DirectionalDerivativeTester}{$f, d, M, \epsilon$}
            \Repeat{$\Theta \left( \frac{\sqrt{d} M^2}{\epsilon^2} \log d \right)$}
                \State Sample $x \in \closedInt^d$ uniformly at random.
                \State Sample $v \in \zo^d$ from distribution $\cD$ given by
                \cref{def:distribution}.
                \State \textbf{Reject} if $\grad f(x) \cdot v < 0$.
            \EndRepeat
            \State \textbf{Accept}.
        \EndProcedure
    \end{algorithmic}
    \label{alg:tester}
\end{algorithm}

\begin{definition}[Distribution of direction vector]
    \label{def:distribution}
    For each $p \in [0,1]$, we define distribution $\cD_p$ over $\zo^d$ as follows. To sample
    $\bm{v} \sim \cD_p$,
    \begin{enumerate}
        \item Sample $\bm{x}_i \sim \Ber(p)$ independently for each $i \in [d]$, where $\Ber(p)$ is
            the Bernoulli distribution;
        \item Produce $\bm{v} = \sum_{i=1}^d \bm{x}_i e_i$, where $e_i$ denotes the $i$-th standard
            basis vector.
    \end{enumerate}
    Then, we define the distribution $\cD$ over $\zo^d$ as follows. Let $P \define \left\{1,
    \frac{1}{2}, \frac{1}{4}, \dotsc, \frac{1}{2^{\ceil{\log_2(4d)}}}\right\}$. To sample $\bm{v}
    \sim \cD$, we first sample $\bm{p}$ uniformly at random from $P$, and then sample $\bm{v}$ from
    $\cD_{\bm{p}}$.
\end{definition}

\begin{theorem}[Refinement of \cref{thm:l2-monotonicity-tester}]
    \label{thm:l2-monotonicity-tester-refined}
    \cref{alg:tester} is a nonadaptive, directional derivative $L^2$ monotonicity tester for
    $M$-Lipschitz functions $f : \closedInt^d \to \bR$ with query complexity $O\left( \frac{\sqrt{d}
    M^2}{\epsilon^2} \log d \right)$ and one-sided error.
\end{theorem}
\begin{proof}
    First note that, since $f$ is Lipschitz and hence differentiable almost everywhere in $(0,1)^d$
    by Rademacher's theorem, with probability $1$ the algorithm only samples points $x$ at which $f$
    is differentiable. Moreover, we have $\|\grad f\|_2 \le M$ almost everywhere since $f$ is
    $M$-Lipschitz.

    The algorithm clearly accepts any monotone function; indeed if $f$ is monotone, then $\grad f
    \succeq 0$ and, since $\cD$ is supported on $\zo^d$, we have $\grad f \cdot \bm{v} \ge 0$ with
    probability $1$. Now, suppose $\dist^\mono_2(f) \ge \epsilon$. Combining
    \cref{lemma:good-vector-distribution,thm:directed-poincare-inequality}, we obtain that the
    probability that any single iteration of the tester rejects is
    \begin{align*}
        \Pr{\text{Iteration rejects}}
        &= \Pru{\substack{\bm{x} \in \closedInt^d \\ \bm{v} \sim \cD}}{
            \grad f(\bm{x}) \cdot \bm{v} < 0} \\
        &\ge \Pru{\substack{\bm{x} \in \closedInt^d \\ \bm{v} \sim \cD}}{
            \grad f(\bm{x}) \cdot \bm{v} < -\frac{\delta}{d} \|\grad^- f(\bm{x})\|_2} \\
        &= \int_{\closedInt^d} \Pru{\bm{v} \sim \cD}{
            \grad f(x) \cdot \bm{v} < -\frac{\delta}{d} \|\grad^- f(x)\|_2} \odif x \\
        &\ge \int_{\closedInt^d} c \cdot \frac{\|\grad^- f(x)\|_2^2}{\sqrt{d} \log(d) M^2} \odif x
            & \text{(\cref{lemma:good-vector-distribution}, $\|\grad f\|_2 \le M$)} \\
        &= \frac{c}{M^2 \sqrt{d} \log d} \int_{\closedInt^d} \|\grad^- f(x)\|_2^2 \odif x \\
        &\ge \frac{c}{M^2 \sqrt{d} \log d} \cdot \frac{1}{C} \dist^\mono_2(f)^2
            & \text{(\cref{thm:directed-poincare-inequality})} \\
        &\ge \frac{c}{C} \cdot \frac{\epsilon^2}{M^2 \sqrt{d} \log d} \,,
    \end{align*}
    where $\delta, c$ are the constants from \cref{lemma:good-vector-distribution} and $C$ is the
    constant from \cref{thm:directed-poincare-inequality}. Thus $\Theta\left( \frac{\sqrt{d}
    M^2}{\epsilon^2} \log d \right)$ iterations suffice to reject with probability at least $2/3$.
\end{proof}

\begin{remark}
    \label{remark:robust-tester}
    A slight modification of the proof of \cref{thm:l2-monotonicity-tester-refined} also shows that,
    if we replace the condition $\grad f(x) \cdot v < 0$ in \cref{alg:tester} with the more
    demanding condition $\grad f(x) \cdot v < -\frac{K \epsilon}{d}$ for some universal constant $K
    > 0$, then we still obtain a tester with the same guarantees; in particular, \cref{alg:tester}
    does not rely on arbitrary precision. Indeed, \cref{thm:directed-poincare-inequality} gives that
    $\int \|\grad^- f(x)\|_2^2 \odif x \ge \tfrac{1}{C} \epsilon^2$ when $f$ is $\epsilon$-far from
    monotone, but the points $x$ for which $\|\grad^- f(x)\|_2^2 \le \frac{\epsilon^2}{2C}$ can only
    contribute at most $\frac{\epsilon^2}{2C}$ to the integral. Therefore the points satisfying
    $\|\grad^- f(x)\|_2^2 > \frac{\epsilon^2}{2C}$ must contribute at least $\frac{\epsilon^2}{2C}$
    and, at each such point, \cref{lemma:good-vector-distribution} guarantees that $\grad^- f(x)
    \cdot \bm{v} < -\frac{\delta}{d} \|\grad^- f(x)\|_2 < -\frac{\delta}{d} \cdot
    \frac{\epsilon}{\sqrt{2C}}$ with at least the probability given in the lemma. Therefore a
    similar calculation to that of \cref{thm:l2-monotonicity-tester-refined} shows that this
    modified tester also rejects with sufficient probability.
\end{remark}

The following lemma, which was used in the proof of \cref{thm:l2-monotonicity-tester-refined},
quantifies the ability of the tester to detect negative entries in the gradient $\grad f(x)$ via
directional derivative queries $v \cdot \grad f(x)$. This setup may be seen, informally, as a signed
version of the classic \emph{group testing} problem \cite{Dor43}.

\begin{lemma}[Detecting negative entries with subset sums]
    \label{lemma:good-vector-distribution}
    There exist universal constants $c, \delta > 0$ such that the distribution $\cD$ from
    \cref{def:distribution} has the following property: for any nonzero $u \in \bR^d$, we have
    \[
        \Pru{\bm{v} \sim \cD}{u \cdot \bm{v} < -\frac{\delta}{d} \|u^-\|_2}
        \ge c \cdot \frac{\|u^-\|_2^2}{\sqrt{d} \log(d) \cdot \|u\|_2^2} \,.
    \]
\end{lemma}
\begin{proof}
    We may assume that $u$ contains at least one strictly negative entry, since otherwise
    $\|u^-\|_2^2 = 0$ and the claim is trivial. Let $\delta \define 1/100$, and define $\delta_d
    \define \delta / d$ for convenience.

    Recall the distributions $\cD$ and $\cD_p$, as well as the set $P$, from
    \cref{def:distribution}. Let $t \define \frac{C \|u^-\|_2^2}{\sqrt{d} \cdot \|u\|_2^2}$, where
    we let $C \define 1/40$. Note that $0 < t < 1$. Letting $c_1 \define 1/10$, we will be done if
    we can show that there exists $p \in P$ such that
    \[
        \Pru{\bm{v} \sim \cD_p}{u \cdot \bm{v} < -\delta_d \|u^-\|_2} \gequestion c_1 t \,,
    \]
    Suppose for a contradiction that this is not the case, \ie that for every $p \in P$,
    \[
        \Pru{\bm{v} \sim \cD_p}{u \cdot \bm{v} < -\delta_d \|u^-\|_2} < c_1 t \,.
    \]
    For convenience of notation, let $a \define u^-$ and $b \define u^+$. Then for every $p \in P$,
    letting $\bm{z}, \bm{w} \sim \cD_p$ independently, we conclude that $u \cdot \bm{v}$ is
    distributed identically to $b \cdot \bm{w} - a \cdot \bm{z}$, and hence
    \[
        \Pru{\bm{z}, \bm{w} \sim \cD_p}{a \cdot \bm{z} > b \cdot \bm{w} + \delta_d \|a\|_2}
        < c_1 t \,.
    \]
    Fix any $\tau \in (2\delta_d \|a\|_2, \|a\|_\infty)$, where the interval is nonempty since
    $\|a\|_2 \le \|a\|_1 \le d \|a\|_\infty$, so that $2\delta_d \|a\|_2 = \frac{2 \|a\|_2}{100 d} <
    \|a\|_\infty$. Define $a_\tau \in \bR^d$ as the vector obtained from $a$ by only preserving
    entries that are strictly larger than $\tau$, and zeroing out entries that are at most
    $\tau$. Since the dot product $a \cdot \bm{z}$ only gets smaller if we omit some of its
    summands, we conclude that for every $p \in P$,
    \[
        \Pru{\bm{z}, \bm{w} \sim \cD_p}{a_\tau \cdot \bm{z} > b \cdot \bm{w} + \delta_d \|a\|_2}
        < c_1 t \,.
    \]
    Let $E$ denote the event that $\supp(\bm{z}) \cap \supp(a_\tau) \ne \emptyset$, so that
    $a_\tau \cdot \bm{z} > \tau$ when $E$ occurs and $a_\tau \cdot \bm{z} = 0$
    otherwise. Fix the smallest $p = p(\tau) \in P$ satisfying $p \ge
    \frac{t}{\|a_\tau\|_0}$, which must exist because $\|a_\tau\|_0 \ge 1$ by the choice
    of range for $\tau$ and $t \le 1$ as observed above, and therefore
    $\frac{t}{\|a_\tau\|_0} \le 1 \in P$. Then
    \[
        \Pru{\bm{z} \sim D_{p(\tau)}}{E}
        = 1 - (1-p(\tau))^{\|a_\tau\|_0}
        \ge 1 - e^{-p(\tau) \|a_\tau\|_0}
        \ge 1 - e^{-t}
        \ge \frac{t}{2} \,,
    \]
    the last inequality since $e^{-x} \le 1 - x/2$ for (say) $0 \le x \le 1$. Therefore
    \begin{align*}
        c_1 t
        &> \Pru{\bm{z}, \bm{w} \sim \cD_{p(\tau)}}{a_\tau \cdot \bm{z}
            > b \cdot \bm{w} + \delta_d \|a\|_2} \\
        &= \Pru{\bm{z} \sim D_{p(\tau)}}{E} \Pruc{\bm{z}, \bm{w} \sim \cD_{p(\tau)}}{
            a_\tau \cdot \bm{z} > b \cdot \bm{w} + \delta_d \|a\|_2}{E} \\
            &\qquad + \Pru{\bm{z} \sim D_{p(\tau)}}{\lnot E}
                \Pruc{\bm{z}, \bm{w} \sim \cD_{p(\tau)}}{
                a_\tau \cdot \bm{z} > b \cdot \bm{w} + \delta_d \|a\|_2}{\lnot E} \\
        &\ge \frac{t}{2} \Pru{\bm{w} \sim \cD_{p(\tau)}}{
            \tau > b \cdot \bm{w} + \delta_d \|a\|_2} \,.
    \end{align*}
    We conclude that
    \[
        \Pru{\bm{w} \sim \cD_{p(\tau)}}{b \cdot \bm{w} < \tau - \delta_d \|a\|_2} < 2 c_1
        \,.
    \]
    We now claim that $p(\tau) \le \frac{2t}{\|a_\tau\|_0}$. Indeed if this were not the
    case, then the choice of $p(\tau)$ would imply that $p_{\min} \define \min P$ satisfies
    $p_{\min} > \frac{2t}{\|a_\tau\|_0}$ with $p(\tau) = p_{\min} \le \frac{1}{4d}$ and
    hence, since $\tau > 2\delta_d \|a\|_2$,
    \[
        \Pru{\bm{w} \sim \cD_{p(\tau)}}{b \cdot \bm{w} < \tau - \delta_d \|a\|_2}
        \ge \Pru{\bm{w} \sim \cD_{p(\tau)}}{b \cdot \bm{w} = 0}
        \ge \Pru{\bm{w} \sim \cD_{p(\tau)}}{\bm{w} = 0}
        = (1 - p_{\min})^d
        \ge 1 - d \cdot p_{\min}
        \ge \frac{3}{4} \,,
    \]
    a contradiction. Thus $p(\tau) \le \frac{2t}{\|a_\tau\|_0}$. Now, the definition of
    $\cD_{p(\tau)}$, linearity of expectation and the Cauchy-Schwarz inequality yield
    \[
        \Exu{\bm{w} \sim \cD_{p(\tau)}}{b \cdot \bm{w}}
        = \sum_{i=1}^d \Pr{\bm{w}_i = 1} b_i
        = p(\tau) \cdot \|b\|_1
        \le p(\tau) \sqrt{d} \|b\|_2 \,,
    \]
    so by Markov's inequality,
    \[
        \Pru{\bm{w} \sim \cD_{p(\tau)}}{b \cdot \bm{w} \ge 10 p(\tau) \sqrt{d} \|b\|_2}
        \le \frac{1}{10} \,.
    \]
    Now, the union bound implies that
    \[
        \Pru{\bm{w} \sim \cD_{p(\tau)}}
            {\tau - \delta_d \|a\|_2 \le b \cdot \bm{w} < 10 p(\tau) \sqrt{d} \|b\|_2}
        \ge 1 - 2c_1 - \frac{1}{10}
        > 0 \,,
    \]
    so there exists $\bm{w}$ satisfying $\tau - \delta_d \|a\|_2 \le b \cdot \bm{w} < 10
    p(\tau) \sqrt{d} \|b\|_2$, and hence $\tau - \delta_d \|a\|_2 < 10 p(\tau) \sqrt{d}
    \|b\|_2$. Now, using the fact that $p(\tau) \le \frac{2t}{\|a_\tau\|_0}$, recalling
    that $t = \frac{C \|u^-\|_2^2}{\sqrt{d} \cdot \|u\|_2^2}$, observing that $\|u\|_2^2 =
    \|u^-\|_2^2 + \|u^+\|_2^2$ while $b = u^+$ by definition, and using the inequality $2xy \le x^2
    + y^2$, we obtain
    \[
        \tau - \delta_d \|a\|_2
        < \frac{20 \sqrt{d} \|b\|_2 t}{\|a_\tau\|_0}
        = \frac{20C \sqrt{d} \|u^+\|_2 \|u^-\|_2^2}
               {\|a_\tau\|_0 \cdot \sqrt{d} \cdot \|u\|_2^2}
        \le \frac{10C \cdot \|u^-\|_2}
                 {\|a_\tau\|_0} \,.
    \]
    In summary, recalling that $a = u^-$ by definition, for all $\tau \in (2\delta_d \|a\|_2,
    \|a\|_\infty)$ we have
    \[
        \frac{\tau}{2} \|a_\tau\|_0
        \le (\tau - \delta_d \|a\|_2) \|a_\tau\|_0
        < 10C \|a\|_2 \,,
    \]
    so $\tau \|a_\tau\|_0 < 20C \|a\|_2$ for all $\tau \in (2\delta_d \|a\|_2,
    \|a\|_\infty)$. On the other hand, if $\tau \in (0, 2\delta_d \|a\|_2]$ then
    \[
        \tau \|a_\tau\|_0
        \le 2\delta_d \|a\|_2 \cdot d
        = 2\delta \|a\|_2
        = \frac{1}{50} \|a\|_2
        < \frac{1}{2} \|a\|_2
        = 20C \|a\|_2 \,,
    \]
    so in fact $\tau \|a_\tau\|_0 < 20C \|a\|_2$ for all $\tau \in (0, \|a\|_\infty)$.
    Integrating over $\tau$ and using Tonelli's theorem, we conclude that
    \begin{align*}
        \|a\|_\infty \cdot 20C \|a\|_2
        &= \int_{(0, \|a\|_\infty)} 20C \|a\|_2 \odif \tau
        > \int_{(0, \|a\|_\infty)} \tau \|a_\tau\|_0 \odif \tau
        = \int_{(0, \|a\|_\infty)} \sum_{i=1}^d
                \left( \tau \cdot \ind{a_i > \tau} \right) \odif \tau \\
        &= \sum_{i=1}^d \int_{(0, \|a\|_\infty)}
                \left( \tau \cdot \ind{a_i > \tau} \right) \odif \tau
        = \sum_{i=1}^d \int_{(0, a_i)} \tau \odif \tau
        = \sum_{i=1}^d \frac{a_i^2}{2}
        = \frac{\|a\|_2^2}{2} \,.
    \end{align*}
    Hence we obtain that $\|a\|_2 < 40C \cdot \|a\|_\infty = \|a\|_\infty$, a contradiction as
    desired.
\end{proof}

\section{Lower bound}
\label{section:lower-bound}

\begin{definition}
    A \emph{derivative-pair tester} for monotonicity of Lipschitz functions $f : \closedInt^d \to
    \bR$ is described by a distribution $\cD$ over \emph{pair tests} and \emph{directional
    derivative tests}, where
    \begin{enumerate}
        \item a pair test is an operation that performs a pair of value queries $f(x), f(y)$ for $x
            \preceq y$, and rejects if and only if $f(x) > f(y)$; and
        \item a directional derivative test is an operation that performs a directional derivative
            query at point $x$ and direction $v \succeq 0$, and rejects if and only if $\grad f
            \cdot v < 0$.
    \end{enumerate}
    The tester independently samples $m$ tests from $\cD$, rejects if any of these tests rejects,
    and accepts otherwise. By definition, every derivative-pair tester is nonadaptive and has
    one-sided error.
\end{definition}

\begin{theorem}
    \label{res:lower-bound}
    Every $L^1$ derivative-pair tester for monotonicity of $M$-Lipschitz functions $f : \closedInt^d
    \to \bR$ must have query complexity $\Omega\left( \frac{\sqrt{d} M}{\epsilon} \right)$.
\end{theorem}
\begin{proof}
    It suffices to give a distribution over $O(M)$-Lipschitz functions $\bm{f}$ that are
    $\Omega(\epsilon)$-far from monotone in $L^1$ distance and such that any fixed pair test or
    directional derivative test has only an $O\left( \frac{\epsilon}{\sqrt{d} M} \right)$
    probability of rejecting a random $\bm{f}$. Moreover, to obtain the asymptotic lower bound we
    may as well assume that (say) $\epsilon \le 1 \le M$.

    For each $i \in [d]$, define the function $f_i : \closedInt^d \to \bR$ by
    \[
        f_i(x) \define -\epsilon x_i + \sum_{j \in [d] \setminus \{i\}} \frac{M}{\sqrt{d}} x_j \,.
    \]
    Then, define $\bm{f} \define f_{\bm{i}}$ where $\bm{i}$ is sampled uniformly at random from
    $[d]$. Note that each $f_i$ is linear and hence Lipschitz with Lipschitz constant
    \[
        |\grad f_i(x)|
        = \sqrt{\epsilon^2 + (d-1) \left(\frac{M}{\sqrt{d}}\right)^2}
        \le \sqrt{2} M
        = O(M) \,.
    \]
    We claim that $\dist^\mono_1(f_i) = \Omega(\epsilon)$ for each $i \in [d]$; by symmetry it
    suffices to consider the case $i=d$. Any restriction $f_d(x_{-d}, \cdot)$ of $f_d$ to an
    axis-aligned line in direction $d$ (where $x_{-d} \define (x_1, \dotsc, x_{d-1})$) is a linear
    function with slope $-\epsilon$, so it is $\Omega(\epsilon)$-far from monotone in $L^1$ distance
    (in fact, $\epsilon/4$-far, and the constant average-valued function achieves this). But by
    Tonelli's theorem, the distance from $f_d$ to any monotone function $g$ is the average of the
    distance over each such axis-aligned line:
    \[
        \int_{\closedInt^d} |f_d-g| \odif x
        = \int_{\closedInt^{d-1}} \odif x_{-d} \int_{[0,1]}
            |f_d(x_{-d}, x_d) - g(x_{-d}, x_d)| \odif x_d
        \ge \int_{\closedInt^{d-1}} \Omega(\epsilon) \odif x_{-d}
        = \Omega(\epsilon) \,.
    \]
    We now show that any pair test or directional derivative test rejects a random $\bm{f}$ only
    with probability $O\left(\frac{\epsilon}{\sqrt{d} M}\right)$. Note that, since each $f_i$ is a
    linear function, any pair test on points $x \preceq y$ may be simulated by a directional
    derivative test on point $x$ and direction $v = y - x \succeq 0$. Therefore it suffices to
    consider directional derivative tests. Moreover, since the gradient of each $f_i$ is constant
    over the points $x$, only the direction $v$ is relevant. For simplicity, we will say that
    direction $v$ \emph{rejects} function $f_i$ if a directional derivative test with direction $v$
    rejects $f_i$.

    Fix any direction $v \succeq 0$. Suppose there exists $j \in [d]$ such that
    \[
        \epsilon v_j \le \frac{M}{\sqrt{d}} \sum_{k \in [d] \setminus \{j\}} v_k \,.
    \]
    We then claim that the vector $v'$ given by $v'_j \define 0$ and $v'_k \define v_k$ for $k \in
    [d] \setminus \{j\}$ rejects every $f_i$ that $v$ rejects. Indeed, suppose $v$ rejects $f_i$.
    There are two cases. If $i = j$, then
    \[
        0 > \grad f_i(x) \cdot v
        = -\epsilon v_j + \sum_{k \in [d] \setminus \{j\}} \frac{M}{\sqrt{d}} v_k \,,
    \]
    a contradiction, so this case cannot happen. On the other hand, if $i \ne j$, then
    \[
        \grad f(x) \cdot v'
        = -\epsilon v'_i + \sum_{k \in [d] \setminus \{i\}} \frac{M}{\sqrt{d}} v'_k
        \le -\epsilon v_i + \sum_{k \in [d] \setminus \{i\}} \frac{M}{\sqrt{d}} v_k
        = \grad f(x) \cdot v
        < 0 \,,
    \]
    the inequality since $v'_i = v_i$ and $v' \preceq v$. Hence $v'$ also rejects $f_i$, as claimed.
    Thus $v'$ rejects a random $\bm{f}$ with no less probability than $v$. Therefore we may assume
    without loss of generality that, for every $j \in [d]$, either $v_j = 0$ or
    \[
        \epsilon v_j > \frac{M}{\sqrt{d}} \sum_{k \in [d] \setminus \{j\}} v_k \,.
    \]
    Let $j^* \in [d]$ be such that $v_{j^*}$ is a minimum nonzero entry of $v$, and let $n \define
    \|v\|_0$ be the number of nonzero entries of $v$. Applying the inequality above to $j = j^*$
    yields
    \[
        \epsilon v_{j^*}
        > \frac{M}{\sqrt{d}} \sum_{k \in [d] \setminus \{j^*\}} v_k
        \ge \frac{M}{\sqrt{d}} \sum_{k \in [d] \setminus \{j^*\}} \ind{v_k > 0} v_{j^*}
        = \frac{M}{\sqrt{d}} \cdot n v_{j^*} \,,
    \]
    and hence $n < \frac{\sqrt{d} \epsilon}{M}$. Finally, note that $v$ does not reject $f_i$ if
    $v_i = 0$, so we conclude that
    \[
        \Pru{\bm{f}}{v \text{ rejects } \bm{f}}
        = \Pru{\bm{i}}{v \text{ rejects } f_{\bm{i}}}
        \le \Pru{\bm{i}}{v_{\bm{i}} > 0}
        = \frac{n}{d}
        \le \frac{\sqrt{d} \epsilon / M}{d}
        = \frac{\epsilon}{\sqrt{d} M} \,,
    \]
    as desired. This concludes the proof.
\end{proof}

\begin{remark}
    By the remark in \cref{section:tester-details}, the lower bound above holds for all $L^p$
    testers, $p \ge 1$.
\end{remark}

\paragraph*{How to read the rest of this paper.}
\cref{section:directed-heat-semigroup,section:transport-energy-inequality,section:poincare-from-transport-energy}
establish the proof of \cref{thm:directed-poincare-inequality} via arguments from partial
differential equations and optimal transport theory, and make up the most technical portion of this
paper. While we give technical preliminaries summarizing the standard theory we require in the
beginning of \cref{section:directed-heat-semigroup,section:transport-energy-inequality}, for a
detailed reading it may also be helpful to have the listed reference texts
\cite{Eva10,Bre11,Vil09,San15} on hand.

We favour a linear presentation in which results are built up progressively (occasionally deferring
a single-use lemma to be given immediately after the result that requires it), culminating in one or
a few key results in each subsection -- these are indicated at the start of each section. While we
defer a few technical lemmas to \cref{section:technical-lemmas}, we include in the main text the
proofs of lemmas that refer to the particular theory developed in this paper (as opposed to standard
results of general interest). Therefore, it may be convenient to skip proofs of lemmas on a first
reading.

\section{Directed heat semigroup}
\label{section:directed-heat-semigroup}

In this section, we make concrete the ideas presented in
\cref{section:overview-directed-heat-equation} of the proof overview. Recall that we wish to analyze
the \emph{directed heat equation}, which we informally specified in \eqref{eq:intro-directed-pde} as
\begin{equation}
    \label{eq:body-directed-pde}
    \partial_t u = \partial_x \partial_x^- u \,.
\end{equation}
We require a more rigorous definition before we can study this PDE, since \eg the quantity
$\partial_x^- u = \min\{0, \partial_x u\}$ is not in general differentiable even if $u$ is smooth.
Indeed, while the classical heat equation has a smoothing effect, we expect any solution to
\eqref{eq:body-directed-pde} (in the appropriate sense) to have non-differentiable points, so a more
technical treatment of this nonlinear PDE is unavoidable.

The first step is to replace the operator $\partial^- u$ with something more amenable to analysis.
In \cref{def:canonical-representation} we define a \emph{canonical representation} $u = \upf{u} +
\downf{u}$ such that $\downf{u}$ is nonincreasing and differentiable (more precisely, it is in the
Sobolev space $H^1$), while $\upf{u}$ is nondecreasing and need not be differentiable (it only needs
to be in the space $L^2$ of square-integrable functions). While a priori such a representation need
not be unique, we show in \cref{prop:optimal-representations} that there is a unique representation
(after fixing the shift by a constant) minimizing the energy $\int (\partial_x \downf{u})^2 \odif
x$; intuitively, this means that we pack all the ``decreasing behaviour'' of $u$ into $\downf{u}$,
and all the ``increasing behaviour'' of $u$ into $\upf{u}$. Note that $\upf{u}$ may even contain
jump discontinuities, which our theory must withstand, so this decomposition captures extremely
relevant properties of our problem.

Then, we are able to restate our PDE more rigorously in
\cref{def:static-neumann-problem,def:neumann-evolution}. Essentially, the ``Neumann problem''
\[
    \begin{cases}
        z = \partial_x \partial_x \downf{u} & \text{in } (0,1) \\
        \partial_x \downf{u} = 0            & \text{on } \{0, 1\}
    \end{cases}
\]
given in \eqref{eq:static-neumann-problem} (where $z$ will take the role of $\partial_t u$ in
\cref{def:neumann-evolution}) replaces $\partial_x^- u$ with $\partial_x \downf{u}$ in
\eqref{eq:body-directed-pde}, and then places a ``Neumann boundary condition'' $\partial_x \downf{u}
= 0$ which, as is standard in PDEs, means that there will be no mass flux at the boundaries of the
domain. For technical reasons that are common to most study of PDEs, we need to give in
\eqref{eq:weak-solution} a notion of ``weak solution'', which is essentially a solution which may
not be regular enough to satisfy the specification of the PDE (\ie it may not be
twice-differentiable), but which behaves (via integration by parts) just like a ``strong solution''
for many analytical purposes, and can sometimes be ``promoted'': a weak solution that is regular
enough is usually also a strong solution (in fact, our problem automatically enjoys such a promotion
thanks to ``elliptic regularity'', see \cref{lemma:elliptic-regularity}).

To study our PDE, we bring in the theory of gradient flows and maximal monotone operators
\cite{Bre73} by recasting the PDE as the problem of finding the direction of steepest descent (in
the space of $L^2$ functions) of the directed Dirichlet energy $\cE^-(u)$, which our formalism
allows us to define in terms of $\downf{u}$ in \cref{def:energy-functional}. The main results in
\cref{section:neumann-gradient-flow-problems} are \cref{res:solution} on the existence of a unique
solution to the gradient flow problem -- which allows us to define a semigroup operator $P_t$, the
\emph{directed heat semigroup}, describing the evolution of this solution -- and
\cref{cor:evolution-solution-gradient-to-neumann}, which shows that this solution is also a solution
to the Neumann problem (\ie our PDE of interest).

To conclude the basic study of our PDE, \cref{section:energy-decay} shows, via the theory of
gradient flows, that the directed Dirichlet energy $\cE^-(P_t u)$ decays exponentially in time
(\cref{prop:exponential-decay}), which allows us to conclude that $P_t u$ converges to some limit as
$t \to \infty$ (\cref{lemma:strong-convergence}).

\cref{section:auxiliary-regularity,section:preservation-of-regularity,section:preservation-of-lipschitz}
are concerned with showing that regular initial states $u$ yield solutions $P_t u$ that are also
regular for all times $t > 0$, where ``regular'' can mean Lipschitz or Sobolev class $H^1$. The main
result in \cref{section:auxiliary-regularity} is \cref{lemma:in-domain-well-behaved}, a technical
lemma which says that for solutions to our PDE, if $\partial_x \downf{u} < 0$ in an interval $(a,
b)$ (which we wish to call ``$u$ is decreasing''), it must be the case that $\upf{u}$ is constant in
$(a, b)$ (legitimizing that wish). This result is far from being vacuous or tautological -- it
implies that in regions where our PDE is ``active'' (because $\partial_x \downf{u}$ is nonzero), the
other component $\upf{u}$ of $u$ is ``inactive'' (constant), which means that solutions to our PDE
avoid a host of pathological scenarios and allows crucial technical arguments to go through. Then,
using technical arguments and tools from the theory of maximal monotone operators,
\cref{section:preservation-of-regularity} builds toward \cref{cor:preservation-of-regularity} --
$H^1$ regularity is preserved over time -- and \cref{section:preservation-of-lipschitz} builds
toward \cref{cor:preservation-of-lipschitz} -- Lipschitz regularity is preserved over time.

In \cref{section:nonexpansiveness}, we turn our attention to nonexpansiveness and order preservation
properties. The main result, \cref{prop:pt-diff-nonincreasing}, is a quantitative order preservation
result which implies that if initial states $u, v$ satisfy $u \le v$, then $P_t u \le P_t v$ for all
times $t > 0$ (\cref{cor:pt-order-preserving}). We conclude \cref{section:nonexpansiveness} with
some useful basic observations about the semigroup $P_t$; in particular,
\cref{prop:pt-transformations} shows that $P_t$ behaves nicely under certain affine transformations,
namely $P_t (\alpha u + \beta) = \alpha P_t u + \beta$ when $\alpha > 0$ and $\beta \in \bR$.

Finally, \cref{section:convergence} studies the convergence of solutions $P_t u$ as $t \to \infty$.
The main result is \cref{prop:convergence-to-equilibrium}, which shows that $P_t u$ converges to a
\emph{monotone} function, which we define as $P_\infty u$ in \cref{def:p-infty}. The remainder of
\cref{section:convergence} shows that the main properties of $P_t$ shown in the previous subsections
-- preservation of regularity, nonexpansiveness, behaviour under affine transformations -- are also
satisfied by $P_\infty$.

\subsection{Preliminaries for PDE}
\label{section:pde-preliminaries}

We briefly outline some of the main concepts required to study our PDE, and refer the reader to \eg
\cite{Eva10,Bre11} for detailed expositions.

Let $J = (a,b)$ be an open interval. The absolutely continuous (AC) functions $f : \overline J \to
\bR$ are precisely those for which there exists a function $g \in L^1(J)$ such that, for all $x \in
\overline J$, $f(x) = f(a) + \int_{(a,x)} g \odif x$. We call $g$ a \emph{weak derivative} of $f$,
and write $\partial_x f$ for any weak derivative of $f$. The weak derivative is almost everywhere
(a.e.) uniquely determined, and moreover, $f$ is classically differentiable \almev and its classical
derivative agrees with $\partial_x f$ \almev.

Let $\Omega \define J^d$. Let $k \in \bZ_{\ge 0}$ and $p \in [1, +\infty]$. The Sobolev space
$W^{k,p}(\Omega)$ is the space of functions $f : \Omega \to \bR$ which have weak derivatives up to
order $k$ in $L^p(\Omega)$. The definition of weak derivative in the multidimensional case is more
involved than in the one-dimensional case, but the details are not relevant here. In one dimension,
we may recursively define $W^{0,p}(J) \define L^p(J)$ and, for $k \ge 1$, $W^{k,p}(J)$ as the set of
functions $f \in W^{k-1,p}(J)$ such that $\partial_x f \in W^{k-1,p}(J)$.

We identify \almev equal functions into equivalence classes in $W^{k,p}(\Omega)$, and we write ``$f
= g$ a.e.'' and ``$f = g$ in $W^{k,p}(\Omega)$'' interchangeably. Often, we are interested in a
continuous representative of $f$, which is unique when it exists, and by abuse of notation write $f$
for its continuous representative as well. With this interpretation, a standard fact is that
$W^{1,\infty}(\Omega)$ is precisely the class of Lipschitz functions on $\Omega$. If $1 \le p \le q
\le +\infty$, then $W^{k,q}(\Omega) \subset W^{k,p}(\Omega)$.

The case $p = 2$ is special because then $H^k(\Omega) \define W^{k,2}(\Omega)$ is a Hilbert space.
In one dimension, the inner product in $H^k(J)$ is
\[
    \inpspace{f}{g}{H^k(J)} \define \sum_{i=0}^k \inpspace{D^i f}{D^i g}{L^2(J)} \,,
\]
where $D^i$ above denotes the $i$-th weak derivative. When the space of integration is clear, we
drop the subscript from the inner product notation. The inner product above also induces the norm
\[
    \|f\|_{H^k(J)} = \sqrt{\sum_{i=0}^k \|D^i f\|_{L^2(J)}^2} \,.
\]

For a sequence $(f_n)_{n \in \bN} \subset H^k(\Omega)$ and $f \in H^k(\Omega)$, the notation ``$f
\to g$ in $H^k(\Omega)$'' means convergence in norm, \ie $\|f_n - f\|_{H^k(\Omega)} \to 0$ as $n \to
\infty$. We write ``$f \weakto g$ weakly in $H^k(\Omega)$'' to denote weak convergence in this
space, which means that $\varphi(f_n) \to \varphi(f)$ as $n \to \infty$ for every $\varphi$ in the
dual space of $H^k(\Omega)$. The dual space of $L^2(\Omega)$ is $L^2(\Omega)$ itself with the
$L^2(\Omega)$ inner product action $\varphi(f) = \inp{\varphi}{f}$. The dual space of $H^1(\Omega)$
is larger, but contains $L^2(\Omega)$ with the same action.

It is often useful to approximate a function by a sequence of smooth functions, and for that we will
use \emph{mollification}. The following comes from \cite[Appendix~C]{Eva10}. The \emph{standard
mollifier} $\eta \in C^\infty(\bR^d)$ is given by
\[
    \eta(x) \define \begin{cases}
        C \exp\left( \frac{1}{|x|^2-1} \right) & \text{if } |x| < 1 \\
        0                                      & \text{if } |x| \ge 1 \,,
    \end{cases}
\]
for constant $C > 0$ chosen so that $\int_{\bR^d} \eta \odif x = 1$. Then for each $\epsilon > 0$,
we let
\[
    \eta_\epsilon(x) \define \frac{1}{\epsilon^d} \eta(x/\epsilon) \,,
\]
which is a $C^\infty(\bR^d)$ function satisfying $\int_{\bR^d} \eta_\epsilon \odif x = 1$ and
$\eta(x) = 0$ for $x \not\in B(0,\epsilon)$, where $B(0,\epsilon)$ is the open ball of radius
$\epsilon$ centered at $0$. We abuse language and also call $\eta_\epsilon$ a standard mollifier.

Let $U \subset \bR^d$ be an open set and let $U_\epsilon \define \{ x \in U : \mathrm{dist}(x,
\partial U) > \epsilon \}$, where $\partial U$ is the boundary of $U$. For locally integrable $f : U
\to \bR$ (meaning that $f$ is integrable on every compact $K \subset U$), we define the
\emph{mollification} $f^\epsilon : U_\epsilon \to \bR$ of $f$ by
\[
    f^\epsilon(x)
    \define (\eta_\epsilon * f)(x)
    = \int_U \eta_\epsilon(x-y) f(y) \odif y
    = \int_{B(0,\epsilon)} \eta(y) f(x-y) \odif y \,.
\]
We then have the following properties (see Theorem~7 of \cite[Appendix~C]{Eva10}):
\begin{enumerate}
    \item $f^\epsilon \in C^\infty(U_\epsilon)$.
    \item $f^\epsilon \to f$ \almev as $\epsilon \to 0$.
    \item If $f \in C(U)$, then $f^\epsilon \to f$ uniformly on compact subsets of $U$.
    \item If $1 \le p < \infty$ and $f \in L^p_\loc(U)$, then $f^\epsilon \to f$ in $L^p_\loc(U)$,
\end{enumerate}
where $L^p_\loc(U)$ is the space of measurable functions whose every restriction to compact $K
\subset U$ is in $L^p(K)$, and convergence in $L^p_\loc(U)$ means convergence in $L^p(K)$ for every
compact $K \subset U$.

We also use the following abstract spaces. Let $X$ be a real Banach space (for us, typically
$L^2(J)$), let $1 \le p \le +\infty$ and let $T \in (0, +\infty]$. Then the \emph{Bochner space}
$L^p(0, T; X)$ is the space of Bochner measurable functions (whose precise definition is not
important here) $\bm{u} : [0, T] \to X$ whose norm $\|\bm{u}\|_{L^p(0, T; X)}$ is finite, where
\[
    \|\bm{u}\|_{L^p(0, T; X)} \define \begin{cases}
        \left( \int_{(0,T)} \|\bm{u}(t)\|_X^p \odif t \right)^{1/p} & 1 \le p < +\infty \\
        \esssup_{t \in (0,T)} \|\bm{u}(t)\|_X                       & p = +\infty \,.
    \end{cases}
\]
As usual, $C([0, T]; X)$ denotes the set of continuous functions $\bm{u} : [0,T] \to X$. When $X$
possesses the so-called Radon-Nikodym property, which in particular is the case for $X = L^2(J)$,
and for $T < +\infty$, the characterization of absolutely continuous functions as those possessing
an integrable weak derivative introduced earlier extends to functions $\bm{u} : [0,T] \to X$ (see
\cite[pp.~217--219]{DU77}).

Above and hereafter, we use boldface symbols, \eg $\bm{u}$, to denote Banach-valued functions whose
domain we think of as the time in an evolution equation.

Hereafter, we write $I$ for the unit interval $(0,1)$ unless otherwise specified, which will always
be clearly indicated.

\subsection{Neumann and gradient flow problems}
\label{section:neumann-gradient-flow-problems}

\begin{definition}[Down-$H^1$ functions]
    \label{def:down-h1-functions}
    We define the set $\cU$ of \emph{down-$H^1$ functions} as
    \[
        \cU \define \{ u \in L^2(I) \stcolon \cR(u) \ne \emptyset \} \,,
    \]
    where $\cR(u)$ is the set of admissible \emph{representations} for $u$ as follows:
    \begin{align*}
        \cR(u) \define \Big\{
            &(\onef{u}, \twof{u}) \in L^2(I) \times H^1(I) \stcolon \\
            &\onef{u} \text{ is nondecreasing, } \twof{u} \text{ is nonincreasing, }
            \int_I \twof{u} \odif x = 0, \text{ and }
            u = \onef{u} + \twof{u} \text{ \almev}
            \Big\} \,.
    \end{align*}
\end{definition}

\begin{remark}
    The properties ``nondecreasing'' and ``nonincreasing'' in the definition above technically do
    not apply directly to members of $L^2(I)$ and $H^1(I)$, which are equivalence classes of \almev
    equal functions. Rather, we mean is that $\onef{u}$ is a member of $L^2(I)$ which has a
    representative function $I \to \bR$ that is nondecreasing, and likewise for $\twof{u}$. More
    generally, when $f$ is an object of any space that identifies \almev equal functions, we abuse
    language and write ``$f$ is nonincreasing'' (resp.\ nondecreasing) to mean that $f$ has a
    nonincreasing (resp.\ nondecreasing) representative.
\end{remark}

\begin{remark}
    Is is easy to check that $H^1(I) \subset \cU$, since for any $u \in H^1(I)$ we may separate the
    positive and negative parts of $\partial_x u$ into $\onef{u}$ and $\twof{u}$, respectively, and
    shift $\onef{u}$ and $\twof{u}$ by appropriate constants so that they satisfy the conditions of
    \cref{def:down-h1-functions}. Then, since $H^1(I)$ is dense in $L^2(I)$, so is $\cU$.
\end{remark}

\begin{definition}[Directed Dirichlet energy]
    For each $u \in \cU$ and $(\onef{u}, \twof{u}) \in \cR(u)$, define
    \[
        \cD^-(\onef{u}, \twof{u})
        \define \frac{1}{2} \int_I \left( \partial_x \twof{u} \right)^2 \odif x \,.
    \]
\end{definition}

\begin{definition}[Optimal representations]
    For each $u \in \cU$, define $\cR^*(u) \subseteq \cR(u)$ as the set of representations
    minimizing $\cD^-$, \ie
    \[
        \cR^*(u) \define \left\{
            (\onef{u}, \twof{u}) \in \cR(u) \stcolon
            \cD^-(\onef{u}, \twof{u})
            = \inf_{\cR(u)} \cD^- \right\} \,.
    \]
\end{definition}

\begin{proposition}[Existence and uniqueness of optimal representation]
    \label{prop:optimal-representations}
    For all $u \in \cU$, $\cR^*(u)$ contains exactly one element.
\end{proposition}
\begin{proof}[Proof of existence]
    Let $u \in \cU$. Let $E \define L^2(I) \oplus H^1(I)$, where $\oplus$ denotes the direct sum of
    Hilbert spaces, so that $E$ is also a Hilbert space (and hence a reflexive Banach space). Note
    that $\cR(u) \subseteq E$ by definition. Our goal is to show that $C = \cR(u)$ and $f = \cD^-$
    satisfy the conditions of \cref{lemma:convex-optimization}.

    It is clear that $\cR(u)$ is nonempty, because $u \in \cU$. It is immediate to verify that
    $\cR(u)$ is convex as well. We now check that $\cR(u)$ is closed in $E$. Let $\big((\onef{u_n},
    \twof{u_n})\big)_{n \in \bN}$ be any sequence in $\cR(u)$ that converges in $E$ to some
    $(\onef{u}, \twof{u}) \in E$; we would like to show that $(\onef{u}, \twof{u}) \in \cR(u)$.
    First, since $(\onef{u_n}, \twof{u_n}) \to (\onef{u}, \twof{u})$ in $E$, we have $\onef{u_n} \to
    \onef{u}$ in $L^2(I)$ and $\twof{u_n} \to \twof{u}$ in $H^1(I)$ and hence in $L^2(I)$, so that
    \begin{align*}
        &\|u - (\onef{u} + \twof{u})\|_{L^2(I)} \\
        &\qquad= \lim_{n \to \infty} \left\|
            (u - (\onef{u_n} + \twof{u_n})) + (\onef{u_n} - \onef{u}) + (\twof{u_n} - \twof{u})
            \right\|_{L^2(I)} \\
        &\qquad\le \lim_{n \to \infty} \left\| u - (\onef{u_n} + \twof{u_n}) \right\|_{L^2(I)}
            + \lim_{n \to \infty} \left\| \onef{u_n} - \onef{u} \right\|_{L^2(I)}
            + \lim_{n \to \infty} \left\| \twof{u_n} - \twof{u} \right\|_{L^2(I)} \\
        &\qquad= 0 \,,
    \end{align*}
    where the last equality used the fact that $u = \onef{u_n} + \twof{u_n}$ \almev for every $n$
    (since $(\onef{u_n}, \twof{u_n}) \in \cR(u)$) and the two convergence observations above. Thus
    $u = \onef{u} + \twof{u}$ in $L^2(I)$ and hence almost everywhere. Moreover, by
    \cref{lemma:weak-limit-monotonic}, $\onef{u}$ is nondecreasing and $\twof{u}$ is nonincreasing.
    Finally, since $\int_I \twof{u_n} \odif x = 0$ for each $n$ and $\twof{u_n} \to \twof{u}$ in
    $L^2(I)$, it follows that $\int_I \twof{u} \odif x = 0$ as well. We conclude that $(\onef{u},
    \twof{u}) \in \cR(u)$, and thus $\cR(u)$ is closed.

    Now, clearly $\cD^-$ is proper since it is finitely valued by definition. It is also continuous
    (and thus lower semicontinuous) because for any $(\onef{u_n}, \twof{u_n}) \to (\onef{u},
    \twof{u})$ in $E$, we have $\twof{u_n} \to \twof{u}$ in $H^1(I)$ and hence $\int_I \left(
    \partial_x \twof{u_n} \right)^2 \odif x \to \int_I \left( \partial_x \twof{u} \right)^2 \odif
    x$, so $\cD^-(\onef{u_n}, \twof{u_n}) \to \cD^-(\onef{u}, \twof{u})$. It is also immediate to
    check that that $\cD^-$ is convex.

    Finally, we need to show that $\cD^-$ is coercive in the sense of
    \cref{lemma:convex-optimization}. Let $\big((\onef{u_n} \twof{u_n})\big)_{n \in \bN}$ be a
    sequence in $\cR(u)$ such that $\left\| (\onef{u_n}, \twof{u_n}) \right\|_E^2 =
    \|\onef{u_n}\|_{L^2(I)}^2 + \|\twof{u_n}\|_{H^1(I)}^2 \to \infty$. We claim that $\|\partial_x
    \twof{u_n}\|_{L^2(I)} \to \infty$. Suppose for a contradiction that this false. There are two
    cases. First, suppose $\|\twof{u_n}\|_{H^1(I)} \to \infty$. Then necessarily
    $\|\twof{u_n}\|_{L^2(I)} \to \infty$, but since each $\twof{u_n}$ has mean zero, the
    Poincaré-Wirtinger inequality gives that
    \[
        \|\partial_x \twof{u_n}\|_{L^2(I)} \ge \frac{1}{C} \|\twof{u_n}\|_{L^2(I)}
    \]
    for some constant $C > 0$. Thus $\|\partial_x \twof{u_n}\|_{L^2(I)} \to \infty$, a
    contradiction. In the second case, $\|\twof{u_n}\|_{H^1(I)}$ remains bounded, so we must have
    $\|\onef{u_n}\|_{L^2(I)} \to \infty$. But since $u = \onef{u_n} + \twof{u_n}$ in $L^2(I)$, the
    reverse triangle inequality implies that $\|\twof{u_n}\|_{L^2(I)} \ge \|\onef{u_n}\|_{L^2(I)} -
    \|u\|_{L^2(I)} \to \infty$, so $\|\twof{u_n}\|_{L^2(I)} \to \infty$. Then again
    Poincaré-Wirtinger implies that $\|\partial_x \twof{u_n}\|_{L^2(I)} \to \infty$, so the claim
    holds. But since $\cD^-\left( \onef{u_n}, \twof{u_n} \right) = \frac{1}{2} \|\partial_x
    \twof{u_n}\|_{L^2(I)}^2$, we have $\cD^-\left( \onef{u_n}, \twof{u_n} \right) \to +\infty$, so
    $\cD^-$ is coercive. Thus $\cD^-$ achieves its minimum on $\cR(u)$ by
    \cref{lemma:convex-optimization}, so $\cR^*(u)$ is nonempty.
\end{proof}
\begin{proof}[Proof of uniqueness]
    We show uniqueness using strict convexity. Suppose $(\onef{u}, \twof{u}), (\onef{v}, \twof{v})
    \in \cR^*(u)$, and suppose for a contradiction that $(\onef{u}, \twof{u}) \ne (\onef{v},
    \twof{v})$ in $E$ (in the notation from the previous part of the proof; note that equality in
    $E$ is equivalent to equality as tuples in $L^2(I) \times H^1(I)$).

    We first claim that $\partial_x \twof{u} \ne \partial_x \twof{v}$ in $L^2(I$). Indeed, suppose
    $\partial_x \twof{u} = \partial_x \twof{v}$ in $L^2(I)$. Note that $\twof{u} \ne \twof{v}$ in
    $L^2(I)$, because otherwise we would have $\twof{u} = \twof{v}$ in $H^1(I)$ and $\onef{u} = u -
    \twof{u} = u - \twof{v} = \onef{v}$ in $L^2(I)$, and thus $(\onef{u}, \twof{u}) = (\onef{v},
    \twof{v})$ in $E$, a contradiction. Now, since $\partial_x \twof{u} = \partial_x \twof{v}$ and
    $\twof{u} \ne \twof{v}$ in $L^2(I)$, we conclude that $\twof{u} = \twof{v} + C$ in $L^2(I)$ for
    some constant $C \ne 0$. But this contradicts the fact that $\int_I \twof{u} \odif x = \int_I
    \twof{v} \odif x = 0$, which holds by the definition of $\cR(u)$. Thus $\partial_x \twof{u} \ne
    \partial_x \twof{v}$ in $L^2(I)$.

    \sloppy
    Define the function $f : L^2(I) \to \bR_{\ge 0}$ by $f(w) \define \frac{1}{2} \|w\|_{L^2(I)}^2$,
    which is strictly convex, and note that $\cD^-(\onef{r}, \twof{r}) = f(\partial_x \twof{r})$ for
    all $r \in \cR(u)$. Now, the element $(\onef{z}, \twof{z}) \define \frac{1}{2} (\onef{u},
    \twof{u}) + \frac{1}{2} (\onef{v}, \twof{v})$ is in $\cR(u)$ by convexity of that set, and
    satisfies $\cD^-(\onef{z}, \twof{z}) = f\left( \frac{1}{2} \partial_x \twof{u} + \frac{1}{2}
    \partial_x \twof{v} \right) < \frac{1}{2} f(\partial_x \twof{u}) + \frac{1}{2} f(\partial_x
    \twof{v}) = \frac{1}{2} \cD^-(\onef{u}, \twof{u}) + \frac{1}{2} \cD^-(\onef{v}, \twof{v})$,
    contradicting the fact that $(\onef{u}, \twof{u}), (\onef{v}, \twof{v}) \in \cR^*(u)$. Thus
    $(\onef{u}, \twof{u}) = (\onef{v}, \twof{v})$ in $E$, as needed.
\end{proof}

\begin{lemma}[See \eg {\cite[Corollary 3.23]{Bre11}}]
    \label{lemma:convex-optimization}
    Let $E$ be a reflexive Banach space and let $C \subseteq E$ be nonempty, closed and convex.
    Suppose $f : C \to (-\infty, +\infty]$ is convex, proper, lower semicontinuous, and coercive in
    the sense that for every sequence $(x_n)_{n \in \bN}$ in $C$ with $\|x_n\|_E \to \infty$, we
    have $f(x_n) \to +\infty$. Then there exists $x^* \in C$ satisfying $f(x^*) = \inf_{x \in C}
    f(x)$.
\end{lemma}

Owing to \cref{prop:optimal-representations}, we may define for each $u \in \cU$ a canonical
representation and directed Dirichlet energy (by a slight abuse of notation):

\begin{definition}[Canonical representation]
    \label{def:canonical-representation}
    For each $u \in \cU$, the \emph{canonical representation} of $u$ is the unique element of
    $\cR^*(u)$, denoted $(\upf{u}, \downf{u})$. We define the \emph{directed Dirichlet energy} of
    $u$ by $\cD^-(u) \define \cD^-(\upf{u}, \downf{u})$.
\end{definition}

The following definition, which extends the directed Dirichlet energy functional to all of $L^2(I)$
by assigning $+\infty$ to functions outside of $\cU$, will enable us to find and study a solution to
our PDE using tools from the theory of maximal monotone operators and gradient flows. Our main
reference for this theory is \cite{Bre73}.

\begin{definition}[Energy functional]
    \label{def:energy-functional}
    Define the functional $\cE^- : L^2(I) \to [0, +\infty]$ by
    \[
        \cE^-(u) \define \begin{cases}
            \cD^-(u) & \text{if } u \in \cU \\
            +\infty  & \text{otherwise.}
        \end{cases}
    \]
\end{definition}

\begin{proposition}
    The functional $\cE^- : L^2(I) \to [0, +\infty]$ is convex, proper and lower semicontinuous.
\end{proposition}
\begin{proof}
    Properness is trivial, since \eg $0 \in \cU$ and $\cE^-(0) = 0$. Convexity also follows
    easily from the convexity of $\cU$ and $\cD^-$.

    It remains to show lower semicontinuity, and since $L^2(I)$ is a metric space, it suffices to
    show sequential lower semicontinuity. Let $u \in L^2(I)$, and let $(u_n)_{n \in \bN}$ be a
    sequence in $L^2(I)$ converging to $u$ in $L^2(I)$. We must show that
    \begin{equation}
        \label{eq:lsc}
        \cE^-(u) \lequestion \liminf_{n \to \infty} \cE^-(u_n) \,.
    \end{equation}
    The only relevant case is when the RHS above is finite, so suppose there exists a subsequence
    $(u_{n_k})_{k \in \bN}$ such that $\lim_{k \to \infty} \cE^-(u_{n_k}) = A < +\infty$. By
    extracting a subsequence if necessary, we may assume that $\cE^-(u_{n_k}) < +\infty$, and thus
    $u_{n_k} \in \cU$, for every $k$.

    \sloppy We claim that $\big( (\upf{u_{n_k}}, \downf{u_{n_k}}) \big)_k$ is bounded as a sequence
    in $L^2(I) \oplus H^1(I)$. First, since $\lim_{k \to \infty} \cE^-(u_{n_k}) = A$, we have that
    $\left( \int_I (\partial_x \downf{u_{n_k}})^2 \odif x \right)_k$ is bounded. Then, since every
    $u_{n_k} \in \cU$ and hence $\int_I \downf{u_{n_k}} \odif x = 0$, the Poincaré-Wirtinger
    inequality implies that
    \[
        \int_I (\downf{u_{n_k}})^2 \odif x \le C \int_I (\partial_x \downf{u_{n_k}})^2 \odif x
    \]
    for some $C > 0$, and hence $(\downf{u_{n_k}})_k$ is bounded in $H^1(I)$. Now, since $u_{n_k}
    \to u$ in $L^2(I)$, we have that $\left( \|u_{n_k}\|_{L^2(I)} \right)_k$ is bounded, and since
    $\left( \|\downf{u_{n_k}}\|_{L^2(I)} \right)_k$ is bounded as a consequence of boundedness in
    $H^1(I)$, it follows that $\left( \|\upf{u_{n_k}}\|_{L^2(I)} \right)_k = \left( \|u_{n_k} -
    \downf{u_{n_k}}\|_{L^2(I)} \right)_k$ is also bounded, that is, $(\upf{u_{n_k}})_k$ is bounded
    in $L^2(I)$. Hence $\big( (\upf{u_{n_k}}, \downf{u_{n_k}}) \big)_k$ is bounded in $L^2(I) \oplus
    H^1(I)$ as claimed.

    Since $L^2(I) \oplus H^1(I)$ is a Hilbert space, we conclude that $\big( (\upf{u_{n_k}},
    \downf{u_{n_k}}) \big)_k$ has a weakly convergent subsequence, which we denote by $\big(
    (\upf{u_{n_{k_m}}}, \downf{u_{n_{k_m}}}) \big)_m$. Then there exist $\onef{v} \in L^2(I)$ and
    $\twof{v} \in H^1(I)$ such that $\upf{u_{n_{k_m}}} \weakto \onef{v}$ weakly in $L^2(I)$ and
    $\downf{u_{n_{k_m}}} \weakto \twof{v}$ weakly in $H^1(I)$. So letting $v \define \onef{v} +
    \twof{v}$, we conclude that $u_{n_{k_m}} \weakto v$ weakly in $L^2(I)$; indeed, for every $f \in
    L^2(I)$, we have
    \[
        \inp{u_{n_{k_m}}}{f} = \inp{\upf{u_{n_{k_m}}} + \downf{u_{n_{k_m}}}}{f}
        = \inp{\upf{u_{n_{k_m}}}}{f} + \inp{\downf{u_{n_{k_m}}}}{f}
        \to \inp{\onef{v}}{f} + \inp{\twof{v}}{f}
        = \inp{v}{f} \,.
    \]
    We now claim that $v \in \cU$. Indeed, by \cref{lemma:weak-limit-monotonic}, $\onef{v}$ is
    nondecreasing and $\twof{v}$ is nonincreasing, and moreover $\int_I \twof{v} \odif x = 0$ by
    weak convergence; hence $(\onef{v}, \twof{v}) \in \cR(v)$. Then, since $u_{n_{k_m}} \weakto v$
    weakly in $L^2(I)$ and $u_{n_{k_m}} \to u$ in $L^2(I)$, we conclude that $u = v$, so $u \in \cU$
    with $(\onef{v}, \twof{v}) \in \cR(u)$. In particular, this means that $\cE^-(u) \le
    \cD^-(\onef{v}, \twof{v}) = \frac{1}{2} \int_I (\partial_x \twof{v})^2 \odif x$.

    Now, the fact that $\downf{u_{n_{k_m}}} \weakto \twof{v}$ weakly in $H^1(I)$ implies that
    $\|\twof{v}\|_{H^1(I)} \le \liminf_{m \to \infty} \|\downf{u_{n_{k_m}}}\|_{H^1(I)}$. Moreover,
    since $H^1(I)$ embeds compactly into $L^2(I)$ by the Rellich-Kondrachov theorem (see \eg
    \cite[Theorem~9.16]{Bre11}), we have that $\downf{u_{n_{k_m}}} \to \twof{v}$ in $L^2(I)$, so
    $\|\downf{u_{n_{k_m}}}\|_{L^2(I)}^2 \to \|\twof{v}\|_{L^2(I)}^2$. Therefore we obtain
    \begin{align*}
        \liminf_{m \to \infty} \cE^-(u_{n_{k_m}})
        &= \frac{1}{2} \liminf_{m \to \infty} \int_I (\partial_x \downf{u_{n_{k_m}}})^2 \odif x
        = \frac{1}{2} \liminf_{m \to \infty} \left( \|\downf{u_{n_{k_m}}}\|_{H^1(I)}^2
            - \|\downf{u_{n_{k_m}}}\|_{L^2(I)}^2 \right) \\
        &= \left( \frac{1}{2} \liminf_{m \to \infty} \|\downf{u_{n_{k_m}}}\|_{H^1(I)}^2 \right)
            - \frac{1}{2} \|\twof{v}\|_{L^2(I)}^2 \\
        &\ge \frac{1}{2} \left( \|\twof{v}\|_{H^1(I)}^2 -  \|\twof{v}\|_{L^2(I)}^2 \right)
        = \frac{1}{2} \int_I (\partial_x \twof{v})^2 \odif x \\
        &\ge \cE^-(u) \,,
    \end{align*}
    and thus \eqref{eq:lsc} holds.
\end{proof}

\begin{definition}[Static Neumann problem]
    \label{def:static-neumann-problem}
    Let $u \in \cU$ and $z \in L^2(I)$. We say that $u, z$ form a \emph{weak solution} to the
    \emph{static Neumann problem}
    \begin{equation}
        \label{eq:static-neumann-problem}
        \begin{cases}
            z = \partial_x \partial_x \downf{u} & \text{in } I \\
            \partial_x \downf{u} = 0            & \text{on } \{0, 1\}
        \end{cases}
    \end{equation}
    if for all $\phi \in H^1(I)$ we have
    \begin{equation}
        \label{eq:weak-solution}
        \int_I z \phi \odif x = -\int_I (\partial_x \downf{u}) (\partial_x \phi) \odif x
        \,.
    \end{equation}
\end{definition}

\begin{remark}
    The equation $z = \partial_x \partial_x \downf{u}$ in \eqref{eq:static-neumann-problem} is also
    called the Poisson equation, while $\partial_x \downf{u} = 0$ on $\{0, 1\}$ is called the
    (homogeneous) Neumann boundary condition.
\end{remark}

\begin{definition}[Nice evolution function]
    \label{def:nice-evolution-functions}
    Let $\bm{u} \in C([0, +\infty); L^2(I))$ and $\bm{u'} : (0, +\infty) \to L^2(I)$. We say that
    $\bm{u}$ is a \emph{nice evolution function} (with weak derivative $\bm{u'}$) if
    1)~$\bm{u}(t) = \bm{u}(0) + \int_{(0, t)} \bm{u'}(s) \odif s$ for all $t > 0$;
    2)~$\bm{u'}_{(\delta, +\infty)} \in L^\infty(\delta, +\infty; L^2(I))$ for all $\delta > 0$,
    where $\bm{u'}_J$ denotes the restriction of $\bm{u'}$ to domain $J \subseteq (0, +\infty)$; and
    3)~$\bm{u'}_{(0, \delta)} \in L^2(0, \delta; L^2(I))$ for all $\delta > 0$.
\end{definition}

\begin{remark}
    \cref{def:nice-evolution-functions} captures functions that are Lipschitz away from zero, and
    absolutely continuous with square-integrable weak derivative near zero.
\end{remark}

\begin{definition}[Neumann evolution problem]
    \label{def:neumann-evolution}
    Let $u_0 \in \cU$ and let $\bm{u} \in C([0, +\infty); L^2(I))$ be a nice evolution function. We
    say $\bm{u}$ is a \emph{weak solution} to the \emph{Neumann evolution problem} with initial
    state (or initial data) $u_0$ if 1)~$\bm{u}(0) = u_0$; 2)~$\bm{u}(t) \in \cU$ for all $t > 0$;
    and 3)~$\bm{u}(t), \bm{u'}(t)$ form a weak solution to the static Neumann problem for \almev $t
    > 0$.
\end{definition}

Before introducing the gradient flow problem, we need some notation. Let $\varphi : L^2(I) \to [0,
+\infty]$ be a convex, proper, lower semicontinuous function. We write $D(\varphi) \define \{ u \in
L^2(I) : \varphi(u) < +\infty \}$ for the \emph{domain} of $\varphi$. For any $u \in L^2(I)$, we
write define the \emph{subdifferential} of $\varphi$ at $u$ as
\[
    \partial \varphi(u) \define \{ z \in L^2(I) :
        \forall v \in L^2(I) \,.\, \varphi(v) \ge \varphi(u) + \inp{z}{v-u} \} \,,
\]
and we write $D(\partial \varphi) \define \{ u \in L^2(I) : \partial \varphi(u) \ne \emptyset \}$
for the \emph{domain} of $\partial \varphi$. Using the fact that $\varphi$ is proper, it is easy to
check that $D(\partial \varphi) \subseteq D(\varphi)$. It is standard that $\partial \varphi$ is a
\emph{maximal monotone operator}; we will not use (the definition of) this property explicitly, but
rather rely on the theory of such operators as presented in \cite{Bre73}.

We can now define the gradient flow problem:

\begin{definition}[Gradient flow problem]
    \label{def:gradient-flow}
    Let $u_0 \in \cU$ and let $\bm{u}$ be a nice evolution function. We say $\bm{u}$ is a
    \emph{solution} to the \emph{gradient flow problem} with initial state (or initial data) $u_0$
    if 1)~$\bm{u}(0) = u_0$; 2)~$\bm{u}(t) \in D(\partial \cE^-)$ for all $t > 0$; and
    3)~$\bm{u'}(t) \in -\partial \cE^-(\bm{u}(t))$ for \almev $t > 0$.
\end{definition}

Note that in both \cref{def:neumann-evolution,def:gradient-flow}, the pointwise condition at $t=0$
makes sense by the requirement that $\bm{u}$ be continuous.

The theory of maximal monotone operators and gradient flows immediately yields that the gradient
flow problem has a unique solution, as follows:

\begin{proposition}
    \label{res:solution}
    Let $u_0 \in \cU$. Then there exists a unique solution $\bm{u} \in C([0, +\infty); L^2(I))$ to
    the gradient flow problem. Moreover, for all $t > 0$ we have
    \[
        \cE^-(u_0) - \cE^-(\bm{u}(t)) = \int_{(0,t)} \|\bm{u'}(s)\|_{L^2(I)}^2 \odif s \,.
    \]
    For $u_0, v_0 \in \cU$, the corresponding solutions $\bm{u}, \bm{v}$ have $\|\bm{u}(t) -
    \bm{v}(t)\|_{L^2(I)} \le \|u_0 - v_0\|_{L^2(I)}$ for all $t > 0$.
\end{proposition}
\begin{proof}
    This is a direct application of \cite[Theorems~3.1~and~3.2 and Proposition~3.1]{Bre73}.
\end{proof}

\paragraph*{Directed heat semigroup.}
As described in \cite{Bre73}, the nonexpansive property of solutions $\bm{u}$ allows us to define an
operator $P_t : L^2(I) \to L^2(I)$ such that $P_t u = \bm{u}(t)$ for each initial state $u \in \cU$
with corresponding solution $\bm{u}$, and with $P_t u$ defined by continuous extension when $u \in
L^2(I) \setminus \cU$ (recall that $\cU$ is dense in $L^2(I)$). Then for every $u \in L^2(I)$ and $t
> 0$ it holds that $P_t u \in D(\partial \cE^-)$, and $(P_t)_{t \ge 0}$ forms a nonexpansive
semigroup in $L^2(I)$, \ie
\begin{enumerate}
    \item $P_{t+s} = P_t P_s$ for all $t, s \in \bR_{\ge 0}$, with $P_0$ the identity;
    \item $P_t u \to u$ in $L^2(I)$ as $t \downarrow 0$, for all $u \in L^2(I)$; and
    \item $\|P_t u - P_t v\|_{L^2(I)} \le \|u - v\|_{L^2(I)}$ for all $u, v \in L^2(I)$ and $t \in
        \bR_{\ge 0}$.
\end{enumerate}
The last property also implies that, for each $t \ge 0$, $P_t : L^2(I) \to L^2(I)$ is continuous. We
say that $P_t$ is the semigroup \emph{generated} by $-\partial \cE^-$.

It turns out that the solution to the gradient flow problem is also a weak solution to the Neumann
evolution problem, as the following results show.

\begin{proposition}
    \label{prop:static-solution-gradient-to-neumann}
    Let $u \in D(\partial \cE^-)$ and $z \in -\partial \cE^-(u)$. Then $u, z$ form a weak solution
    to the static Neumann problem.
\end{proposition}
\begin{proof}
    Note that we have $u \in D(\partial \cE^-) \subseteq D(\cE^-) = \cU$. Let $\phi \in H^1(I)$; we
    will prove that \eqref{eq:weak-solution} holds. Let $\alpha \ne 0$, let $\psi_\alpha \define
    \alpha \phi$ and let $v_\alpha \define u + \psi_\alpha$. Since $-z \in \partial \cE^-(u)$, we
    have
    \begin{equation}
        \label{eq:subgradient-v-alpha}
        \cE^-(v_\alpha) \ge \cE^-(u) + \inp{-z}{v_\alpha-u}
        = \cE^-(u) + \inp{-z}{\psi_\alpha} \,.
    \end{equation}
    Now, the fact that $u \in \cU$ implies that $\cE^-(u) = \cD^-(u)$. We now claim that $v_\alpha
    \in \cU$ as well, which will imply that $\cE^-(v_\alpha) = \cD^-(v_\alpha)$. In fact, we
    establish the following:

    \begin{claim}
        We have $v_\alpha \in \cU$. Moreover, $\cD^-(v_\alpha) \le \frac{1}{2} \int_I (\partial_x
        \downf{u} + \partial_x \psi_\alpha)^2 \odif x$.
    \end{claim}
    \begin{proof}
        We define a nondecreasing function $\onef{v} \in L^2(I)$ and a nonincreasing function
        $\twof{v} \in H^1(I)$ as follows: for each $x \in I$,
        \begin{align*}
            \onef{v}(x) &\define \upf{u}(x) + \psi_\alpha(0) + \int_{(0, x)}
                    \big( \partial_x \psi_\alpha(y)
                            - \abs{\partial_x \downf{u}(y)} \big)^+ \odif y
                \qquad \text{and} \\
            \twof{v}(x) &\define \downf{u}(x) + \int_{(0, x)}
                    \big( \partial_x \psi_\alpha(y)
                            \land \abs{\partial_x \downf{u}(y)} \big) \odif y \,.
        \end{align*}
        It is clear that $\onef{v}, \twof{v} \in L^2(I)$ with $\onef{v}$ nondecreasing. Since
        $\downf{u} \in H^1(I)$ by definition of $\cR(u)$, while $\psi_\alpha \in H^1(I)$ because
        $\phi \in H^1(I)$, we also obtain that $\twof{v} \in H^1(I)$. To see that $\twof{v}$ is
        nonincreasing, note that for \almev $x \in I$ we have
        \[
            \partial_x \twof{v}(x)
            = \partial_x \downf{u}(x)
                + \big( \partial_x \psi_\alpha(x) \land \abs{\partial_x \downf{u}(x)} \big)
            \le 0 \,.
        \]
        Observe also that $\onef{v} + \twof{v} = u + \psi_\alpha = v_\alpha$. Finally, by
        translating $\onef{v}$ and $\twof{v}$ by a constant if necessary, we can ensure that $\int_I
        \twof{v} \odif x = 0$ without invalidating the other properties. Thus $v_\alpha \in \cU$.

        To show the second part of the claim, we note that for \almev $x \in I$,
        \[
            \Big| \underbrace{\partial_x \twof{v}(x)}_{\le 0} \Big|
            = \Big| \underbrace{\partial_x \downf{u}(x)}_{\le 0}
                + \underbrace{\big( \partial_x \psi_\alpha(x) \land
                    \abs{\partial_x \downf{u}(x)} \big)}_{\le \abs{\partial_x \downf{u}(x)}} \Big|
            \le \Big| \partial_x \downf{u}(x) + \partial_x \psi_\alpha(x) \Big| \,,
        \]
        where the inequality follows from inspecting the cases $\partial_x \psi_\alpha(x) \le
        \abs{\partial_x \downf{u}(x)}$ and $\partial_x \psi_\alpha(x) > \abs{\partial_x
        \downf{u}(x)}$. We conclude that
        \[
            \cD^-(v_\alpha) \le \cD^-(\onef{v}, \twof{v})
            = \frac{1}{2} \int_I (\partial_x \twof{v})^2 \odif x
            \le \frac{1}{2} \int_I (\partial_x \downf{u}(x) + \partial_x \psi_\alpha(x))^2 \odif x
            \,. \qedhere
        \]
    \end{proof}

    \noindent
    Now, putting together \eqref{eq:subgradient-v-alpha} and the claim, we get
    \[
        \frac{1}{2} \int_I (\partial_x \downf{u})^2 \odif x - \int_I z \psi_\alpha \odif x
        \le \cD^-(v_\alpha) \le \frac{1}{2} \int_I (\partial_x \downf{u} + \partial_x \psi_\alpha)^2
        \odif x \,.
    \]
    Simplifying and substituting $\psi_\alpha = \alpha \phi$, we obtain
    \[
        \alpha \int_I z \phi \odif x
        \ge -\frac{1}{2} \alpha^2 \int_I (\partial_x \phi)^2 \odif x
            - \alpha \int_I (\partial_x \downf{u})(\partial_x \phi) \odif x \,.
    \]
    Taking the limits $\alpha \to 0^+$ and $\alpha \to 0^-$, we conclude that
    \[
        \int_I z \phi \odif x = -\int_I (\partial_x \downf{u})(\partial_x \phi) \odif x \,,
    \]
    which is \eqref{eq:weak-solution} as needed.
\end{proof}

\begin{corollary}
    \label{cor:evolution-solution-gradient-to-neumann}
    Let $u_0 \in \cU$ and suppose $\bm{u}$ is a solution to the gradient flow problem. Then $\bm{u}$
    is also a weak solution to the Neumann evolution problem.
\end{corollary}
\begin{proof}
    The condition $\bm{u}(0) = u_0$ holds by definition of solution to the gradient flow problem,
    and for all $t > 0$ we have $\bm{u}(t) \in D(\partial \cE^-) \subseteq D(\cE^-) = \cU$. Finally,
    for \almev $t > 0$ we have $\bm{u'}(t) \in -\partial \cE^-(\bm{u}(t))$, which by
    \cref{prop:static-solution-gradient-to-neumann} implies that $\bm{u}(t), \bm{u'}(t)$ form a
    solution to the static Neumann problem. Hence $\bm{u}$ is a weak solution to the Neumann
    evolution problem.
\end{proof}

\subsection{Exponential decay of directed Dirichlet energy}
\label{section:energy-decay}

The following elliptic regularity result is standard, and essentially shows that weak solutions to
the static Neumann problem are in fact strong solutions whose regularity is two degrees higher than
that of $z$; in particular, even if we only have $z \in L^2(I)$, we gain one degree of regularity by
obtaining $\downf{u} \in H^2(I)$ when we only assumed $\downf{u} \in H^1(I)$. Recall that $H^0(I)$
is the same as $L^2(I)$.

\begin{lemma}[Elliptic regularity; see {\cite[Chapter IV Section 2, Theorems 3 and 4]{Mik78}}]
    \label{lemma:elliptic-regularity}
    Let $k \ge 0$ be an integer and let $u \in \cU, z \in H^k(I)$ form a weak solution to the static
    Neumann problem
    \[
        \begin{cases}
            z = \partial_x \partial_x \downf{u} & \text{in } I \\
            \partial_x \downf{u} = 0            & \text{on } \{0, 1\} \,.
        \end{cases}
    \]
    Then $\downf{u} \in H^{k+2}(I)$, and moreover $\partial_x \downf{u} = 0$ on $\{0,1\}$ and $z =
    \partial_x \partial_x \downf{u}$ \almev in $I$.
\end{lemma}

We now give our energy decay result.

\begin{proposition}
    \label{prop:exponential-decay}
    There exists a constant $K > 0$ such that the following holds. Let $u \in \cU$. Then for all $t
    > 0$,
    \[
        \cE^-(P_t u) \le e^{-Kt} \cE^-(u) \,.
    \]
\end{proposition}
\begin{proof}
    Let $u \in \cU$ and let $\bm{u}(t) = P_t u$ be the corresponding solution to the gradient flow
    problem. Recall that \cref{res:solution} gives, for all $t > 0$,
    \[
        \cE^-(\bm{u}(t)) = \cE^-(u) + \int_{(0,t)} -\|\bm{u'}(s)\|_{L^2(I)}^2 \odif s \,.
    \]
    It follows that $t \mapsto \cE^-(\bm{u}(t))$ is absolutely continuous on every interval $[0, T]$
    with weak derivative $\partial_t \cE^-(\bm{u}(t)) = -\|\bm{u'}(t)\|_{L^2(I)}^2$. Moreover for
    \almev $t > 0$ we have that $\bm{u}(t), \bm{u'}(t)$ form a weak solution to the static Neumann
    problem, and thus $\bm{u'}(t) = \partial_x \partial_x \downf{\bm{u}(t)}$ \almev in $I$ and
    $\partial_x \downf{\bm{u}(t)} = 0$ on $\{0,1\}$ by \cref{lemma:elliptic-regularity}. Therefore
    the Poincaré inequality (for zero-on-the-boundary functions) yields
    \[
        \cE^-(\bm{u}(t))
        = \frac{1}{2} \|\partial_x \downf{\bm{u}(t)}\|_{L^2(I)}^2
        \le \frac{1}{2} C \|\partial_x \partial_x \downf{\bm{u}(t)}\|_{L^2(I)}^2
        = \frac{C}{2} \|\bm{u'}(t)\|_{L^2(I)}^2
    \]
    for some constant $C > 0$. It follows that for \almev $t > 0$, we have
    \[
        \partial_t \cE^-(\bm{u}(t))
        = -\|\bm{u'}(t)\|_{L^2(I)}^2
        \le -\frac{2}{C} \cE^-(\bm{u}(t)) \,,
    \]
    which implies that for all $t > 0$,
    \[
        \cE^-(\bm{u}(t)) \le e^{-2t/C} \cE^-(u) \,. \qedhere
    \]
\end{proof}

The exponential decay of $t \mapsto \cE^-(P_t u)$ allows us to find a Cauchy sequence in $(P_t u)_{t
\ge 0}$, and thus establish its strong convergence to some limit in $L^2(I)$. Later on, we will say
more about this limit by reasoning about the weak convergence of $P_t u$ (to the same limit).

\begin{lemma}[Cauchy sequence]
    \label{lemma:cauchy}
    Let $u \in \cU$, and let $(u_k)_{k \in \bN}$ be given by $u_k \define P_k u$. Then $(u_k)_k$ is
    Cauchy as a sequence in $L^2(I)$. As a consequence, $u_k \to u^*$ in $L^2(I)$ for some $u^* \in
    L^2(I)$.
\end{lemma}
\begin{proof}
    The existence of a strong limit from the Cauchy property follows from the fact that $L^2(I)$ is
    a complete normed space. Let us now establish the Cauchy property. Fix any $n > m$ in $\bN$.
    Then
    {\allowdisplaybreaks
        \begin{align*}
            \|u_n - u_m\|_{L^2(I)}
            &= \left\| \int_{(m,n)} \bm{u'}(t) \odif t \right\|_{L^2(I)}
                & \text{(Absolute continuity)} \\
            &\le \sum_{j=m}^{n-1} \int_{(j,j+1)} \|\bm{u'}(t)\|_{L^2(I)} \odif t
                & \text{(Triangle inequality)} \\
            &\le \sum_{j=m}^{n-1}
                \left( \int_{(j,j+1)} \|\bm{u'}(t)\|_{L^2(I)}^2 \odif t \right)^{1/2}
                & \text{(Jensen's inequality)} \\
            &= \sum_{j=m}^{n-1} \left( \cE^-(u_j) - \cE^-(u_{j+1}) \right)^{1/2}
                & \text{(\cref{res:solution})} \\
            &\le \sum_{j=m}^\infty \cE^-(u_j)^{1/2}
                & \text{($\cE^-$ is nonnegative)} \\
            &\le \sum_{j=m}^\infty \left( e^{-Kj} \cE^-(u) \right)^{1/2}
                & \text{(\cref{prop:exponential-decay})} \\
            &= \cE^-(u)^{1/2} \frac{e^{-Km/2}}{1 - e^{-K/2}}
                & \text{(Geometric series)} \\
            &= A \cE^-(u)^{1/2} e^{-Bm}
        \end{align*}
    }%
    for constants $A, B > 0$ that only depend on the constant $K$ from
    \cref{prop:exponential-decay}. It follows that, for every $\epsilon > 0$, there exists $N \in
    \bN$ such that, for all $n > m > N$,
    \[
        \|u_n - u_m\|_{L^2(I)}
        \le A \cE^-(u)^{1/2} e^{-Bm}
        \le A \cE^-(u)^{1/2} e^{-BN}
        < \epsilon \,,
    \]
    so $(u_k)_k$ is Cauchy.
\end{proof}

\begin{proposition}[Strong convergence]
    \label{lemma:strong-convergence}
    Let $u \in \cU$. Then $P_t u$ converges in $L^2(I)$ to some $u^* \in L^2(I)$ as $t \to \infty$.
\end{proposition}
\begin{proof}
    Let $u^*$ be such that $(P_k u)_{k \in \bN}$ converges to $u^*$ in $L^2(I)$ as $k \to \infty$,
    as given by \cref{lemma:cauchy}. We claim that $P_t u \to u^*$. Let $\epsilon > 0$, and let $N
    \in \bN$ be such that $\|P_k u - u^*\|_{L^2(I)} < \epsilon$ for all $k \ge N$, as given by the
    convergence of the sequence $(P_k u)_k$. Then for any $t \ge N$, letting $j \define \floor{t}$,
    we have
    \[
        \|P_t u - u^*\|_{L^2(I)} \le \|P_j u - u^*\|_{L^2(I)} + \|P_t u - P_j u\|_{L^2(I)}
    \]
    by the triangle inequality. We have $\|P_j u - u^*\|_{L^2(I)} < \epsilon$ by the choice of $N$,
    and on the other hand,
    \begin{align*}
        \|P_t u - P_j u\|_{L^2(I)}
        &= \left\| \int_{(j,t)} \bm{u'}(s) \odif s \right\|_{L^2(I)}
            & \text{(Absolute continuity)} \\
        &\le \int_{(j,t)} \|\bm{u'}(s)\|_{L^2(I)} \odif s
            & \text{(Triangle inequality)} \\
        &\le (t-j)^{1/2} \left( \int_{(j,t)} \|\bm{u'}(s)\|_{L^2(I)}^2 \odif s \right)^{1/2}
            & \text{(Jensen's inequality)} \\
        &\le \left( \cE^-(P_j u) - \cE^-(P_t u) \right)^{1/2}
            & \text{($t-j \in [0, 1)$, \cref{res:solution})} \\
        &\le e^{-Kj/2} \cE^-(u)^{1/2}
            & \text{($\cE^- \ge 0$, \cref{prop:exponential-decay})} \,.
    \end{align*}
    Thus by letting $N$ be large enough, we can ensure that
    \[
        \|P_t u - P_j u\|_{L^2(I)}
        \le e^{-Kj/2} \cE^-(u)^{1/2}
        \le e^{-KN/2} \cE^-(u)^{1/2}
        < \epsilon \,,
    \]
    and hence $\|P_t u - u^*\|_{L^2(I)} < 2\epsilon$. Thus $P_t u \to u^*$ in $L^2(I)$.
\end{proof}

\subsection{Auxiliary results for studying regularity of solutions}
\label{section:auxiliary-regularity}

\begin{lemma}[Functions in $D(\partial \cE^-)$ are well-behaved]
    \label{lemma:in-domain-well-behaved}
    Suppose $u \in D(\partial \cE^-)$. Let $(a,b) \in I$ be a nonempty interval and suppose that
    $\partial_x \downf{u}(x) < 0$ for all $x \in (a,b)$. Then $\upf{u}$ is constant in $(a,b)$.
\end{lemma}
\begin{proof}
    Let $z \in L^2(I)$ be such that $-z \in \partial \cE^-(u)$, which is nonempty by hypothesis.
    By \cref{prop:static-solution-gradient-to-neumann} and \cref{lemma:elliptic-regularity},
    we conclude that $\downf{u} \in H^2(I)$ (which in particular justifies writing the condition
    that $\partial_x \downf{u} < 0$ on $(a,b)$, which we take to mean in terms of the continuous
    representative of $\partial_x \downf{u} \in H^1(I)$).

    By the continuity of $\partial_x \downf{u}$ and the boundary conditions $\partial_x \downf{u}(0)
    = \partial_x \downf{u}(1) = 0$ (which follows from \cref{lemma:elliptic-regularity}), we may
    assume without loss of generality (by extending the interval $(a,b)$ if necessary) that
    $\partial_x \downf{u}(a) = \partial_x \downf{u}(b) = 0$.

    Let
    \[
        Z \define \begin{cases}
            \inf_{[a,b)} \upf{u} & \text{if } a > 0 \\
            \frac{1}{b-a} \int_{(a,b)} \upf{u} \odif x & \text{if } a = 0 \,,
        \end{cases}
    \]
    which is finite in the first case because the infimum must be no smaller than (say)
    $\upf{u}(\frac{a}{2})$, and in the second case because $\upf{u} \in L^2(I)$. Next, define $v \in
    L^2(I)$ by
    \[
        v(x) \define \begin{cases}
            u(x) & \text{if } x \in (0, a) \cup [b, 1) \\
            \downf{u}(x) + Z & \text{if } x \in [a, b) \,.
        \end{cases}
    \]
    We first observe that $v \in \cU$, since $v = \onef{v} + \twof{v}$ with $\twof{v} \define
    \downf{u}$ and $\onef{v} \in L^2(I)$ given by
    \[
        \onef{v}(x) \define \begin{cases}
            \upf{u}(x) & \text{if } x \in (0, a) \cup [b, 1) \\
            Z & \text{if } x \in [a, b) \,,
        \end{cases}
    \]
    which is nondecreasing because $\upf{u}$ is and by the definition of $Z$. Therefore $v \in
    D(\cE^-)$ and
    \[
        \cE^-(v) \le \cD^-(\onef{v}, \twof{v}) 
        = \frac{1}{2} \int_I (\partial_x \twof{v})^2 \odif x
        = \frac{1}{2} \int_i (\partial_x \downf{u})^2 \odif x
        = \cD^-(\upf{u}, \downf{u})
        = \cE^-(u) \,.
    \]
    On the other hand, we claim that $\cE^-(v) \ge \cE^-(u)$ as well. Indeed, let $(\onef{w},
    \twof{w}) \in \cR(v)$ be arbitrary. Construct $\onef{u} \in L^2(I)$ and $\twof{u} \in H^1(I)$ by
    $\twof{u} \define \twof{w}$ and
    \[
        \onef{u}(x) \define \begin{cases}
            \onef{w}(x) & \text{if } x \in (0, a) \cup [b, 1) \\
            u(x) - \twof{u}(x) & \text{if } x \in [a, b) \,,
        \end{cases}
    \]
    That $u = \onef{u} + \twof{u}$ is clear by construction. To show that $(\onef{u}, \twof{u}) \in
    \cR(u)$, it remains to show that $\onef{u}$ is nondecreasing. Certainly it is nondecreasing on
    $(0, a) \cup [b, 1)$ since $\onef{w}$ is nondecreasing. Note that for $x \in [a, b)$, we have
    \[
        \onef{u}(x)
        = u(x) - \twof{w}(x) - \onef{w}(x) + \onef{w}(x)
        = u(x) - v(x) + \onef{w}(x)
        = \upf{u}(x) + \onef{w}(x) - Z \,.
    \]
    Thus $\onef{u}$ is nondecreasing on $[a, b)$, since $\upf{u}$ and $\onef{w}$ are. For $x \in (0,
    a)$ and $y \in [a, b)$, which only applies when $a > 0$, the inequality $\onef{u}(x) \le
    \onef{u}(y)$ follows from the nondecreasing monotonicity of $\onef{w}$ and the inequality
    $\upf{u}(y) \ge Z$, which holds by the (first case of the) definition of $Z$.

    Now, we must show that $\onef{u}(x) \le \onef{u}(y)$ if $x \in [a, b)$ and $y \in [b, 1)$. In
    fact, it suffices to consider $y \in (b, 1)$, because if the inequality holds in all such cases,
    then $\onef{u}$ can be made monotone on all of $I$ by possibly changing its value at $b$, which
    does not affect the \almev equality $u = \onef{u} + \twof{u}$. Therefore let $x \in [a, b)$ and
    $y \in (b, 1)$, which in particular only applies when $b < 1$. Recalling that every monotone
    function has limits from the left and from the right at every point, the key observation is that
    \begin{align*}
        \lim_{s \to b^+} \onef{w}(s) - \lim_{s \to b^-} \onef{w}(s)
        &= \left[ \lim_{s \to b^+} v(s) - \twof{w}(s) \right]
            - \left[ \lim_{s \to b^-} v(s) - \twof{w}(s) \right] \\
        &= \left[ \lim_{s \to b^+} \upf{u}(s) + \downf{u}(s) \right] - \twof{w}(b)
            - \left[ \lim_{s \to b^-} \downf{u}(s) + Z \right] + \twof{w}(b) \\
        &= \lim_{s \to b^+} \upf{u}(s) - Z \,,
    \end{align*}
    where we used the fact that $(\onef{w}, \twof{w}) \in \cR(v)$ in the first equality, the
    definition of $v$ and continuity of $\twof{w}$ in the second equality, and the continuity of
    $\downf{u}$ in the third equality. Now, we have
    \begin{align*}
        \onef{u}(x) \le \onef{u}(y)
        &\iff \upf{u}(x) + \downf{u}(x) - \twof{w}(x) + Z - Z \le \onef{w}(y) \\
        &\iff \upf{u}(x) + \onef{w}(x) - Z \le \onef{w}(y) \\
        &\impliedby \upf{u}(x) + \onef{w}(x) - Z \le \lim_{s \to b^+} \onef{w}(s) \\
        &\iff \upf{u}(x) + \onef{w}(x) - Z
            \le \lim_{s \to b^-} \onef{w}(s) + \lim_{s \to b^+} \upf{u}(s) - Z \\
        &\impliedby \upf{u}(x) \le \lim_{s \to b^+} \upf{u}(s)
            \quad \text{and} \quad \onef{w}(x) \le \lim_{s \to b^-} \onef{w}(s) \,,
    \end{align*}
    which is true because $\upf{u}$ and $\onef{w}$ are nondecreasing and $x < b$. This establishes
    that $\onef{u}$ is nondecreasing and thus $(\onef{u}, \twof{u}) \in \cR(u)$. It follows that
    \[
        \cE^-(u)
        \le \cD^-(\onef{u}, \twof{u})
        = \frac{1}{2} \int_I (\partial_x \twof{u})^2 \odif x
        = \frac{1}{2} \int_I (\partial_x \twof{w})^2 \odif x
        = \cD^-(\onef{w}, \twof{w}) \,,
    \]
    and thus $\cE^-(u) \le \cE^-(v)$ as claimed. Thus $\cE^-(u) = \cE^-(v)$, which in particular
    implies that $\cE^-(v) = \cD^-(\onef{v}, \twof{v})$ and hence $(\upf{v}, \downf{v}) = (\onef{v},
    \twof{v})$.

    By definition of subdifferential, the fact that $-z \in \partial \cE^-(u)$ implies that
    \[
        \frac{1}{2} \int_I (\partial_x \downf{v})^2 \odif x
        = \cE^-(v)
        \ge \cE^-(u) + \inp{-z}{v-u}
        = \frac{1}{2} \int_I (\partial_x \downf{u})^2 \odif x - \inp{z}{v-u} \,.
    \]
    Since $v = u$ on $I \setminus [a,b]$ and $\downf{v} = \downf{u}$, we conclude that
    \[
        0 \le \inp{z}{v-u}
        = \inpspace{z}{v-u}{L^2(a,b)}
        = \inpspace{z}{Z - \upf{u}}{L^2(a,b)} \,.
    \]
    Now, since $z = \partial_x \partial_x \downf{u}$ \almev and $\partial_x \downf{u}(a) =
    \partial_x \downf{u}(b) = 0$, we have
    \[
        \inpspace{z}{Z}{L^2(a,b)}
        = Z \inpspace{\partial_x \partial_x \downf{u}}{1}{L^2(a,b)}
        = Z(\partial_x \downf{u}(b) - \partial_x \downf{u}(a))
        = 0 \,,
    \]
    and hence
    \[
        \label{eq:inner-product-inequality-1}
        \inpspace{z}{\upf{u}}{L^2(a,b)} \le 0 \,.
    \]
    Now, suppose for a contradiction that $\upf{u}$ is not constant in $(a, b)$. Then since
    $\upf{u}$ is nondecreasing, there must exist $a', b'$ with $a < a' < b' < b$ and $\upf{u}(a') <
    \upf{u}(b')$, \ie $\delta \define \upf{u}(b') - \upf{u}(a') > 0$. Let $\alpha \define
    -\sup_{(a',b')} \partial_x \downf{u}$, and note that $\alpha > 0$ by the extreme value theorem
    together with the continuity of $\partial_x \downf{u}$ and the fact that $\partial_x \downf{u} <
    0$ in $[a',b']$. Then \cref{res:inner-product-bound} applied to $f = \upf{u}$, $g = \partial_x
    \downf{u}$ and $a', b' \in J = (a,b)$ implies that $\inpspace{z}{\upf{u}}{L^2(a,b)} =
    \inpspace{\partial_x \partial_x \downf{u}}{\upf{u}}{L^2(a,b)} \ge \alpha \delta > 0$, which is
    the desired contradiction.
\end{proof}

\begin{lemma}
    \label{lemma:inner-product-bound}
    Let $J \subset \bR$ be a finite, nonempty open interval. Let $f \in L^\infty(J)$ be
    nondecreasing and let $g \in H^1(J)$ be such that $g < 0$ in $J$ and $g = 0$ on $\partial J$,
    \ie the endpoints of $J$. Let $[a,b] \subset J$, let $\delta \define f(b) - f(a)$ and let
    $\alpha \define -\sup_{(a,b)} g$. Then $\inpspace{f}{\partial_x g}{L^2(J)} \ge \alpha \delta$.
\end{lemma}
\begin{proof}
    Without loss of generality, we may assume that $J = I$.
    \ignore{
        Indeed suppose the lemma holds for $I$, and let $J = (c, d)$ be arbitrary. Define $f^* \in
        L^\infty(I)$ and $g^* \in H^1(I)$ by $f^*(x) \define f(h(x))$ and $g^*(x) \define g(h(x))$,
        where $h : I \to J$ is given by $h(x) = c + (d-c)x$. Note that $\partial_x g^*(x) =
        \partial_y g(h(x)) \partial_x h(x) = (d-c) \partial_y g(h(x))$. Now, $f^*$ and $g^*$ satisfy
        the same conditions as $f$ and $g$ with respect to $\alpha$ and $\delta$ and inner interval
        $[h^{-1}(a), h^{-1}(b)]$, where the inverse function $h^{-1} : J \to I$ is given by
        $h^{-1}(y) = \frac{y-c}{d-c}$. Hence $\inpspace{f^*}{\partial_x g^*}{L^2(I)} \ge \alpha
        \delta$, and by change of variables we obtain
        \begin{align*}
            \inpspace{f}{\partial_y g}{L^2(J)}
            &= \int_J f(y) (\partial_y g(y)) \odif y
            = \int_{(c,d)} f^*(h^{-1}(y))
                \left( \frac{1}{d-c} \partial_x g^*(h^{-1}(y)) \right) \odif y \\
            &= \int_{(0,1)} \frac{1}{d-c} f^*(x) (\partial_x g^*(x)) (d-c) \odif x
            = \int_I f^* (\partial_x g^*) \odif x
            = \inpspace{f^*}{\partial_x g^*}{L^2(I)}
            \ge \alpha \delta \,.
        \end{align*}
    }%
    Now, for each sufficiently small $\epsilon > 0$, recall that $I_\epsilon = (\epsilon,
    1-\epsilon)$ and let $f_\epsilon \in C^\infty(I_\epsilon)$ be the mollification of $f$. Let
    $f_\epsilon^* : I \to \bR$ be given by
    \[
        f_\epsilon^*(x) \define \begin{cases}
            f_{\epsilon/2}(\epsilon)   & \text{if } x \in (0, \epsilon) \\
            f_{\epsilon/2}(x)          & \text{if } x \in [\epsilon, 1-\epsilon] \\
            f_{\epsilon/2}(1-\epsilon) & \text{if } x \in (1-\epsilon, 1) \,.
        \end{cases}
    \]
    Note that each $f_\epsilon^* \in H^1(I)$, in particular because the piecewise definition is
    continuous and $f_{\epsilon/2}$ is smooth on $[\epsilon, 1-\epsilon]$. Also, each
    $f_{\epsilon/2}$ is nondecreasing and hence so is each $f_\epsilon^*$. Moreover, $f_\epsilon^*
    \to f$ \almev as $\epsilon \to 0$ since this is true of $(f_\epsilon)_{\epsilon > 0}$. Finally,
    we have $f_\epsilon^* \to f$ in $L^2_\loc(I)$ since $f_\epsilon \to f$ in $L^2_\loc(\Omega)$,
    and since $\|f_\epsilon^*\|_{L^\infty(I)} \le \|f_{\epsilon/2}\|_{L^\infty(I)} \le
    \|f\|_{L^\infty(I)}$, we have that $(f_\epsilon^*)_{\epsilon > 0}$ is bounded in $L^2(I)$ and
    \cref{lemma:bounded-loc-to-weak} implies that $f_\epsilon^* \weakto f$ weakly in $L^2(I)$ as
    $\epsilon \to 0$.

    Let $(a_n)_{n \in \bN}, (b_n)_{n \in \bN}$ be two sequences such that $a_n \uparrow a$, $b_n
    \downarrow b$, and moreover, for every $n \in \bN$, $f_\epsilon^*(a_n) \to f(a_n)$ and
    $f_\epsilon^*(b_n) \to f(b_n)$ as $\epsilon \to 0$; the existence of such sequences is
    guaranteed by the fact that $f_\epsilon^* \to f$ almost everywhere. For each $n \in \bN$, let
    $\delta_n \define f(b_n) - f(a_n)$ and $\alpha_n \define -\sup_{(a_n,b_n)} g$. Note that
    $\delta_n \ge \delta$ because $f$ is nondecreasing while $a_n \le a$ and $b \le b_n$, and that
    $\alpha_n \to \alpha$ as $n \to \infty$ by the continuity of $g$.

    For each $\epsilon > 0$, integration by parts gives
    \[
        \inp{f_\epsilon^*}{\partial_x g}
        = \int_I f_\epsilon^* (\partial_x g) \odif x
        = \left. f_\epsilon^* g \right|_0^1 - \int_I (\partial_x f_\epsilon^*) g \odif x \,.
    \]
    Recall that $g = 0$ on $\{0,1\}$. Moreover, since $f_\epsilon^*$ is nondecreasing while $g < 0$
    in $I$, the integrand in the RHS above is nonpositive. Hence we can only make the RHS smaller by
    restricting the range of integration. Thus, fixing any $n \in \bN$ and using the definition of
    $\alpha_n$,
    \[
        \inp{f_\epsilon^*}{\partial_x g}
        \ge - \int_{(a_n,b_n)} (\partial_x f_\epsilon^*) g \odif x
        \ge - \int_{(a_n,b_n)} (\partial_x f_\epsilon^*) (-\alpha_n) \odif x
        = \alpha_n \left. f_\epsilon^* \right|_{a_n}^{b_n} \,.
    \]
    Since $f_\epsilon^* \weakto f$ weakly in $L^2(I)$, we obtain
    \[
        \inp{f}{\partial_x g}
        = \lim_{\epsilon \to 0} \inp{f_\epsilon^*}{\partial_x g}
        \ge \lim_{\epsilon \to 0} \alpha_n \left. f_\epsilon^* \right|_{a_n}^{b_n}
        = \alpha_n \delta_n
        \ge \alpha_n \delta \,,
    \]
    the second equality by the choice of sequences $(a_n)_n, (b_n)_n$ and definition of $\delta_n$.
    We conclude that
    \[
        \inp{f}{\partial_x g} \ge \lim_{n \to \infty} \alpha_n \delta = \alpha \delta \,.
        \qedhere
    \]
\end{proof}

\begin{lemma}
    \label{res:inner-product-bound}
    The statement of \cref{lemma:inner-product-bound} still holds if we replace the condition $f \in
    L^\infty(J)$ with $f \in L^2(J)$.
\end{lemma}
\begin{proof}
    Again let $J = I$ without loss of generality. We proceed by an approximation argument. Let
    $(a_n)_{n \in \bN}$ be a strictly decreasing sequence satisfying $a_1 = a$ and $a_n \to 0$.
    Similarly, let $(b_n)_{n \in \bN}$ be a strictly increasing sequence satisfying $b_1 = b$ and
    $b_n \to 1$. For each $n \in \bN$, define $f_n \in L^\infty(I)$ by
    \[
        f_n(x) \define \begin{cases}
            \inf_{[a_n, b_n]} f & \text{if } x \in (0, a_n) \\
            f(x)                & \text{if } x \in [a_n, b_n] \\
            \sup_{[a_n, b_n]} f & \text{if } x \in (b_n, 1) \,.
        \end{cases}
    \]
    Note that the infimum and supremum above are finite by virtue of the monotonicity of $f$ and the
    observation that each $[a_n, b_n] \subset I$; thus we indeed have $f_n \in L^\infty(I)$.
    Moreover, each $f_n$ is nondecreasing and, since $a_n \le a < b \le b_n$, we have $f_n(a) =
    f(a)$ and $f_n(b) = f(b)$. \cref{lemma:inner-product-bound} implies that $\inp{f_n}{\partial_x
    g} \ge \alpha \delta$ for $\alpha$ as in that statement and $\delta = f_n(b) - f_n(a)
    = f(b) - f(a)$. Finally, we have $f_n \to f$ in $L^2(I)$; indeed, letting $c \define
    \frac{a+b}{2}$ for convenience and using the monotonicity of $f$,
    \begin{align*}
        \|f - f_n\|_{L^2(I)}^2
        &= \int_{(0,a_n)} \left( f(x) - \inf_{[a_n,b_n]} f \right)^2 \odif x
            + \int_{(b_n,1)} \left( f(x) - \sup_{[a_n,b_n]} f \right)^2 \odif x \\
        &\le \int_{(0,a_n)} \left( f(x) - f(c) \right)^2 \odif x
            + \int_{(b_n,1)} \left( f(x) - f(c) \right)^2 \odif x \\
        &\le 2\left[ \|f\|_{L^2(0,a_n)}^2 + \|f\|_{L^2(b_n,1)}^2 + (a_n + 1 - b_n) f(c)^2 \right]
        \to 0 \,,
    \end{align*}
    the last step by the continuity of the functions $x \to \|f\|_{L^2(0,x)}^2$ and $x \to
    \|f\|_{L^2(x,1)}^2$ and the fact that $a_n + (1-b_n) \to 0$ as $n \to \infty$. Hence
    $\inp{f_n}{\partial_x g} \to \inp{f}{\partial_x g}$, and the conclusion follows.
\end{proof}

The following auxiliary result confirms the natural expectation that, if $u$ has a first derivative
and some degree of regularity, then its components $\upf{u}$ and $\downf{u}$ enjoy the same
regularity and, moreover, can be obtained by taking the positive and negative parts of the
derivative $\partial_x u$, respectively.

\begin{proposition}
    \label{prop:w1p-representation}
    Let $1 \le p \le \infty$, and let $u \in \cU \cap W^{1,p}(I)$. Then $\upf{u}, \downf{u} \in
    W^{1,p}(I)$ and, moreover, we have $\partial_x \upf{u} = \partial_x^+ u$ and $\partial_x
    \downf{u} = \partial_x^- u$ \almev in $I$.
\end{proposition}
\begin{proof}
    Let $\onef{u}, \twof{u} : I \to \bR$ be given by
    \[
        \onef{u}(y) \define u(0) + C + \int_{(0,y)} \partial_x^+ u(x) \odif x
        \qquad \text{and} \qquad
        \twof{u}(y) \define -C + \int_{(0,y)} \partial_x^- u(x) \odif x
    \]
    for each $y \in I$, where $C \in \bR$ is implicitly defined so as to satisfy
    \[
        \int_I \twof{u} \odif x = 0 \,.
    \]
    Then $\onef{u}$ is nondecreasing, $\twof{u}$ is nonincreasing, both are absolutely continuous
    with $\partial_x \onef{u} = \partial_x^+ u$ and $\partial_x \twof{u} = \partial_x^- u$ \almev in
    $I$, and
    \[
        (\onef{u} + \twof{u})(y)
        = u(0) + \int_{(0,y)} \left( \partial_x^+ u(x) + \partial_x^- u(x) \right) \odif x
        = u(x)
    \]
    for all $y \in I$ by the absolute continuity of $u$. It is also clear that $\onef{u}, \twof{u}
    \in L^\infty(I) \subset L^p(I)$, since they are pointwise upper bounded in magnitude by
    $\abs*{u(0)} + \abs*{C} + \|\partial_x u\|_{L^1(I)} < +\infty$. Moreover, since $u \in
    W^{1,p}(I)$, we have $\partial_x u \in L^p(I)$ and hence $\partial_x^+ u, \partial_x^- u \in
    L^p(I)$, yielding that $\onef{u}, \twof{u} \in W^{1,p}(I)$.

    The proof will be concluded if we show that $(\upf{u}, \downf{u}) = (\onef{u}, \twof{u})$.
    However, at this point we cannot even state that $(\onef{u}, \twof{u}) \in \cR(u)$ because we
    have not established that $\twof{u} \in H^1(I) = W^{1,2}(I)$ (unless $p \ge 2$, of course).
    However, if we can show that any $(\onef{v}, \twof{v}) \in \cR(u)$ satisfies
    \begin{equation}
        \label{eq:wts-rep}
        \int_I (\partial_x \twof{v})^2 \odif x \gequestion \int_I (\partial_x \twof{u})^2 \odif x
        \,,
    \end{equation}
    then using the assumption that $u \in \cU$ and the definition of $\upf{u}, \downf{u}$, we will
    conclude that indeed $(\upf{u}, \downf{u}) = (\onef{u}, \twof{u})$, as needed.

    Let $(\onef{v}, \twof{v}) \in \cR(u)$. We claim that $\abs*{\partial_x \twof{v}(x)} \ge
    \abs*{\partial_x \twof{u}(x)}$ for \almev $x \in I$, which will imply \eqref{eq:wts-rep}. Recall
    that $\twof{v} \in H^1(I) \subset W^{1,1}(I)$. Since $u = \onef{v} + \twof{v}$ \almev and $u \in
    W^{1,p}(I) \subset W^{1,1}(I)$, we conclude that $\onef{v} \in W^{1,1}(I)$. Hence $u, \onef{u},
    \twof{u}, \onef{v}, \twof{v}$ are all absolutely continuous. The fundamental theorem of calculus
    for the Lebesgue integral implies that, almost everywhere in $I$, these functions are all
    differentiable, and their classical and weak derivatives agree with $\partial_x u = \partial_x
    \onef{u} + \partial_x \twof{u} = \partial_x \onef{v} + \partial_x \twof{v}$. On any such point
    $x$, the monotonicity of $\onef{u}, \twof{u}, \onef{v}, \twof{v}$ implies that $\partial_x
    \onef{u}(x), \partial_x \onef{v}(x) \ge 0$ and $\partial_x \twof{u}(x), \partial_x \twof{v}(x)
    \le 0$. Therefore
    \begin{align*}
        \abs*{\partial_x \twof{v}(x)}
        &= -\partial_x \twof{v}(x)
        = -\partial_x u(x) + \partial_x \onef{v}(x)
        = -\partial_x \twof{u}(x) - \partial_x \onef{u}(x) + \partial_x \onef{v}(x) \\
        &= \abs*{\partial_x \twof{u}(x)} - \partial_x \onef{u}(x) + \partial_x \onef{v}(x) \,.
    \end{align*}
    Now, if $\partial_x u(x) \ge 0$, then $\partial_x \twof{u}(x) = 0$ by definition, so
    $\abs*{\partial_x \twof{v}(x)} \ge \abs*{\partial_x \twof{u}(x)}$ holds trivially. Otherwise, we
    conversely have $\partial_x \onef{u}(x) = 0$ while $\partial_x \onef{v} \ge 0$, and hence, by
    the above,
    \[
        \abs*{\partial_x \twof{v}(x)} \ge \abs*{\partial_x \twof{u}(x)} \,,
    \]
    which concludes the proof.
\end{proof}

The following standard facts will also be useful.

\begin{fact}[See \eg {\cite[Theorem 4.4]{EG15}}]
    \label{res:derivative-on-zero}
    Let $1 \le p < \infty$ and let $f \in W^{1,p}(I)$. Then $\partial_x f = 0$ \almev on $\{f =
    0\}$.
\end{fact}

The following fact is an immediate application of the Sobolev embedding theorem:

\begin{fact}
    \label{res:sobolev-embedding}
    Let $f \in H^2(I)$. Then $f \in C^{1,1/2}(\overline I)$. In particular, $f$ is continuously
    differentiable.
\end{fact}

\begin{fact}[See \eg {\cite[Chapter II, Corollary 2.1]{Bar76}}]
    \label{fact:domains}
    The set $D(\partial \cE^-)$ is a dense subset of $D(\cE^-) = \cU$.
\end{fact}

\begin{observation}
    \label{obs:closure-of-domains}
    Since $\cU$ contains $H^1(I)$, which is dense in $L^2(I)$, \cref{fact:domains} implies that
    \[
        \overline{D(\partial \cE^-)} = \overline{D(\cE^-)} = \overline{\cU} = L^2(I) \,.
    \]
\end{observation}

\subsection{Preservation of $H^1$ regularity}
\label{section:preservation-of-regularity}

We wish to show that if the initial state $u$ is in $H^1(I)$, then $P_t u$ remains in $H^1(I)$ for
all times $t > 0$. To that end, define $\varphi : L^2(I) \to [0, +\infty]$ by
\[
    \varphi(u) \define \begin{cases}
        \|\partial_x u\|_{L^2(I)}^2 & \text{if } u \in H^1(I) \\
        +\infty                     & \text{otherwise.}
    \end{cases}
\]

The theory of maximal monotone operators gives us a recipe to establish that $t \mapsto \varphi(P_t
u)$ is nonincreasing. The key ingredients are
\cref{claim:varphi-properties,lemma:varphi-inequality}.

\begin{claim}
    \label{claim:varphi-properties}
    The functional $\varphi$ is convex, proper and lower semicontinuous.
\end{claim}
\begin{proof}
    Convexity and properness are straightforward; it remains to verify lower semicontinuity. Since
    $L^2(I)$ is a metric space, it suffices to check sequential lower semicontinuity. Let $(u_n)_{n
    \in \bN}$ be a sequence in $L^2(I)$ such that $u_n \to u$ in $L^2(I)$. We need to show that
    \[
        \varphi(u) \lequestion \liminf_{n \to \infty} \varphi(u_n) \,.
    \]
    The only relevant case is when the RHS above is finite, so suppose there exists a subsequence
    $(u_{n_k})_{k \in \bN}$ such that $\lim_{k \to \infty} \varphi(u_{n_k}) = A < +\infty$. By
    extracting a subsequence if necessary, we may assume that $\varphi(u_{n_k}) < +\infty$, and thus
    $u_{n_k} \in H^1(I)$, for every $k$.

    We claim that $(u_{n_k})_k$ is bounded in $H^1(I)$. Indeed suppose this is not the case. Then
    since $\|u_{n_k}\|_{H^1(I)}^2 = \|u_{n_k}\|_{L^2(I)}^2 + \|\partial_x u_{n_k}\|_{L^2(I)}^2$ and
    $\|\partial_x u_{n_k}\|_{L^2(I)}^2$ remains bounded due to the fact that $\varphi(u_{n_k}) \to
    A$, we conclude that $\|u_{n_k}\|_{L^2(I)}$ gets arbitrarily large as $k \to \infty$. But this
    contradicts the fact that $u_{n_k} \to u$ in $L^2(I)$, so the claim holds.

    It follows that we may extract a weakly convergent subsequence $(u_{n_{k_\ell}})_{\ell \in
    \bN}$, and by uniqueness of weak limits, we obtain that $u \in H^1(I)$ and $u_{n_{k_\ell}}
    \weakto u$ weakly in $H^1(I)$. By weak lower semicontinuity the norm in $H^1(I)$,
    \[
        \|u\|_{L^2(I)}^2 + \|\partial_x u\|_{L^2(I)}^2
        = \|u\|_{H^1(I)}^2
        \le \liminf_{\ell \to \infty} \|u_{n_{k_\ell}}\|_{H^1(I)}^2
        = \liminf_{\ell \to \infty}
            \|u_{n_{k_\ell}}\|_{L^2(I)}^2 + \|\partial_x u_{n_{k_\ell}}\|_{L^2(I)}^2 \,.
    \]
    Since $u_{n_{k_\ell}} \to u$ in $L^2(I)$, we have $\|u_{n_{k_\ell}}\|_{L^2(I)}^2 \to
    \|u\|_{L^2(I)}^2$ and hence
    \[
        \varphi(u)
        = \|\partial_x u\|_{L^2(I)}^2
        \le -\|u\|_{L^2(I)}^2 + \liminf_{\ell \to \infty}
            \|u_{n_{k_\ell}}\|_{L^2(I)}^2 + \|\partial_x u_{n_{k_\ell}}\|_{L^2(I)}^2
        = \liminf_{\ell \to \infty} \varphi(u_{n_{k_\ell}}) \,.
        \qedhere
    \]
\end{proof}

\begin{lemma}
    \label{lemma:varphi-inequality}
    Let $u \in D(\partial \cE^-)$ and let $z \in -\partial \cE^-(u)$. Then for all $\lambda > 0$,
    $\varphi(u - \lambda z) \ge \varphi(u)$.
\end{lemma}
\begin{proof}
    By \cref{prop:static-solution-gradient-to-neumann}, $z$ is a weak solution to the static Neumann
    problem, and by \cref{lemma:elliptic-regularity} we have $\downf{u} \in H^2(I)$ with $\partial_x
    \downf{u} = 0$ on $\{0,1\}$ and $z = \partial_x \partial_x \downf{u}$ in $L^2(I)$.

    Note that the result holds trivially if $u - \lambda z \not\in H^1(I)$, in which case $\varphi(u
    - \lambda z) = +\infty$. Therefore assume that $u - \lambda z \in H^1(I)$. We consider four
    cases: $u \in H^1(I)$; $u \in W^{1,1}(I) \setminus H^1(I)$, \ie $u$ is AC but not in $H^1(I)$;
    $u$ is continuous but not AC; and $u$ is not continuous.\footnote{As usual, phrases such as
    ``$u$ is continuous'' should be understood as ``the object $u \in L^2(I)$ has a continuous
    representative'', and in particular the condition ``$u$ is continuous but not AC'' makes sense
    because the continuous representative, if it exists, is unique.}

    \paragraph*{Case 1.} Suppose $u \in H^1(I)$. Note that in this case the assumption that $u -
    \lambda z \in H^1(I)$ implies that $z \in H^1(I)$ as well, so in particular $\partial_x u,
    \partial_x z \in L^2(I)$. Additionally, \cref{lemma:elliptic-regularity} also yields $\downf{u}
    \in H^3(I)$. We have
    \[
        \varphi(u - \lambda z)
        = \int_I \left( \partial_x (u - \lambda z) \right)^2 \odif x
        = \varphi(u) + \lambda^2 \|\partial_x z\|_{L^2(I)}^2
            - 2\lambda \int_I (\partial_x u)(\partial_x z) \odif x \,.
    \]
    Therefore it suffices to show that
    \begin{equation}
        \label{eq:inequality-preservation}
        \int_I (\partial_x u) (\partial_x z) \odif x \lequestion 0 \,.
    \end{equation}
    Recalling that $u = \upf{u} + \downf{u}$ with $\upf{u}, \downf{u} \in H^1(I) = W^{1,2}(I)$ by
    \cref{prop:w1p-representation}, we have
    \[
        \int_I (\partial_x u) (\partial_x z) \odif x
        = \int_I (\partial_x \upf{u}) (\partial_x z) \odif x
            + \int_I (\partial_x \downf{u}) (\partial_x z) \odif x \,,
    \]
    with the second term in the RHS satisfying
    \[
        \int_I (\partial_x \downf{u}) (\partial_x z) \odif x
        = \left. (\partial_x \downf{u}) z \right|_0^1
            - \int_I (\partial_x \partial_x \downf{u}) z \odif x
        = - \|\partial_x \partial_x \downf{u}\|_{L^2(I)}^2
        \le 0 \,.
    \]
    By \cref{prop:w1p-representation}, $\partial_x \upf{u} = \partial_x^+ u$ and $\partial_x
    \downf{u} = \partial_x^- u$ in $L^2(I)$. Since $z = \partial_x \partial_x \downf{u}$ in
    $L^2(I)$, we get
    \[
        \int_I (\partial_x \upf{u}) (\partial_x z) \odif x
        = \int_I (\partial_x^+ u) (\partial_x \partial_x \partial_x \downf{u}) \odif x \,.
    \]
    We claim that the quantity above is zero. Indeed, fixing any representative of $\partial_x u \in
    L^2(I)$, let $S \define \{ x \in I : \partial_x u > 0 \}$. First, we have $(\partial_x^+ u)
    (\partial_x \partial_x \partial_x^- u) = \partial_x^+ u = 0$ on $I \setminus S$. Second, we have
    $\partial_x \downf{u} = \partial_x^- u = 0$ on $S$, so applying \cref{res:derivative-on-zero}
    twice (recall that $\downf{u} \in H^3(I)$) gives that $\partial_x \partial_x \partial_x
    \downf{u} = 0$ \almev in $S$, thus establishing the claim. Thus
    \eqref{eq:inequality-preservation} indeed holds, which concludes the proof in Case~1.

    \paragraph*{Case 2.} Suppose $u \in W^{1,1}(I) \setminus H^1(I)$. We will derive a
    contradiction, showing that this case cannot happen. Note that, since $u - \lambda z \in H^1(I)$
    by assumption, we conclude that $z \in W^{1,1}(I) \setminus H^1(I)$, and in particular $u, z$
    are AC with $\partial_x u, \partial_x z \in L^1(I)$. As in the previous case, we have
    \begin{align}
        \varphi(u - \lambda z)
        &= \int_I \left( \partial_x (u - \lambda z) \right)^2 \odif x \\
        &= \int_I \left( \partial_x u - \lambda \partial_x z \right)^2 \odif x \\
        \label{eq:integral-case-2}
        &= \int_I \left(
            (\partial_x u)^2 - 2\lambda (\partial_x u)(\partial_x z) + \lambda^2 (\partial_x z)^2
            \right) \odif x \,.
    \end{align}
    We claim that the function $(\partial_x u)(\partial_x z) \in L^1(I)$. First, by
    \cref{prop:w1p-representation} we have $\upf{u} \in W^{1,1}(I)$ (while $\downf{u} \in H^1(I)$
    since $\downf{u} \in H^2(I)$), as well as $\partial_x \upf{u} = \partial_x^+ u$ and $\partial_x
    \downf{u} = \partial_x^- u$ \almev in $I$. Hence
    \[
        \int_I \abs*{(\partial_x u) (\partial_x z)} \odif x
        = \int_I \abs*{(\partial_x \upf{u})(\partial_x z)
                            + (\partial_x \downf{u})(\partial_x z)} \odif x
        \le \int_I \Big[ \abs*{(\partial_x \upf{u})(\partial_x z)}
                        + \abs*{(\partial_x \downf{u})(\partial_x z)} \Big] \odif x \,.
    \]
    \sloppy We claim that $(\partial_x \upf{u})(\partial_x z), (\partial_x \downf{u})(\partial_x z)
    \in L^1(I)$. First, we again have that $(\partial_x \upf{u})(\partial_x z) = (\partial_x
    \upf{u})(\partial_x \partial_x \partial_x \downf{u}) = 0$ \almev as in the previous case, where
    in particular we are allowed to apply \cref{res:derivative-on-zero} twice to $\partial_x
    \downf{u}$ because $\partial_x \downf{u} \in W^{2,1}(I)$ by virtue of the fact that $\partial_x
    \partial_x \downf{u} = z \in W^{1,1}(I)$ in the current case. Hence $(\partial_x
    \upf{u})(\partial_x z) \in L^1(I)$. Second, note that $\partial_x \downf{u}$ is AC and hence
    bounded, while $\partial_x z \in L^1(I)$ since $z$ is AC. Therefore
    \[
        \int_I \abs*{(\partial_x \downf{u})(\partial_x z)} \odif x
        \le \|\partial_x \downf{u}\|_{L^\infty(I)} \|\partial_x z\|_{L^1(I)}
        < +\infty \,,
    \]
    and hence $(\partial_x \downf{u})(\partial_x z) \in L^1(I)$. Hence $\abs*{(\partial_x
    \upf{u})(\partial_x z)} + \abs*{(\partial_x \downf{u})(\partial_x z)} \in L^1(I)$ and
    \[
        \int_I \abs*{(\partial_x u) (\partial_x z)} \odif x
        \le \int_I \abs*{(\partial_x \upf{u})(\partial_x z)} \odif x
                + \int_I \abs*{(\partial_x \downf{u})(\partial_x z)} \odif x
        < +\infty \,,
    \]
    so $(\partial_x u)(\partial_x z) \in L^1(I)$ as claimed. Now, since $\varphi(u - \lambda z) <
    +\infty$ by assumption, \eqref{eq:integral-case-2} shows that $(\partial_x u)^2 - 2\lambda
    (\partial_x u)(\partial_x z) + \lambda^2 (\partial_x z)^2 \in L^1(I)$, while we have just
    established that $2\lambda (\partial_x u)(\partial_x z) \in L^1(I)$. We conclude that
    $(\partial_x u)^2 + \lambda^2 (\partial_x z)^2 \in L^1(I)$, \ie
    \[
        \int_I \left( (\partial_x u)^2 + \lambda^2 (\partial_x z)^2 \right) \odif x < +\infty \,.
    \]
    On the other hand, the fact that $u, z \in W^{1,1}(I) \setminus H^1(I)$ implies that
    \[
        \int_I (\partial_x u)^2 \odif x = +\infty
        \qquad \text{and} \qquad
        \int_I (\partial_x z)^2 \odif x = +\infty \,,
    \]
    which is the desired contradiction. This concludes the proof in Case~2.

    \paragraph*{Case 3.} Suppose $u$ is continuous but not AC (and hence the same is true of
    $\upf{u}$). Then by the definition of absolute continuity, there exists $\epsilon > 0$ such
    that, for all $\delta > 0$, there exists a set of pairwise disjoint intervals $\left( (a_i, b_i)
    \right)_{i \in [k]}$ in $I$ such that $\sum_{i=1}^k (b_i-a_i) < \delta$ and $\sum_{i=1}^k
    \abs*{u(a_i) - u(b_i)} > \epsilon$.

    We claim that, moreover, the sequences $\left( (a_i, b_i) \right)_{i \in [k]}$ above can always
    be taken to satisfy $u(a_i) < u(b_i)$ for every $i \in [k]$. Indeed, let $\epsilon > 0$ be as in
    the paragraph above, let $\delta > 0$, and let $\left( (a_i, b_i) \right)_{i \in [k]}$ be the
    corresponding sequence. Let $S \define \{ i \in [k] : u(a_i) \ge u(b_i) \}$. Then for each $i
    \in S$, we have
    \[
        \abs*{u(a_i) - u(b_i)}
        = u(a_i) - u(b_i)
        = \underbrace{\Big[ \downf{u}(a_i) - \downf{u}(b_i) \Big]}_{\ge 0}
            + \underbrace{\Big[ \upf{u}(a_i) - \upf{u}(b_i) \Big]}_{\le 0}
        \le \abs*{\downf{u}(a_i) - \downf{u}(b_i)} \,.
    \]
    Now, since $\downf{u}$ is AC, let $\delta > 0$ be small enough so that
    \[
        \sum_{i \in S} \abs*{\downf{u}(a_i) - \downf{u}(b_i)} \le \frac{\epsilon}{2} \,.
    \]
    It follows that
    \[
        \epsilon < \sum_{i=1}^k \abs*{u(a_i) - u(b_i)}
        \le \sum_{i \in S} \abs*{\downf{u}(a_i) - \downf{u}(b_i)}
            + \sum_{i \in [k] \setminus S} \abs*{u(a_i) - u(b_i)}
        \le \frac{\epsilon}{2} + \sum_{i \in [k] \setminus S} \abs*{u(a_i) - u(b_i)} \,,
    \]
    and hence
    \[
        \sum_{i \in [k] \setminus S} \abs*{u(a_i) - u(b_i)} > \frac{\epsilon}{2} \,,
    \]
    and of course $u(a_i) < u(b_i)$ for each $i \in [k] \setminus S$ and $\sum_{i \in [k] \setminus
    S} (b_i - a_i) < \delta$. This establishes the claim.

    Let $v \define u - \lambda z$ for convenience, and fix any sequence $\left( (a_i, b_i)
    \right)_{i \in [k]}$ of pairwise disjoint intervals satisfying $u(a_i) < u(b_i)$ for each $i \in
    [k]$. Using \cref{lemma:case-continuous}, we map each interval $(a_i, b_i)$ into an interval
    $(a'_i, b'_i) \subseteq (a_i, b_i)$ such that
    \[
        \abs*{v(a'_i) - v(b'_i)} \ge \frac{\abs*{u(a_i) - u(b_i)}}{2} \,.
    \]
    This implies that $v$ is not AC and hence $\varphi(v) = +\infty$, thus concluding the proof in
    this case.

    \paragraph*{Case 4.} Suppose $u$ is not continuous. By \cref{lemma:case-not-continuous}, $u -
    \lambda z$ is not continuous, which implies that $u - \lambda z \not\in H^1(I)$ and $\varphi(u -
    \lambda z) = +\infty$, concluding the proof.
\end{proof}

\begin{lemma}
    \label{lemma:case-continuous}
    Let $u \in D(\partial \cE^-)$, suppose $u$ is continuous, and let $z \in -\partial \cE^-(u)$.
    Let $\lambda > 0$ and let $v \define u - \lambda z$. Let $0 < a < b < 1$ and suppose $u(a) <
    u(b)$. Then for all $\epsilon > 0$, there exist $a', b'$ with $a \le a' < b' \le b$ such that
    \[
        v(b') - v(a') \ge \abs*{u(a) - u(b)} - \epsilon \,.
    \]
\end{lemma}
\begin{proof}
    Recall that, by \cref{prop:static-solution-gradient-to-neumann}, $u, z$ form a weak solution to
    the static Neumann problem, and by \cref{lemma:elliptic-regularity} we have $\downf{u} \in
    H^2(I)$ with $\partial_x \downf{u} = 0$ on $\{0,1\}$ and $z = \partial_x \partial_x \downf{u}$
    in $L^2(I)$. In particular, $\downf{u}$ is continuously differentiable by
    \cref{res:sobolev-embedding}.

    We first observe that it cannot be the case that $\partial_x \downf{u}(x) < 0$ for all $x \in
    (a, b)$, since otherwise \cref{lemma:in-domain-well-behaved} would imply that $\upf{u}$ is
    constant in $(a, b)$, which by continuity would imply that $\upf{u}(a) = \upf{u}(b)$ and hence
    $u(a) \ge u(b)$, a contradiction.

    We first construct $a' \in (a, b)$. In particular, we wish $a'$ to satisfy
    \begin{equation}
        \label{eq:a'-satisfies}
        u(a') \lequestion u(a) + \frac{\epsilon}{2}
        \qquad \text{and} \qquad
        z(a') \gequestion 0 \,.
    \end{equation}
    Let $x^* \in [a, b)$ be given by
    \[
        x^* \define \inf \{ x \in (a, b) : \partial_x \downf{u}(x) = 0 \} \,,
    \]
    which is well-defined by the observation above and the fact that $\partial_x \downf{u} \le 0$ in
    all of $I$ (since $\downf{u}$ is nonincreasing and continuously differentiable). By the
    continuity of $\partial_x \downf{u}$, we have that $\partial_x \downf{u}(x^*) = 0$. Since
    $\partial_x \downf{u}(a) \le 0$, it must be the case that either $a = x^*$ or
    \begin{equation}
        \label{eq:positive-measure}
        \cL \left\{ x \in (a, x^*) : \partial_x \partial_x \downf{u}(x) \ge 0 \right\} > 0 \,.
    \end{equation}
    We consider each case separately. First, suppose $a = x^*$. We claim that for all $\delta > 0$,
    there exists $x \in (a, a+\delta)$ such that $z(x) \ge 0$. Suppose for a contradiction that this
    is not the case, and fix $\delta > 0$ such that for all $x \in (a, a+\delta)$, $z(x) < 0$. Then
    since $z = \partial_x \partial_x \downf{u}$ almost everywhere, we conclude that $\partial_x
    \partial_x \downf{u} < 0$ \almev in $(a, a+\delta)$. Hence $\partial_x \downf{u}(x) < 0$ for all
    $x \in (a, a+\delta)$, contradicting the assumption that $a = x^*$ given the definition of
    $x^*$. Thus the claim holds. Now, using the continuity of $u$, choose $\delta > 0$ small enough
    and choose $a' \in (a, a+\delta)$ so that $z(a') \ge 0$ and moreover $u(a') \le u(a) +
    \epsilon/2$. This choice of $a'$ satisfies \eqref{eq:a'-satisfies}.

    Second, suppose $a < x^*$ and \eqref{eq:positive-measure} holds. Since $z = \partial_x
    \partial_x \downf{u}$ almost everywhere, choose $a' \in (a, x^*)$ such that $z(a') \ge 0$, which
    is possible by \eqref{eq:positive-measure}. Now, since $\partial_x \downf{u} < 0$ for all $x \in
    (a, a') \subset (a, x^*)$ by the choice of $x^*$, \cref{lemma:in-domain-well-behaved} implies
    that $\upf{u}$ is constant in $(a, a')$. By the continuity of $\upf{u}$, we conclude that
    $\upf{u}(a) = \upf{u}(a')$ and hence, since $\downf{u}$ is nonincreasing, we have
    \[
        u(a') = \upf{u}(a') + \downf{u}(a')
        \le \upf{u}(a) + \downf{u}(a)
        = u(a) \,,
    \]
    and again \eqref{eq:a'-satisfies} is satisfied. This concludes the choice of $a'$.

    Now, we may assume without loss of generality that $\epsilon < \abs*{u(a) - u(b)}/2$. Therefore
    our choice of $a'$ yields an interval $(a', b)$ which, using \eqref{eq:a'-satisfies} and
    recalling that $u(a) < u(b)$, satisfies $u(a') < u(b)$. Therefore, repeating a symmetric version
    of the argument above yields a choice of $b' \in (a', b)$ satisfying
    \begin{equation}
        \label{eq:b'-satisfies}
        u(b') \ge u(b) - \frac{\epsilon}{2}
        \qquad \text{and} \qquad
        z(b') \le 0 \,.
    \end{equation}
    Combining \eqref{eq:a'-satisfies} and \eqref{eq:b'-satisfies}, we conclude that
    \[
        v(b') - v(a')
        = \Big[ u(b') - u(a') \Big]
            - \lambda \underbrace{\Big[ z(b') - z(a') \Big]}_{\le 0}
        \ge u(b) - u(a) - \epsilon
        = \abs*{u(a) - u(b)} - \epsilon \,. \qedhere
    \]
\end{proof}

\begin{lemma}
    \label{lemma:case-not-continuous}
    Let $u \in D(\partial \cE^-)$, suppose $u$ is not continuous, and let $z \in -\partial
    \cE^-(u)$. Let $\lambda > 0$. Then $u - \lambda z$ is not continuous.
\end{lemma}
\begin{proof}
    Recall that, by \cref{prop:static-solution-gradient-to-neumann}, $u, z$ form a weak solution to
    the static Neumann problem, and by \cref{lemma:elliptic-regularity} we have $\downf{u} \in
    H^2(I)$ with $\partial_x \downf{u} = 0$ on $\{0,1\}$ and $z = \partial_x \partial_x \downf{u}$
    in $L^2(I)$. In particular, $\downf{u}$ is continuously differentiable by
    \cref{res:sobolev-embedding}.

    Since $\downf{u} \in H^2(I)$ is continuous while $u$ is not continuous, we conclude that
    $\upf{u}$ is not continuous and, since it is monotone, it contains only jump discontinuities.
    Let $x_0 \in I$ be a point of jump discontinuity of $\upf{u}$, \ie a point such that $L_- < L_+$
    where
    \[
        L_- \define \lim_{x \to x_0^-} \upf{u}(x)
        \qquad \text{and} \qquad
        L_+ \define \lim_{x \to x_0^+} \upf{u}(x) \,.
    \]
    We first claim that $\partial_x \downf{u}(x_0) = 0$. Indeed, it is clear that $\partial_x
    \downf{u} \le 0$ since $\downf{u}$ is nonincreasing and continuously differentiable. If we had
    $\partial_x \downf{u}(x_0) < 0$, then by continuity $\partial_x \downf{u}$ would be strictly
    negative in a neighbourhood of $x_0$, in which case \cref{lemma:in-domain-well-behaved} would
    imply that $\upf{u}$ is constant in a neighbourhood of $x_0$, contradicting the fact that $L_- <
    L_+$. Hence the claim holds.

    We now claim that, for all $\delta > 0$, there exists $x \in (x_0, x_0+\delta)$ such that $z(x)
    \le 0$. Suppose for a contradiction that there exists $\delta > 0$ such that for all $x \in
    (x_0, x_0+\delta)$, $z(x) > 0$. Then, since $z = \partial_x \partial_x \downf{u}$ \almev and
    $\partial_x \downf{u}(x_0) = 0$, we conclude that $\partial_x \downf{u} > 0$ in $(x_0,
    x_0+\delta)$, contradicting the fact that $\partial_x \downf{u} \le 0$. Hence the claim holds.
    By the same reasoning, we conclude that for all $\delta > 0$ there exists $x \in (x_0-\delta,
    x_0)$ such that $z(x) \ge 0$.

    Thus, for all $\delta > 0$ there exist points $x_- \in (x_0-\delta, x_0), x_+ \in (x_0,
    x_0+\delta)$ such that $z(x_+) - z(x_-) \le 0$ and hence
    \begin{align*}
        (u - \lambda z)(x_+) - (u - \lambda z)(x_-)
        &= \Big[ \upf{u}(x_+) - \upf{u}(x_-) \Big]
            + \Big[ \downf{u}(x_+) - \downf{u}(x_-) \Big]
            - \lambda \Big[ z(x_+) - z(x_-) \Big] \\
        &\ge \Big[ \upf{u}(x_+) - \upf{u}(x_-) \Big]
            + \Big[ \downf{u}(x_+) - \downf{u}(x_-) \Big] \,.
    \end{align*}
    Since $\upf{u} \to L_+$ as $x \downarrow x_0$ and $\upf{u} \to L_-$ as $x \uparrow x_0$ and
    $\upf{u}$ is nondecreasing, we have $\upf{u}(x_+) - \upf{u}(x_-) \ge L_+ - L_-$. By the
    continuity of $\downf{u}$, we can let $\delta > 0$ be small enough so that $\downf{u}(x_+) -
    \downf{u}(x_-) \ge -\frac{L_+ - L_-}{2}$. We conclude that, for all sufficiently small $\delta >
    0$, there exist points $x_- \in (x_0-\delta, x_0), x_+ \in (x_0, x_0+\delta)$ such that
    \[
        (u - \lambda z)(x_+) - (u - \lambda z)(x_-) \ge \frac{L_+ - L_-}{2} \,,
    \]
    and thus $u - \lambda z$ is not continuous.
\end{proof}

\begin{lemma}[Specialization of {\cite[Theorem~4.4]{Bre73}}]
    \label{res:brezis-preservation}
    Let $H$ be a Hilbert space and let $\tau : H \to [0, +\infty]$ be a convex, proper, and lower
    semicontinuous functional such that $\tau(\mathrm{Proj}_{\overline{D(A)}} x) \le \tau(x)$ for
    all $x \in H$. Let $A : H \to 2^H$ be a maximal monotone operator and let $S_t$ be the semigroup
    generated by $-A$. Then the following are equivalent:
    \begin{enumerate}
        \item $\tau((I + \lambda A)^{-1} x) \le \tau(x)$ for all $x \in H$ and $\lambda > 0$; and
        \item $\tau(S_t x) \le \tau(x)$ for all $x \in \overline{D(A)}$ and $t \ge 0$.
    \end{enumerate}
\end{lemma}

In the statement above, $(I + \lambda A)^{-1} : H \to D(A)$ is the \emph{resolvent} of $A$ (also
denoted by $J_\lambda$, see \eg \cite[p.~563]{Eva10}).

\begin{proposition}[$\varphi$-monotonicity of solutions]
    \label{prop:varphi-monotonicity}
    Let $u_0 \in \cU$ and let $\bm{u} \in C([0, +\infty); L^2(I))$ be the solution to the gradient
    flow problem with initial data $u_0$. Then for all $0 \le t_1 \le t_2 < +\infty$, we have
    $\varphi(\bm{u}(t_1)) \ge \varphi(\bm{u}(t_2))$.
\end{proposition}
\begin{proof}
    This is a direct consequence of \cref{res:brezis-preservation}. Indeed, letting $H \define
    L^2(I)$ and $A \define \partial \cE^-$ (which is maximal monotone as remarked earlier), we first
    observe that indeed $\varphi(\mathrm{Proj}_{\overline{D(A)}} f) \le \varphi(f)$ for all $f \in
    L^2(I)$, since by \cref{obs:closure-of-domains} we have $\mathrm{Proj}_{\overline{D(A)}} f = f$.
    Thus it suffices to show that for all $f \in L^2(I)$ and all $\lambda > 0$,
    \[
        \varphi((I + \lambda A)^{-1} f) \lequestion \varphi(f) \,.
    \]
    But this is equivalent to showing that, for all $u
    \in D(A)$, $-z \in A(u)$ and $\lambda > 0$,
    \[
        \varphi(u) \lequestion \varphi(u - \lambda z) \,,
    \]
    which is precisely \cref{lemma:varphi-inequality}.
\end{proof}

\begin{corollary}[Preservation of $H^1$ regularity]
    \label{cor:preservation-of-regularity}
    Suppose $u_0 \in H^1(I)$, and let $\bm{u} \in C([0, +\infty); L^2(I))$ be the solution to the
    gradient flow problem with initial data $u_0$. Then $\bm{u}(t) \in H^1(I)$ for all $t > 0$.
\end{corollary}
\begin{proof}
    This is an immediate consequence of \cref{prop:varphi-monotonicity} and the definition of
    $\varphi$.
\end{proof}


\subsection{Preservation of Lipschitz regularity}
\label{section:preservation-of-lipschitz}

It will also be useful to control the Lipschitz regularity of solutions. At a high level, we follow
a similar strategy as in \cref{section:preservation-of-regularity}. Recall that $W^{1,\infty}(I)$ is
equivalent to the space of Lipschitz real-valued functions on $I$, up to identification of almost
everywhere equal functions. Define $\psi : L^2(I) \to [0, +\infty]$ by
\[
    \psi(u) \define \begin{cases}
        \|\partial_x u\|_{L^\infty(I)} & \text{if } u \in W^{1,\infty}(I) \\
        +\infty                        & \text{otherwise.}
    \end{cases}
\]

\begin{fact}
    \label{fact:lipschitz}
    Let $u \in L^2(I)$ and let $M \in \bR_{\ge 0}$. Then $\psi(u) \le M$ if and only if $u = f$
    \almev for some $M$-Lipschitz function $f : I \to \bR$. Moreover, in this case $f$ is the
    (unique) continuous representative of $u$ and $\psi(u)$ is the Lipschitz constant of $f$.
\end{fact}

\begin{claim}
    \label{claim:psi-properties}
    The functional $\psi$ is convex, proper and lower semicontinuous.
\end{claim}
\begin{proof}
    Convexity and properness are straightforward; it remains to verify lower semicontinuity. Since
    $L^2(I)$ is a metric space, it suffices to check sequential lower semicontinuity. Let $(u_n)_{n
    \in \bN}$ be a sequence in $L^2(I)$ such that $u_n \to u$ in $L^2(I)$. We need to show that
    \[
        \psi(u) \lequestion \liminf_{n \to \infty} \psi(u_n) \,.
    \]
    The only relevant case is when the RHS above is finite, so suppose there exists a subsequence
    $(u_{n_k})_{k \in \bN}$ such that $\lim_{k \to \infty} \psi(u_{n_k}) = M < +\infty$. By
    extracting a subsequence if necessary, we may assume that $\psi(u_{n_k}) < +\infty$, and thus
    $u_{n_k} \in W^{1,\infty}(I)$, for every $k$. For simplicity and using \cref{fact:lipschitz},
    fix for each $n_k$ the continuous representative of $u_{n_k}$, which we denote by the same name.
    Then each $u_{n_k}$ is $M_k$-Lipschitz with $M_k \to M$.

    Since $u_{n_k} \to u$ in $L^2(I)$, it is standard that we may extract a subsequence that
    converges to $u$ pointwise almost everywhere. Denote this further subsequence again by
    $(u_{n_k})_k$, and let $N \subset I$ be a measure zero set such that $u_{n_k} \to u$ pointwise
    in $I \setminus N$. Then for each $x \ne y$ in $I \setminus N$, we have
    \[
        \abs*{u(x) - u(y)}
        = \lim_{k \to \infty} \abs*{u_{n_k}(x) - u_{n_k}(y)}
        \le \lim_{k \to \infty} M_k |x-y|
        = M |x-y| \,.
    \]
    By \cref{lemma:almost-lipschitz-extension}, $u$ is \almev equal to an $M$-Lipschitz function,
    so $\psi(u) \le M$ by \cref{fact:lipschitz}, as needed.
\end{proof}

\begin{definition}
    For any $w : I \to \bR$ and distinct $x, y \in I$, let
    \[
        \slope_w(x, y) \define \frac{w(y) - w(x)}{y - x} \,.
    \]
\end{definition}

\begin{lemma}
    \label{lemma:psi-inequality}
    Let $u \in D(\partial \cE^-)$ and let $z \in -\partial \cE^-(u)$. Then for all $\lambda > 0$,
    $\psi(u - \lambda z) \ge \psi(u)$.
\end{lemma}
\begin{proof}
    We follow the proof outline from \cref{lemma:varphi-inequality}, but some of the technical
    details are different. As in that proof, we have $\downf{u} \in H^2(I)$ with $\partial_x
    \downf{u} = 0$ on $\{0,1\}$ and $z = \partial_x \partial_x \downf{u}$ in $L^2(I)$.

    We may assume that $u - \lambda z \in W^{1,\infty}(I)$, since otherwise $\psi(u - \lambda z) =
    +\infty$ and there is nothing to prove. In particular, this establishes that $u - \lambda z$ is
    continuous.

    When $u$ is not continuous, \cref{lemma:case-not-continuous} yields that $u - \lambda z$ is not
    continuous, a contradiction. Therefore we may assume that $u$ is continuous. Since both $u$ and
    $u - \lambda z$ are continuous, we conclude that $z$ is also continuous. Let $v \define u -
    \lambda z$. By \cref{fact:lipschitz}, it suffices to show that for every $(a, b) \subset I$ and
    all $\epsilon > 0$, there exists $(a', b') \subset I$ such that
    \begin{equation}
        \label{eq:desired-slope}
        \abs*{\slope_v(a', b')} \gequestion \abs*{\slope_u(a, b)} - \epsilon \,.
    \end{equation}
    If $u(a) < u(b)$, then the interval $(a', b') \subseteq (a, b)$ given by
    \cref{lemma:case-continuous} with parameter $\epsilon(b-a)$ satisfies \eqref{eq:desired-slope};
    and if $u(a) = u(b)$, then any interval will do. Therefore suppose $u(a) > u(b)$, so in
    particular $\slope_u(a, b) < 0$.

    Note that a sufficient condition for \eqref{eq:desired-slope} is that $\abs*{\slope_u(a', b')}
    \ge \abs*{\slope_u(a, b)} - \epsilon$ with $u(a') > u(b')$, $z(a') \le 0$ and $z(b') \ge 0$,
    since in that case, 
    \begin{align*}
        \abs*{\slope_v(a', b')}
        &= \frac{\abs*{(u - \lambda z)(b') - (u - \lambda z)(a')}}{b' - a'}
        = \Big| \underbrace{u(b') - u(a')}_{< 0} - \lambda \underbrace{(z(b') - z(a'))}_{\ge 0}
            \Big| \left(\frac{1}{b' - a'}\right) \\
        &\ge \frac{\abs*{u(b') - u(a')}}{b' - a'}
        = \abs*{\slope_u(a', b')}
        \ge \abs*{\slope_u(a, b)} - \epsilon \,,
    \end{align*}
    which is \eqref{eq:desired-slope}. We now find $a'$ and $b'$ satisfying the aforementioned
    conditions.

    For each point $x \in I$, say $x$ is \emph{left-favourable} if every neighbourhood $(x-\delta,
    x+\delta)$ contains a point $x'$ such that $z(x') \le 0$. Similarly, say $x$ is
    \emph{right-favourable} if every neighbourhood $(x-\delta, x+\delta)$ contains a point $x'$ such
    that $z(x') \ge 0$.

    We claim that there exist points $a^* < b^*$ in $I$ such that $a^*$ is left-favourable, $b^*$ is
    right-favourable, and $\abs*{\slope_u(a^*, b^*)} \ge \abs*{\slope_u(a, b)}$ with $u(a^*) >
    u(b^*)$. Let us first show that this claim yields the desired points $a'$ and $b'$, and then
    proceed to prove the claim. Suppose we have $a^*$ and $b^*$ as claimed. Then by the continuity
    of $u$, we may fix sufficiently small $\delta > 0$ and find points $a' \in (a^*-\delta,
    a^*+\delta)$ and $b' \in (b^*-\delta, b^*+\delta)$ such that $a' < b'$, $z(a') \le 0$, $z(b')
    \ge 0$, $u(a') > u(b')$, and, for $\alpha \define \frac{\epsilon}{2} \cdot \frac{b^* -
    a^*}{u(a^*) - u(b^*)}$ and $\beta \define \frac{\epsilon (b^* - a^*)}{2}$,
    \begin{align*}
        \abs*{\slope_u(a', b')}
        &= \frac{u(a') - u(b')}{b' - a'}
        \ge \frac{u(a^*) - u(b^*) - \beta}{(b^* - a^*)(1 + \alpha)}
        \ge \frac{u(a^*) - u(b^*)}{b^* - a^*}(1 - \alpha) - \frac{\epsilon/2}{1 + \alpha} \\
        &\ge \frac{u(a^*) - u(b^*)}{b^* - a^*} - \epsilon
        = \abs*{\slope_u(a^*, b^*)} - \epsilon
        \ge \abs*{\slope_u(a, b)} - \epsilon \,,
    \end{align*}
    as desired, where we used the inequality $\frac{1}{1 + \alpha} \ge 1 - \alpha$. We now establish
    the existence of $a^*$ and $b^*$.

    We first find $a^*$. If $a$ is left-favourable, then choose $a^* = a$; note that, of course, we
    have $u(a^*) > u(b)$ and $\abs*{\slope_u(a^*, b)} \ge \abs*{\slope_u(a, b)}$. Otherwise, $a$ is
    not left-favourable, which gives some $\delta > 0$ such that $z(x) > 0$ for every $x \in
    (a-\delta, a+\delta)$, and hence $\partial_x \partial_x \downf{u} > 0$ \almev in $(a-\delta,
    a+\delta)$. We now consider two cases. Recall that $\slope_u(a, b) < 0$.

    First, suppose $\partial_x \downf{u}(a) \ge \slope_u(a, b)$. Note that it cannot be the case
    that $\partial_x \downf{u}(x) > \slope_u(a, b)$ for all $x \in (a, b)$, since otherwise we would
    have
    \[
        u(b)
        = \upf{u}(b) + \downf{u}(b)
        \ge \upf{u}(a) + \downf{u}(a) + \int_{(a,b)} \partial_x \downf{u} \odif x
        > u(a) + (b-a) \slope_u(a, b)
        = u(b) \,,
    \]
    a contradiction. Therefore we may let
    \[
        a^* \define \inf\left\{ x \in (a, b) : \partial_x \downf{u}(x) \le \slope_u(a, b) \right\}
        \,.
    \]
    By the continuity of $\partial_x \downf{u}$ (recall that $\downf{u} \in H^2(I)$ is continuously
    differentiable), we conclude that $\partial_x \downf{u}(a^*) = \slope_u(a, b)$. We also observe
    that we must have $a^* > a$ since, as noted above, we have $\partial_x \partial_x \downf{u}(x) >
    0$ \almev in some neighbourhood $(a-\delta, a+\delta)$, which implies that $\partial_x
    \downf{u}(x) > \partial_x \downf{u}(a) \ge \slope_u(a, b)$ for all $x \in (a, a+\delta)$. We
    claim that $a^*$ is left-favourable. Indeed, otherwise there would be some $\delta > 0$ such
    that $z(x) > 0$ for all $x \in (a^*-\delta, a^*+\delta)$, and since $z = \partial_x \partial_x
    \downf{u}$ \almev we would conclude that $\partial_x \downf{u}(x) < \partial_x \downf{u}(a^*) =
    \slope_u(a, b)$ for $x \in (a^*-\delta, a^*)$, contradicting the choice of $a^*$. Hence we have
    found a left-favourable $a^*$ in this case; we claim that $a^*$ also satisfies $u(a^*) > u(b)$
    and $\abs*{\slope_u(a^*, b)} \ge \abs*{\slope_u(a, b)}$. The first inequality holds since
    \begin{align*}
        u(a^*) - u(b)
        &= \downf{u}(a^*) + \upf{u}(a^*) - \big[ u(a) + (b-a) \slope_u(a, b) \big] \\
        &\ge \downf{u}(a) + \int_{(a,a^*)} \partial_x \downf{u} \odif x
            + \upf{u}(a) - \upf{u}(a) - \downf{u}(a) - (b-a) \slope_u(a, b) \\
        &= \int_{(a, a^*)} \left( \underbrace{\partial_x \downf{u}(x)}_{> \slope_u(a, b)}
                - \slope_u(a, b) \right) \odif x
            - (b-a^*) \underbrace{\slope_u(a, b)}_{< 0} \\
        &> 0 \,,
    \end{align*}
    and the second inequality holds since
    \begin{align*}
        &(b-a^*)(b-a) \Big( \abs*{\slope_u(a^*, b)} - \abs*{\slope_u(a, b)} \Big) \\
        &\qquad = (b-a^*)(b-a) \left(
            \frac{u(a^*) - u(b)}{b - a^*} - \frac{u(a) - u(b)}{b - a} \right) \\
        &\qquad = \Big( u(a^*) - u(b) \Big)(b - a) - \Big( u(a) - u(b) \Big)(b - a^*) \\
        &\qquad = \Bigg(
                \upf{u}(a^*) + \downf{u}(a)
                + \int_{(a,a^*)} \underbrace{\partial_x \downf{u}(x)}_{> \slope_u(a, b)} \odif x
                - u(b)
            \Bigg)(b - a)
            - \Big( u(a) - u(b) \Big)(b - a^*) \\
        &\qquad > \Big(
                \upf{u}(a) + \downf{u}(a) + (a^*-a) \slope_u(a, b) - u(b)
            \Big)(b - a)
            - \Big( u(a) - u(b) \Big)(b - a^*) \\
        &\qquad = \Big(
                u(a) + (a^*-a) \frac{u(b) - u(a)}{b - a} - u(b)
            \Big)(b - a)
            - \Big( u(a) - u(b) \Big)(b - a^*) \\
        &\qquad = \Big( u(a) - u(b) \Big) \Big( (b-a) - (b-a^*) - (a^*-a) \Big) \\
        &\qquad = 0 \,.
    \end{align*}
    Therefore in the first case we have found a left-favourable $a^* \in (a, b)$ such that $u(a^*) >
    u(b)$ and $\abs*{\slope_u(a^*, b)} \ge \abs*{\slope_u(a, b)}$.

    Second, suppose $\partial_x \downf{u}(a) < \slope_u(a, b)$. We proceed similarly, but with
    points to the left of $a$ instead. Namely, it cannot be the case that $\partial_x \downf{u}(x) <
    \slope_u(a, b)$ for all $x \in (0, a)$, since by continuity this would imply that $\partial_x
    \downf{u}(0) \le \slope_u(a, b) < 0$, a contradiction. Hence we may define
    \[
        a^* \define \sup\left\{ x \in (0, a) : \partial_x \downf{u}(x) \ge \slope_u(a, b) \right\}
        \,.
    \]
    By continuity, we conclude that $\partial_x \downf{u}(a^*) = \slope_u(a, b)$, which also implies
    that $a^* < a$. We claim that $a^*$ is left-favourable. Indeed, if it was not, then for some
    $\delta > 0$ we would have $z(x) > 0$ for all $x \in (a^*-\delta, a^*+\delta)$, and since $z =
    \partial_x \partial_x \downf{u}$ \almev we would conclude that $\partial_x \downf{u}(x) >
    \partial_x \downf{u}(a^*) = \slope_u(a, b)$ for all $x \in (a^*, a^*+\delta)$, contradicting the
    choice of $a^*$. We now claim that $a^*$ also satisfies $u(a^*) > u(b)$ and $\abs*{\slope_u(a^*,
    b)} \ge \abs*{\slope_u(a, b)}$.

    To prove the first inequality, we first observe that $\partial_x \downf{u}(x) < \slope_u(a, b) <
    0$ for all $x \in (a^*, a)$, which by \cref{lemma:in-domain-well-behaved} implies that $\upf{u}$
    is constant in $(a^*, a)$. Since $\upf{u}$ is continuous (because $u$ and $\downf{u}$ are), we
    conclude that $\upf{u}(a^*) = \upf{u}(a)$. Therefore we have
    \[
        u(a^*) - u(a) = \downf{u}(a^*) - \downf{u}(a) \ge 0
    \]
    since $\downf{u}$ is nonincreasing, and hence $u(a^*) \ge u(a) > u(b)$. As for the second
    inequality, we have
    {\allowdisplaybreaks
    \begin{align*}
        &(b-a^*)(b-a) \Big( \abs*{\slope_u(a^*, b)} - \abs*{\slope_u(a, b)} \Big) \\
        &\qquad = (b-a^*)(b-a) \left(
            \frac{u(a^*) - u(b)}{b - a^*} - \frac{u(a) - u(b)}{b - a} \right) \\
        &\qquad = \Big( u(a^*) - u(b) \Big)(b - a) - \Big( u(a) - u(b) \Big)(b - a^*) \\
        &\qquad = \Bigg( \upf{u}(a) + \downf{u}(a)
                - \int_{(a^*,a)} \underbrace{\partial_x \downf{u}(x)}_{< \slope_u(a,b)} \odif x
                - u(b) \Bigg)(b - a)
            - \Big( u(a) - u(b) \Big)(b - a^*) \\
        &\qquad > \Bigg( u(a) - (a-a^*) \frac{u(b) - u(a)}{b-a} - u(b) \Bigg)(b - a)
            - \Big( u(a) - u(b) \Big)(b - a^*) \\
        &\qquad = \Big( u(a) - u(b) \Big) \Big( (b-a) - (b-a^*) + (a-a^*) \Big) \\
        &\qquad = 0 \,.
    \end{align*}
    }%
    Therefore in any case we have found a left-favourable $a^* \in (0, b)$ such that $u(a^*) > u(b)$
    and $\abs*{\slope_u(a^*, b)} \ge \abs*{\slope_u(a, b)}$.

    Repeating an analogous argument for the right endpoint of the interval $(a^*, b)$, we find a
    right-favourable $b^* \in (a^*, 1)$ such that $u(a^*) > u(b^*)$ and $\abs*{\slope_u(a^*, b^*)}
    \ge \abs*{\slope_u(a^*, b)}$, which concludes the proof as explained above.
\end{proof}

\begin{proposition}[$\psi$-monotonicity of solutions]
    \label{prop:psi-monotonicity}
    Let $u_0 \in \cU$ and let $\bm{u} \in C([0, +\infty); L^2(I))$ be the solution to the gradient
    flow problem with initial data $u_0$. Then for all $0 \le t_1 \le t_2 < +\infty$, we have
    $\psi(\bm{u}(t_1)) \ge \psi(\bm{u}(t_2))$.
\end{proposition}
\begin{proof}
    As in the proof of \cref{prop:varphi-monotonicity}, this follows from
    \cref{lemma:psi-inequality,res:brezis-preservation}.
\end{proof}

\begin{corollary}[Preservation of Lipschitz regularity]
    \label{cor:preservation-of-lipschitz}
    Suppose $u_0 \in W^{1,\infty}(I)$, and let $\bm{u} \in C([0, +\infty); L^2(I))$ be the solution
    to the gradient flow problem with initial data $u_0$. Then $\bm{u}(t) \in W^{1,\infty}(I)$ for
    all $t > 0$.
\end{corollary}
\begin{proof}
    This is an immediate consequence of \cref{prop:psi-monotonicity} and the definition of $\psi$.
\end{proof}

\subsection{Nonexpansiveness and order preservation}
\label{section:nonexpansiveness}

Our goal in this section is to establish that the semigroup $P_t$ is nonexpansive and order
preserving. These properties will help us show that applying $P_t$ to line restrictions in the
multidimensional setting behaves as expected by ``making progress'' toward monotonicity with each
application.

In the context of PDEs, one desirable way to show that a property is preserved through time is to
differentiate in time, and then pass the derivative inside the integral to exploit the definition of
the PDE. However, we need to establish some technical results before we can justify such
calculations.

We start with the following lemma, slightly adapted from \cite{CT80}, which reveals a close
connection between order preservation and nonexpansiveness (here, in the supremum norm) of
operators. Hence our strategy will be to establish order preservation, and conclude
nonexpansiveness.

\begin{lemma}
    \label{lemma:ct80-prop2}
    Let $(\Omega, \Sigma, \mu)$ be a finite measure space, and let $1 \le p \le +\infty$. Let
    $L^\infty(\Omega) \subset C \subset L^p(\Omega)$ have the property that for all $f \in C$ and $r
    \in \bR$, $f + r \in C$. Let $T : C \to L^p(\Omega)$ be continuous as a map from $L^p(\Omega)$
    to $L^p(\Omega)$ (the continuity requirement may be dropped if $p = \infty$) satisfying
    \begin{equation}
        \label{eq:translation-invariance}
        T(f + r) = T(f) + r \text{ \almev}
    \end{equation}
    for every $f \in C$ and $r \in \bR$. Then the following are equivalent:
    \begin{enumerate}[label=(\alph*)]
        \item \label{item:ct-a}
            For all $f, g \in C$, if $f \le g$ \almev then $Tf \le Tg$ \almev.
        \item \label{item:ct-b}
            For all $f, g \in C$, $(Tf - Tg)^+ \le \esssup (f-g)^+$ \almev.
        \item \label{item:ct-c}
            For all $f, g \in C$, $\abs*{Tf - Tg} \le \esssup \abs*{f-g}$ \almev.
    \end{enumerate}
\end{lemma}
\begin{proof}
    We show \ref*{item:ct-a} $\implies$ \ref*{item:ct-b} $\implies$ \ref*{item:ct-c} $\implies$
    \ref*{item:ct-a}. First, assume \ref*{item:ct-a} holds, and let $f, g \in C$. Suppose $r \define
    \esssup (f-g)^+ < +\infty$, since otherwise there is nothing to prove. Note that $g + r \ge f$
    and $g + r \ge g$ both \almev. Thus by \ref*{item:ct-a} and \eqref{eq:translation-invariance},
    we have
    \[
        Tg + (Tf - Tg)^+ = Tf \lor Tg \le T(g+r) = Tg + r \text{ \almev} \,,
    \]
    so $(Tf - Tg)^+ \le r$ \almev as needed. Next, suppose \ref*{item:ct-b} holds, and let $f, g \in
    C$. Let $r_1 \define \esssup (f-g)^+$ and $r_2 \define \esssup (g-f)^+$; note that $\esssup
    \abs*{f-g} = \max\{ r_1, r_2 \}$, so we may assume that $r_1, r_2 < +\infty$, since otherwise
    there is nothing to show. Then indeed, using \ref*{item:ct-b},
    \[
        \abs*{Tf - Tg}
        = \max\left\{ (Tf - Tg)^+, (Tg - Tf)^+ \right\}
        \le \max\{ r_1, r_2 \}
        = \esssup \abs*{f-g} \text{ \almev,}
    \]
    as needed. Finally, suppose \ref*{item:ct-c} holds, and let $f, g \in C$ satisfy $f \le g$
    \almev. First, suppose $f, g \in L^\infty(\Omega)$. Let $r \define \esssup (g-f)$, so that $0
    \le r < +\infty$ by assumption. Then, by \ref*{item:ct-c} and \eqref{eq:translation-invariance},
    \begin{align*}
        \esssup\{ Tf - Tg + r \}
        &= \esssup\{ T(f+r) - Tg \}
        \le \esssup \abs*{T(f+r) - Tg} \\
        &\le \esssup \abs*{f - g + r}
        \le r \,,
    \end{align*}
    so $Tf \le Tg$ \almev as needed.

    Now, for general $f, g \in C$ with $f \le g$ \almev, we may approximate them by sequences
    $(f_n)_n, (g_n)_n$ in $L^\infty(\Omega) \subset C$ such that $f_n \to f$ and $g_n \to g$ in
    $L^p(\Omega)$ and moreover $f_n \le g_n$ \almev for each $n$; indeed, letting $f_n \define
    \max\{ -n, \min\{ n, f \} \}$ and likewise for $g_n$, it is immediate that $f_n, g_n \in
    L^\infty(\Omega)$ with $f_n \le g_n$ \almev, and the convergences $f_n \to f$ and $g_n \to g$ in
    $L^p(\Omega)$ follow from the dominated convergence theorem: since $\abs*{f-f_n}^p \to 0$
    pointwise and $\abs*{f-f_n} \le \abs*{f}$ pointwise for each $n$, we have $\lim_{n \to \infty}
    \int_\Omega \abs*{f_n-f}^p \odif \mu = 0$.

    The previous case yields $Tf_n \le Tg_n$ \almev for each $n$, and we claim that the continuity
    of $T$ implies that $Tf \le Tg$ \almev as well. Indeed, we have $Tf_n \to Tf$ an $Tg_n \to Tg$
    in $L^p(\Omega)$ and thus
    \begin{align*}
        \left\| (Tf-Tg)^+ - (Tf_n - Tg_n)^+ \right\|_{L^p(\Omega)}
        &\le \left\| (Tf-Tg) - (Tf_n - Tg_n) \right\|_{L^p(\Omega)} \\
        &\le \|Tf - Tf_n\|_{L^p(\Omega)} + \|Tg - Tg_n\|_{L^p(\Omega)} \to 0 \,,
    \end{align*}
    which means that $(Tf_n - Tg_n)^+ \to (Tf - Tg)^+$ in $L^p(\Omega)$. But since $(Tf_n - Tg_n)^+
    = 0$ \almev and hence in $L^p(\Omega)$ for each $n$, we conclude that $(Tf - Tg)^+ = 0$ \almev,
    \ie $Tf \le Tg$ \almev as needed.
\end{proof}

The following lemma is a specific formulation of the well-known Leibniz rule for differentiating
under the integral sign. It is a specialization of the version stated in \cite{Che13}, and can be
proved by a standard argument using the Fubini-Tonelli theorem.

\begin{lemma}[Differentiating under the integral sign]
    \label{lemma:leibniz}
    Let $(\Omega, \Sigma, \mu)$ be a measure space and let $[a, b] \subset \bR$. Let $f : [a, b]
    \times \Omega \to \bR$ be a jointly measurable function satisfying the following:
    \begin{enumerate}
        \item $f(t, \cdot) \in L^1(\Omega)$ for \almev $t \in (a, b)$.
        \item $f(\cdot, x)$ is AC for \almev $x \in \Omega$. (Its \almev defined partial derivative
            is denoted by $\partial_t f$ as usual.)
        \item It holds that
            \[
                \int_{(a, b)} \int_\Omega \abs*{\partial_t f(t, x)} \odif \mu \odif t < +\infty \,.
            \]
    \end{enumerate}
    Then the function $t \mapsto \int_\Omega f(t, x) \odif \mu$ is AC and
    \[
        \partial_t \int_\Omega f(t, x) \odif \mu
        = \int_\Omega \partial_t f(t, x) \odif \mu
        \qquad \text{for \almev $t \in (a, b)$.}
    \]
\end{lemma}

The following fact is a standard consequence of the definition of absolute continuity:

\begin{fact}
    \label{fact:max-ac}
    Let $J \subset \bR$ be a compact interval. If $f, g : J \to \bR$ are AC, then $f \lor g$ is AC.
\end{fact}


\begin{fact}[See \eg {\cite[Theorem 4.4]{EG15}}]
    \label{fact:derivative-of-positive-part}
    Let $J \subset \bR$ be a compact interval. Let $1 \le p < \infty$ and let $f \in W^{1,p}(J)$.
    Then $f^+ \in W^{1,p}(J)$ and $\partial_x (f^+) = \chi_{\{f > 0\}} \partial_x f$ \almev in $J$.
\end{fact}

We also use the following standard formulation of a chain rule for Sobolev functions. This version
follows \eg from \cite[Theorem 4.4]{EG15}, which is stated for globally Lipschitz functions $F$, by
using the fact that the image of $f$ on $J$ is bounded, so that the local Lipschitz condition
suffices (\eg extend $F$ linearly outside the image of $f$ to obtain a $C^1(\bR)$, globally
Lipschitz function $\tilde F$).

\begin{fact}
    \label{fact:weak-derivative-chain-rule}
    Let $J \subset \bR$ be a compact interval. Let $1 \le p < \infty$, let $f \in W^{1,p}(J)$, and
    let $F \in C^1(\bR)$ be locally Lipschitz. Then $F \circ f \in W^{1,p}(J)$ and $\partial_x (F
    \circ f)(x) = F'(f(x)) \partial_x f(x)$ for \almev $x \in J$.
\end{fact}

The following lemma essentially says that, given an element of a Bochner space, we can get a handle
on a concrete jointly measurable function ``representing'' that element in a precise sense. This
makes concrete the intuitive expectation that a solution $\bm{u}$ to our PDEs, which maps each point
in time to an element of $L^2(I)$, should also give us a specific value at each point in time and
space ``$\bm{u}(t,x)$''.

\begin{lemma}{{\cite[Theorem 17, p. 198]{DS58}}}
    \label{lemma:ds58}
    Let $(S, \Sigma_S, \mu)$ and $(T, \Sigma_T, \lambda)$ be measure spaces which are either both
    finite or both positive and $\sigma$-finite, and let $(R, \Sigma_R, \rho)$ be their product. Let
    $1 \le p \le \infty$ and let $F$ be a $\mu$-integrable function on $S$ to $L^p(T, \Sigma_T,
    \lambda, \cX)$ where $\cX$ is a real or complex Banach space. Then there is a $\rho$-measurable
    function $f$ on $R$ to $\cX$, which is uniquely determined except for a set of $\rho$-measure
    zero, such that $f(s, \cdot) = F(s)$ for $\mu$-almost all $s$ in $S$. Moreover $f(\cdot, t)$ is
    $\mu$-integrable on $S$ for $\lambda$-almost all $t$ and the integral $\int_S f(s, t) \mu(\odif
    s)$, as a function of $t$ is equal to the element $\int_S F(s) \mu(\odif s)$ of $L^p(T,
    \Sigma_T, \lambda, \cX)$.
\end{lemma}

\begin{corollary}
    \label{cor:jointly-measurable-function}
    Let $[a, b] \subset \bR$ be a compact interval endowed with the Lebesgue measure, and let $F \in
    L^1(a, b; L^2(I))$. Then there exists a jointly measurable function $f^* : [a, b] \times I \to
    \bR$ satisfying
    \begin{enumerate}
        \item $f^*(s, \cdot) = F(s)$ in $L^2(I)$ for \almev $s \in (a, b)$.
        \item $f^*(\cdot, x) \in L^1(a, b)$ for all $x \in I$.
        \item For all $s \in [a, b]$, the functions $x \mapsto \int_{(a,s)} f^*(r, x) \odif r$ and
            $\int_{(a,s)} F(r) \odif r$ are equal in $L^2(I)$.
    \end{enumerate}
\end{corollary}
\begin{proof}
    Let condition $\tilde 2$ denote condition 2 with ``all $x \in I$'' replaced by ``\almev $x \in
    I$''\!\!. We first apply \cref{lemma:ds58} with $S = [a, b]$ and $T = I$, both endowed with the
    Lebesgue measure, to obtain a function $f$ satisfying conditions 1 and $\tilde 2$, as well
    condition 3 for $s = b$.

    We now show that condition 3 also holds for other values of $s \in [a, b)$ using the \almev
    uniqueness given by \cref{lemma:ds58}. For any $s \in [a, b)$, we may apply that lemma with
    $\tilde S = [a, s]$ instead to obtain a function $\tilde f$ satisfying $x \mapsto \int_{(a,s)}
    {\tilde f}(r, x) \odif r = \int_{(a,s)} F(r) \odif r$ in $L^2(I)$. But since both $g = \tilde f$
    and $g = f$ satisfy that $g(r, \cdot) = F(r)$ in $L^2(I)$ for \almev $r \in (a, s)$,
    \cref{lemma:ds58} implies that $\tilde f = f$ except for a subset of $[a, s] \times I$ of joint
    measure zero. We conclude that, for \almev $x \in I$, we have $\tilde f(\cdot, x) = f(\cdot, x)$
    \almev in $(a, s)$ and hence $\int_{(a,s)} f(r, x) \odif r = \int_{(a,s)} {\tilde f}(r, x) \odif
    r$. It follows that $x \mapsto \int_{(a,s)} f(r, x) \odif r = x \mapsto \int_{(a,s)} {\tilde
    f}(r, x) \odif r = \int_{(a,s)} F(r)$ in $L^2(I)$, as claimed.

    The final step is to construct a jointly measurable function $f^* : [a, b] \times I \to \bR$
    satisfying condition 2 rather than just $\tilde 2$, while preserving conditions 1 and 3. Let $N
    \subset I$ be a measure zero set such that $f(\cdot, x) \in L^1(a, b)$ for all $x \in I
    \setminus N$. Define $f^*$ by
    \[
        f^*(s, x) \define \begin{cases}
            f(s, x) & \text{if } x \in I \setminus N \\
            0       & \text{otherwise.}
        \end{cases}
    \]
    We note that $f^*$ is jointly measurable; indeed, this follows from the facts that $f = f^*$ on
    $[a, b] \times (I \setminus N)$ and that the set $[a, b] \times (I \setminus N)$ is jointly
    measurable. By construction, $f^*$ satisfies condition 2, and it is clear that $f^*$ also
    satisfies conditions 1 and 3 since $f$ does and $N$ is a null set.
\end{proof}

\begin{lemma}
    \label{lemma:joint-representation}
    Let $0 \le a < b < +\infty$, and let $\bm{u} \in C([a, b]; L^2(I))$ be the restriction to domain
    $[a, b]$ of any solution to the gradient flow problem, with $\bm{u'} : [a, b] \to L^2(I)$ its
    weak derivative restricted in the same way. Then there exists a jointly measurable function
    $\bm{\tilde u'} : [a, b] \times I \to \bR$ satisfying
    \begin{enumerate}
        \item $\bm{\tilde u'}(t, \cdot) = \bm{u'}(t)$ in $L^2(I)$ for \almev $t \in (a, b)$.
        \item $\bm{\tilde u'}(\cdot, x) \in L^1(a, b)$ for all $x \in I$.
        \item For all $t \in [a, b]$, the functions $x \mapsto \int_{(a,t)} \bm{\tilde u'}(s, x)
            \odif s$ and $\int_{(a,t)} \bm{u'}(s) \odif s$ are equal in $L^2(I)$.
    \end{enumerate}
    Moreover, fixing any representative of $\bm{u}(a) \in L^2(I)$, the function $\bm{\tilde u} : [a,
    b] \times I \to \bR$ given by
    \[
        \bm{\tilde u}(t, x) \define \bm{u}(a)(x) + \int_{(a,t)} \bm{\tilde u'}(s, x) \odif s
    \]
    is jointly measurable and satisfies
    \begin{enumerate}
            \setcounter{enumi}{3}
        \item For each $t \in [a, b]$, $\bm{\tilde u}(t, \cdot) = \bm{u}(t)$ in $L^2(I)$.
        \item For each $x \in I$, $\bm{\tilde u}(\cdot, x)$ is absolutely continuous with weak
            derivative $\partial_t \bm{\tilde u}(t, x) = \bm{\tilde u'}(t, x)$.
    \end{enumerate}
\end{lemma}
\begin{proof}
    Apply \cref{cor:jointly-measurable-function} to $\bm{u'}$, which is integrable since it is the
    weak derivative of $\bm{u}$, to obtain jointly measurable $\bm{\tilde u'} : [a, b] \times I \to
    \bR$ satisfying properties 1--3. We now verify that $\bm{\tilde u}$ satisfies properties 4 and
    5. Note that these two properties then imply that $\bm{\tilde u}$ is a Carathéodory function and
    hence jointly measurable (see \eg \cite[Lemma 4.51]{AB06}).

    For each $t \in [a, b]$, the definition of $\bm{u}$ (in particular its absolute continuity)
    implies that
    \[
        \bm{u}(t) - \bm{u}(a)
        = \left[ \bm{u}(0) + \int_{(0,t)} \bm{u'}(s) \odif s \right]
            - \left[ \bm{u}(0) + \int_{(0,a)} \bm{u'}(s) \odif s \right]
        = \int_{(a,t)} \bm{u'}(s) \odif s
    \]
    in $L^2(I)$. Property 3 implies that, for \almev $x \in I$,
    \[
        \bm{u}(t)(x)
        = \bm{u}(a)(x) + \int_{(a,t)} \bm{\tilde u'}(s, x) \odif s
        = \bm{\tilde u}(t, x) \,,
    \]
    which is property 4. Finally, property 5 is an immediate consequence of the definition of
    $\bm{\tilde u}$.
\end{proof}

We are now prepared to differentiate in time in order to establish that $P_t$ is order preserving.

\begin{proposition}
    \label{prop:pt-diff-nonincreasing}
    Let $u_0, v_0 \in H^1(I)$, and let $\bm{u}, \bm{v}$ be the solutions to the gradient flow
    problem with initial data $u_0, v_0$ respectively. Then the function $\Delta : [0, +\infty) \to
    [0, +\infty)$ given by
    \[
        \Delta(t) \define \frac{1}{2} \int_I \left[ (\bm{u}(t) - \bm{v}(t))^+ \right]^2 \odif x
    \]
    is nonincreasing.
\end{proposition}
\begin{proof}
    Let $0 \le a < b < +\infty$, so that it suffices to show that $\Delta(a) \ge \Delta(b)$. Apply
    \cref{lemma:joint-representation} to $\bm{u}$ and $\bm{v}$ to obtain functions $\bm{\tilde u'}$,
    $\bm{\tilde u}$, $\bm{\tilde v'}$, and $\bm{\tilde v}$ with the properties stated in that lemma.
    Define the function $f : [a, b] \times I \to \bR$ by
    \[
        f(t, x) \define
        \frac{1}{2} \left[ (\bm{\tilde u}(t, x) - \bm{\tilde v}(t, x))^+ \right]^2 \,.
    \]
    Note that $f$ is jointly measurable, since both $\bm{\tilde u}$ and $\bm{\tilde v}$ are. We also
    obtain that, for each $t \in [a, b]$,
    \[
        \Delta(t) = \int_I f(t, x) \odif x \,.
    \]
    We now verify that $f$ satisfies the conditions of \cref{lemma:leibniz}. Let us verify the first
    condition. Since $\bm{\tilde u}(t, \cdot), \bm{\tilde v}(t, \cdot) \in L^2(I)$, it follows that
    $f(t, \cdot) \in L^1(I)$, as desired.

    We now verify the second condition. We already have that, for all $x \in I$, $\bm{\tilde
    u}(\cdot, x)$ and $\bm{\tilde v}(\cdot, x)$ are AC. \cref{fact:derivative-of-positive-part}
    implies that, for all $x \in I$, the function $t \mapsto (\bm{\tilde u}(\cdot, x) - \bm{\tilde
    v}(\cdot, x))^+$ is AC. But the function $y \mapsto y^2$ is locally Lipschitz, so
    \cref{fact:weak-derivative-chain-rule} implies that $f(\cdot, x)$ is AC for each $x$, so the
    second condition is satisfied. Moreover, using
    \cref{fact:weak-derivative-chain-rule,fact:derivative-of-positive-part}, its weak derivative is
    \begin{align*}
        \partial_t f(t, x)
        &= \left[ (\bm{\tilde u}(t, x) - \bm{\tilde v}(t, x))^+ \right]
            \partial_t \left[ (\bm{\tilde u}(t, x) - \bm{\tilde v}(t, x))^+ \right] \\
        &= \left[ (\bm{\tilde u}(t, x) - \bm{\tilde v}(t, x))^+ \right]
            \chi_{\{\bm{\tilde u}(t, x) > \bm{\tilde v}(t, x)\}}
            \left[ \partial_t \bm{\tilde u}(t, x) - \partial_t \bm{\tilde v}(t, x) \right] \\
        &= \chi_{\{\bm{\tilde u}(t, x) > \bm{\tilde v}(t, x)\}}
            \left[ \bm{\tilde u}(t, x) - \bm{\tilde v}(t, x) \right]
            \left[ \bm{\tilde u'}(t, x) - \bm{\tilde v'}(t, x) \right] \,.
    \end{align*}
    We claim that $\int_{(a, b)} \int_I \abs*{\partial_t f(t, x)} \odif x \odif t < +\infty$.
    Indeed, we have
    \begin{align*}
        &\int_{(a, b)} \int_I \abs*{\partial_t f(t, x)} \odif x \odif t \\
        &\quad = \int_{(a, b)} \int_I \Big|
            \chi_{\{\bm{\tilde u}(t, x) > \bm{\tilde v}(t, x)\}}
            \left[ \bm{\tilde u}(t, x) - \bm{\tilde v}(t, x) \right]
            \left[ \bm{\tilde u'}(t, x) - \bm{\tilde v'}(t, x) \right] \Big|
            \odif x \odif t \\
        &\quad \le
            \frac{1}{2} \int_{(a, b)} \int_I
                (\bm{\tilde u}(t, x) - \bm{\tilde v}(t, x))^2 \odif x \odif t
            + \frac{1}{2} \int_{(a, b)} \int_I
                (\bm{\tilde u'}(t, x) - \bm{\tilde v'}(t, x))^2 \odif x \odif t \\
        &\quad \le
            \int_{(a, b)} \int_I \bm{\tilde u}(t, x)^2 \odif x \odif t
            + \int_{(a, b)} \int_I \bm{\tilde v}(t, x)^2 \odif x \odif t
            + \int_{(a, b)} \int_I \bm{\tilde u'}(t, x)^2 \odif x \odif t
            + \int_{(a, b)} \int_I \bm{\tilde v'}(t, x)^2 \odif x \odif t \,.
    \end{align*}
    We claim that each of the four terms above is finite. First, using Tonelli's theorem, the
    definition of $\bm{\tilde u}$ and Jensen's inequality, and letting $L \define b-a$, we have
    \begin{align*}
        \int_{(a, b)} \int_I \bm{\tilde u}(t, x)^2 \odif x \odif t
        &= \int_I \int_{(a, b)} \left(
                \bm{u}(a)(x) + \int_{(a,t)} \bm{\tilde u'}(s, x) \odif s
            \right)^2 \odif t \odif x \\
        &\le 2 \int_I \int_{(a, b)} \bm{u}(a)(x)^2 \odif t \odif x
            + 2L \int_I \int_{(a, b)}
                \int_{(a,t)} \bm{\tilde u'}(s, x)^2 \odif s \odif t \odif x \\
        &\le 2 \int_{(a, b)} \int_I \bm{u}(a)(x)^2 \odif x \odif t
            + 2L \int_{(a, b)} \int_{(a, b)} \int_I
                \bm{\tilde u'}(s, x)^2 \odif x \odif s \odif t \\
        &= 2L \|\bm{u}(a)\|_{L^2(I)}^2
            + 2L^2 \int_{(a, b)} \int_I \bm{\tilde u'}(s, x)^2 \odif x \odif s \,,
    \end{align*}
    where in the last line we write $\|\bm{u}(a)\|_{L^2(I)}^2$, which is finite, since $\bm{u}(a)
    \in L^2(I)$. Hence, since the argument for the $\bm{v}$ terms proceeds identically, all we need
    to show is that
    \[
        \int_{(a, b)} \int_I \bm{\tilde u'}(t, x)^2 \odif x \odif t \ltquestion +\infty \,.
    \]
    Since $\bm{\tilde u'}(t, \cdot) = \bm{u'}(t)$ in $L^2(I)$ for \almev $t \in (a, b)$, this is
    equivalent to showing that
    \begin{equation}
        \label{eq:ineq-uprime-norm}
        \int_{(a, b)} \|\bm{u'}(t)\|_{L^2(I)}^2 \odif t \ltquestion +\infty \,.
    \end{equation}
    Recall that, by \cref{def:nice-evolution-functions}, $\bm{u'} \in L^2(0, b; L^2(I))$, so in
    particular $\bm{u'} \in L^2(a, b; L^2(I))$, which by definition implies
    \eqref{eq:ineq-uprime-norm}. Hence the claim holds and the third condition of
    \cref{lemma:leibniz} is satisfied.

    Therefore, \cref{lemma:leibniz} implies that $\Delta(\cdot) = \int_I f(\cdot, x) \odif x$ is AC
    and, for \almev $t \in (a, b)$,
    \begin{align*}
        \partial_t \Delta(t)
        &= \partial_t \int_I f(t, x) \odif x
        = \int_I \partial_t f(t, x) \odif x \\
        &= \int_I \chi_{\{\bm{\tilde u}(t, x) > \bm{\tilde v}(t, x)\}}
                \left[ \bm{\tilde u}(t, x) - \bm{\tilde v}(t, x) \right]
                \left[ \bm{\tilde u'}(t, x) - \bm{\tilde v'}(t, x) \right] \odif x
            & \text{(Shown above)} \\
        &= \int_I \chi_{\{\bm{u}(t) > \bm{v}(t)\}}
                \left[ \bm{u}(t) - \bm{v}(t) \right]
                \left[ \bm{u'}(t) - \bm{v'}(t) \right] \odif x
            & \text{(By \cref{lemma:joint-representation})} \\
        &= \int_I \left[ \bm{u}(t) - \bm{v}(t) \right]^+
                \partial_x \partial_x \left[ \downf{\bm{u}(t)} - \downf{\bm{v}(t)} \right] \odif x
            & \text{(By \cref{lemma:elliptic-regularity}).}
    \end{align*}
    By \cref{cor:preservation-of-regularity}, $\bm{u}(t), \bm{v}(t) \in H^1(I)$ and hence are AC,
    and by \cref{fact:derivative-of-positive-part}, $\left[ \bm{u}(t) - \bm{v}(t) \right]^+$ is AC.
    Therefore we can integrate by parts and, using \cref{fact:derivative-of-positive-part} and the
    boundary condition $\partial_x \downf{\bm{u}(t)} = \partial_x \downf{\bm{v}(t)} = 0$ on
    $\{0,1\}$ from \cref{lemma:elliptic-regularity}, we obtain
    \begin{align*}
        \partial_t \Delta(t)
        &= -\int_I \chi_{\{\bm{u}(t) > \bm{v}(t)\}}
            \left[ \partial_x (\bm{u}(t) - \bm{v}(t)) \right]
            \left[ \partial_x (\downf{\bm{u}(t)} - \downf{\bm{u}(t)}) \right] \\
        &= -\int_I \chi_{\{\bm{u}(t) > \bm{v}(t)\}}
            \left[ \partial_x \bm{u}(t) - \partial_x \bm{v}(t) \right]
            \left[ \partial_x^- \bm{u}(t) - \partial_x^- \bm{u}(t) \right]
            & \text{(By \cref{prop:w1p-representation}).}
    \end{align*}
    But for any numbers $\alpha, \beta \in \bR$, it is the case that $\alpha \ge \beta \implies
    (\alpha \land 0) \ge (\beta \land 0)$ and $\alpha \le \beta \implies (\alpha \land 0) \le (\beta
    \land 0)$. Hence the integrand above is pointwise nonnegative, and we conclude that $\partial_t
    \Delta(t) \le 0$ for \almev $t \in (a, b)$. Hence $\Delta(a) \ge \Delta(b)$, as needed.
\end{proof}

\begin{corollary}[``Directed nonexpansiveness'' of $P_t$]
    \label{cor:directed-nonexpansiveness}
    Let $u, v \in L^2(I)$. Then for all $t > 0$,
    \[
        \int_I \left[ (P_t u - P_t v)^+ \right]^2 \odif x
        \le \int_I \left[ (u - v)^+ \right]^2 \odif x \,.
    \]
\end{corollary}
\begin{proof}
    This is an immediate consequence of \cref{prop:pt-diff-nonincreasing} when $u, v \in H^1(I)$,
    and the general case follows by approximating $u$ and $v$ by $H^1(I)$ functions and using the
    continuity of the map $w \mapsto P_t w$ from $L^2(I)$ to $L^2(I)$.
\end{proof}

\begin{corollary}[$P_t$ is order preserving]
    \label{cor:pt-order-preserving}
    Let $u, v \in L^2(I)$, and suppose $u \le v$ \almev. Then $P_t u \le P_t v$ \almev for all $t >
    0$.
\end{corollary}
\begin{proof}
    This is an immediate consequence of \cref{cor:directed-nonexpansiveness}.
\end{proof}

\begin{lemma}
    \label{lemma:stationary-points}
    Every nondecreasing $u \in L^2(I)$ is a stationary point of $P_t$, \ie $P_t u = u$ for all $t$.
\end{lemma}
\begin{proof}
    By definition, $u = \upf{u} + \downf{u}$ with $\upf{u} = u$ and $\downf{u} = 0$, so $\cE^-(u) =
    0$. Thus for every $v \in L^2(I)$,
    \[
        \cE^-(v) \ge 0 = \cE^-(u) + \inp{0}{v-u} \,,
    \]
    so $0 \in \partial \cE^-(u)$. Thus $\bm{u}(t) = u$ is the solution to the gradient flow problem
    with initial state $u$.
\end{proof}

\begin{observation}
    \label{obs:uniformly-bounded-away-from-zero}
    \cref{cor:pt-order-preserving} in particular implies that if $u \ge a$ \almev for some $a \in
    \bR$, then $P_t u \ge P_t a = a$ \almev for all $t > 0$, the equality by
    \cref{lemma:stationary-points}.
\end{observation}

We also observe below that $P_t$ is degree one positively homogeneous and additive when one argument
is a constant function, as the following results show.

\begin{lemma}
    \label{lemma:shift-rep}
    Let $u \in \cU$, $\beta \in \bR$, and $v \define u + \beta$. Then $v \in \cU$ with $\upf{v} =
    \upf{u} + \beta$ and $\downf{v} = \downf{u}$, and thus $\cE^-(v) = \cE^-(u)$.
\end{lemma}
\begin{proof}
    This is a straightforward consequence of the definition of $\upf{v}, \downf{v}$.
\end{proof}

\begin{lemma}
    \label{lemma:scale-rep}
    Let $u \in \cU$, $\alpha > 0$, and $v \define \alpha u$. Then $v \in \cU$ with $\upf{v} = \alpha
    \upf{u}$ and $\downf{v} = \alpha \downf{u}$, and thus $\cE^-(v) = \alpha^2 \cE^-(u)$.
\end{lemma}
\begin{proof}
    This is a straightforward consequence of the definitions of $\upf{v}, \downf{v}$ and $\cE^-$.
\end{proof}

\begin{lemma}
    \label{lemma:shift-subdiff}
    Let $u \in D(\partial \cE^-)$, $\beta \in \bR$, and $v \define u + \beta$. Then $\partial
    \cE^-(u) = \partial \cE^-(v)$.
\end{lemma}
\begin{proof}
    By symmetry, it suffices to prove that $\partial \cE^-(u) \subseteq \partial \cE^-(v)$. Let $z
    \in \partial \cE^-(u)$. Then for all $w \in \cU$, using \cref{lemma:shift-rep} twice and the
    definition of subdifferential,
    \[
        \cE^-(w)
        = \cE^-(w - \beta)
        \ge \cE^-(u) + \inp{z}{w - \beta - u}
        = \cE^-(v) + \inp{z}{w - v} \,,
    \]
    and hence $z \in \partial \cE^-(v)$ as desired.
\end{proof}

\begin{lemma}
    \label{lemma:scale-subdiff}
    Let $u \in D(\partial \cE^-)$, $\alpha > 0$, and $v \define \alpha u$. Then $\partial \cE^-(v) =
    \alpha \partial \cE^-(u)$.
\end{lemma}
\begin{proof}
    By symmetry, it suffices to prove that $\alpha \partial \cE^-(u) \subseteq \partial \cE^-(v)$.
    Let $z \in \partial \cE^-(u)$; we claim that $\alpha z \in \partial \cE^-(v)$. Indeed for all $w
    \in \cU$, using \cref{lemma:scale-rep} twice and the definition of subdifferential,
    \[
        \cE^-(w)
        = \alpha^2 \cE^-\left( \tfrac{1}{\alpha} w \right)
        \ge \alpha^2 \left( \cE^-(u) + \inp{z}{\tfrac{1}{\alpha} w - u} \right)
        = \cE^-(v) + \inp{\alpha z}{w - v} \,,
    \]
    and hence $\alpha z \in \partial \cE^-(v)$ as desired.
\end{proof}

\begin{proposition}[Effect of certain affine transformations on $P_t$]
    \label{prop:pt-transformations}
    Let $u \in L^2(I)$, $\alpha > 0$, and $\beta \in \bR$. Then $P_t (\alpha u + \beta) = \alpha P_t
    u + \beta$.
\end{proposition}
\begin{proof}
    We suffices to prove the statement for $u \in \cU$; the general case then follows by
    approximating $\cU$ by $H^1(I)$ functions and the continuity of $P_t$ from $L^2(I)$ to $L^2(I)$.
    For $v \define u + \beta$, we observe that $\bm{v}(t) \define \bm{u}(t) + \beta$ with
    $\bm{v'}(t) \define \bm{u'}(t)$ is a solution to the gradient flow problem with initial state
    $v$, since for each $t > 0$ for which $-\bm{u'}(t) \in \partial \cE^-(\bm{u}(t))$,
    \cref{lemma:shift-subdiff} implies that $-\bm{v'}(t) \in \partial \cE^-(\bm{v}(t))$.

    Similarly, let $w \define \alpha u$. Then $\bm{w}(t) \define \alpha \bm{u}(t)$ with $\bm{w'}(t)
    \define \alpha \bm{u'}(t)$ is the solution to the gradient flow problem with initial state $w$,
    since for each $t > 0$ for which $-\bm{u'}(t) \in \partial \cE^-(\bm{u}(t))$,
    \cref{lemma:scale-subdiff} implies that $-\bm{w'}(t) \in \partial \cE^-(w)$.
\end{proof}

\subsection{Convergence to monotone equilibrium}
\label{section:convergence}

Since the directed Dirichlet energy $\cE^-(P_t u)$ decays over time and $P_t u$ converges to some
limit as $t \to \infty$ (by \cref{lemma:strong-convergence}), we expect this limit to be a monotone
function. Let us establish this fact and other properties of that limit.

\begin{lemma}
    \label{lemma:minimizers-nondecreasing}
    Let $u \in \cU$ satisfy $\cE^-(u) = 0$. Then $u$ is nondecreasing.
\end{lemma}
\begin{proof}
    Since $\cE^-(u) = 0$, we have $\partial_x \downf{u} = 0$ \almev in $I$. Hence $\downf{u}$ is a
    constant function, while $\upf{u}$ is nondecreasing by definition.
\end{proof}

\begin{proposition}
    \label{prop:convergence-to-equilibrium}
    Let $u \in L^2(I)$. Then there exists a nondecreasing $u^* \in \cU$, unique as an element of
    $L^2(I)$, such that $P_t u \to u^*$ in $L^2(I)$ as $t \to \infty$.
\end{proposition}
\begin{proof}
    Since $\cE^-$ achieves its minimum (namely $0$, on \eg constant functions) and $P_t u \in
    D(\partial \cE^-)$ for all $t > 0$, \cite[Theorem 2, p. 160]{AC84} implies that there exists a
    minimizer $u^*$ of $\cE^-$ such that $P_t u \weakto u^*$ weakly in $L^2(I)$. This means that
    $u^* \in \cU$ with $\cE^-(u^*) = 0$, so by \cref{lemma:minimizers-nondecreasing} $u^*$ is
    nondecreasing. By \cref{lemma:strong-convergence} (which we may apply because, fixing any $t_0 >
    0$, we have $P_{t_0} u \in D(\partial \cE^-) \subseteq D(\cE^-) = \cU$), $P_t u$ also converges
    strongly in $L^2(I)$, and it is standard that the weak and strong limits agree and that this
    limit is unique.
\end{proof}

Therefore the following definition is justified:

\begin{definition}[Monotone equilibrium]
    \label{def:p-infty}
    Let $P_\infty : L^2(I) \to \cU$ be the operator mapping each $u \in L^2(I)$ to the unique (as an
    element of $L^2(I)$) nondecreasing $u^* \in \cU$ such that $P_t u \to u^*$ in $L^2(I)$ as $t \to
    \infty$. We call $P_\infty u$ the \emph{monotone equilibrium} of $u$.
\end{definition}

We now pass to the limit $P_\infty$ some useful properties of $P_t$.

\begin{proposition}[``Directed nonexpansiveness'' of $P_\infty$]
    \label{prop:p-infty-directed-nonexpansiveness}
    Let $u, v \in L^2(I)$. Then
    \[
        \int_I \left[ (P_\infty u - P_\infty v)^+ \right]^2 \odif x
        \le \int_I \left[ (u - v)^+ \right]^2 \odif x \,.
    \]
\end{proposition}
\begin{proof}
    This follows by a standard limit argument as follows. Let $u, v \in L^2(I)$. By
    \cref{cor:directed-nonexpansiveness}, $\|(P_t u - P_t v)^+\|_{L^2(I)} \le \|(u -
    v)^+\|_{L^2(I)}$ for all $t > 0$. By definition of $P_\infty$, we have $P_t u \to P_\infty u$
    and $P_t v \to P_\infty v$ in $L^2(I)$ as $t \to \infty$. Using the triangle inequality, we
    obtain
    \begin{align*}
        \|(P_t u - P_t v)^+ - (P_\infty u - P_\infty v)^+\|_{L^2(I)}
        &\le \|(P_t u - P_t v) - (P_\infty u - P_\infty v)\|_{L^2(I)} \\
        &\le \|P_t u - P_\infty u\|_{L^(I)} + \|P_t v - P_\infty v\|_{L^2(I)}
        \to 0 \,,
    \end{align*}
    so $(P_t u - P_t v)^+ \to (P_\infty u - P_\infty v)^+$ in $L^2(I)$ as $t \to \infty$. Since
    $\|(P_t u - P_t v)^+\|_{L^2(I)} \le \|(u - v)^+\|_{L^2(I)}$ for all $t > 0$, we conclude that
    $\|(P_\infty u - P_\infty v)^+\|_{L^2(I)} \le \|(u - v)^+\|_{L^2(I)}$, which gives the
    conclusion.
\end{proof}

\begin{corollary}[$P_\infty$ is order preserving]
    \label{cor:p-infty-order-preserving}
    Let $u, v \in L^2(I)$, and suppose $u \le v$ \almev. Then $P_\infty u \le P_\infty v$ \almev.
\end{corollary}
\begin{proof}
    This is a direct consequence of \cref{prop:p-infty-directed-nonexpansiveness}.
\end{proof}

\begin{proposition}
    \label{prop:p-infty-nonexpansive}
    $P_\infty$ is nonexpansive, and therefore continuous, as an $L^2(I) \to L^2(I)$ map.
\end{proposition}
\begin{proof}
    Let $u, v \in L^2(I)$. In the proof of \cref{prop:p-infty-directed-nonexpansiveness}, we showed
    that $\|(P_t u - P_t v) - (P_\infty u - P_\infty v)\|_{L^2(I)} \to 0$, which implies that $\|P_t
    u - P_t v\|_{L^2(I)} \to \|P_\infty u - P_\infty v\|_{L^2(I)}$. But $\|P_t u - P_t v\|_{L^2(I)}
    \le \|u - v\|_{L^2(I)}$ for all $t > 0$ since $(P_t)_{t \ge 0}$ is a nonexpansive semigroup, so
    the conclusion follows.
\end{proof}

\begin{lemma}
    \label{lemma:p-infty-stationary-points}
    For every nondecreasing $u \in L^2(I)$, we have $P_\infty u = u$.
\end{lemma}
\begin{proof}
    This follows from \cref{lemma:stationary-points} along with the convergence $P_t u \to P_\infty
    u$ in $L^2(I)$.
\end{proof}

Strong convergence to the monotone equilibrium, together with the preservation of regularity results
from \cref{section:preservation-of-regularity}, allows us to obtain regularity of the monotone
equilibrium as well:

\begin{proposition}[$H^1$ regularity of the monotone equilibrium]
    \label{prop:regularity-monotone-equilibrium}
    Let $u \in H^1(I)$. Then $P_\infty u \in H^1(I)$ with $\varphi(P_\infty u) \le \varphi(u)$.
\end{proposition}
\begin{proof}
    Recall the functional $\varphi$ from \cref{section:preservation-of-regularity}. Since $u \in
    H^1(I)$, $\varphi(u) < +\infty$. By \cref{prop:varphi-monotonicity}, we conclude that
    $\varphi(P_t u) \le \varphi(u) < +\infty$ for all $t > 0$. Since $\varphi$ is lower
    semicontinuous by \cref{claim:varphi-properties} and $P_t u \to P_\infty u$ in $L^2(I)$ by
    \cref{lemma:strong-convergence}, we conclude that $\varphi(P_\infty u) \le \varphi(u) <
    +\infty$. Hence $P_\infty u \in H^1(I)$.
\end{proof}

\begin{proposition}[Lipschitz regularity of the monotone equilibrium]
    \label{prop:lipschitz-regularity-monotone-equilibrium}
    Let $u \in W^{1,\infty}(I)$. Then $P_\infty u \in W^{1,\infty}(I)$ with $\psi(P_\infty u) \le
    \psi(u)$.
\end{proposition}
\begin{proof}
    Recall the functional $\psi$ from \cref{section:preservation-of-lipschitz}. Since $u \in
    W^{1,\infty}(I)$, $\psi(u) < +\infty$. By \cref{prop:psi-monotonicity}, we conclude that
    $\psi(P_t u) \le \psi(u) < +\infty$ for all $t > 0$. Since $\psi$ is lower semicontinuous by
    \cref{claim:psi-properties} and $P_t u \to P_\infty u$ in $L^2(I)$ by
    \cref{lemma:strong-convergence}, we conclude that $\psi(P_\infty u) \le \psi(u) < +\infty$.
    Hence $P_\infty u \in W^{1,\infty}(I)$.
\end{proof}

We observe that $P_\infty$ also behaves nicely under the appropriate class of affine
transformations.

\begin{proposition}[Effect of certain affine transformations on $P_\infty$]
    \label{prop:p-infty-transformations}
    Let $u \in L^2(I)$, $\alpha > 0$, and $\beta \in \bR$. Then $P_\infty (\alpha u + \beta) =
    \alpha P_\infty u + \beta$.
\end{proposition}
\begin{proof}
    Since $P_t v \to P_\infty v$ as $t \to \infty$ for each $v \in L^2(I)$, applying
    \cref{prop:pt-transformations} yields
    \[
        P_\infty (\alpha u + \beta)
        = \lim_{t \to \infty} P_t (\alpha u + \beta)
        = \lim_{t \to \infty} \left[ \alpha P_t u + \beta \right]
        = \alpha \left[ \lim_{t \to \infty} P_t u \right] + \beta
        = \alpha P_\infty u + \beta \,. \qedhere
    \]
\end{proof}

We can finally conclude, via \cref{lemma:ct80-prop2}, that $P_\infty$ is nonexpansive in the
$L^\infty$ norm.

\begin{proposition}[$P_\infty$ is nonexpansive in $L^\infty$ norm]
    \label{prop:p-infty-l-infty-nonexpansive}
    Let $u, v \in L^2(I)$. Then $\abs*{P_\infty u - P_\infty v} \le \esssup \abs*{u - v}$ \almev in
    $I$.
\end{proposition}
\begin{proof}
    We verify the conditions of \cref{lemma:ct80-prop2} with $(\Omega, \Sigma, \mu)$ the set $I$
    endowed with the Lebesgue measure, $p = 2$, $C = L^2(I)$, and $T = P_\infty$. It is clear that
    for all $f \in L^2(I)$ and $r \in \bR$, $f + r \in L^2(I)$. By \cref{prop:p-infty-nonexpansive},
    $P_\infty$ is a continuous $L^2(I) \to L^2(I)$ map. Moreover, for all $f \in L^2(I)$ and $r \in
    \bR$, the condition $P_\infty(f + r) = P_\infty f + r$ holds by
    \cref{prop:p-infty-transformations}. Finally, for all $f, g \in L^2(I)$ with $f \le g$ \almev,
    we have that $P_\infty f \le P_\infty g$ \almev by \cref{cor:p-infty-order-preserving}. The
    conclusion follows.
\end{proof}

\section{Directed transport-energy inequality}
\label{section:transport-energy-inequality}

In this section, we establish a connection between the PDE studied above and the Wasserstein
distance between the initial state $u$ and its monotone equilibrium $P_\infty u$, via the dynamical
approach embodied by the \emph{Benamou-Brenier formula}, and then \emph{tensorize} this result from
one dimension to $[0,1]^d$. For optimal transport concepts, we follow the presentation and formalism
of \cite{San15}, and also refer to \cite{AGS05,Vil09}.

This section is organized as follows. After preliminary considerations in
\cref{section:preliminaries-optimal-transport}, \cref{section:benamou-brenier} shows that the PDE
results from \cref{section:directed-heat-semigroup} imply a result about optimal transport
(Wasserstein) distance, namely a one-dimensional \emph{transport-energy inequality}, thereby
implementing the ideas described in \cref{section:overview-transport-energy-one-dimension} of the
proof overview. The idea is to show that our solution to the directed heat equation
\[
    \partial_t u = \partial_x \partial_x^- u
\]
satisfies the so-called \emph{continuity equation}
\[
    \partial_t \varrho_t + \partial_x (\varrho_t v_t) = 0 \,,
\]
which prescribes that ``particles'' (distributed according to measure $\varrho_t$) move according to
the velocity field $v_t$; in our setting, we have $\odif \varrho_t = u \odif x$ and $v_t = -
\frac{\partial_x \downf{u}}{u}$ (\cref{prop:pde-solves-continuity}), so that indeed $\partial_x
(\varrho_t v_t) = -\partial_x \partial_x \downf{u} = -\partial_t \varrho_t$. For functions $u$
bounded close to $1$, we can informally disregard the denominator in $v_t$ and imagine, for
simplicity, that $v_t \approx -\partial_x \downf{u}$. Then thanks to the exponential decay of
\[
    \cE^-(u) = \frac{1}{2} \int (\partial_x \downf{u})^2 \odif x
\]
from the previous section, we conclude that the right-hand side of the Benamou-Brenier formula
\[
    W_2^2(\mu, \nu)
    = \min\left\{ \int_0^1 \|v_t\|_{L^2(\varrho_t)}^2 \odif t :
                    \partial_t \varrho_t + \partial_x (\varrho_t v_t) = 0,
                    \varrho_0 = \mu, \varrho_1 = \nu \right\}
\]
from \cref{prop:benamou-brenier} is roughly upper bounded by $\cE^-(u)$, where (informally) the
measures $\odif \mu = u \odif x$ and $\odif \nu = (P_\infty u) \odif x$ correspond to the initial
state and monotone equilibrium of the directed heat equation, as probability measures\footnote{More
    formally, we decompose the evolution from time $t=0$ to time $t \to \infty$ into intervals of
    unit time length, but by the exponential energy decay, this summation is dominated by the
initial term $\cE^-(u)$.}\!\!. Therefore we conclude in \cref{thm:undirected-transport-energy} the
transport-energy inequality in one dimension,
\[
    W_2^2(\mu, \mu_\infty) \lesssim \cE^-(u) \,.
\]
Finally, the continuity equation also lets us conclude natural and useful properties of solutions of
the directed heat equation, such as mass conservation
(\cref{cor:pt-mass-conserving,cor:p-infty-mass-conserving}) and the property that solutions can only
ever lose mass in any given prefix of the interval $[0,1]$ (\cref{cor:mu-infty-dominated}).

The next goal is to establish the \emph{tensorized} (multidimensional) transport-energy inequality,
thereby implementing the ideas described in \cref{section:overview-tensorizing} of the proof
overview. In \cref{section:directed-optimal-transport}, we develop the definitions and technical
tools required to show our ``Pythagorean composition'' result
(\cref{lemma:pythagorean-composition}), which yields a recipe for proving multidimensional
transport-energy inequalities (with respect to the ``directed Wasserstein distance'') by induction
(\cref{lemma:induction}).

Finally, \cref{section:tensorization} ties together the tensorization recipe described above with
the one-dimensional results specific to our PDE from the previous sections. We show that, by taking
an initial state $f : [0,1]^d \to \bR$ with corresponding measure $\odif \mu = f \odif x$, applying
the directed heat semigroup one coordinate at a time, and putting together the one-dimensional
transport-energy inequalities at each of these steps via properties of our PDE such as order
preservation and nonexpansiveness (which enable the induction to go through), we indeed obtain a
\emph{monotone} final state $f^* : [0,1]^d \to \bR$ with corresponding measure $\odif \mu^* = f^*
\odif x$ for which the multidimensional transport-energy inequality applies
(\cref{thm:transport-energy}):
\[
        W_2^2(\mu \to \mu^*) \lesssim \int_{\closedInt^d} \left| \grad^- f \right|^2 \odif x \,.
\]

\subsection{Preliminaries for optimal transport}
\label{section:preliminaries-optimal-transport}

We start by introducing notation and definitions relevant to the theory of optimal transport.

\paragraph*{Projections.} Say $X \times Y$ is a product space and $z = (x,y)$ is an element in this
space. Then the projection operator onto the first coordinate, denoted interchangeably by $\pi_1$ or
$\pi_x$ depending on the context, is given by $\pi_1(z) = \pi_x(z) = x$. Similarly, $\pi_2(z) =
\pi_y(z) = y$. We extend this definition in the natural way to projections from larger product
spaces onto smaller product spaces. For example, given the product space $X_1 \times \dotsm \times
X_d$ and index set $I \subset [d]$, $I = \{i_1, \dotsc, i_n\}$, the operator $\pi_I$ projects any
element in this space down to an element of $X_{i_1} \times \dotsm \times X_{i_n}$. We also use the
shorthand $\pi_{-I} \define \pi_{[d] \setminus I}$.

Now, let $\Omega \subset \bR^d$ be a Borel set and let $I \subseteq [d]$. For a point $x \in
\Omega$, we write the projections $x_I \define \pi_I(x)$ and $x_{-I} \define \pi_{-I}(x)$, and we
also write $x = (x_I, x_{-I})$. Then, for the set $\Omega$, we write the projections $\Omega_I
\define \{ x_I : x \in \Omega \}$ and $\Omega_{-I} \define \{ x_{-I} : x \in \Omega \}$. For small
sets $I$ or $[d] \setminus I$, we also use shorthand notation such as $x_i \define x_{\{i\}}$ and
$\Omega_{-i-j} \define \Omega_{-\{i,j\}}$.

\paragraph*{Pushforward measure.} Given measurable spaces $(X, \Sigma_X)$ and $(Y, \Sigma_Y)$,
measure $\mu$ on $X$, and measurable map $T : X \to Y$, the \emph{pushforward measure}
$\transmap{T}{\mu}$ on $Y$ is the measure satisfying
\begin{align*}
    (\transmap{T}{\mu})(B) &= \mu(T^{-1}(B))
        &\text{for all $B \in \Sigma_Y$, or equivalently,} \\
    \int_Y \phi(y) \odif {(\transmap{T}{\mu})(y)} &= \int_X \phi(T(x)) \odif \mu(x)
        &\text{for all measurable $\phi : Y \to \bR$.}
\end{align*}

For a measure $\gamma$ on product space $X \times Y$, we say that $\transmap{(\pi_1)}{\gamma}$ and
$\transmap{(\pi_2)}{\gamma}$, which are measures on $X$ and $Y$ respectively, are the first and
second \emph{marginals} of $\gamma$, respectively.

If $(\Omega_3, \Sigma_3)$ is another measurable space and $S : Y \to Z$ a measurable map, then it
holds that $\transmap{(S \circ T)}{\mu} = \transmap{S}{(\transmap{T}{\mu})}$.

\paragraph*{Transport plans.}
For two probability spaces $(\Omega_1, \Sigma_1, \mu_1)$ and $(\Omega_2, \Sigma_2, \mu_2)$, we write
$\Pi(\mu_1, \mu_2)$ for the set of \emph{couplings}, or \emph{transport plans}, between $\mu_1$ and
$\mu_2$, namely probability measures $\gamma$ on the product space $\Omega_1 \times \Omega_2$ whose
first and second marginals are $\mu_1$ and $\mu_2$, respectively.

\paragraph*{Space of probability measures.}
For any Borel set $\Omega \subset \bR^d$, let $P(\Omega)$ denote the space of all (Borel)
probability measures on $\Omega$. We endow $P(\Omega)$ with the weak topology, which is the topology
of weak convergence with respect to bounded continuous functionals. Namely, we say $\mu_n$ converges
weakly to $\mu$ in $P(\Omega)$, and write $\mu_n \weakto \mu$, if $\int_\Omega \phi \odif \mu_n \to
\int_\Omega \phi \odif \mu$ for all bounded continuous $\phi : \Omega \to \bR$.

\paragraph*{Wasserstein distances.}
Let $\Omega \subset \bR^d$ be a \emph{bounded} Borel set and let $p \in [1, +\infty)$. Let $\mu, \nu
\in P(\Omega)$. Given transport plan $\gamma \in \Pi(\mu, \nu)$, we define the \emph{cost}
$C_p(\gamma)$ by
\[
    C_p(\gamma) \define \left( \int_{\Omega \times \Omega} |x-y|^p \odif \gamma(x,y) \right)^{1/p}
    \,.
\]
We then define the \emph{$p$-Wasserstein} distance between $\mu$ and $\nu$ by
\[
    W_p(\mu, \nu) \define \inf_{\gamma \in \Pi(\mu, \nu)} C_p(\gamma) \,,
\]
and we also often refer to the quantity $W_p^p(\mu, \nu) \define W_p(\mu, \nu)^p$. It is standard
that $W_p(\cdot, \cdot)$ is a distance metric on $P(\Omega)$.

\begin{fact}[Wasserstein distance metrizes weak convergence; see \eg
    {\cite[Theorems~5.10~and~5.11]{San15}}]
    \label{fact:wp-metrizes}
    Let $\Omega \subset \bR^d$ be a bounded Borel set, and let $p \in [1, \infty)$. Then for a
    sequence $(\mu_n)_n$ in $P(\Omega)$ and $\mu \in P(\Omega)$, we have $\mu_n \weakto \mu$ if and
    only if $W_p(\mu_n, \mu) \to 0$.
\end{fact}

\begin{corollary}
    \label{cor:wp-continuous}
    Let $\Omega \subset \bR^d$ be a bounded Borel set, and let $p \in [1, \infty)$. Then for
    sequences $(\mu_n)_n$ and $(\nu_n)_n$ in $P(\Omega)$ with $\mu_n \weakto \mu$ and $\nu_n \weakto
    \nu$ for $\mu, \nu \in P(\Omega)$, we have $W_p(\mu_n, \nu_n) \to W_p(\mu, \nu)$. In other
    words, the $p$-Wasserstein distance is continuous in (the weak topology on) $P(\Omega)$.
\end{corollary}
\begin{proof}
    By the triangle inequality, we have
    \[
        W_p(\mu, \nu) - W_p(\mu, \mu_n) - W_p(\nu_n, \nu)
        \le W_p(\mu_n, \nu_n)
        \le W_p(\mu_n, \mu) + W_p(\mu, \nu) + W_p(\nu, \nu_n) \,.
    \]
    By \cref{fact:wp-metrizes}, the LHS and RHS converge to $W_p(\mu, \nu)$ as $n \to \infty$, so
    the conclusion follows.
\end{proof}

\subsection{Optimal transport via Benamou-Brenier}
\label{section:benamou-brenier}

Let $\{\varrho_t : t \in [0,T]\}$ be a family of measures on $\overline I$ and $\{ v_t : t \in
[0,T]\}$ be a family of \emph{velocity fields} such that $v_t \in L^1(\varrho_t)$ for each $t$. We
start by defining what it means for the family $(\varrho_t, v_t)$ to solve the \emph{continuity
equation}
\[
    \partial_t \varrho_t + \partial_x (\varrho_t v_t) = 0 \,.
\]

\begin{definition}[Weak solution; see {\cite[Section 4.1.2]{San15}}]
    \label{def:continuity-weak}
    Let $(\varrho_t, v_t)$ be a family of measure/velocity field pairs indexed by $t \in [0, T]$
    such that $v_t \in L^1(\varrho)$ for each $t$. We say that $(\varrho_t, v_t)$ is a \emph{weak
    solution to the continuity equation} if, for every test function $\psi \in C^1(\overline I)$,
    the function $t \mapsto \int_{\overline I} \psi \odif \varrho_t$ is absolutely continuous in $t$
    and, for \almev $t$, we have
    \begin{equation}
        \label{eq:continuity-weak}
        \partial_t \int_{\overline I} \psi \odif \varrho_t
        = \int_{\overline I} (\partial_x \psi) v_t \odif \varrho_t \,.
    \end{equation}
    In this case, we call $\varrho_0$ and $\varrho_T$ the \emph{initial} and \emph{final states} of
    the solution, respectively (this makes sense because the above implies that $t \mapsto
    \varrho_t$ is continuous for the weak convergence of measures).
\end{definition}

\begin{remark}
    We will only work with absolutely continuous measures $\odif \varrho_t = u \odif x$, $u \in
    L^2(I)$.
\end{remark}

\begin{remark}
    Applying \eqref{eq:continuity-weak} with a constant test function $\psi$ shows that every
    solution to the continuity equation is \emph{mass conserving}. In particular,
    \cref{prop:pde-solves-continuity} below implies that this is true of the semigroup $P_t$; see
    \cref{cor:pt-mass-conserving}.
\end{remark}

The continuity equation is connected to Wasserstein distance via the \emph{Benamou-Brenier formula}:

\begin{proposition}[Benamou-Brenier formula; see \eg {\cite[Theorem 5.28]{San15}}]
    \label{prop:benamou-brenier}
    Let $\mu, \nu$ be probability measures on $\overline I$. Then
    \[
        W_2^2(\mu, \nu)
        = \min\left\{ \int_0^1 \|v_t\|_{L^2(\varrho_t)}^2 \odif t :
                        \partial_t \varrho_t + \partial_x (\varrho_t v_t) = 0,
                        \varrho_0 = \mu, \varrho_1 = \nu \right\} \,,
    \]
    where the constraint above means that $(\varrho_t, v_t)_{t \in [0,1]}$ is a weak solution to the
    continuity equation with initial state $\mu$ and final state $\nu$.
\end{proposition}

We now show our solution to the directed heat equation yields a solution to the continuity equation.
This allows us to conclude several nice properties of this solution, such as mass conservation, as
well as apply the Benamou-Brenier formula to obtain a \emph{transport-energy inequality}.

\begin{condition}
    \label{cond:initial-prob-state}
    $u \in \cU$ is \almev positive and bounded away from zero, and satisfies $\int_I u \odif x = 1$.
\end{condition}

\begin{proposition}
    \label{prop:pde-solves-continuity}
    Let $u \in \cU$ satisfy \cref{cond:initial-prob-state}. Let $\bm{u}(t)$ be the solution to the
    gradient flow problem with initial state $u$, and $\bm{u'}(t)$ its weak derivative for each $t
    \ge 0$. Let $T > 0$ and define the measure/velocity field family $(\varrho_t, v_t)_{t \in
    [0,T]}$ by
    \[
        \odif \varrho_t \define \bm{u}(t) \odif x
        \quad \text{and} \quad
        v_t \define -\frac{\partial_x \downf{\bm{u}(t)}}{\bm{u}(t)}
    \]
    for each $t \in [0, T]$. Then $(\varrho_t, v_t)$ is a weak solution to the continuity equation
    with initial state $\bm{u}(0) \odif x$ and final state $\bm{u}(T) \odif x$.
\end{proposition}
\begin{proof}
    We first note that, by \cref{obs:uniformly-bounded-away-from-zero}, each $\bm{u}(t)$ is \almev
    positive and bounded away from zero, uniformly in $t$. In particular, this justifies the
    denominator in the definition of $v_t$. Furthermore, recall that $\bm{u}(t) \in \cU$ for all $t
    \ge 0$ by \cref{cor:evolution-solution-gradient-to-neumann,def:neumann-evolution}, so the
    numerator in the definition of $v_t$ is defined up to sets of measure zero for each $t$. (In
    fact, by
    \cref{def:gradient-flow,prop:static-solution-gradient-to-neumann,lemma:elliptic-regularity} we
    have $\downf{\bm{u}(t)} \in H^2(I)$ for \almev $t > 0$, in which case the numerator is even
    pointwise well-defined.)

    We now verify that $(\varrho_t, v_t)$ satisfies \cref{def:continuity-weak}. The initial and
    final states are as claimed by construction. Let $\psi \in C^1(\overline I)$. Let $\bm{\tilde
    u'}, \bm{\tilde u} : [0, T] \times I \to \bR$ be obtained by applying
    \cref{lemma:joint-representation} to $\bm{u}$ with time domain $[0, T]$. Define the (jointly
    measurable) function $f : [0, T] \times I \to \bR$ by
    \[
        f(t, x) \define \psi(x) \bm{\tilde u}(t, x) \,.
    \]
    We claim that $f$ satisfies the properties of \cref{lemma:leibniz}. It is clear that $f(t,
    \cdot) \in L^2(I) \subset L^1(I)$ for each $t \in [0, T]$ since $\bm{u}(t, \cdot) \in L^2(I)$
    while $\psi$ is bounded, so the first property is satisfied. Also, for each $x \in I$ we have
    that $\bm{\tilde u}(\cdot, x)$ is AC and hence so is $f(\cdot, x)$, so the second property is
    satisfied. Finally, by properties 1 and 5 of \cref{lemma:joint-representation} and the
    Cauchy-Schwarz inequality,
    \[
        \int_{(0, T)} \int_I \abs*{\partial_t f(t, x)} \odif x \odif t
        = \int_{(0, T)} \int_I \abs*{\psi(x) \bm{\tilde u'}(t, x)} \odif x \odif t
        \le \|\psi\|_{L^2(I)} \int_{(0, T)} \|\bm{u'}(t)\|_{L^2(I)} \odif t < +\infty \,,
    \]
    the last inequality since $\psi \in C^1(\overline I)$ while $\bm{u'} \in L^2(0, T; L^2(I))$ by
    \cref{def:nice-evolution-functions}. Hence the third property is satisfied and
    \cref{lemma:leibniz} applies. Thus the function mapping each $t \in [0, T]$ to
    \[
        \int_{\overline I} \psi \odif \varrho_t
        = \int_I \psi \bm{u}(t) \odif x
        = \int_I \psi(x) \bm{\tilde u}(t, x) \odif x
        = \int_I f(t, x) \odif x
    \]
    is absolutely continuous and, for \almev $t \in (0, T)$,
    \begin{align*}
        \partial_t \int_{\overline I} \psi \odif \varrho_t
        &= \int_I \partial_t f(t, x) \odif x
        = \int_I \psi(x) \bm{\tilde u'}(t, x) \odif x
            & \text{(Property 5 of \cref{lemma:joint-representation})} \\
        &= \int_I \psi \bm{u'}(t) \odif x
            & \text{(Property 1 of \cref{lemma:joint-representation})} \\
        &= \int_I \psi \partial_x \partial_x \downf{\bm{u}(t)} \odif x
            & \text{(\cref{prop:static-solution-gradient-to-neumann,lemma:elliptic-regularity})}
            \\
        &= -\int_I (\partial_x \psi) (\partial_x \downf{\bm u(t)}) \odif x
            & \text{(Integration by parts, \cref{lemma:elliptic-regularity})} \\
        &= \int_I (\partial_x \psi)
                  \left( -\frac{\partial_x \downf{\bm{u}(t)}}{\bm{u}(t)} \right)
                  \bm{u}(t) \odif x
        = \int_{\overline I} (\partial_x \psi) v_t \odif \varrho_t \,. 
        && \qedhere
    \end{align*}
\end{proof}

\begin{corollary}[$P_t$ is mass conserving]
    \label{cor:pt-mass-conserving}
    Let $u \in L^2(I)$. Then for all $t > 0$, $\int_I P_t u \odif x = \int_I u \odif x$.
\end{corollary}
\begin{proof}
    If $u$ satisfies \cref{cond:initial-prob-state}, then this follows from
    \cref{prop:pde-solves-continuity} by taking any constant test function $\psi$ in
    \cref{def:continuity-weak}.

    If $u \in H^1(I)$, then it is bounded, so let $\alpha > 0$, $\beta \in \bR$ be such that $v
    \define \alpha u + \beta \in H^1(I)$ satisfies \cref{cond:initial-prob-state}. Then $u =
    \frac{1}{\alpha} v - \frac{\beta}{\alpha}$ and, by \cref{prop:pt-transformations} and the above,
    \[
        \int_I P_t u \odif x
        = \int_I \left( \frac{1}{\alpha} P_t v - \frac{\beta}{\alpha} \right) \odif x
        = -\frac{\beta}{\alpha} + \frac{1}{\alpha} \int_I P_t v \odif x
        = -\frac{\beta}{\alpha} + \frac{1}{\alpha} \int_I v \odif x
        = \int_I \left( \frac{1}{\alpha} v - \frac{\beta}{\alpha} \right) \odif x
        = \int_I u \odif x \,.
    \]
    Finally, let $u \in \cU$ be arbitrary. Since $H^1(I)$ is dense in $L^2(I)$, let $(u_n)_n$ be a
    sequence in $H^1(I)$ such that $u_n \to u$ in $L^2(I)$. By the continuity of $P_t$ from $L^2(I)$
    to $L^2(I)$, and using the above,
    \[
        \int_I P_t u \odif x
        = \lim_{n \to \infty} \int_I P_t u_n \odif x
        = \lim_{n \to \infty} \int_I u_n \odif x
        = \int_I u \odif x \,. \qedhere
    \]
\end{proof}

By passing to the limit $P_t u \to P_\infty u$ in $L^2(I)$, we also conclude

\begin{corollary}[$P_\infty$ is mass conserving]
    \label{cor:p-infty-mass-conserving}
    Let $u \in L^2(I)$. Then $\int_I P_\infty u \odif x = \int_I u \odif x$.
\end{corollary}

We now upper bound the Wasserstein distance between $u \in \cU$ and its monotone equilibrium by
combining the Benamou-Brenier formula with the exponential decay of the directed Dirichlet energy,
at least as long as $u$ is positive and bounded away from zero.

\thmundirectedtransportenergy
\begin{proof}
    Let $a \define \inf u$. By \cref{cor:pt-order-preserving,cor:pt-mass-conserving}, $P_t u$
    satisfies \cref{cond:initial-prob-state} with $P_t u \ge a$ \almev for all $t \ge 0$. For each
    $t \ge 0$, define the measure $\odif \varrho_t \define (P_t u) \odif x$. Let $\ell \in \bZ_{\ge
    0}$. Then \cref{prop:pde-solves-continuity} applies to time interval $[\ell, \ell+1]$ (we
    refrain from introducing excessive notation for such relabeling of intervals) with initial/final
    states $\varrho_\ell$ and $\varrho_{\ell+1}$, respectively, so that \cref{prop:benamou-brenier}
    gives
    \begin{align*}
        W_2^2(\varrho_\ell, \varrho_{\ell+1})
        &\le \int_\ell^{\ell+1} \int_I
            \left( -\frac{\partial_x \downf{\bm{u}(t)}}{\bm{u}(t)} \right)^2
            \bm{u}(t) \odif x \odif t
        = \int_\ell^{\ell+1} \int_I \frac{(\partial_x \downf{\bm{u}(t)})^2}{\bm{u}(t)}
            \odif x \odif t \\
        &\le \frac{1}{a} \int_\ell^{\ell+1} \int_I (\partial_x \downf{\bm{u}(t)})^2 \odif x \odif t
        = \frac{2}{a} \int_\ell^{\ell+1} \cE^-(P_t u) \odif t \\
        &\le \frac{2}{a} \cE^-(u) \int_\ell^{\ell+1} e^{-Kt} \odif t
        < \frac{2}{a} \cE^-(u) \int_\ell^\infty e^{-Kt} \odif t
        = \frac{2 e^{-K \ell}}{a K} \cE^-(u) \,,
    \end{align*}
    where the penultimate inequality and the constant $K$ come from \cref{prop:exponential-decay}.
    Now, by the triangle inequality (on the distance $W_2$, not the squared distance $W_2^2$), for
    every $\ell \in \bZ_{\ge 1}$ we have
    \begin{align*}
        W_2(\varrho_0, \varrho_\ell)
        &\le \sum_{j=0}^{\ell-1} W_2(\varrho_j, \varrho_{j+1})
        \le \left( \frac{2}{aK} \cE^-(u) \right)^{1/2} \sum_{j=0}^\infty e^{-Kj/2}
        = \left( \frac{2}{aK} \cE^-(u) \right)^{1/2} \frac{1}{1 - e^{-K/2}} \,.
    \end{align*}
    Therefore, extracting the appropriate constant $C$ from the terms involving $K$, we obtain
    \[
        W_2(\varrho_0, \varrho_\ell) \le \frac{C}{a^{1/2}} (\cE^-(u))^{1/2} \,.
    \]
    Recall that $P_t u \weakto P_\infty u$ weakly in $L^2(I)$ by
    \cref{prop:convergence-to-equilibrium}. Hence, defining the absolutely continuous measure
    $\odif \varrho_\infty \define (P_\infty u) \odif x$, we conclude that $\varrho_t \weakto
    \varrho_\infty$ as $t \to \infty$. Then by \cref{cor:wp-continuous},
    \[
        W_2(\varrho_0, \varrho_\infty)
        = \lim_{\ell \to \infty} W_2(\varrho_0, \varrho_\ell)
        \le \frac{C}{a^{1/2}} (\cE^-(u))^{1/2} \,. \qedhere
    \]
\end{proof}

The definition of the velocity field $v_t$ in \cref{prop:pde-solves-continuity}, along with the fact
that $\partial_x \downf{\bm{u}(t)} \le 0$, gives $v_t \ge 0$, which intuitively says that
``particles only move to the right''\!\!. This suggests that the solution $\bm{u}(t)$ can only lose
mass in any prefix of the interval $(0,1)$ over time. The following items make this observation
rigorous.

\begin{lemma}
    \label{lemma:pt-prefix}
    Let $\delta \in I$ and let $J \define (0, \delta)$. Let $u \in L^2(I)$. Then for all $t > 0$,
    \[
        \int_J (P_t u) \odif x \le \int_J u \odif x \,.
    \]
\end{lemma}
\begin{proof}
    Suppose $u$ satisfies \cref{cond:initial-prob-state}; the general case will then follow by the
    same arguments as in the proof of \cref{cor:pt-mass-conserving}. Let $J' \define (-\infty,
    \delta)$, let $\chi_{J'} : \bR \to [0,1]$ be the characteristic function for the set $J'$, and
    for each $\epsilon > 0$, let $\chi_{J'}^\epsilon$ be its mollification by the standard
    mollifier. Note that each $\chi_{J'}^\epsilon$ is nonincreasing. Also note that $\chi_{J'} =
    \chi_J$ pointwise in $I$. Now, fix $T > 0$. Since the family $(\varrho_t, v_t)_{t \in [0,T]}$
    given by \cref{prop:pde-solves-continuity} solves the continuity equation, applying
    \eqref{eq:continuity-weak} with test function $\chi_{J'}^\epsilon \in C^\infty(\overline I)$
    gives that $t \mapsto \int_I \chi_{J'}^\epsilon \bm{u}(t) \odif x$ is AC and, for \almev $t \in
    [0, T]$,
    \[
        \partial_t \int_I \chi_{J'}^\epsilon \bm{u}(t) \odif x
        = \int_I
            \underbrace{(\partial_x \chi_{J'}^\epsilon)}_{\le 0}
            \underbrace{v_t}_{\ge 0} 
            \underbrace{\bm{u}(t)}_{\ge 0} \odif x
        \le 0 \,.
    \]
    Hence $t \mapsto \int_I \chi_{J'}^\epsilon \bm{u}(t) \odif x$ is nonincreasing, and obtain
    \[
        \inpspace{\chi_{J'}^\epsilon}{u}{L^2(I)}
        = \int_I \chi_{J'}^\epsilon \bm{u}(0) \odif x
        \ge \int_I \chi_{J'}^\epsilon \bm{u}(T) \odif x
        = \inpspace{\chi_{J'}^\epsilon}{P_T u}{L^2(I)} \,.
    \]
    Now, since $\chi_{J'} \in L^2_\loc(\bR)$, we have that $\chi_{J'}^\epsilon \to \chi_{J'}$ in
    $L^2_\loc(\bR)$, so in particular $\chi_{J'}^\epsilon \to \chi_{J'} = \chi_J$ in $L^2(I)$. Thus
    we have
    \[
        \int_J (P_T u) \odif x
        = \inpspace{\chi_J}{P_T u}{L^2(I)}
        = \lim_{\epsilon \to 0} \inpspace{\chi_{J'}^\epsilon}{P_T u}{L^2(I)}
        \le \lim_{\epsilon \to 0} \inpspace{\chi_{J'}^\epsilon}{u}{L^2(I)}
        = \inpspace{\chi_J}{u}{L^2(I)}
        = \int_J u \odif x \,. \qedhere
    \]
\end{proof}

\begin{corollary}
    \label{cor:p-infty-prefix}
    Let $\delta \in I$ and let $J \define (0, \delta)$. Let $u \in L^2(I)$. Then
    \[
        \int_J (P_\infty u) \odif x \le \int_J u \odif x \,.
    \]
\end{corollary}
\begin{proof}
    This follows from \cref{lemma:pt-prefix} along with the convergence $P_t u \to P_\infty u$ in
    $L^2(I)$.
\end{proof}

\begin{definition}
    \label{def:domination}
    Let $\mu$ and $\nu$ be probability measures on $\bR$. We say $\mu$ \emph{dominates} $\nu$, and
    write $\mu \succeq \nu$, if for every $x \in \bR$ we have $\mu(-\infty, x) \ge \nu(-\infty, x)$.
\end{definition}


\begin{corollary}
    \label{cor:mu-infty-dominated}
    Let $u \in \cU$ satisfy \cref{cond:initial-prob-state}. Define the measures $\odif \mu \define u
    \odif x$ and $\odif \mu_\infty \define (P_\infty u) \odif x$. Then $\mu \succeq \mu_\infty$.
\end{corollary}
\begin{proof}
    We may view $\mu$ and $\mu_\infty$ as absolutely continuous measures on all of $\bR$, taking
    zero outside $[0,1]$. Now, for each $x \in [0,1]$, \cref{cor:p-infty-prefix} gives
    \[
        \mu(-\infty, x)
        = \int_{(0,x)} u \odif y
        \ge \int_{(0,x)} (P_\infty u) \odif y
        = \mu_\infty(-\infty, x) \,.
    \]
    On the other hand, for $x < 0$ both sides are zero, and for $x > 1$, both sides are $1$ by
    \cref{cor:p-infty-mass-conserving}.
\end{proof}

\subsection{Directed optimal transport and Pythagorean composition}
\label{section:directed-optimal-transport}

We now introduce directed versions of basic optimal transport concepts and theory, with the goal of
combining a directed version of weak Kantorovich duality with a multidimensional directed version of
the transport-energy inequality from \cref{thm:undirected-transport-energy} to obtain our desired
directed Poincar\'e inequality via a perturbation argument. In the interest of space, we refrain
starting a systematic study of directed optimal transport, but rather limit ourselves to the results
we require.

\begin{definition}[Quasimetric]
    Let $X$ be a set. An extended real-valued function $d : X \times X \to [0, +\infty]$, not
    necessarily symmetric, is called a \emph{quasimetric} on $X$ if it satisfies the following:
    \begin{enumerate}
        \item For all $x, y \in X$, $d(x, y) \ge 0$ with $d(x, y) = 0$ if and only if $x = y$.
        \item For all $x, y, z \in X$, $d(x, z) \le d(x, y) + d(y, z)$.
    \end{enumerate}
\end{definition}

\begin{definition}
    Let $\Omega \subset \bR^d$ be a Borel set, and let $\mu$ and $\nu$ be probability measures on
    $\Omega$. We define the set $\Pi(\mu \to \nu)$ of \emph{directed couplings}, or \emph{directed
    transport plans}, from $\mu$ to $\nu$ as
    \[
        \Pi(\mu \to \nu) \define \left\{
            \gamma \in \Pi(\mu, \nu)
            : \int_{\Omega \times \Omega} \chi_{\{x \not\preceq y\}} \odif \gamma(x, y) = 0
            \right\} \,,
    \]
    where the integral is well-defined because the set $\{(x, y) \in \bR^d \times \bR^d : x
    \not\preceq y\}$ is open and hence Borel measurable. Note that the condition could also be
    written as $\gamma\left( \{x \not\preceq y \} \right) = 0$, or $\gamma\left( \{x \preceq y\}
    \right) = 1$.
\end{definition}

\begin{definition}[Directed Wasserstein distance]
    Let $\Omega \subset \bR^d$ be a bounded Borel set and let $p \in [1, \infty)$. Given two
    probability distributions $\mu, \nu$ over $\Omega$, we define the \emph{directed $p$-Wasserstein
    distance from $\mu$ to $\nu$} by
    \[
        W_p(\mu \to \nu) \define \inf_{\gamma \in \Pi(\mu \to \nu)} C_p(\gamma) \,,
    \]
    and we write $W_p^p(\mu \to \nu) \define W_p(\mu \to \nu)^p$. Note that we may have $W_p(\mu \to
    \nu) = +\infty$, and that $W_p(\cdot \to \cdot)$ is not symmetric in general.
\end{definition}

\begin{observation}
    \label{obs:wp-ineq}
    For all measures $\mu$ and $\nu$, we have $\Pi(\mu \to \nu) \subset \Pi(\mu, \nu)$ and hence
    $W_p(\mu, \nu) \le W_p(\mu \to \nu)$.
\end{observation}

\begin{lemma}[Gluing lemma; see \eg {\cite[Lemma 5.5]{San15}}]
    \label{lemma:gluing}
    Let $\Omega \subset \bR^d$ be a Borel set. Let $\mu, \varrho, \nu$ be probability measures on
    $\Omega$, and let $\gamma^+ \in \Pi(\mu, \varrho)$ and $\gamma^- \in \Pi(\varrho, \nu)$. Then
    there exists a probability measure $\sigma$ on $\Omega \times \Omega \times \Omega$ such that
    $\transmap{(\pi_{x,y})}{\sigma} = \gamma^+$ and $\transmap{(\pi_{y,z})}{\sigma} = \gamma^-$.
\end{lemma}

The proof of that $W_p(\cdot \to \cdot)$ is a quasimetric follows the presentation of \cite[Lemma
5.4]{San15}, with a simple additional argument to handle \emph{directed} couplings.


\begin{lemma}[Composition of directed transport plans]
    \label{lemma:directed-composition}
    In \cref{lemma:gluing}, if $\gamma^+ \in \Pi(\mu \to \varrho)$ and $\gamma^- \in \Pi(\varrho \to
    \nu)$, then $\transmap{(\pi_{x,z})}{\sigma} \in \Pi(\mu \to \nu)$.
\end{lemma}
\begin{proof}
    Let $\gamma \define \transmap{(\pi_{x,z})}{\sigma}$. First, since $\transmap{(\pi_x)}{\gamma} =
    \transmap{(\pi_x \circ \pi_{x,z})}{\sigma} = \transmap{(\pi_x \circ \pi_{x,y})}{\sigma} =
    \transmap{(\pi_x)}{\gamma^+} = \mu$, and similarly $\transmap{(\pi_z)}{\gamma} = \nu$, we have
    $\gamma \in \Pi(\mu, \nu)$. Moreover, by definition of pushforward measure we have
    \begin{align*}
        \int_{\Omega \times \Omega} \chi_{\{x \not\preceq z\}} \odif \gamma(x, z)
        &= \int_{\Omega \times \Omega \times \Omega}
            \chi_{\{x \not\preceq z\}} \odif \sigma(x, y, z)
        \le \int_{\Omega \times \Omega \times \Omega}
            \left( \chi_{\{x \not\preceq y\}} + \chi_{\{y \not\preceq z\}} \right)
            \odif \sigma(x, y, z) \\
        &= \int_{\Omega \times \Omega \times \Omega}
                \chi_{\{x \not\preceq y\}} \odif \sigma(x, y, z)
            + \int_{\Omega \times \Omega \times \Omega}
                \chi_{\{y \not\preceq z\}} \odif \sigma(x, y, z) \\
        &= \int_{\Omega \times \Omega}
                \chi_{\{x \not\preceq y\}} \odif \gamma^+(x, y)
            + \int_{\Omega \times \Omega}
                \chi_{\{y \not\preceq z\}} \odif \gamma^-(y, z) = 0 \,,
    \end{align*}
    the last equality since $\gamma^+$ and $\gamma^-$ are directed couplings. Hence $\gamma \in
    \Pi(\mu \to \nu)$ as claimed.
\end{proof}

\begin{proposition}
    Let $\Omega \subset \bR^d$ be a bounded Borel set and $p \in [1, \infty)$. Then $W_p(\cdot \to
    \cdot)$ is a quasimetric on $P(\Omega)$.
\end{proposition}
\begin{proof}
    It is clear that $W_p(\mu \to \nu) \ge 0$ always. If $\mu = \nu$ then the identity coupling
    shows that $W_p(\mu \to \nu) = 0$, and conversely if $W_p(\mu \to \nu) = 0$ then, by
    \cref{obs:wp-ineq}, $W_p(\mu, \nu) = 0$ and hence $\mu = \nu$. It remains to show that $W_p(\cdot
    \to \cdot)$ satisfies the triangle inequality.

    Let $\mu, \varrho, \nu$ be probability measures on $\Omega$. Let $\gamma^+ \in \Pi(\mu \to
    \varrho)$ and $\gamma^- \in \Pi(\varrho \to \nu)$. Then since we also have $\gamma^+ \in
    \Pi(\mu, \varrho)$ and $\gamma^- \in \Pi(\varrho, \nu)$, apply \cref{lemma:gluing} to obtain a
    probability measure $\sigma$ on $\Omega \times \Omega \times \Omega$ such that
    $\transmap{(\pi_{x,y})}{\sigma} = \gamma^+$ and $\transmap{(\pi_{y,z})}{\sigma} = \gamma^-$. Let
    $\gamma \define \transmap{(\pi_{x,z})}{\sigma}$, so that $\gamma \in \Pi(\mu \to \nu)$ by
    \cref{lemma:directed-composition}. Then
    {\allowdisplaybreaks
    \begin{align*}
        W_p(\mu \to \nu)
        &\le \left( \int_{\Omega \times \Omega} |x-z|^p \odif \gamma(x, z) \right)^{1/p}
        = \left( \int_{\Omega \times \Omega \times \Omega}
                        |x-z|^p \odif \sigma(x, y, z) \right)^{1/p} \\
        &= \| |x-z| \|_{L^p(\sigma)}
        \le \| |x-y| + |y-z| \|_{L^p(\sigma)}
        \le \| |x-y| \|_{L^p(\sigma)} + \| |y-z| \|_{L^p(\sigma)} \\
        &= \left( \int_{\Omega \times \Omega \times \Omega}
                        |x-y|^p \odif \sigma(x, y, z) \right)^{1/p}
            + \left( \int_{\Omega \times \Omega \times \Omega}
                        |y-z|^p \odif \sigma(x, y, z) \right)^{1/p} \\
        &= \left( \int_{\Omega \times \Omega} |x-y|^p \odif \gamma^+(x, y) \right)^{1/p}
            + \left( \int_{\Omega \times \Omega} |y-z|^p \odif \gamma^-(y, z) \right)^{1/p} \,.
    \end{align*}
    }%
    Since $\gamma^+ \in \Pi(\mu \to \varrho), \gamma^- \in \Pi(\varrho \to \nu)$ were arbitrary,
    $W_p(\mu \to \nu) \le W_p(\mu \to \varrho) + W_p(\varrho \to \nu)$.
\end{proof}

The directed Wasserstein distance arises naturally in the one-dimensional case when one probability
measure dominates the other in the sense of \cref{def:domination}, as we now show.

\begin{proposition}[Specialization of {\cite[Theorem 2.9]{San15}}]
    \label{prop:mon-maps}
    Let $p \in (1, +\infty)$. Let $\mu$ and $\nu$ be two absolutely continuous probability measures
    on $\overline I$ with strictly positive densities. Then there exists a unique $\gamma \in
    \Pi(\mu, \nu)$ attaining $C_p(\gamma) = W_p(\mu, \nu)$, and in fact $\gamma = \transmap{(\idmap,
    T_\mon)}{\mu}$ where $T_\mon : \overline I \to \overline I$ is a nondecreasing function that is
    \almev uniquely determined, and $\idmap$ denotes the identity map.
\end{proposition}

\begin{corollary}
    \label{cor:directed-wp-1d}
    If in \cref{prop:mon-maps} we have $\mu \succeq \nu$, then in fact $\gamma \in \Pi(\mu \to
    \nu)$. As a consequence, $W_p(\mu \to \nu) = W_p(\mu, \nu)$.
\end{corollary}
\begin{proof}
    We claim that $T_\mon(x) \ge x$ for every $x \in I$. Indeed, suppose $T_\mon(x) < x$ for some $x
    \in I$. Then
    \[
        \nu(0, T_\mon(x))
        = (\transmap{T_\mon}{\mu})(0, T_\mon(x))
        = \mu(T_\mon^{-1}(0, T_\mon(x)))
        \ge \mu(0, x) \,,
    \]
    the inequality because certainly $(0,x) \subseteq T_\mon^{-1}(0, T_\mon(x))$, but we do not rule
    out at this point that $T$ remains constant for a while after $x$. But since $T_\mon(x) < x$ and
    $\nu$ has strictly positive density by assumption, we conclude that
    \[
        \nu(0, x) > \nu(0, T_\mon(x)) \ge \mu(0, x) \,,
    \]
    contradicting the assumption that $\mu \succeq \nu$. Hence $T_\mon(x) \ge x$ for every $x \in I$
    as claimed. We now show that $\gamma \in \Pi(\mu \to \nu)$. Let $S \define \{(x, y) \in
    \overline I \times \overline I : x \le y\}$. Since $\gamma = \transmap{(\idmap, T_\mon)}{\mu}$,
    we have
    \[
        \gamma(S)
        = \mu((\idmap, T_\mon)^{-1}(S))
        = \mu(\{x \in \overline I : x \le T_\mon(x)\})
        = \mu(\overline I)
        = 1 \,,
    \]
    so $\gamma \in \Pi(\mu \to \nu)$ as claimed.
\end{proof}

We now introduce the formal language for two related operations: constructing a probability measure
by its marginal and conditional distributions, and the inverse process of extracting marginal and
conditionals from a probability measure, which is also called \emph{disintegration}. We refer the
reader to \cite[Section 5.3]{AGS05} for an overview of these ideas from a measure-theoretic
perspective.

We use the following definition from \cite{AGS05}. For $X$ and $Y$ separable metric spaces and $x
\in X \mapsto \mu_x \in P(Y)$ a measure-valued map, we say $\mu_x$ is a Borel map if $x \mapsto
\mu_x(B)$ is a Borel map for any Borel set $B \subset Y$, or equivalently if this holds for any open
set $A \subset Y$. In this case we also have that
\[
    x \mapsto \int_Y f(x,y) \odif {\mu_x(y)}
\]
is Borel for any bounded (or nonnegative) Borel function $f : X \times Y \to \bR$.

\begin{definition}[Construction and disintegration]
    \label{def:disintegration}
    Let $S$ be a finite set and let $I \subset S$. Let $x_{S \setminus I} \in \bR^{S \setminus I}
    \mapsto \cond{\mu}{x_{S \setminus I}} \in P(\bR^I)$ be a Borel map (the \emph{conditionals}).
    Then for any bounded (or nonnegative) Borel function $f : \bR^S \to \bR$, the $\bR^{S \setminus
    I} \to \bR$ map
    \begin{equation}
        \label{eq:disintegration-map}
        x_{S \setminus I}
            \mapsto \int_{\bR^I} f(x_{S \setminus I}, x_I)
            \odif {\cond{\mu}{x_{S \setminus I}}(x_I)}
    \end{equation}
    is Borel. Therefore, for any $\marginal{\mu}{S \setminus I} \in P(\bR^{S \setminus I})$ (the
    \emph{marginal}), we define $\mu \in P(\bR^S)$ implicitly by
    \begin{equation}
        \label{eq:disintegration-condition}
        \int_{\bR^S} f \odif \mu
        = \int_{\bR^{S \setminus I}} \odif {\marginal{\mu}{S \setminus I}(x_{S \setminus I})}
                \int_{\bR^I} f(x_{S \setminus I}, x_I) \odif {\cond{\mu}{x_{S \setminus I}}(x_I)} \,,
    \end{equation}
    and we formally write $\mu = \int_{\bR^{S \setminus I}} \cond{\mu}{x_{S \setminus I}} \odif
    {\marginal{\mu}{S \setminus I}}$.

    Conversely, given $\mu \in P(\bR^S)$, we let $\marginal{\mu}{S \setminus I} \define
    \transmap{(\pi_{S \setminus I})}{\mu} \in P(\bR^{S \setminus I})$ and write $(\cond{\mu}{x_{S
    \setminus I}})_{x_{S \setminus I} \in \bR^{S \setminus I}} \subset P(\bR^I)$ for the
    $\marginal{\mu}{S \setminus I}$-\almev uniquely determined Borel family of probability measures
    such that $\mu = \int_{\bR^{S \setminus I}} \cond{\mu}{x_{S \setminus I}} \odif
    {\marginal{\mu}{S \setminus I}}$, and we call this decomposition a \emph{disintegration} of
    $\mu$.
\end{definition}

\paragraph*{Examples and simplified notational conventions.} Since the index notation in the
definition above can be somewhat laborious to parse, let us briefly give two concrete settings we
will use below, and introduce simplified notational conventions for those settings. The
simplifications are supposed to be mnemonic for what we have already defined for projections.
\begin{enumerate}
    \item For a probability measure $\mu \in P(\bR^d)$ and index $i \in [d]$, we may disintegrate
        $\mu$ into a marginal on index set $[d] \setminus \{i\}$ and conditionals supported along
        $i$-th coordinate, and write $\mu = \int_{\bR^{[d] \setminus \{i\}}} \cond{\mu}{x_{[d]
        \setminus \{i\}}} \odif {\marginal{\mu}{[d] \setminus \{i\}}}$. We simplify this notation by
        writing $\mu = \int_{\bR^d_{-i}} \cond{\mu}{x_{-i}} \odif {\marginal{\mu}{-i}}$.
    \item For a transport plan $\gamma \in P(\bR^d \times \bR^d)$, we write the index set as $[2d] =
        S = S_x \cup S_y$ for $S_x \define [d]$ and $S_y \define [2d] \setminus [d]$, so we write
        each element in the support of $\gamma$ as a tuple $(x, y)$ with $x \in \bR^{S_x}$ and $y
        \in \bR^{S_y}$. Then, given index set $I \subset [d]$, we may disintegrate $\gamma$ into a
        marginal on index set $S_M \define (S_x \setminus I_x) \cup (S_y \setminus I_y)$ for $I_x
        \define I$ and $I_y \define \{d+i : i \in I\}$, so that we think of the marginal as
        determining all but the $I$-indexed coordinates of the points $x$ and $y$; and we view the
        conditionals as supported along the $I_x \cup I_y$ directions. This disintegration is
        $\gamma = \int_{\bR^{S_M}} \cond{\gamma}{x_{S_M}} \odif {\marginal{\gamma}{S_M}}$. To be
        clear, here we have $\marginal{\gamma}{S_M} \in P(\bR^{S_M})$ and $\cond{\gamma}{x_{S_M}}
        \in P(\bR^{I_x \cup I_y})$ for each $x_{S_M} \in \bR^{S_M}$. We make the notation somewhat
        more mnemonic by writing $\gamma = \int_{\bR^d_{-I} \times \bR^d_{-I}} \cond{\gamma}{x_{-I},
        y_{-I}} \odif {\marginal{\gamma}{-I_x -I_y}}$. When we are considering a singleton $I =
        \{i\}$, we further simplify the notation by directly writing $i$ in the place of $I$
        in this formula.
\end{enumerate}

For absolutely continuous probability measures on $\closedInt^d$, the following more familiar
characterization of disintegrations will be useful:

\begin{proposition}
    \label{prop:disintegration-absolutely-continuous}
    Let $\mu$ be an absolutely continuous probability measure whose density is Borel measurable and
    supported in $\closedInt^d$. Namely, write $\odif \mu = u \odif x$ where $u : \closedInt^d \to
    [0, +\infty)$ is a Borel map. Let $i \in [d]$. Then the disintegration $\mu = \int_{\bR^d_{-i}}
    \cond{\mu}{x_{-i}} \odif {\marginal{\mu}{-i}}$ is as follows: $\marginal{\mu}{-i}$ and each
    $\cond{\mu}{x_{-i}}$ are absolutely continuous probability measures with Borel densities
    supported in $\closedInt^d_{-i}$ and $[0,1]$ respectively; and writing $\odif
    {\marginal{\mu}{-i}} = \marginal{u}{-i} \odif x_{-i}$, and $\odif {\cond{\mu}{x_{-i}}} =
    \cond{u}{x_{-i}} \odif x_i$ for each $x_{-i}$, where the (Borel) density functions are
    $\marginal{u}{-i} : \closedInt^d_{-i} \to [0, +\infty)$ and $\cond{u}{x_{-i}} : [0,1] \to [0,
    +\infty)$ for each $x_{-i} \in \closedInt^d_{-i}$, we have (marginal density)
    \[
        \marginal{u}{-i}(x_{-i}) = \int_I u(x_{-i}, x_i) \odif x_i
    \]
    for each $x_{-i} \in \closedInt^d_{-i}$, and (conditional density)
    \[
        \cond{u}{x_{-i}}(x_i) = \frac{u(x)}{\marginal{u}{-i}(x_{-i})}
    \]
    for each $x \in \closedInt^d$ where the denominator is nonzero, or $1$ by convention otherwise.
\end{proposition}
\begin{proof}
    First, note that it is standard that $\marginal{u}{-i}$ is a Borel map since it is the integral
    of the section $u(x_{-i}, \cdot)$ at each point $x_{-i}$. It is hence immediate that each
    $\cond{u}{x_{-i}}$ is also Borel. Thus $\marginal{\mu}{-i}$ and all $\cond{\mu}{x_{-i}}$ are
    indeed absolutely continuous probability measures.

    We now check the conditions of \cref{def:disintegration}. We claim that the map $x_{-i} \in
    \bR^d_{-i} \mapsto \cond{\mu}{x_{-i}} \in P(\bR)$ is Borel. Let $B \subset \bR$ be a Borel set;
    we may assume that $B \subset [0,1]$ without loss of generality, since all measures assign zero
    outside this interval. Then the map $x_{-i} \in \bR^d_{-i} \mapsto \cond{\mu}{x_{-i}}(B)$ takes
    value zero outside $\closedInt^d_{-i}$, and for $x_{-i} \in \closedInt^d_{-i}$, it is
    \[
        x_{-i} \mapsto \frac{1}{\marginal{u}{-i}(x_{-i})} \int_B u(x_{-i}, x_i) \odif x_i \,.
    \]
    This is again, at each point, the integral of a section over the subspace $B$, so it is standard
    that this map is Borel. Hence the claim holds.

    It remains to verify \eqref{eq:disintegration-condition}. Let $f : \bR^d \to [0, +\infty)$ be a
    nonnegative Borel function (and the result for bounded $f$ will also follow). By the definition
    of the measures and Tonelli's theorem, we have
    \begin{align*}
        &\int_{\bR^{-i}} \odif {\marginal{\mu}{-i}(x_{-i})}
                \int_\bR f(x_{-i}, x_i) \odif {\cond{u}{x_{-i}}(x_i)}
        = \int_{\closedInt^d_{-i}} \marginal{u}{-i}(x_{-i}) \odif x_{-i}
                \int_I f(x_{-i}, x_i) \cond{u}{x_{-i}}(x_i) \odif x_i \\
        &\qquad = \int_{\closedInt^d}
                f(x) \marginal{u}{-i}(x_{-i}) \cond{u}{x_{-i}}(x_i) \odif x
        = \int_{\closedInt^d} f(x) u(x) \odif x \,,
    \end{align*}
    where the last equality holds because $u(x) = \marginal{u}{-i}(x_{-i}) \cond{u}{x_{-i}}(x_i)$
    whenever $\marginal{u}{-i}(x_{-i}) > 0$, and if $\marginal{u}{-i}(x_{-i}) = 0$ then $u(x_{-i},
    \cdot) = 0$ almost everywhere, so $\left\{ x \in \closedInt^d : u(x) > 0 \text{ and }
    \marginal{u}{-i}(x_{-i}) = 0 \right\}$ is Borel and has measure zero (by another application of
    Tonelli's theorem).
\end{proof}

\begin{definition}[Aligned transport plans]
    Let $I \subsetneq [d]$ be nonempty and let $\mu, \nu \in P(\bR^d)$. Let $\gamma \in \Pi(\mu,
    \nu)$ and write its disintegration $\gamma = \int_{\bR^d_{-I} \times \bR^d_{-I}}
    \cond{\gamma}{x_{-I}, y_{-I}} \odif {\marginal{\gamma}{-I_x -I_y}}$. We say $\gamma$ is
    \emph{$I$-aligned} if its marginal $\marginal{\gamma}{-I_x -I_y} \in P(\bR^d_{-I} \times
    \bR^d_{-I})$ satisfies the following:
    \[
        \marginal{\gamma}{-I_x -I_y}\left(
            \left\{ (x_{-I}, y_{-I}) \in \bR^d_{-I} \times \bR^d_{-I} : x_{-I} \ne y_{-I} \right\}
        \right)
        = 0 \,.
    \]
    Note that this condition is well-defined because the set being measured is open and hence Borel.
    By convention, we say that every $\gamma \in \Pi(\mu, \nu)$ is $[d]$-aligned. For any nonempty
    $I \subseteq [d]$, we denote the set of $I$-aligned transport plans by $\Pi_I(\mu, \nu)$. When
    we have a singleton $I = \{i\}$, we write $i$ directly in the place of $I$ in this definition.
\end{definition}

\begin{definition}
    For nonempty $I \subseteq [d]$ and $\mu, \nu \in P(\bR^d)$, write $\Pi_I(\mu \to \nu) \define
    \Pi(\mu \to \nu) \cap \Pi_I(\mu, \nu)$.
\end{definition}

\begin{lemma}[Alternative characterization of $I$-aligned plans]
    \label{lemma:aligned-alternative}
    Let $I \subsetneq [d]$ be nonempty and let $\mu, \nu \in P(\bR^d)$. Then $\gamma \in \Pi(\mu,
    \nu)$ is $I$-aligned if and only if
    \[
        \int_{\bR^d \times \bR^d} \chi_{\{x_{-I} \ne y_{-I}\}} \odif \gamma(x,y) = 0 \,.
    \]
\end{lemma}
\begin{proof}
    Writing the disintegration $\gamma = \int_{\bR^d_{-I} \times \bR^d_{-I}} \cond{\gamma}{x_{-I},
    y_{-I}} \odif {\marginal{\gamma}{-I_x -I_y}}$, we have
    {\allowdisplaybreaks
    \begin{align*}
        \marginal{\gamma}{-I_x -I_y}(\{x_{-I} \ne y_{-I}\})
        &= \int_{\bR^d_{-I} \times \bR^d_{-I}}
                \chi_{\{x_{-I} \ne y_{-I}\}}
                \odif {\marginal{\gamma}{-I_x -I_y}}
                \underbrace{\int_{\bR^I \times \bR^I}
                \odif {{\cond{\gamma}{x_{-I}, y_{-I}}(x_I, y_I)}}}_{=1} \\
        &= \int_{\bR^d_{-I} \times \bR^d_{-I}}
                \odif {\marginal{\gamma}{-I_x -I_y}}
                \int_{\bR^I \times \bR^I}
                    \chi_{\{x_{-I} \ne y_{-I}\}}
                    \odif {{\cond{\gamma}{x_{-I}, y_{-I}}(x_I, y_I)}} \\
        &= \int_{\bR^d \times \bR^d} \chi_{\{x_{-I} \ne y_{-I}\}} \odif \gamma \,. \qedhere
    \end{align*}
    }%
\end{proof}

While we defined $I$-aligned transport plans for probability measures on all of $\bR^d$, we extend
the definition to measures on subsets $\Omega \subset \bR^d$ (\eg on the cube) by canonically
extending any such measure to all of $\bR^d$ by assigning zero outside of $\Omega$.

\begin{lemma}[Composition of aligned transport plans]
    \label{lemma:aligned-composition}
    Let $I, J \subseteq [d]$ be nonempty. Then in \cref{lemma:gluing}, if $\gamma^+ \in \Pi_I(\mu,
    \varrho)$ and $\gamma^- \in \Pi_J(\varrho, \nu)$, then $\transmap{(\pi_{x,z})}{\sigma} \in
    \Pi_{I \cup J}(\mu, \nu)$.
\end{lemma}
\begin{proof}
    Let $\gamma \define \transmap{(\pi_{x,z})}{\sigma}$ and $K \define I \cup J$. As in the proof of
    \cref{lemma:directed-composition}, we do have $\gamma \in \Pi(\mu, \nu)$, so it remains to show
    that $\gamma$ is $K$-aligned (and we may assume that $K \subsetneq [n]$, otherwise there is
    nothing to prove). We have
    \begin{align*}
        &\int_{\bR^d \times \bR^d} \chi_{\{x_{-K} \ne z_{-K}\}} \odif \gamma \\
        &\qquad = \int_{\bR^d \times \bR^d \times \bR^d} \chi_{\{x_{-K} \ne z_{-K}\}}(x, z)
                \odif \sigma(x, y, z)
            & \text{(Pushforward)} \\
        &\qquad \le \int_{\bR^d \times \bR^d \times \bR^d}
                \left( \chi_{\{x_{-K} \ne y_{-K}\}}(x, y)
                    + \chi_{\{y_{-K} \ne z_{-K}\}}(y, z) \right)
                \odif \sigma(x, y, z) \\
        &\qquad
            = \int_{\bR^d \times \bR^d \times \bR^d} \chi_{\{x_{-K} \ne y_{-K}\}}
                \odif \sigma
            + \int_{\bR^d \times \bR^d \times \bR^d} \chi_{\{y_{-K} \ne z_{-K}\}}
                \odif \sigma \\
        &\qquad
            = \int_{\bR^d \times \bR^d} \chi_{\{x_{-K} \ne y_{-K}\}} \odif \gamma^+
            + \int_{\bR^d \times \bR^d} \chi_{\{y_{-K} \ne z_{-K}\}} \odif \gamma^-
            & \text{(Pushforward)} \\
        &\qquad
            \le \int_{\bR^d \times \bR^d} \chi_{\{x_{-I} \ne y_{-I}\}} \odif \gamma^+
            + \int_{\bR^d \times \bR^d} \chi_{\{y_{-J} \ne z_{-J}\}} \odif \gamma^-
            & \text{($I, J \subseteq K$)} \\
        &\qquad = 0 \,,
    \end{align*}
    the last step since $\gamma^+$ is $I$-aligned and $\gamma^-$ is $J$-aligned and by
    \cref{lemma:aligned-alternative}; and again by the latter, $\gamma$ is $K$-aligned.
\end{proof}

The following results let us find aligned transport plans by transporting mass only within
$i$-aligned lines.


\begin{lemma}[Measurable selection; specialization of {\cite[Corollary 5.22]{Vil09}}]
    \label{lemma:measurable-selection}
    Let $\Omega \subset \bR^d$ be a bounded set. Let $p \in [1, \infty)$. Let $A$ be a measurable
    space and let $a \mapsto (\mu_a, \nu_a)$ be a measurable function $A \to P(\bR) \times P(\bR)$
    with each $\mu_a$ and $\nu_a$ having bounded support. Then there is a measurable choice $a
    \mapsto \gamma_a$ such that, for each $a \in A$, $\gamma_a \in \Pi(\mu_a, \nu_a)$ and
    $C_p(\gamma_a) = W_p(\mu_a, \nu_a)$.
\end{lemma}

\begin{lemma}
    \label{lemma:i-aligned-plan}
    Let $p \in [1, \infty)$, $i \in [d]$, and let $\mu, \nu \in P(\bR^d)$ be supported inside some
    bounded set. Suppose $\marginal{\mu}{-i} = \marginal{\nu}{-i}$. Then there exists $\gamma \in
    \Pi_i(\mu, \nu)$ satisfying
    \[
        C_p(\gamma)^p
        = \int_{\bR^d_{-i}}
            W_p^p(\cond{\mu}{x_{-i}}, \cond{\nu}{x_{-i}}) \odif \mu_{-i}(x_{-i}) \,.
    \]
\end{lemma}
\begin{proof}
    Let $A \define \bR^d_{-i}$ be equipped with the Borel $\sigma$-algebra on $\bR^{[d] \setminus
    \{i\}}$. Then the $A \to P(\bR) \times P(\bR)$ function $x_{-i} \mapsto (\cond{\mu}{x_{-i}},
    \cond{\nu}{x_{-i}})$ is Borel measurable because, by disintegration, both $x_{-i} \mapsto
    \cond{\mu}{x_{-i}}$ and $x_{-i} \mapsto \cond{\nu}{x_{-i}}$ are Borel measurable. By
    \cref{lemma:measurable-selection}, we obtain a Borel measurable map $x_{-i} \mapsto
    \gamma_{x_{-i}}$ such that, for each $x_{-i} \in \bR^d_{-i}$, $\gamma_{x_{-i}} \in
    \Pi(\cond{\mu}{x_{-i}}, \cond{\nu}{x_{-i}})$ and $C_p(\gamma_{x_{-i}})^p =
    W_p^p(\cond{\mu}{x_{-i}}, \cond{\nu}{x_{-i}})$.

    Define $\gamma \in P(\bR^d \times \bR^d)$ by its disintegration $\gamma = \int_{\bR^d_{-i}
    \times \bR^d_{-i}} \cond{\gamma}{x_{-i}, y_{-i}} \odif {\marginal{\gamma}{-i_x -i_y}}$ as
    follows. First, for the marginal $\marginal{\gamma}{-i_x -i_y} \in P(\bR^d_{-i} \times
    \bR^d_{-i})$, for each Borel $Z \subset \bR^d_{-i} \times \bR^d_{-i}$ we set
    \[
        \marginal{\gamma}{-i_x -i_y}(Z)
        \define \int_{\bR^d_{-i}} \chi_Z(x_{-i}, x_{-i}) \odif {\marginal{\mu}{-i}(x_{-i})} \,,
    \]
    where the integral is well-defined because the set $\{x_{-i} \in \bR^d_{-i} : (x_{-i}, x_{-i})
    \in Z\}$ is Borel by standard arguments. We remark that $\marginal{\gamma}{-i_x -i_y}$ is
    indeed a probability measure because it is nonnegative, it assigns zero to the empty set, it is
    additive over countable disjoint unions, and
    \[
        \marginal{\gamma}{-i_x -i_y}(\bR^d_{-i} \times \bR^d_{-i})
        = \int_{\bR^d_{-i}} \odif {\marginal{\mu}{-i}}
        = 1 \,.
    \]
    It follows that for every bounded (or nonnegative) Borel function $f : \bR^d_{-i} \times
    \bR^d_{-i} \to \bR$, we have
    \begin{equation}
        \label{eq:gamma-marginal-evaluation}
        \int_{\bR^d_{-i} \times \bR^d_{-i}} f(x_{-i}, y_{-i})
                \odif {{\marginal{\gamma}{-i_x -i_y}(x_{-i}, y_{-i})}}
        = \int_{\bR^d_{-i}} f(x_{-i}, x_{-i}) \odif {\marginal{\mu}{-i}(x_{-i})} \,.
    \end{equation}
    Second, for the conditionals, we simply set $\cond{\gamma}{x_{-i}, y_{-i}} \define
    \gamma_{x_{-i}}$ for each $(x_{-i}, y_{-i}) \in \bR^d_{-i} \times \bR^d_{-i}$.

    We claim that $\gamma \in \Pi_i(\mu, \nu)$. The $i$-aligned condition holds by construction, so
    it remains to show that $\gamma \in \Pi(\mu, \nu)$. We first verify that
    $\transmap{(\pi_1)}{\gamma} = \mu$. It suffices to verify that these measures agree on each
    rectangle $X \define X_{-i} \times X_i$ of Borel sets $X_{-i} \subset \bR^d_{-i}$, $X_i \subset
    \bR$. We have
    {\allowdisplaybreaks
        \begin{align*}
            &(\transmap{(\pi_1)}{\gamma})(X) \\
            &\quad = \int_{\bR^d \times \bR^d} \chi_X(x) \odif \gamma(x, y)
                & \text{(Pushforward)} \\
            &\quad = \int_{\bR^d_{-i} \times \bR^d_{-i}}
                    \odif {{\marginal{\gamma}{-i_x -i_y}(x_{-i}, y_{-i})}}
                    \int_{\bR \times \bR} \chi_X(x_{-i}, x_i)
                            \odif {{\cond{\gamma}{x_{-i}, y_{-i}}(x_i, y_i)}}
                & \text{(Disintegration)} \\
            &\quad = \int_{\bR^d_{-i} \times \bR^d_{-i}}
                    \chi_{X_{-i}}(x_{-i})
                    \odif {{\marginal{\gamma}{-i_x -i_y}(x_{-i}, y_{-i})}}
                    \int_{\bR \times \bR} \chi_{X_i}(x_i)
                            \odif {{\gamma_{x_{-i}}(x_i, y_i)}}
                & \text{(Definition of $X, \cond{\gamma}{x_{-i}, y_{-i}}$)} \\
            &\quad = \int_{\bR^d_{-i} \times \bR^d_{-i}}
                    \chi_{X_{-i}}(x_{-i})
                    (\transmap{(\pi_1)}{\gamma_{x_{-i}}})(X_i)
                    \odif {{\marginal{\gamma}{-i_x -i_y}(x_{-i}, y_{-i})}}
                & \text{(Pushforward)} \\
            &\quad = \int_{\bR^d_{-i} \times \bR^d_{-i}}
                    \chi_{X_{-i}}(x_{-i})
                    \cond{\mu}{x_{-i}}(X_i)
                    \odif {{\marginal{\gamma}{-i_x -i_y}(x_{-i}, y_{-i})}}
                & \text{($\gamma_{x_{-i}} \in \Pi(\cond{\mu}{x_{-i}}, \cond{\nu}{x_{-i}})$)} \\
            &\quad = \int_{\bR^d_{-i}} \chi_{X_{-i}}(x_{-i}) \cond{\mu}{x_{-i}}(X_i)
                    \odif {\marginal{\mu}{-i}(x_{-i})}
                & \text{(Application of \eqref{eq:gamma-marginal-evaluation})} \\
            &\quad = \int_{\bR^d_{-i}} \chi_{X_{-i}}(x_{-i})
                    \odif {\marginal{\mu}{-i}(x_{-i})}
                    \int_{\bR} \chi_{X_i}(x_i) \odif {{\cond{\mu}{x_{-i}}(x_i)}} \\
            &\quad = \int_{\bR^d_{-i}}
                    \odif {\marginal{\mu}{-i}(x_{-i})}
                    \int_{\bR} \chi_X(x_{-i}, x_i) \odif {{\cond{\mu}{x_{-i}}(x_i)}}
                & \text{(Definition of $X$)} \\
            &\quad = \int_{\bR^d} \chi_X \odif \mu
                = \mu(X)
                & \text{(Disintegration)} \,,
        \end{align*}
    }%
    as desired. Thanks to the hypothesis that $\marginal{\mu}{-i} = \marginal{\nu}{-i}$, an
    analogous calculation on the second variable yields that $\transmap{(\pi_2)}{\gamma} = \nu$ and
    hence $\gamma \in \Pi(\mu, \nu)$. Thus $\gamma \in \Pi_i(\mu, \nu)$ as claimed.

    Finally, we compute the cost of the plan $\gamma$. Again using its disintegration and
    \eqref{eq:gamma-marginal-evaluation}, we have
    {\allowdisplaybreaks
    \begin{align*}
        C_p(\gamma)^p
        &= \int_{\bR^d \times \bR^d} |x-y|^p \odif \gamma(x, y) \\
        &= \int_{\bR^d_{-i} \times \bR^d_{-i}}
                \odif {{\marginal{\gamma}{-i_x -i_y}(x_{-i}, y_{-i})}}
                \int_{\bR \times \bR} |(x_{-i}, x_i) - (y_{-i}, y_i)|^p
                    \odif {{\cond{\gamma}{x_{-i}, y_{-i}}(x_i, y_i)}} \\
        &= \int_{\bR^d_{-i}} \odif {{\marginal{\mu}{-i}(x_{-i})}}
                \int_{\bR \times \bR} |(x_{-i}, x_i) - (x_{-i}, y_i)|^p
                    \odif {{\cond{\gamma}{x_{-i}, x_{-i}}(x_i, y_i)}} \\
        &= \int_{\bR^d_{-i}} \odif {{\marginal{\mu}{-i}(x_{-i})}}
                \int_{\bR \times \bR} |x_i - y_i|^p
                    \odif {{\gamma_{x_{-i}}(x_i, y_i)}} \\
        &= \int_{\bR^d_{-i}} \odif {{\marginal{\mu}{-i}(x_{-i})}} C_p(\gamma_{x_{-i}})^p
        = \int_{\bR^d_{-i}} W_p^p(\cond{\mu}{x_{-i}}, \cond{\nu}{x_{-i}})
                \odif {{\marginal{\mu}{-i}(x_{-i})}} \,. \qedhere
    \end{align*}
    }%
\end{proof}

We may combine the above to find \emph{directed} transport plans by transporting mass within
axis-aligned lines, as long as the first probability distribution dominates the second in each such
line.

\begin{lemma}
    \label{lemma:directed-plan-one-coordinate}
    Let $p \in (1, +\infty)$, and let $\mu, \nu \in P(\closedInt^d)$ be absolutely continuous
    probability measures with strictly positive densities. Let $i \in [d]$ and suppose that
    1)~$\marginal{\mu}{-i} = \marginal{\nu}{-i}$ and 2)~$\cond{\mu}{x_{-i}} \succeq
    \cond{\nu}{x_{-i}}$ for each $x_{-i} \in \closedInt^d_{-i}$. Then there exists $\gamma \in
    \Pi_i(\mu \to \nu)$ satisfying
    \begin{equation}
        \label{eq:directed-plan-one-coordinate}
        C_p(\gamma)^p
        = \int_{\closedInt^d_{-i}} W_p^p(\cond{\mu}{x_{-i}}, \cond{\nu}{x_{-i}})
                \odif {\marginal{\mu}{-i}(x_{-i})} \,.
    \end{equation}
\end{lemma}
\begin{proof}
    Let $\gamma \in \Pi_i(\mu, \nu)$ be the plan obtained from \cref{lemma:i-aligned-plan}, which
    satisfies \eqref{eq:directed-plan-one-coordinate}. After recalling the definition of $\gamma$ in
    \cref{lemma:i-aligned-plan}, \cref{cor:directed-wp-1d} implies that $\cond{\gamma}{x_{-i},
    y_{-i}} \in \Pi(\cond{\mu}{x_{-i}} \to \cond{\nu}{x_{-i}})$ for each $(x_{-i}, y_{-i}) \in
    \closedInt^d_{-i} \times \closedInt^d_{-i}$. It remains to confirm that $\gamma \in \Pi(\mu \to
    \nu)$. Indeed,
    \begin{align*}
        &\int_{\bR^d \times \bR^d} \chi_{\{x \not\preceq y\}} \odif \gamma(x, y) \\
        &\quad = \int_{\bR^d_{-i} \times \bR^d_{-i}}
                \odif {{\marginal{\gamma}{-i_x -i_y}(x_{-i}, y_{-i})}}
                \int_{\bR \times \bR} \chi_{\{x \not\preceq y\}}((x_{-i}, x_i), (y_{-i}, y_i))
                        \odif {{\cond{\gamma}{x_{-i}, y_{-i}}(x_i, y_i)}}
            & \text{(Disintegration)} \\
        &\quad \le \int_{\bR^d_{-i} \times \bR^d_{-i}}
                \odif {{\marginal{\gamma}{-i_x -i_y}(x_{-i}, y_{-i})}}
                \int_{\bR \times \bR} \chi_{\{x_{-i} \not\preceq y_{-i}\}}(x_{-i}, y_{-i})
                        \odif {{\cond{\gamma}{x_{-i}, y_{-i}}(x_i, y_i)}} \\
        &\quad\qquad + \int_{\bR^d_{-i} \times \bR^d_{-i}}
                \odif {{\marginal{\gamma}{-i_x -i_y}(x_{-i}, y_{-i})}}
                \underbrace{\int_{\bR \times \bR} \chi_{\{x_i > y_i\}}(x_i, y_i)
                        \odif {{\cond{\gamma}{x_{-i}, y_{-i}}(x_i, y_i)}}}_{\text{$= 0$
                            since $\cond{\gamma}{x_{-i}, y_{-i}}$ is a directed plan}} \\
        &\quad = \int_{\bR^d_{-i} \times \bR^d_{-i}}
                \chi_{\{x_{-i} \not\preceq y_{-i}\}}
                \odif {{\marginal{\gamma}{-i_x -i_y}}}
                \underbrace{\int_{\bR \times \bR} \odif {{\cond{\gamma}{x_{-i}, y_{-i}}}}}_{= 1}
        = \int_{\bR^d_{-i}} \underbrace{\chi_{\{x_{-i} \not\preceq x_{-i}\}}}_{= 0}
                \odif {\marginal{\mu}{-i}(x_{-i})}
            & \text{(By \eqref{eq:gamma-marginal-evaluation})} \\
        &\quad = 0 \,. & \qedhere
    \end{align*}
\end{proof}

Now, we show that the directed, aligned $2$-Wasserstein distance enjoys a nice ``Pythagorean''
composition property: composing $I$-aligned and $J$-aligned directed transport plans, for $I$ and
$J$ disjoint, yields an $I \cup J$-aligned directed transport plan whose cost is obtained from the
costs of the two other plans via the Pythagorean theorem. Arguments of this nature are well-known in
the undirected case, \eg a similar statement appears in \cite[p.~572]{Vil09}.

\lemmapythagoreancomposition
\begin{proof}
    Using \cref{lemma:gluing}, we obtain $\sigma \in P(\bR^d \times \bR^d \times \bR^d)$ such that
    $\transmap{(\pi_{1,2})}{\sigma} = \gamma^+$ and $\transmap{(\pi_{2,3})}{\sigma} = \gamma^-$, and
    by \cref{lemma:directed-composition} and \cref{lemma:aligned-composition}, we have $\gamma
    \define \transmap{(\pi_{1,3})}{\sigma} \in \Pi_{I \cup J}(\mu \to \nu)$. To compute the cost of
    $\gamma$, we first claim that $\sigma$ assigns zero measure to points $x, y, z$ such that
    $x_{-I} \ne y_{-I}$ or $y_{-J} \ne z_{-J}$. Indeed,
    \begin{align*}
        0
        &\le \int_{\bR^d \times \bR^d \times \bR^d}
                (1 - \chi_{\{x_{-I}=y_{-I} \text{ and } y_{-J}=z_{-J}\}}) \odif \sigma(x,y,z) \\
        &\le \int_{\bR^d \times \bR^d \times \bR^d}
                \chi_{\{x_{-I} \ne y_{-I}\}} \odif \sigma(x,y,z)
            + \int_{\bR^d \times \bR^d \times \bR^d}
                \chi_{\{y_{-J} \ne z_{-J}\}} \odif \sigma(x,y,z) \\
        &= \int_{\bR^d \times \bR^d}
                \chi_{\{x_{-I} \ne y_{-I}\}} \odif \gamma^+(x,y)
            + \int_{\bR^d \times \bR^d}
                \chi_{\{y_{-J} \ne z_{-J}\}} \odif \gamma^-(y,z)
        = 0 \,,
    \end{align*}
    \sloppy
    the last step by \cref{lemma:aligned-alternative} since $\gamma^+$ and $\gamma^-$ are $I$- and
    $J$-aligned, respectively. Hence $\int f \odif \sigma = \int f \chi_{\{x_{-I}=y_{-I} \text{ and
    } y_{-J}=z_{-J}\}} \odif \sigma$ for any bounded or nonnegative Borel function $f$. Therefore,
    using the assumption that $I$ and $J$ are disjoint to apply the Pythagorean theorem, we obtain
    {\allowdisplaybreaks
        \begin{align*}
            C_2(\gamma)^2
            &= \int_{\bR^d \times \bR^d} |x-z|^2 \odif \gamma(x,z)
            = \int_{\bR^d \times \bR^d \times \bR^d} |x-z|^2 \odif \sigma(x,y,z) \\
            &= \int_{\bR^d \times \bR^d \times \bR^d}
                \chi_{\{x_{-I}=y_{-I} \text{ and } y_{-J}=z_{-J}\}} |x-z|^2 \odif \sigma(x,y,z) \\
            &= \int_{\bR^d \times \bR^d \times \bR^d}
                \chi_{\{x_{-I}=y_{-I} \text{ and } y_{-J}=z_{-J}\}}
                \left( |x-y|^2 + |y-z|^2 \right) \odif \sigma(x,y,z) \\
            &= \int_{\bR^d \times \bR^d \times \bR^d}
                \left( |x-y|^2 + |y-z|^2 \right) \odif \sigma(x,y,z) \\
            &= \int_{\bR^d \times \bR^d \times \bR^d}
                    |x-y|^2 \odif \sigma(x,y,z)
                + \int_{\bR^d \times \bR^d \times \bR^d}
                    |y-z|^2 \odif \sigma(x,y,z) \\
            &= \int_{\bR^d \times \bR^d}
                    |x-y|^2 \odif \gamma^+(x,y)
                + \int_{\bR^d \times \bR^d}
                    |y-z|^2 \odif \gamma^-(y,z) \\
            &= C_2(\gamma^+)^2 + C_2(\gamma^-)^2 \,. \qedhere
        \end{align*}
    }%
\end{proof}

\begin{lemma}[Induction over coordinates]
    \label{lemma:induction}
    Let $(\mu^{(i)})_{i \in [d+1]}$ be a family of absolutely continuous probability measures on
    $\closedInt^d$ with strictly positive densities. Suppose that, for each $i \in [d]$, we have
    1)~$\marginal{\mu^{(i)}}{-i} = \marginal{\mu^{(i+1)}}{-i}$ and 2)~$\cond{\mu^{(i)}}{x_{-i}}
    \succeq \cond{\mu^{(i+1)}}{x_{-i}}$ for each $x_{-i} \in \closedInt^d_{-i}$. Then
    \[
        W_2^2(\mu^{(1)} \to \mu^{(d+1)})
        \le \sum_{i=1}^d \int_{\closedInt^d_{-i}}
                W_2^2(\cond{\mu^{(i)}}{x_{-i}}, \cond{\mu^{(i+1)}}{x_{-i}})
                \odif {\marginal{\mu^{(i)}}{-i}(x_{-i})} \,.
    \]
\end{lemma}
\begin{proof}
    For each $k \in [d+1]$, let $A(k)$ be the following proposition: there exists $\gamma \in
    \Pi_{[k]}(\mu^{(1)} \to \mu^{(k+1)})$ satisfying
    \[
        C_2(\gamma)^2
        = \sum_{i=1}^k \int_{\closedInt^d_{-i}}
                W_2^2(\cond{\mu^{(i)}}{x_{-i}}, \cond{\mu^{(i+1)}}{x_{-i}})
                \odif {\marginal{\mu^{(i)}}{-i}(x_{-i})} \,.
    \]
    Note that $A(d)$ implies the result we want to prove, while $A(1)$ holds by
    \cref{lemma:directed-plan-one-coordinate}. Now let $2 \le k \le d$. Suppose $A(k-1)$ holds, and
    thus let $\gamma^+ \in \Pi_{[k-1]}(\mu^{(1)} \to \mu^{(k)})$ satisfy
    \[
        C_2(\gamma^+)^2
        = \sum_{i=1}^{k-1} \int_{\closedInt^d_{-i}}
                W_2^2(\cond{\mu^{(i)}}{x_{-i}}, \cond{\mu^{(i+1)}}{x_{-i}})
                \odif {\marginal{\mu^{(i)}}{-i}(x_{-i})} \,.
    \]
    Apply \cref{lemma:directed-plan-one-coordinate} with $i=k$ to obtain a plan $\gamma^- \in
    \Pi_k(\mu^{(k)} \to \mu^{(k+1)})$ satisfying
    \[
        C_2(\gamma^-)^2 = \int_{\closedInt^d_{-k}}
                W_2^2(\cond{\mu^{(k)}}{x_{-k}}, \cond{\mu^{(k+1)}}{x_{-k}})
                \odif {\marginal{\mu^{(k)}}{-k}(x_{-k})} \,.
    \]
    Then, use \cref{lemma:pythagorean-composition} with $I = [k-1]$ and $J = \{k\}$ to obtain
    $\gamma \in \Pi_{[k]}(\mu^{(1)} \to \mu^{(k+1)})$ satisfying
    \[
        C_2(\gamma)^2
        = C_2(\gamma^+)^2 + C_2(\gamma^-)^2
        = \sum_{i=1}^k \int_{\closedInt^d_{-i}}
                W_2^2(\cond{\mu^{(i)}}{x_{-i}}, \cond{\mu^{(i+1)}}{x_{-i}})
                \odif {\marginal{\mu^{(i)}}{-i}(x_{-i})} \,,
    \]
    which implies that $A(k)$ holds. Thus $A(d)$ follows by induction.
\end{proof}

\subsection{From one-dimensional PDE to optimal transport in the cube}
\label{section:tensorization}

We now tie most of the foregoing theory together by showing how to inductively transform a function
on the unit cube into a monotone function, one coordinate at a time via the directed heat semigroup
machinery, keeping track of the cost using the directed optimal transport machinery.

For this part of the proof, it is convenient to restrict our attention to Lipschitz functions. As we
will see, the Lipschitz property of $f : \closedInt^d \to \bR$ is preserved under taking the
monotone equilibrium along each axis-aligned direction, which enables an inductive argument over the
coordinates.

Denote the set of Lipschitz functions $f : \closedInt^d \to \bR$ by $\Lip$ (the dimension $d$ will
always be clear from context).

\begin{observation}
    A function $f : \closedInt^d \to \bR$ is Lipschitz if and only if its restriction to every
    axis-aligned line is uniformly Lipschitz, \ie there exists some $M > 0$ such that, for every $i
    \in [d]$ and $x_{-i} \in \closedInt^d_{-i}$, $f(x_{-i}, \cdot) : [0,1] \to \bR$ is
    $M$-Lipschitz.
\end{observation}

\begin{definition}[$i$- and $I$-monotonicity]
    Let $i \in [d]$. We say $f : \closedInt^d \to \bR$ is \emph{$i$-monotone} if its every
    restriction along direction $i$ is nondecreasing, \ie if for every $x_{-i} \in
    \closedInt^d_{-i}$, $f(x_{-i}, \cdot)$ is nondecreasing. For $I \subseteq [d]$, we say $f$ is
    \emph{$I$-monotone} if it is $i$-monotone for every $i \in I$. We write $\Mon_I$ for the set of
    $I$-monotone functions.
\end{definition}

\begin{observation}
    $\Mon_{[d]}$ is simply the set of monotone functions on $\closedInt^d$.
\end{observation}

The following useful characterization of $i$-monotone functions in terms of $k$-aligned lines
follows by definition:

\begin{observation}
    \label{obs:mon-k}
    Let $i, k \in [d]$ be distinct. Then for all $f : \closedInt^d \to \bR$, $f$ is $i$-monotone if
    and only if the following holds: for every $z_{-i-k} \in \closedInt^d_{-i-k}$, and for all $x_i
    \le y_i$ in $[0,1]$, $f(z_{-i-k}, x_i, \cdot) \le f(z_{-i-k}, y_i, \cdot)$.
\end{observation}

\begin{definition}[Monotone equilibrium operator]
    Let $k \in [d]$. We define the operator $\cM_k : \Lip \to (\closedInt^d \to \bR)$ as follows:
    for each function $f \in \Lip$, $x_{-k} \in \closedInt^d_{-k}$ and $x_k \in [0,1]$, the function
    $\cM_k f : \closedInt^d \to \bR$ satisfies
    \[
        (\cM_k f)(x_{-k}, x_k) \define (P_\infty f(x_{-k}, \cdot))(x_k) \,,
    \]
    where the application of $P_\infty$ is well-defined because each $f(x_{-k}, \cdot)$ is Lipschitz
    and hence in $L^2(I)$.
\end{definition}

The following lemma shows that $\cM_k$ preserves Lipschitzness, confirming that $\cM_k f$ is
well-defined pointwise as a real-valued function, as opposed to only defined up to sets of measure
zero.

\begin{lemma}
    \label{lemma:mk-lipschitz}
    Let $k \in [d]$. Let $f \in \Lip$. Then $\cM_k f \in \Lip$.
\end{lemma}
\begin{proof}
    Let $M > 0$ be such that every axis-aligned line restriction of $f$ is $M$-Lipschitz. We claim
    that every axis-aligned line restriction of $\cM_k f$ is $M$-Lipschitz. We first check the
    $k$-aligned lines. For any $x_{-k} \in \closedInt^d_{-k}$
    \cref{prop:lipschitz-regularity-monotone-equilibrium} implies that
    \[
        \psi((\cM_k f)(x_{-k}, \cdot))
        = \psi(P_\infty f(x_{-k}, \cdot))
        \le \psi(f(x_{-k}, \cdot))
        \le M \,,
    \]
    the last inequality and the conclusion that $(\cM_k f)(x_{-k}, \cdot)$ is $M$-Lipschitz by
    \cref{fact:lipschitz}.

    It remains to check the other line restrictions. Let $i \in [d] \setminus \{k\}$ and let
    $z_{-i-k} \in \closedInt^d_{-i -k}$. Let $x_i \ne y_i \in [0,1]$. Since the axis-aligned line
    restrictions of $f$ (in particular, in direction $i$) are $M$-Lipschitz, we have
    \[
        \sup_{w_k \in [0,1]} \abs*{f(z_{-i-k}, x_i, w_k) - f(z_{-i-k}, y_i, w_k)}
        \le M |x_i - y_i| \,.
    \]
    \cref{prop:p-infty-l-infty-nonexpansive} implies that
    \begin{equation}
        \label{eq:lipschitz}
        \abs*{P_\infty f(z_{-i-k}, x_i, \cdot) - P_\infty f(z_{-i-k}, y_i, \cdot)}
        \le \esssup_{w_k \in [0,1]} \abs*{f(z_{-i-k}, x_i, w_k) - f(z_{-i-k}, y_i, w_k)}
        \le M |x_i - y_i|
    \end{equation}
    \almev in $I$, and since $P_\infty f(z_{-i-k}, x_i, \cdot)$ and $P_\infty f(z_{-i-k}, y_i,
    \cdot)$ are (Lipschitz) continuous as observed above, we conclude that \eqref{eq:lipschitz}
    holds pointwise. Since \eqref{eq:lipschitz} holds for every $i \in [d] \setminus \{k\}$ and
    $z_{-i-k} \in \closedInt^d_{-i-k}$, and since
    \[
        \abs*{(\cM_k f)(z_{-i-k}, x_i, \cdot) - (\cM_k f)(z_{-i-k}, y_i, \cdot)}
        = \abs*{P_\infty f(z_{-i-k}, x_i, \cdot) - P_\infty f(z_{-i-k}, y_i, \cdot)}
    \]
    pointwise, we get that every axis-aligned line restriction of $\cM_k f$ is $M$-Lipschitz, and
    $\cM_k f \in \Lip$.
\end{proof}

\begin{lemma}[$\cM_k$ application makes progress]
    \label{lemma:mk-progress}
    Let $k \in [d]$, and let $I \subseteq [d] \setminus \{k\}$. Let $f \in \Lip \cap \Mon_I$. Then
    $\cM_k f \in \Lip \cap \Mon_{I \cup \{k\}}$.
\end{lemma}
\begin{proof}
    The fact that $\cM_k f \in \Lip$ is given by \cref{lemma:mk-lipschitz}, and the fact that $\cM_k f$
    is $k$-monotone follows from the definition of $\cM_k$ and the fact that the monotone equilibrium
    $P_\infty f(x_{-k}, \cdot)$ is nondecreasing. Now, let $i \in I$, so that $f$ is $i$-monotone.
    By \cref{obs:mon-k}, we have that for every $z_{-i-k} \in \closedInt^d_{-i-k}$, and for all $x_i
    \le y_i$ in $[0,1]$, $f(z_{-i-k}, x_i, \cdot) \le f(z_{-i-k}, y_i, \cdot)$. By
    \cref{cor:p-infty-order-preserving}, we obtain that $P_\infty f(z_{-i-k}, x_i, \cdot) \le
    P_\infty f(z_{-i-k}, y_i, \cdot)$ \almev in $I$, and since these two functions are (Lipschitz)
    continuous, this inequality holds pointwise. Thus $(\cM_k f)(z_{-i-k}, x_i, \cdot) \le (\cM_k
    f)(z_{-i-k}, y_i, \cdot)$, and again by \cref{obs:mon-k}, $\cM_k f$ is $i$-monotone. This
    concludes the proof.
\end{proof}

\begin{lemma}[$\cM_k$ preserves bounds]
    \label{lemma:mk-bounds}
    Let $f \in \Lip$. Let $a \le b$ be real numbers and suppose $a \le f \le b$. Then $a \le \cM_k f
    \le b$.
\end{lemma}
\begin{proof}
    This follows by the definition of $\cM_k$ together with
    \cref{cor:p-infty-order-preserving,lemma:p-infty-stationary-points}.
\end{proof}

\begin{lemma}[$\cM_k$ is nonexpansive in $L^2$]
    \label{lemma:mk-nonexpansive}
    Let $f, g \in \Lip$. Then
    \[
        \int_{\closedInt^d} (\cM_k f- \cM_k g)^2 \odif x
        \le \int_{\closedInt^d} (f - g)^2 \odif x \,.
    \]
\end{lemma}
\begin{proof}
    This follows from \cref{prop:p-infty-nonexpansive} and Tonelli's theorem, as follows:
    \begin{align*}
        \int_{\closedInt^d} (\cM_k f- \cM_k g)^2 \odif x
        &= \int_{\closedInt^d_{-k}} \odif x_{-k}
            \int_I \left[ (\cM_k f)(x_{-k}, x_k) - (\cM_k g)(x_{-k}, x_k) \right]^2 \odif x_k \\
        &= \int_{\closedInt^d_{-k}} \odif x_{-k}
            \int_I \left[
                (P_\infty f(x_{-k}, \cdot))(x_k) - (P_\infty g(x_{-k}, \cdot))(x_k)
                \right]^2 \odif x_k \\
        &\le \int_{\closedInt^d_{-k}} \odif x_{-k}
            \int_I \left[ f(x_{-k}, x_k) - g(x_{-k}, x_k) \right]^2 \odif x_k
        = \int_{\closedInt^d} (f-g)^2 \odif x \,. && \qedhere
    \end{align*}
\end{proof}

Similarly, combining \cref{cor:p-infty-mass-conserving} with Fubini's theorem yields

\begin{lemma}
    \label{lemma:mk-mass-conserving}
    Let $f \in \Lip$. Then
    \[
        \int_{\closedInt^d} (\cM_k f) \odif x = \int_{\closedInt^d} f \odif x \,.
    \]
\end{lemma}

The following result is essentially an $L^2$ version of \cite[Proposition~3.14]{Fer23}.

\begin{lemma}[Effect of $\cM_k$ on the directed Dirichlet energy]
    \label{lemma:mk-energy}
    Let $i, k \in [d]$ be distinct and let $f \in \Lip$. Then
    \[
        \int_{\closedInt^d_{-i}} \cE^-((\cM_k f)(x_{-i}, \cdot)) \odif x_{-i}
        \le
        \int_{\closedInt^d_{-i}} \cE^-(f(x_{-i}, \cdot)) \odif x_{-i} \,.
    \]
\end{lemma}
\begin{proof}
    For any $g \in \Lip$, $i \in [d]$ and $h \in \bR \setminus \{0\}$, define the $Q_{g,i,h} :
    \closedInt^d \to \bR$ by
    \begin{equation}
        \label{eq:q-def}
        Q_{g,i,h}(x) \define \begin{cases}
            \frac{g(x_{-i}, x_i + h) - g(x_{-i}, x_i)}{h} & \text{if $x_i+h \in (0,1)$} \\
            0                                             & \text{otherwise.}
        \end{cases}
    \end{equation}
    By Rademacher's theorem on the open domain $(0,1)^d$, $g$ is differentiable almost everywhere.
    In particular, the function $D_{g,i} : \closedInt^d \to \bR$ given by
    \[
        D_{g,i}(x) \define \begin{cases}
            \lim_{h \to 0} Q_{g,i,h}(x) & \text{if the limit exists} \\
            0                           & \text{otherwise}
        \end{cases}
    \]
    gives the partial derivative of $g$ in direction $i$ for \almev $x \in \closedInt^d$, and is
    measurable as the limit of the measurable functions $Q_{g,i,h}$. Moreover, letting $M > 0$ be
    such that the axis-aligned line restrictions of $g$ are $M$-Lipschitz, we have $\abs*{Q_{g,i,h}}
    \le M$ pointwise. Finally, fixing any $x_{-i} \in \closedInt^d_{-i}$, $g(x_{-i}, \cdot)$ is
    (Lipschitz and hence) in $H^1(I)$, and the partial derivative $D_{g,i}(x_{-i}, \cdot)$ agrees
    with any weak derivative $\partial_{x_i} g(x_{-i}, \cdot)$ \almev in $[0,1]$. In particular, all
    of these considerations apply to the functions $f$ and $\cM_k f$ by \cref{lemma:mk-lipschitz}. We
    conclude that, for $g = f$ or $g = \cM_k f$ and each $x_{-i} \in \closedInt^d_{-i}$,
    \begin{align*}
        \cE^-(g(x_{-i}, \cdot))
        &= \frac{1}{2} \int_I \left[ \partial_{x_i}^- g(x_{-i}, x_i) \right]^2 \odif x_i
        = \frac{1}{2} \int_I \left[ D_{g,i}(x_{-i}, x_i)^- \right]^2 \odif x_i \\
        &= \frac{1}{2} \int_I \lim_{h \to 0} \left[ Q_{g,i,h}(x_{-i}, x_i)^- \right]^2 \odif x_i \,,
    \end{align*}
    where for simplicity we write the limit in the last expression with the understanding that it is
    only defined almost everywhere in $[0,1]$. Note that the measurability of the last integrand
    over $[0,1]^d$ justifies writing the integrals in the statement of the lemma via Tonelli's
    theorem (as the computation below shows). Define for each $h \in \bR$ the set
    \[
        I_h \define \begin{cases}
            (0, 1-h) & \text{if } h \ge 0 \\
            (-h, 1)  & \text{otherwise,}
        \end{cases}
    \]
    so that $Q_{g,i,h}(x)$ is defined by the first case of \eqref{eq:q-def} if and only if $x_i \in
    I_h$. Then, by repeated applications of Tonelli's theorem and the dominated convergence theorem,
    we obtain
    {\allowdisplaybreaks
        \begin{align*}
            &\int_{\closedInt^d_{-i}} \cE^-((\cM_k f)(x_{-i}, \cdot)) \odif x_{-i}
            = \frac{1}{2} \int_{\closedInt^d_{-i}} \odif x_{-i} \int_I \lim_{h \to 0} \left[
                    Q_{\cM_k f, i, h}(x_{-i}, x_i)^- \right]^2 \odif x_i \\
            &\, = \frac{1}{2} \int_{\closedInt^d} \lim_{h \to 0} \underbrace{\left[
                    Q_{\cM_k f, i, h}(x)^- \right]^2}_{\le M^2} \odif x
            = \frac{1}{2} \lim_{h \to 0} \int_{\closedInt^d} \left[
                    Q_{\cM_k f, i, h}(x)^- \right]^2 \odif x \\
            &\, = \frac{1}{2} \lim_{h \to 0}
                \int_{\closedInt^d_{-i-k}} \odif x_{-i-k}
                \int_I \odif x_i
                \int_I \left[ Q_{\cM_k f, i, h}(x_{-i-k}, x_i, x_k)^- \right]^2 \odif x_k \\
            &\, = \frac{1}{2} \lim_{h \to 0}
                \int_{\closedInt^d_{-i-k}} \odif x_{-i-k}
                \int_{I_h} \odif x_i
                \int_I \left[ \left(
                        \frac{(\cM_k f)(x_{-i-k}, x_i+h, x_k) - (\cM_k f)(x_{-i-k}, x_i, x_k)}{h}
                    \right)^- \right]^2 \odif x_k \\
            &\, = \frac{1}{2} \lim_{h \to 0}
                \int_{\closedInt^d_{-i-k}} \odif x_{-i-k}
                \int_{I_h} \odif x_i
                \int_I \left[ \left(
                        \frac{(P_\infty f(x_{-i-k}, x_i+h, \cdot))(x_k)
                                - (P_\infty f(x_{-i-k}, x_i, \cdot))(x_k)} {h}
                    \right)^- \right]^2 \odif x_k \\
            &\, \le \frac{1}{2} \lim_{h \to 0}
                \int_{\closedInt^d_{-i-k}} \odif x_{-i-k}
                \int_{I_h} \odif x_i
                \int_I \left[ \left(
                        \frac{f(x_{-i-k}, x_i+h, x_k) - f(x_{-i-k}, x_i, x_k)} {h}
                    \right)^- \right]^2 \odif x_k \\
            &\, = \frac{1}{2} \lim_{h \to 0}
                \int_{\closedInt^d_{-i-k}} \odif x_{-i-k}
                \int_I \odif x_i
                \int_I \left[ Q_{f,i,h}(x_{-i-k}, x_i, x_k)^- \right]^2 \odif x_k \\
            &\, = \frac{1}{2} \lim_{h \to 0} \int_{\closedInt^d} \underbrace{\left[
                    Q_{f,i,h}(x)^- \right]^2}_{\le M^2} \odif x
            \, = \frac{1}{2} \int_{\closedInt^d} \lim_{h \to 0} \left[
                    Q_{f,i,h}(x)^- \right]^2 \odif x \\
            &\, = \frac{1}{2} \int_{\closedInt^d_{-i}} \odif x_{-i} \int_I \lim_{h \to 0} \left[
                    Q_{f,i,h}(x_{-i}, x_i)^- \right]^2 \odif x_i
            = \int_{\closedInt^d_{-i}} \cE^-(f(x_{-i}, \cdot)) \odif x_{-i} \,,
        \end{align*}
    }%
    the inequality by \cref{prop:p-infty-directed-nonexpansiveness} via the identity $\left[ (a-b)^+
    \right]^2 = \left[ (b-a)^- \right]^2$.
\end{proof}

\begin{definition}[Coordinate-wise monotone equilibrium]
    For each $f \in \Lip$, define $f^* \in \Lip$, the \emph{coordinate-wise monotone equilibrium} of
    $f$, by
    \[
        f^* \define \cM_d \cM_{d-1} \dotsm \cM_1 f \,.
    \]
\end{definition}

Note that indeed $f^*$ is indeed Lipschitz by \cref{lemma:mk-lipschitz}, and it is monotone:

\begin{proposition}
    Let $f \in \Lip$. Then $f^*$ is monotone.
\end{proposition}
\begin{proof}
    This follows by repeatedly applying \cref{lemma:mk-progress}, and recalling that $\Mon_{[d]}$ is
    the set of monotone functions on $\closedInt^d$.
\end{proof}

By \cref{prop:p-infty-transformations} and the definition of $\cM_k$, we also observe that the
coordinate-wise monotone equilibrium behaves nicely with respect to certain affine transformations:

\begin{observation}
    \label{obs:fstar-transformations}
    Let $f \in \Lip$. Let $\alpha > 0$ and $\beta \in \bR$. Then $(\alpha f + \beta)^* = \alpha f^*
    + \beta$.
\end{observation}

Similarly, \cref{lemma:mk-nonexpansive} implies that taking the coordinate-wise monotone equilibrium
is a nonexpansive operation:

\begin{observation}
    \label{obs:fstar-nonexpansive}
    Let $f, g \in \Lip$. Then $\|f^* - g^*\|_{L^2((0,1)^d)} \le \|f - g\|_{L^2((0,1)^d)}$. Since
    $0^* = 0$ by \cref{lemma:p-infty-stationary-points} (where we write $0$ for the constant zero
    function), we conclude in particular that $\|f^*\|_{L^2((0,1)^d)} \le \|f\|_{L^2((0,1)^d)}$.
\end{observation}

We are now ready to combine the ingredients from the previous sections and establish the
transport-energy inequality. Note that this inequality is most effective when the function $f$ is
pointwise bounded close to $1$.

\thmtransportenergy
\begin{proof}
    Define the family $(f^{(i)})_{i \in [d+1]} \subset \Lip$ by $f^{(1)} \define f$ and $f^{(i+1)}
    \define \cM_i f^{(i)}$ for each $i \in [d]$, so that $f^{(d+1)} = f^*$. Note that we have $1-a \le
    f^{(i)} \le 1+a$ and $\int_{\closedInt^d} f^{(i)} \odif x = 1$ for each $i \in [d+1]$ by
    \cref{lemma:mk-bounds,lemma:mk-mass-conserving}. Then, define the family $(\mu^{(i)})_{i \in
    [d+1]}$ of probability measures on $\closedInt^d$ by $\odif \mu^{(i)} \define f^{(i)} \odif x$,
    so that $\mu^{(1)} = \mu$ and $\mu^{(d+1)} = \mu^*$.

    We claim $(\mu^{(i)})_{i \in [d+1]}$ satisfies the conditions of \cref{lemma:induction}. First,
    they have strictly positive densities because $1-a \le f^{(i)} \le 1+a$ for each $i \in [d+1]$
    as observed above. They are also absolutely continuous with Borel measurable density because
    each $f^{(i)}$ is continuous and hence Borel.

    Recall that, by \cref{prop:disintegration-absolutely-continuous}, for each $i,j \in [d+1]$ we
    may characterize the disintegration $\mu^{(j)} = \int_{\bR^d_{-i}} \cond{\mu^{(j)}}{x_{-i}}
    \odif {\marginal{\mu^{(j)}}{-i}}$ as follows: the marginal and all conditional measures are both
    absolutely continuous supported in $\closedInt^d_{-i}$ and $[0,1]$ respectively, and we have
    $\odif {\marginal{\mu^{(j)}}{-i}} = \marginal{f^{(j)}}{-i} \odif x_{-i}$ where
    \begin{equation}
        \label{eq:f-marginal}
        \marginal{f^{(j)}}{-i}(x_{-i}) = \int_I f^{(j)}(x_{-i}, x_i) \odif x_i
    \end{equation}
    for each $x_{-i} \in \closedInt^d_{-i}$, and $\odif {\cond{\mu^{(j)}}{x_{-i}}} =
    \cond{f^{(j)}}{x_{-i}} \odif x_i$ where
    \begin{equation}
        \label{eq:f-cond}
        \cond{f^{(j)}}{x_{-i}}(x_i) = \frac{f^{(j)}(x)}{\marginal{f^{(j)}}{-i}(x_{-i})}
    \end{equation}
    for each $x \in \closedInt^d$ (note that the denominator is nonzero since $f^{(j)} \ge 1-a$).

    Now, fix any $i \in [d]$. Then the condition $\marginal{\mu^{(i)}}{-i} =
    \marginal{\mu^{(i+1)}}{-i}$ is implied by the pointwise condition
    \begin{equation}
        \label{eq:marginals}
        \int_I f^{(i)}(x_{-i}, x_i) \odif x_i
        = \int_I f^{(i+1)}(x_{-i}, x_i) \odif x_i
        \qquad
        \forall x_{-i} \in \closedInt^d_{-i} \,,
    \end{equation}
    which in turn holds by the definitions of $f^{(i+1)}$ and $\cM_i$, and
    \cref{cor:p-infty-mass-conserving}. Finally, fixing $x_{-i} \in \closedInt^d_{-i}$, we verify
    the condition $\cond{\mu^{(i)}}{x_{-i}} \succeq \cond{\mu^{(i+1)}}{x_{-i}}$. By
    \cref{cor:mu-infty-dominated}, it suffices to show that
    \begin{equation}
        \label{eq:cond-f-p-infty}
        \cond{f^{(i+1)}}{x_{-i}} \eqquestion P_\infty \cond{f^{(i)}}{x_{-i}} \,.
    \end{equation}
    By \eqref{eq:f-marginal} and \eqref{eq:f-cond}, this is equivalent to
    \[
        \frac{1}{\int_I f^{(i+1)}(x_{-i}, x_i) \odif x_i} f^{(i+1)}(x_{-i}, \cdot)
        \eqquestion P_\infty \left(
            \frac{1}{\int_I f^{(i)}(x_{-i}, x_i) \odif x_i} f^{(i)}(x_{-i}, \cdot) \right) \,.
    \]
    Since $f^{(i+1)}(x_{-i}, \cdot) = P_\infty f^{(i)}(x_{-i}, \cdot)$, and applying
    \cref{prop:p-infty-transformations}, the above is equivalent to
    \[
        \frac{1}{\int_I f^{(i+1)}(x_{-i}, x_i) \odif x_i} P_\infty f^{(i)}(x_{-i}, \cdot)
        \eqquestion \frac{1}{\int_I f^{(i)}(x_{-i}, x_i) \odif x_i}
            P_\infty f^{(i)}(x_{-i}, \cdot) \,,
    \]
    which is true again by \eqref{eq:marginals}, so indeed $\cond{\mu^{(i)}}{x_{-i}} \succeq
    \cond{\mu^{(i+1)}}{x_{-i}}$. Thus, by \cref{lemma:induction}, we have
    \begin{align*}
        W_2^2(\mu \to \mu^*)
        &\le \sum_{i=1}^d \int_{\closedInt^d_{-i}}
                W_2^2(\cond{\mu^{(i)}}{x_{-i}}, \cond{\mu^{(i+1)}}{x_{-i}})
                \odif {\marginal{\mu^{(i)}}{-i}(x_{-i})} \\
        &= \sum_{i=1}^d \int_{\closedInt^d_{-i}}
                W_2^2(\cond{\mu^{(i)}}{x_{-i}}, \cond{\mu^{(i+1)}}{x_{-i}})
                \marginal{f^{(i)}}{-i}(x_{-i}) \odif x_{-i} \\
        &\le (1+a) \sum_{i=1}^d \int_{\closedInt^d_{-i}}
                W_2^2(\cond{\mu^{(i)}}{x_{-i}}, \cond{\mu^{(i+1)}}{x_{-i}}) \odif x_{-i} \,,
    \end{align*}
    the last inequality because $1-a \le f^{(i)} \le 1+a$ also implies that $1-a \le
    \marginal{f^{(i)}}{-i} \le 1+a$. Now, applying \cref{thm:undirected-transport-energy} via
    \eqref{eq:cond-f-p-infty}, we have
    \[
        W_2^2(\cond{\mu^{(i)}}{x_{-i}}, \cond{\mu^{(i+1)}}{x_{-i}})
        \le \frac{C}{\inf \cond{f^{(i)}}{x_{-i}}} \cE^-(\cond{f^{(i)}}{x_{-i}})
    \]
    for each $i \in [d]$ and $x_{-i} \in \closedInt^d_{-i}$. Note that $\cond{f^{(i)}}{x_{-i}} \ge
    \frac{1-a}{1+a}$ by the pointwise bounds on $f^{(i)}$ and $\marginal{f^{(i)}}{-i}$, and that
    $\cE^-(\alpha g) = \alpha^2 \cE^-(g)$ for any $\alpha \ge 0$ and (say) Lipschitz $g$ by the
    definition of $\cE^-$. Therefore, applying \eqref{eq:f-cond}, we obtain
    \[
        W_2^2(\cond{\mu^{(i)}}{x_{-i}}, \cond{\mu^{(i+1)}}{x_{-i}})
        \le \frac{C (1+a)}{1-a} \left( \frac{1}{\marginal{f^{(i)}}{-i}(x_{-i})} \right)^2
            \cE^-(f^{(i)}(x_{-i}, \cdot))
        \le \frac{C (1+a)}{(1-a)^3} \cE^-(f^{(i)}(x_{-i}, \cdot)) \,.
    \]
    Hence we have
    \[
        W_2^2(\mu \to \mu^*)
        \le \frac{C (1+a)^2}{(1-a)^3} \sum_{i=1}^d \int_{\closedInt^d_{-i}}
            \cE^-(f^{(i)}(x_{-i}, \cdot)) \odif x_{-i} \,.
    \]
    Recall that $f^{(1)} = f$ while $f^{(i)} = \cM_{i-1} \dotsm \cM_1 f$ for each $i = 2, \dotsc, d+1$.
    Hence, by inductively applying \cref{lemma:mk-energy} for each $i \ge 2$, and by Tonelli's
    theorem and the definition of $\cE^-$, we obtain
    \begin{align*}
        W_2^2(\mu \to \mu^*)
        &\le \frac{C (1+a)^2}{(1-a)^3} \sum_{i=1}^d \int_{\closedInt^d_{-i}}
            \cE^-(f(x_{-i}, \cdot)) \odif x_{-i} \\
        &= \frac{C (1+a)^2}{2 (1-a)^3} \sum_{i=1}^d \int_{\closedInt^d_{-i}} \odif x_{-i}
            \int_I \left( \partial_{x_i}^- f(x_{-i}, x_i) \right)^2 \odif x_i \\
        &= \frac{C (1+a)^2}{2 (1-a)^3} \int_{\closedInt^d} \sum_{i=1}^d
            \left( \partial_i^- f \right)^2 \odif x \\
        &= \frac{C (1+a)^2}{2 (1-a)^3} \int_{\closedInt^d} \left| \grad^- f \right|^2 \odif x \,.
        \qedhere
    \end{align*}
\end{proof}

\section{Directed Poincaré from Wasserstein and Kantorovich}
\label{section:poincare-from-transport-energy}

In this section, we establish our directed Poincaré inequality
(\cref{thm:directed-poincare-inequality}) by combining the transport-energy inequality from
\cref{thm:transport-energy} with a directed version of (weak) Kantorovich duality via a perturbation
argument, thereby implementing the ideas from \cref{section:overview-duality} of the proof overview.
The general theme of going between transport inequalities and Poincaré inequalities, as well as the
spirit of the technical arguments presented here, are well-known in the undirected case; we refer
the reader to \cite[Chapters 7 and 22]{Vil09} for a comprehensive presentation and primary
references, and to \cite{Liu20} for an undirected version of the precise reduction we instantiate
here (that is, from a transport-energy inequality to a Poincaré inequality).

For simplicity, we do not attempt to initiate a systematic study of duality in directed optimal
transport, or to state results in the broadest possible generality, but rather limit ourselves to
what is required for the aforementioned goal. In \cref{section:duality}, we introduce the notion of
directed weak Kantorovich duality and establish some of its properties. The main results are
\cref{prop:h-duality}, which gives a \emph{lower bound} on the directed Wasserstein distance between
two measures in terms the so-called (directed) \emph{Hamilton-Jacobi operator} $\dirH$ (recall that
the previous section gave an \emph{upper bound} on this distance via the transport-energy
inequality), and \cref{prop:properties-h}, which relates the action $\dirH_t h(x)$ of this operator
back to the directed gradient $\grad^+ h(x)$ (which, recall, we wish to have in the right-hand side
of our Poincaré inequality). Then, in \cref{section:perturbation}, we apply these tools via a
careful choice of dual function (a ``perturbation argument'') to conclude the directed Poincaré
inequality, first for continuously differentiable functions (\cref{thm:c1-poincare}) and then, via a
standard approximation argument, for all $H^1$ functions (\cref{thm:directed-poincare-inequality}).

\subsection{Directed weak Kantorovich duality and Hamilton-Jacobi operator}
\label{section:duality}

Kantorovich duality plays a fundamental role in the study of Wasserstein distance. The weak part of
this duality allows one to lower bound the Wasserstein distance between two distributions by means
of ``test functions'' satisfying a certain constraint. Here, we translate a small part of this rich
theory to the directed setting.

\begin{lemma}[Weak duality]
    \label{lemma:weak-duality}
    Let $\Omega \subset \bR^d$ be a bounded Borel set and let $\mu, \nu \in P(\Omega)$. Suppose
    $\phi \in L^1(\nu)$ and $\psi \in L^1(\mu)$ satisfy
    \[
        \phi(y) - \psi(x) \le |x-y|^2 \qquad \forall x \preceq y \text{ in } \Omega \,.
    \]
    Then
    \[
        W_2^2(\mu \to \nu) \ge \int_\Omega \phi \odif \nu - \int_\Omega \psi \odif \mu \,.
    \]
\end{lemma}
\begin{proof}
    Suppose $\gamma \in \Pi(\mu \to \nu)$. Then
    {\allowdisplaybreaks
    \begin{align*}
        C_2(\gamma)^2
        &= \int_{\Omega \times \Omega} |x-y|^2 \odif \gamma(x,y) \\
        &= \int_{\Omega \times \Omega} \chi_{\{x \preceq y\}} |x-y|^2 \odif \gamma(x,y)
            & \text{(Since $\gamma \in \Pi(\mu \to \nu)$)} \\
        &\ge \int_{\Omega \times \Omega} \chi_{\{x \preceq y\}}
            (\phi(y) - \psi(x)) \odif \gamma(x,y)
            & \text{(Hypothesis)} \\
        &= \int_{\Omega \times \Omega} (\phi(y) - \psi(x)) \odif \gamma(x,y)
            & \text{(Since $\gamma \in \Pi(\mu \to \nu)$)} \\
        &= \int_\Omega \phi \odif \nu - \int_\Omega \psi \odif \mu
            & \text{(Since $\transmap{(\pi_1)}{\gamma} = \mu, \transmap{(\pi_2)}{\gamma} = \nu$)}
            \,.
    \end{align*}
    }%
    Since this holds for every $\gamma \in \Pi(\mu \to \nu)$, the result follows.
\end{proof}

This duality results motivates the following definition. Let $C_b(\Omega)$ denote the set of bounded
continuous functions $\Omega \to \bR$.

\begin{definition}[Directed Hamilton-Jacobi operator]
    Let $\Omega \subset \bR^d$. For each $t \ge 0$, define the \emph{directed Hamilton-Jacobi
    operator} $\dirH_t : C_b(\Omega) \to (\Omega \to \bR)$ as follows: for each $h \in C_b(\Omega)$
    and $x \in \Omega$, we set
    \begin{equation}
        \label{eq:hamilton-jacobi}
        (\dirH_t h)(x) \define \begin{cases}
            h(x)                                                            & \text{if $t = 0$} \\
            \sup_{y \succeq x} \left\{ h(y) - \frac{1}{2t} |x-y|^2 \right\} & \text{otherwise.}
        \end{cases}
    \end{equation}
\end{definition}

\begin{remark}
    The definition above is a directed analogue of the so-called (backward)
    \emph{Hamilton-Jacobi-Hopf-Lax-Oleinik semigroup}. We do not claim that $\dirH_t$ forms a
    semigroup.
\end{remark}

\begin{proposition}
    \label{prop:h-duality}
    Let $\Omega \subset \bR^d$ be a bounded Borel set, and let $\mu, \nu \in P(\Omega)$. Let $h \in
    C_b(\Omega)$. Then
    \[
        \tfrac{1}{2} W_2^2(\mu \to \nu)
        \ge \int_\Omega h \odif \nu - \int_\Omega (\dirH_1 h) \odif \mu \,.
    \]
\end{proposition}
\begin{proof}
    Let $\phi \define 2h$ and $\psi \define 2 \dirH_1 h$. Note that $\phi \in L^1(\nu)$ since $h$ is
    bounded. The fact that $h$ is bounded also implies that $\dirH_1 h$ is bounded, as can be
    verified from its definition, and hence $\psi \in L^1(\mu)$. Finally, for each $x \preceq y$ in
    $\Omega$ we have
    \begin{align*}
        \phi(y) - \psi(x)
        &= 2h(y) - 2 \sup_{y' \succeq x} \left\{ h(y') - \frac{1}{2} |x - y'|^2 \right\} \\
        &= |x-y|^2 + \underbrace{2\left[
            \left( h(y) - \frac{1}{2} |x-y|^2 \right)
            - \sup_{y' \succeq x}
                \left\{ h(y') - \frac{1}{2} |x - y'|^2 \right\} \right]}_{\le 0}
        \le |x-y|^2 \,.
    \end{align*}
    Thus \cref{lemma:weak-duality} gives the conclusion.
\end{proof}

The following simple lemma plays a key role in the perturbation argument.

\begin{lemma}
    \label{lemma:h-scaling}
    Let $\Omega \subset \bR^d$. Let $t > 0$ and let $h \in C_b(\Omega)$. Then $\dirH_1 (th) = t
    \dirH_t h$.
\end{lemma}
\begin{proof}
    Indeed, for each $x \in \Omega$ we have
    \[
        (\dirH_1 (th))(x)
        = \sup_{y \succeq x} \left\{ th(y) - \frac{1}{2} |x-y|^2 \right\}
        = t \sup_{y \succeq x} \left\{ h(y) - \frac{1}{2t} |x-y|^2 \right\}
        = t ((\dirH_t h)(x)) \,. \qedhere
    \]
\end{proof}

\paragraph*{Notation.} For $\Omega \subset \bR^d$ an open set and $k \in \bN \cup \{+\infty\}$,
write $C^k(\overline \Omega)$ for the set of restrictions to $\overline \Omega$ of functions in
$C^k(\bR^d)$.

For $\Omega \subset \bR^d$, point $x \in \Omega$, and $r > 0$, let $B^\circ_\Omega(x,r) \define
B^\circ(x,r) \cap \Omega$, where $B^\circ(x,r) \define \{ y \in \bR^d : 0 < |x-y| < r \}$ is the
open ball of radius $r$ centered at $x$ with $x$ itself excluded, and let $B^+_\Omega(x,r) \define
B^\circ_\Omega(x,r) \cap \{ y \in \bR^d : y \succeq x\}$.

The following result gives a directed analogue to some of the properties of the Hamilton-Jacobi
semigroup, as presented in \cite[Theorem~22.16]{Vil09}.\footnote{Note that our use of the notation
$\grad^+ h= 0 \lor \grad h$, $\grad^- h = 0 \land \grad h$ is unrelated to the use of similar
notation in \cite{Vil09}, where it denotes a notion of norm of the gradient for functions in more
general spaces where the usual derivative may not be defined. In particular, their definition agrees
with the norm of the gradient for differentiable functions, while the point is the our definition
(the directed gradient) does not.}

\begin{proposition}[Some properties of the directed Hamilton-Jacobi operator]
    \label{prop:properties-h}
    Let $\Omega \subset \bR^d$ be a bounded open set. Let $h \in C^1(\overline \Omega)$.
    \begin{enumerate}[label=(\alph*)]
        \item \label{item:h-a}
            Let $C \define \sup h - \inf h$, which is well-defined because $h$ is continuous over
            the compact set $\overline \Omega$ and thus bounded. Then for each $t > 0$ and $x \in
            \Omega$, the supremum in \eqref{eq:hamilton-jacobi} may be taken over the set
            $B^+_\Omega(x, \sqrt{2Ct})$.
        \item \label{item:h-b}
            For each $x \in \Omega$,
            \[
                \limsup_{t \to 0^+} \frac{\dirH_t h(x) - h(x)}{t}
                \le \frac{|\grad^+ h(x)|^2}{2} \,.
            \]
        \item \label{item:h-c}
            The quotient
            \[
                \frac{\dirH_t h(x) - h(x)}{t}
            \]
            is nonnegative and bounded, uniformly in $t > 0$ and $x \in \Omega$.
    \end{enumerate}
\end{proposition}
\begin{proof}
    Let us first show \ref*{item:h-a}. For any $x \preceq y$ in $\Omega$ and $t > 0$, if $|x-y| \ge
    \sqrt{2Ct}$ then
    \[
        h(y) - \frac{1}{2t}|x-y|^2
        \le h(x) + C - \frac{1}{2t} 2Ct
        = h(x) \,,
    \]
    so $y$ is irrelevant to the supremum in \eqref{eq:hamilton-jacobi}. Moreover, by the continuity
    of $h$, we may drop $x$ itself from the supremum. Thus $y$ is irrelevant to the supremum
    whenever $y \not\in B^+_\Omega(x, \sqrt{2Ct})$, as claimed.

    We now show \ref*{item:h-b} using \ref*{item:h-a}. Let $x \in \Omega$ and $t > 0$. We have
    {\allowdisplaybreaks
        \begin{align*}
            \frac{\dirH_t h(x) - h(x)}{t}
            &= \frac{\sup_{y \in B^+_\Omega(x, \sqrt{2Ct})}
                \left\{ h(y) - \frac{1}{2t}|x-y|^2 \right\} - h(x)}{t} \\
            &= \sup_{y \in B^+_\Omega(x, \sqrt{2Ct})}
                \left\{ \frac{h(y) - h(x)}{t} - \frac{1}{2t^2}|x-y|^2 \right\} \\
            &\le \sup_{y \in B^+_\Omega(x, \sqrt{2Ct})}
                \left\{ \frac{(h(y) - h(x))^+}{t} - \frac{1}{2t^2}|x-y|^2 \right\} \\
            &= \sup_{y \in B^+_\Omega(x, \sqrt{2Ct})}
                \left\{ \frac{(h(y) - h(x))^+}{|x-y|} \frac{|x-y|}{t}
                        - \frac{1}{2}\left(\frac{|x-y|}{t}\right)^2 \right\} \,.
        \end{align*}
    }%
    Using the inequality $\alpha \beta - \frac{1}{2} \beta^2 \le \frac{1}{2} \alpha^2$, we obtain
    \begin{equation}
        \label{eq:h-quotient-bound}
        \frac{\dirH_t h(x) - h(x)}{t}
        \le \frac{1}{2} \sup_{y \in B^+_\Omega(x, \sqrt{2Ct})}
            \left( \frac{(h(y) - h(x))^+}{|x-y|} \right)^2 \,.
    \end{equation}
    Using \cref{lemma:quotients-directed-gradient} and the fact that the term inside the supremum is
    nonnegative, we conclude that
    \[
        \limsup_{t \to 0^+} \frac{\dirH_t h(x) - h(x)}{t}
        \le \frac{1}{2} \left( \lim_{t \to 0^+} \sup_{y \in B^+_\Omega(x, \sqrt{2Ct})}
            \frac{(h(y) - h(x))^+}{|x-y|} \right)^2
        = \frac{|\grad^+ h(x)|^2}{2} \,,
    \]
    as desired. We may also obtain \ref*{item:h-c} from \eqref{eq:h-quotient-bound} as follows.
    First note that $\dirH_t h(x) \ge h(x)$ for all $t \ge 0$ and $x \in \Omega$. It is standard
    that $C^1$ functions are Lipschitz on compact sets, so $h$ is Lipschitz on $\overline \Omega$.
    Let $M > 0$ be its Lipschitz constant. Then from \eqref{eq:h-quotient-bound} we conclude that,
    for each $x \in \Omega$ and $t > 0$,
    \[
        0
        \le \frac{\dirH_t h(x) - h(x)}{t}
        \le \frac{1}{2} \sup_{y \in B^+_\Omega(x, \sqrt{2Ct})}
            \left( \frac{M |x-y|}{|x-y|} \right)^2
        = \frac{M^2}{2} \,,
    \]
    which establishes \ref*{item:h-c}.
\end{proof}

\begin{lemma}
    \label{lemma:quotients-directed-gradient}
    Let $\Omega \subset \bR^d$ be an open set. Let $h \in C^1(\overline \Omega)$. Then for each $x
    \in \Omega$,
    \[
        |\grad h(x)|
        = \lim_{r \to 0^+} \sup_{y \in B^\circ_\Omega(x, r)} \frac{(h(y) - h(x))^+}{|x-y|}
    \]
    and
    \[
        |\grad^+ h(x)|
        = \lim_{r \to 0^+} \sup_{y \in B^+_\Omega(x, r)} \frac{(h(y) - h(x))^+}{|x-y|} \,.
    \]
\end{lemma}
\begin{proof}
    The first part of the statement is known, but let us give a proof for completeness. Let $x \in
    \Omega$. We have
    \[
        \lim_{r \to 0^+} \sup_{y \in B^\circ_\Omega(x,r)} \frac{h(y) - h(x)}{|x-y|}
        = \lim_{r \to 0^+} \sup_{y \in B^\circ_\Omega(x,r)}
            \frac{h(y) - h(x) - \grad h(x)^\top (y-x) + \grad h(x)^\top (y-x)}{|x-y|} \,.
    \]
    By definition of derivative, we have
    \[
        \lim_{y \to x} \frac{\abs*{h(y) - h(x) - \grad h(x)^\top (y-x)}}{|x-y|} = 0 \,,
    \]
    and hence, by $\ell^2$ norm duality,
    \[
        \lim_{r \to 0^+} \sup_{y \in B^\circ_\Omega(x,r)} \frac{h(y) - h(x)}{|x-y|}
        = \lim_{r \to 0^+} \sup_{y \in B^\circ_\Omega(x,r)} \frac{\grad h(x)^\top (y-x) }{|x-y|}
        = \sup_{v \in \bR^d : |v|=1} \grad h(x)^\top v
        = |\grad h(x)| \,.
    \]
    Similarly, we also have
    \begin{align*}
        \lim_{r \to 0^+} \sup_{y \in B^\circ_\Omega(x,r)} \frac{\abs*{h(y) - h(x)}}{|x-y|}
        &= \lim_{r \to 0^+} \sup_{y \in B^\circ_\Omega(x,r)}
            \frac{\abs*{h(y) - h(x) - \grad h(x)^\top (y-x) + \grad h(x)^\top (y-x)}}{|x-y|} \\
        &\le \lim_{r \to 0^+} \sup_{y \in B^\circ_\Omega(x,r)}
            \frac{\abs*{h(y) - h(x) - \grad h(x)^\top (y-x)}
                + \abs*{\grad h(x)^\top (y-x)}}{|x-y|} \\
        &= \lim_{r \to 0^+} \sup_{y \in B^\circ_\Omega(x,r)}
            \frac{\abs*{\grad h(x)^\top (y-x)}}{|x-y|}
        = \sup_{v \in \bR^d : |v|=1} \abs*{\grad h(x)^\top v}
        = |\grad h(x)| \,.
    \end{align*}
    Since $h(y) - h(x) \le (h(y) - h(x))^+ \le \abs*{h(y) - h(x)}$, the first part of the statement
    follows.

    We now establish the second part. Write the partition $[d] = P \cup N$ for
    \begin{align*}
        P &\define \{ i \in [d] : \partial_i h(x) > 0 \} \,, \\
        N &\define \{ i \in [d] : \partial_i h(x) \le 0 \} \,.
    \end{align*}
    Assume for a moment that $P \ne \emptyset$\footnote{We track the edge case $P = \emptyset$ as we
    go along the proof; the main idea is the same.}\!\!. Define $\Omega^P \subset \bR^P$ by
    \[
        \Omega^P \define \{ y_P : y \in \Omega \text{ satisfies } \supp(y-x) \subseteq P \} \,,
    \]
    and let $h^P \in C^1(\overline {\Omega^P})$ be given by $h^P(y_P) \define h(y_P, x_{-P})$ for
    each $y_P \in \Omega_P$, \ie $h^P$ is the restriction of $h$ to the space obtained by fixing the
    $N$-coordinates of inputs to those of $x$. Note that $\partial_i h(x) = \partial_i h^P(x_P)$ for
    each $i \in P$, and hence $|\grad^+ h(x)| = |\grad h^P(x_P)|$. Applying the first part of the
    statement to $\Omega^P$, $h^P$ and $x_P$ gives
    \begin{equation}
        \label{eq:gradplus-1}
        |\grad^+ h(x)|
        = |\grad h^P(x_P)|
        = \lim_{r \to 0^+} \sup_{y_P \in B^\circ_{\Omega^P}(x_P, r)}
            \frac{(h^P(y_P) - h^P(x_P))^+}{|x_P - y_P|}
    \end{equation}
    We claim that
    \begin{equation}
        \label{eq:gradplus-2}
        \lim_{r \to 0^+} \sup_{y_P \in B^\circ_{\Omega^P}(x_P, r)}
            \frac{(h^P(y_P) - h^P(x_P))^+}{|x_P - y_P|}
        \eqquestion
        \lim_{r \to 0^+} \sup_{y \in B^+_\Omega(x, r)} \frac{(h(y) - h(x))^+}{|x-y|} \,,
    \end{equation}
    with the LHS replaced by $0$ if $P = \emptyset$, which will conclude the proof.

    For each set $A \subseteq [d]$, let $X_A \define \{y \in \Omega : y_A \succeq x_A\}$ and $Y_A
    \define \{y \in \Omega : y_A \preceq x_A\}$. Note that for all $r > 0$, $B^+_\Omega(x,r) =
    B^\circ_\Omega(x,r) \cap X_P \cap X_N$ and moreover, defining $B'_\Omega(x,r) \define
    B^\circ_\Omega(x,r) \cap X_N \cap Y_N$, when $P \ne \emptyset$ we have
    \begin{equation}
        \label{eq:gradplus-3}
        \lim_{r \to 0^+} \sup_{y_P \in B^\circ_{\Omega^P}(x_P, r)}
            \frac{(h^P(y_P) - h^P(x_P))^+}{|x_P - y_P|}
        = \lim_{r \to 0^+} \sup_{y \in B'_\Omega(x, r)} \frac{(h(y) - h(x))^+}{|x-y|}
    \end{equation}
    by definition of the objects involved.

    Note that, when $P \ne \emptyset$, $\min_{i \in P} \partial_i h(x) > 0$ by the definition of
    $P$. Since $h$ is continuously differentiable, it follows that there exists $r_0 > 0$ such that,
    for all $y \in B^\circ_\Omega(x,r_0)$ and every $i \in P$, we have $\partial_i h(y) > 0$. Since
    $\Omega$ is open, we can also let $r_0$ be small enough so that $B(x,r_0) \subset \Omega$.

    Now, suppose $y \in B^\circ_\Omega(x,r_0)$ and $i \in P$ are such that $y_i < x_i$. Then define
    $z \in \bR^d$ by $z_i \define y_i + 2(x_i - y_i) > x_i$, and $z_j \define y_j$ for $j \in [d]
    \setminus \{i\}$. Then $|x-z| = |x-y|$, and in particular $z \in B^\circ(x,r_0) \subset \Omega$.
    Also, since $i \in P$ and $z$ agrees with $y$ on all $j \ne i$, we have that $y \in X_N$ (resp.
    $Y_N$) if and only if $z \in X_N$ (resp. $Y_N$). Moreover, since $\partial_i h(w) > 0$ for all
    $w$ in the line segment connecting $y$ and $z$ (which is contained in $B^\circ_\Omega(x,r_0)
    \cup \{x\}$), the fundamental theorem of calculus implies that $h(z) > h(y)$.

    Repeating this argument inductively for each index $i \in P$ for which $y_i < x_i$, we conclude
    that for all $r \in (0, r_0)$ and $y \in B'_\Omega(x,r)$, there exists $z \in B'_\Omega(x,r)
    \cap X_P$ such that $|x-z| = |x-y|$ and $h(z) > h(y)$. It follows that, when $P \ne \emptyset$,
    \begin{equation}
        \label{eq:gradplus-4}
        \lim_{r \to 0^+} \sup_{y \in B'_\Omega(x, r)} \frac{(h(y) - h(x))^+}{|x-y|}
        = \lim_{r \to 0^+} \sup_{y \in B'_\Omega(x, r) \cap X_P} \frac{(h(y) - h(x))^+}{|x-y|} \,.
    \end{equation}
    Combining \eqref{eq:gradplus-2}, \eqref{eq:gradplus-3} and \eqref{eq:gradplus-4}, we see that it
    remains to show that
    \begin{equation}
        \label{eq:gradplus-5}
        \lim_{r \to 0^+} \sup_{y \in B'_\Omega(x, r) \cap X_P} \frac{(h(y) - h(x))^+}{|x-y|}
        \eqquestion
        \lim_{r \to 0^+} \sup_{y \in B^+_\Omega(x, r)} \frac{(h(y) - h(x))^+}{|x-y|} \,,
    \end{equation}
    again with the LHS replaced by $0$ if $P = \emptyset$. Note that for each $r > 0$ it holds by
    definition that $B'_\Omega(x, r) \cap X_P \subseteq B^+_\Omega(x, r)$, namely $B'_\Omega(x, r)
    \cap X_P = B^+_\Omega(x, r) \cap Y_N$. Thus we already have
    \[
        \lim_{r \to 0^+} \sup_{y \in B'_\Omega(x, r) \cap X_P} \frac{(h(y) - h(x))^+}{|x-y|}
        \le \lim_{r \to 0^+} \sup_{y \in B^+_\Omega(x, r)} \frac{(h(y) - h(x))^+}{|x-y|} \,,
    \]
    which also holds when the LHS is replaced by $0$, and it remains to prove the reverse
    inequality. Let $\epsilon > 0$. As before, recalling that $\partial_i h(x) \le 0$ for every $i
    \in N$ and using the fact that $h$ is continuously differentiable, we can let $r_0 > 0$ be small
    enough so that for all $y \in B^\circ(x,r_0)$ and every $i \in N$, we have $\partial_i h(y) \le
    \epsilon$. Again since $\Omega$ is open, we may also let $r_0$ be small enough so that $B(x,r_0)
    \subset \Omega$.

    Now, let $r \in (0, r_0)$ and $y \in B_\Omega^+(x,r)$. Define $y^P \in \bR^d$ as follows: $y^P_i
    \define y_i$ for each $i \in P$, and $y^P_j \define x_j$ for each $j \in N$. Then $\supp(y^P-x)
    \subseteq P$ by construction, so $y^P \in X_N \cap Y_N$. We also have $y^P \succeq x$ and $|x -
    y^P| \le |x - y|$, and hence $y^P \in (B^+(x,r) \cap Y_N) \cup \{x\} = (B'_\Omega(x,r) \cap X_P)
    \cup \{x\}$.

    By a standard multivariate version of the mean value theorem, there exists a point $z$ in the
    line segment connecting $y^P$ and $y$ such that
    \[
        h(y) - h(y^P) = \grad h(z)^\top (y - y^P) \,.
    \]
    By construction, we have $\supp(y - y^P) \subseteq N$. Moreover, since $y \in B^+_\Omega(x,r)$
    and hence $y \succeq x$, while $y^P_j = x_j$ for each $j \in N$, we have that $y - y^P \succeq
    \vec 0$. Finally, it is clear from the construction of $y^P$ that $|y - y^P| \le |x - y|$. We
    conclude that
    \begin{align*}
        h(y) - h(y^P)
        &= \grad h(z)^\top (y - y^P) \\
        &\le \grad^+ h(z)^\top (y - y^P)
            & \text{(Since $y - y^P \succeq \vec 0$)} \\
        &= \sum_{i \in N} (\partial_i^+ h(z)) (y_i - y^P_i)
            & \text{(Since $\supp(y - y^P) \subseteq N$)} \\
        &\le \left( \sum_{i \in N} (\partial_i^+ h(z))^2 \right)^{1/2} |y - y^P|
            & \text{(Cauchy-Schwarz inequality)} \\
        &\le \left( \sum_{i \in N} \epsilon^2 \right)^{1/2} |x-y|
            & \text{(Choice of $r_0$ and last observation)} \\
        &\le \epsilon \sqrt{d} |x-y| \,.
    \end{align*}
    Therefore, when $P \ne \emptyset$ we have
    \begin{align*}
        \lim_{r \to 0^+} \sup_{y \in B^+_\Omega(x, r)} \frac{(h(y) - h(x))^+}{|x-y|}
        &\le \lim_{r \to 0^+} \sup_{y \in B^+_\Omega(x, r)}
            \frac{\left( h(y^P) - h(x) + \epsilon \sqrt{d} |x-y| \right)^+}{|x-y|} \\
        &\le \epsilon \sqrt{d}
            + \lim_{r \to 0^+} \sup_{y \in B'_\Omega(x, r) \cap X_P}
                \frac{(h(y) - h(x))^+}{|x-y|} \,,
    \end{align*}
    the second inequality since $|x-y^P| \le |x-y|$ and since we showed that, for all $r < r_0$ and
    $y \in B^+_\Omega(x,r)$, we have $y^P \in (B'_\Omega(x,r) \cap X_P) \cup \{x\}$, and moreover
    $(h(y^P) - h(x))^+ = 0$ when $y^P = x$. Similarly, when $P = \emptyset$ we have $y^P = x$
    always, so in this case we obtain the upper bound $\epsilon \sqrt{d}$, \ie the last limit
    superior may indeed be replaced by $0$. Since this inequality holds for every $\epsilon > 0$,
    \eqref{eq:gradplus-5} follows, and this concludes the proof.
\end{proof}

\subsection{Perturbation argument}
\label{section:perturbation}

We can now obtain the main result, first for $C^1(\closedInt^d)$ functions and then for all of
$H^1((0,1)^d)$ by an approximation argument. The key idea of the perturbation argument is to use the
same (mean zero) function $h$ in two different roles: 1)~to construct absolutely continuous
probability measures from $h$ and its coordinate-wise monotone equilibrium, namely $\odif \mu =
(1+th) \odif x$ and $\odif \mu^* = (1+th^*) \odif x$, for small $t > 0$; and 2)~in the ``test
function'' $-th$ for weak duality on $W_2^2(\mu \to \mu^*)$ via the directed Hamilton-Jacobi
operator.

\begin{theorem}
    \label{thm:c1-poincare}
    There exists a universal constant $C > 0$ such that the following holds. Define $\Omega \define
    (0,1)^d$, and let $h \in C^1(\overline \Omega)$. Then
    \[
        \dist^\mono_2(h)^2
        \le \int_\Omega (h - h^*)^2 \odif x
        \le C \int_\Omega |\grad^- h|^2 \odif x \,,
    \]
    where $h^*$ is the coordinate-wise monotone equilibrium of $h$.
\end{theorem}
\begin{proof}
    Note that $h^*$ is well-defined because $C^1$ functions are Lipschitz on compact sets. We first
    show that we may assume without loss of generality that $h$ is bounded and has mean zero. For
    any $\alpha > 0$ and $\beta \in \bR$, we have that $\alpha g + \beta$ is monotone if and only if
    $g$ is monotone, so it follows that
    \[
        \dist^\mono_2(\alpha h + \beta)^2 = \alpha^2 \dist^\mono_2(h)^2 \,.
    \]
    Moreover, \cref{obs:fstar-transformations} gives that $(\alpha h + \beta)^* = \alpha h^* +
    \beta$, and clearly $|\grad^- (\alpha h + \beta)|^2 = \alpha^2 |\grad^- h|^2$ pointwise. Thus as
    long as we show the result for $h$ satisfying $\int_\Omega h \odif x = 0$ and (say) $-0.1 \le h
    \le 0.1$, then we may write any other function in $C^1(\overline \Omega)$ as $\alpha h + \beta$
    for some $h$ satisfying these conditions, and conclude the result. Hence assume that
    $\int_\Omega h \odif x = 0$ and that $-0.1 \le h \le 0.1$ pointwise.

    Let $t \in (0, 1)$, and define $f \in C^1(\overline \Omega)$ by $f \define 1 + t h$. Note that
    $f$ is Lipschitz and bounded between $0.9$ and $1.1$, and it satisfies $\int_\Omega f \odif x =
    1$. Define the absolutely continuous probability measures $\odif \mu \define f \odif x$ and
    $\odif \mu^* \define f^* \odif x$. \cref{thm:transport-energy} then implies that, for some
    universal constant $C > 0$,
    \[
        W_2^2(\mu \to \mu^*)
        \le C \int_\Omega \left| \grad^- f \right|^2 \odif x
        = C t^2 \int_\Omega \left| \grad^- h \right|^2 \odif x
        = C t^2 \int_\Omega \left| \grad^+ (-h) \right|^2 \odif x \,,
    \]
    the first equality by the definition of $f$ and the fact that $t > 0$. On the other hand, noting
    that $-th \in C^1(\overline \Omega)$ and hence $-th \in C_b(\Omega)$ as well, we have
    \begin{align*}
        \tfrac{1}{2} W_2^2(\mu \to \mu^*)
        &\ge \int_\Omega (-th) \odif \mu^* - \int_\Omega (\dirH_1 (-th)) \odif \mu
            & \text{(\cref{prop:h-duality})} \\
        &= -t \int_\Omega h (1 + t h)^* \odif x
            - \int_\Omega (\dirH_1 (-th)) (1 + t h) \odif x
            & \text{(Definition of $\mu, \mu^*$).}
    \end{align*}
    Recalling that $\int_\Omega h \odif x = 0$ and that $(1 + t h)^* = 1 + t h^*$ by
    \cref{obs:fstar-transformations}, the first term in the last line above is
    \[
        -t \int_\Omega h (1 + t h)^* \odif x
        = -t \int_\Omega h (1 + t h^*) \odif x
        = -t^2 \int_\Omega h h^* \odif x \,.
    \]
    Recalling that $1 + t h$ is bounded between $0.9$ and $1.1$, the second term is
    \begin{align*}
        &-\int_\Omega (\dirH_1 (-th)) (1 + t h) \odif x \\
        &\qquad = -\int_\Omega (t \dirH_t (-h)) (1 + t h) \odif x
            & \text{(\cref{lemma:h-scaling})} \\
        &\qquad = -t^2 \int_\Omega \left( \frac{\dirH_t (-h) - (-h) - h}{t} \right)
            (1 + t h) \odif x \\
        &\qquad = - t^2 \int_\Omega \underbrace{
            \left( \frac{\dirH_t (-h) - (-h)}{t} \right)}_{
                \ge 0 \text{ by \cref{prop:properties-h}\ref*{item:h-c}}}
                \underbrace{(1 + t h)}_{\le 2} \odif x
            + t \int_\Omega h (1 + t h) \odif x \\
        &\qquad \ge - 2 t^2 \int_\Omega \frac{\dirH_t (-h) - (-h)}{t} \odif x
            + t \underbrace{\int_\Omega h \odif x}_{=0}
            + t^2 \int_\Omega h^2 \odif x \\
        &\qquad = - 2 t^2 \int_\Omega \frac{\dirH_t (-h) - (-h)}{t} \odif x
            + t^2 \int_\Omega h^2 \odif x \,.
    \end{align*}
    Putting all of the above together, we conclude that
    \[
        \tfrac{C}{2} t^2 \int_\Omega \left| \grad^+ (-h) \right|^2 \odif x
        \ge \tfrac{1}{2} W_2^2(\mu \to \mu^*)
        \ge - t^2 \int_\Omega h h^* \odif x
            - 2 t^2 \int_\Omega \frac{\dirH_t (-h) - (-h)}{t} \odif x
            + t^2 \int_\Omega h^2 \odif x
    \]
    and hence, since $t > 0$,
    \[
        \int_\Omega h^2 \odif x - \int_\Omega h h^* \odif x
        \le \tfrac{C}{2} \int_\Omega \left| \grad^+ (-h) \right|^2 \odif x
            + 2 \int_\Omega \frac{\dirH_t (-h) - (-h)}{t} \odif x \,.
    \]
    Since this holds for all sufficiently small $t > 0$, we may pass the inequality to the limit
    superior as $t \to 0^+$. Since $\frac{\dirH_t (-h) - (-h)}{t}$ is uniformly bounded by
    \cref{prop:properties-h}\ref*{item:h-c}, we may apply the reverse Fatou lemma and
    \cref{prop:properties-h}\ref*{item:h-b} to obtain
    \begin{align*}
        \int_\Omega h^2 \odif x - \int_\Omega h h^* \odif x
        &\le \tfrac{C}{2} \int_\Omega \left| \grad^+ (-h) \right|^2 \odif x
            + 2 \limsup_{t \to 0^+} \int_\Omega \frac{\dirH_t (-h) - (-h)}{t} \odif x \\
        &\le \tfrac{C}{2} \int_\Omega \left| \grad^+ (-h) \right|^2 \odif x
            + 2 \int_\Omega \limsup_{t \to 0^+} \frac{(\dirH_t (-h))(x) - (-h)(x)}{t} \odif x \\
        &\le \tfrac{C}{2} \int_\Omega \left| \grad^+ (-h) \right|^2 \odif x
            + \int_\Omega |\grad^+ (-h)(x)|^2 \odif x
        = \left(1 + \tfrac{C}{2}\right) \int_\Omega \left| \grad^- h \right|^2 \odif x
        \,.
    \end{align*}
    By \cref{obs:fstar-nonexpansive}, we have that $\int_\Omega (h^*)^2 \odif x \le \int_\Omega h^2
    \odif x$ and hence, since $h^*$ is monotone,
    \begin{align*}
        \tfrac{1}{2} \dist^\mono_2(h)^2
        &\le \tfrac{1}{2} \int_\Omega (h - h^*)^2 \odif x
        = \tfrac{1}{2} \int_\Omega h^2 \odif x
            + \tfrac{1}{2} \int_\Omega (h^*)^2 \odif x
            - \int_\Omega h h^* \odif x \\
        &\le \int_\Omega h^2 \odif x
            - \int_\Omega h h^* \odif x
        \le \left(1 + \tfrac{C}{2}\right) \int_\Omega \left| \grad^- h \right|^2 \odif x
        \,. \qedhere
    \end{align*}
\end{proof}

The following denseness result is an immediate application of \cite[Theorem~1, p.~10]{Maz11}.

\begin{fact}
    \label{fact:smooth-functions-dense}
    Define $\Omega \define (0,1)^d$. The space $C^\infty(\overline \Omega)$ is dense in
    $H^1(\Omega)$.
\end{fact}

Thus we may promote the main result to all of $H^1(\Omega)$ by an approximation argument:

\thmdirectedpoincareinequality*
\begin{proof}
    By \cref{fact:smooth-functions-dense}, we may find a sequence $(f_n)_{n \in \bN} \subset
    C^\infty(\overline \Omega)$ such that $f_n \to f$ in $H^1(\Omega)$. Since each $f_n$ belongs in
    particular to $C^1(\overline \Omega)$, \cref{thm:c1-poincare} gives that, for each $n \in \bN$,
    \[
        \|f_n - f_n^*\|_{L^2(\Omega)}^2 \le C \int_\Omega |\grad^- f_n|^2 \odif x \,.
    \]
    Let $\epsilon > 0$. Since $f_n \to f$ in $H^1(\Omega)$ implies in particular that $f_n \to f$ in
    $L^2(\Omega)$, we have that, for all sufficiently large $n$, $\|f - f_n\|_{L^2(\Omega)} \le
    \epsilon$ and hence, by the triangle inequality,
    \begin{equation}
        \label{eq:dist-mono-1}
        \dist^\mono_2(f)
        \le \|f - f_n^*\|_{L^2(\Omega)}
        \le \|f - f_n\|_{L^2(\Omega)} + \|f_n - f_n^*\|_{L^2(\Omega)}
        \le \epsilon + \sqrt{C \int_\Omega |\grad^- f_n|^2 \odif x} \,.
    \end{equation}
    Now, since the norm in $H^1(\Omega)$ is given by
    \[
        \|u\|_{H^1(\Omega)}^2
        = \|u\|_{L^2(\Omega)}^2 + \sum_{i=1}^d \int_\Omega |\partial_i u|^2 \odif x \,,
    \]
    the fact that $f_n \to f$ in $H^1(\Omega)$ implies that $\|f - f_n\|_{H^1(\Omega)} \to 0$ and
    hence
    \[
        0
        = \lim_{n \to \infty} \sum_{i=1}^d \int_\Omega |\partial_i (f - f_n)|^2 \odif x
        = \lim_{n \to \infty}
            \sum_{i=1}^d \int_\Omega |(\partial_i f) - (\partial_i f_n)|^2 \odif x \,.
    \]
    Since $0 \le |(a \land 0) - (b \land 0)| \le |a - b|$ for any $a, b \in \bR$, we conclude that
    \begin{equation}
        \label{eq:convergence-negative-partials}
        \lim_{n \to \infty}
            \sum_{i=1}^d \int_\Omega |(\partial^-_i f) - (\partial^-_i f_n)|^2 \odif x
        = 0 \,.
    \end{equation}
    Observing that $\partial^-_i f, \partial^-_i f_n \in L^2(\Omega)$ for each $i \in [d]$ and $n
    \in \bN$ (because $\partial_i f, \partial_i f_n \in L^2(\Omega)$ by the definition of
    $H^1(\Omega)$, and if $u \in L^2(\Omega)$ then $u \land 0 \in L^2(\Omega)$),
    \eqref{eq:convergence-negative-partials} implies that $\partial^-_i f_n \to \partial^-_i f$ in
    $L^2(\Omega)$ for each $i \in [d]$, and hence
    \[
        \int_\Omega |\grad^- f|^2 \odif x
        = \sum_{i=1}^d \int_\Omega |\partial^-_i f|^2 \odif x
        = \lim_{n \to \infty} \sum_{i=1}^d \int_\Omega |\partial^-_i f_n|^2 \odif x
        = \lim_{n \to \infty} \int_\Omega |\grad^- f_n|^2 \odif x \,.
    \]
    Thus, for all sufficiently large $n$ we have
    \[
        \int_\Omega |\grad^- f_n|^2 \odif x
        \le \epsilon + \int_\Omega |\grad^- f|^2 \odif x \,.
    \]
    Combining with \eqref{eq:dist-mono-1} we obtain, for all sufficiently large $n$,
    \begin{align*}
        \dist^\mono_2(f)^2
        &\le \left( \epsilon + \sqrt{C \int_\Omega |\grad^- f_n|^2 \odif x} \right)^2
        \le 2\epsilon^2 + 2C \int_\Omega |\grad^- f_n|^2 \odif x \\
        &\le 2\epsilon^2 + 2C \left( \epsilon + \int_\Omega |\grad^- f|^2 \odif x \right)
        = 2\epsilon^2 + 2C \epsilon + 2C \int_\Omega |\grad^- f|^2 \odif x \,.
    \end{align*}
    Since this holds for all $\epsilon > 0$, we conclude that
    \[
        \dist^\mono_2(f)^2 \le 2C \int_\Omega |\grad^- f|^2 \odif x \,. \qedhere
    \]
\end{proof}

\iftoggle{anonymous}{}{
\section*{Acknowledgments}

We thank Eric Blais for extensive discussions throughout the development of this project and for
helpful feedback on preliminary versions of this article; Nathan Harms for offering feedback and
identifying typos in the introduction of the paper; Lap Chi Lau for extensive feedback on the
presentation of this article and for pointing out the connection to the nonlinear Laplacian of
\cite{Yos16}; and C.~Seshadhri for suggesting the connection between
\cref{lemma:good-vector-distribution} and combinatorial search problems.
}

\printbibliography

\appendix

\section{Technical lemmas}
\label{section:technical-lemmas}

\begin{lemma}
    \label{lemma:weak-limit-monotonic}
    Define $\Omega \define (0,1)^d$. Suppose $f_n \weakto f$ weakly in $L^2(\Omega)$, and each $f_n$
    is monotone nondecreasing (resp.\ monotone nonincreasing). Then $f$ is monotone nondecreasing
    (resp.\ monotone nonincreasing).
\end{lemma}
\begin{proof}
    Recall that for each $\epsilon > 0$ and $g \in L^2(\Omega)$, $g^\epsilon = \eta_\epsilon * g$ is
    the mollification of $g$, defined on $\Omega_\epsilon = (\epsilon, 1-\epsilon)^d$.

    Let $\epsilon \in (0, 1/3)$. We claim that $f_n^\epsilon \to f^\epsilon$ pointwise in
    $\Omega_\epsilon$. Indeed, fix any $x \in \Omega_\epsilon$ and define $\eta_{\epsilon,x} \in
    C^\infty(\bR^d)$ by $\eta_{\epsilon,x}(y) \define \eta_\epsilon(x-y)$. Since $f_n \weakto f$
    weakly in $L^2(\Omega)$, we obtain
    \[
        f^\epsilon(x) = \int_\Omega \eta_\epsilon(x-y) f(y) \odif y
        = \inp{\eta_{\epsilon,x}}{f}
        = \lim_{n \to \infty} \inp{\eta_{\epsilon,x}}{f_n}
        = \lim_{n \to \infty} \int_\Omega \eta_\epsilon(x-y) f_n(y) \odif y
        = \lim_{n \to \infty} f_n^\epsilon(x) \,.
    \]

    Now, suppose without loss of generality that each $f_n$ is monotone nondecreasing. Let $\epsilon
    \in (0, 1/3)$. It is immediate that each $f_n^\epsilon$ is also monotone nondecreasing (in
    $\Omega_\epsilon$). Since $f_n^\epsilon \to f^\epsilon$ pointwise in $\Omega_\epsilon$, we
    conclude that $f^\epsilon$ is monotone nondecreasing as well.

    The conclusion follows the fact that $f^\epsilon \to f$ almost everywhere as $\epsilon \to 0$.
\end{proof}

The following lemma is essentially standard, and below we present the proof sketched in
\cite{Fen}.

\begin{lemma}
    \label{lemma:bounded-loc-to-weak}
    Let $\Omega \subset \bR^d$ be an open set, let $u \in L^2(\Omega)$ and let $(u_n)_{n \in \bN}$
    be a bounded sequence in $L^2(\Omega)$ such that $u_n \to u$ in $L^2_\loc(\Omega)$. Then $u_n
    \weakto u$ weakly in $L^2(\Omega)$.
\end{lemma}
\begin{proof}
    It suffices to show that every subsequence $(u_{n_k})_{k \in \bN}$ has a subsequence that weakly
    converges to $u$. Fix any subsequence $(u_{n_k})_{k \in \bN}$. First, since $u_{n_k} \to u$ in
    $L^2_\loc(\Omega)$, we have that $\inp{u_{n_k}}{\phi} \to \inp{u}{\phi}$ for all $\phi \in
    C^\infty_c(\Omega)$. Moreover, since $L^2(\Omega)$ is a Hilbert space and $(u_{n_k})_{k \in
    \bN}$ is bounded, let $(u_{n_{k_\ell}})_{\ell \in \bN}$ be a subsequence such that
    $u_{n_{k_\ell}} \weakto w$ weakly in $L^2(\Omega)$ for some $w \in L^2(\Omega)$. Now, let $v \in
    L^2(\Omega)$. Since $C^\infty_c(\Omega)$ is dense in $L^2(\Omega)$, let $(\phi_m)_{m \in \bN}$
    be a sequence in $C^\infty_c(\Omega)$ with $\phi_m \to v$ in $L^2(\Omega)$. Then
    \begin{align*}
        \inp{u}{v}
        &= \lim_{m \to \infty} \inp{u}{\phi_m}
            & \text{(Since $\phi_m \to v$ in $L^2(\Omega)$)} \\
        &= \lim_{m \to \infty} \lim_{k \to \infty} \inp{u_{n_k}}{\phi_m}
            & \text{(As observed above)} \\
        &= \lim_{m \to \infty} \lim_{\ell \to \infty} \inp{u_{n_{k_\ell}}}{\phi_m}
            & \text{(Taking subsequence preserves the limit)} \\
        &= \lim_{m \to \infty} \inp{w}{\phi_m}
            & \text{(Since $u_{n_{k_\ell}} \weakto w$ weakly in $L^2(\Omega)$)} \\
        &= \inp{w}{v} & \text{(Since $\phi_m \to v$ in $L^2(\Omega)$)} \\
        &= \lim_{\ell \to \infty} \inp{u_{n_{k_\ell}}}{v}
            & \text{(Since $u_{n_{k_\ell}} \weakto w$ weakly in $L^2(\Omega)$),}
    \end{align*}
    so $u_{n_{k_\ell}} \weakto u$ weakly in $L^2(\Omega)$ as needed.
\end{proof}

\begin{lemma}[From ``almost Lipschitz'' to Lipschitz]
    \label{lemma:almost-lipschitz-extension}
    Let $M \in \bR_{\ge 0}$ and let $N \subset I$ be a measure zero set. Suppose $f : I \to \bR$
    satisfies $\abs*{f(x) - f(y)} \le M |x-y|$ for all $x, y \in I \setminus N$. Then there exists a
    $M$-Lipschitz function $g : I \to \bR$ such that $f = g$ in $I \setminus N$.
\end{lemma}
\begin{proof}
    Let $\overline f  : I \to (-\infty, +\infty]$ be given by $\overline f  \define f$ in $I
    \setminus N$ and $\overline f  \define +\infty$ in $N$. Define $g : I \to \bR$ by
    \[
        g(x) \define \begin{cases}
            f(x) & \text{if } x \in I \setminus N \\
            \liminf_{z \to x} \overline f (z) & \text{otherwise.}
        \end{cases}
    \]
    Note that $g$ is real-valued because, when $x \in N$, we may find a sequence $(z_n)_{n \in \bN}
    \subset I \setminus N$ such that $z_n \to x$ and then, by the definition of $\liminf$ by
    subsequential limits and the assumption on $f$,
    \[
        g(x)
        = \liminf_{z \to x} \overline f(z)
        \le \lim_{n \to \infty} \overline f(z_n)
        = \lim_{n \to \infty} f(z_n)
        \le \lim_{n \to \infty} \big[ f(z_1) + M |z_n - z_1| \big]
        = f(z_1) + M |x - z_1|
        < +\infty \,.
    \]
    Clearly $g = f$ in $I \setminus N$. We claim that $g$ is $M$-Lipschitz. By the assumption on
    $f$, we do have $\abs*{g(x) - g(y)} \le M |x-y|$ for all $x, y \in I \setminus N$. Now, let $x
    \in N$ and $y \in I \setminus N$. By the definition of $g$, fix a sequence $(z_n)_{n \in \bN}
    \subset I \setminus N$ such that $z_n \to x$ and moreover $\lim_{n \to \infty} f(z_n) = g(x)$.
    Then
    \[
        \abs*{g(x) - g(y)}
        = \abs*{\lim_{n \to \infty} f(z_n) - f(y)}
        = \lim_{n \to \infty} \abs*{f(z_n) - f(y)}
        \le \lim_{n \to \infty} M |z_n - y|
        = M |x - y| \,,
    \]
    as desired. The case when $x, y \in N$ follows by the triangle inequality: suppose $x < y$
    without loss of generality, choose any $z \in (x, y) \setminus N$, which is possible since $N$
    has measure zero, and use the above to conclude that
    \[
        \abs*{g(x) - g(y)}
        \le \abs*{g(x) - g(z)} + \abs*{g(z) - g(y)}
        \le M|x-z| + M|z-y|
        = M|x-y| \,. \qedhere
    \]
\end{proof}

\ignore{
\begin{proposition}
    \label{res:dist-achieved}
    The infimum in the definition of $\dist^\mono_2 : L^2((0,1)^d) \to [0, +\infty)$ is achieved,
    \ie for each $f \in L^2((0,1)^d)$ there exists a monotone $g \in L^2((0,1)^d)$ such that
    $\|f-g\|_{L^2((0,1)^d)} = \dist^\mono_2(f)$.
\end{proposition}
\begin{proof}
    Let $f \in L^2((0,1)^d)$. Let $E \define L^2((0,1)^d)$ and let $C \subset E$ be the set of
    monotone functions in $E$. Let $\delta : C \to [0, +\infty)$ be given by $\delta(g) \define
    \|f-g\|_{L^2((0,1)^d)}$. We verify that $E$, $C$ and $\delta$ meet the conditions of
    \cref{lemma:convex-optimization}. First, $E$ is a Hilbert space and hence a reflexive Banach
    space. It is clear that $C$ is nonempty and convex, and it is also closed by
    \cref{lemma:weak-limit-monotonic}. That $\delta$ is convex and lower semicontinuous (in fact,
    continuous) is immediate from the properties of the $L^2$ norm, and it is finitely-valued and
    thus proper. Finally, $\delta$ is coercive because if $(g_n)_{n \in \bN} \subset C$ is such that
    $\|g_n\|_{L^2((0,1)^d)} \to +\infty$, then $\delta(g_n) = \|f-g_n\|_{L^2((0,1)^d)} \to +\infty$.
    \cref{lemma:convex-optimization} gives the conclusion.
\end{proof}
}

\end{document}